\documentclass[10pt,reqno,hidelinks]{amsart}
\usepackage{amssymb}
\usepackage{amsmath, amssymb, verbatim, url, mathrsfs}
\usepackage{mathrsfs}
\usepackage{verbatim}
\usepackage{amsthm}
\usepackage{wasysym}
\usepackage{upgreek}
\usepackage{color}
\usepackage{accents}
\usepackage{dsfont}
\usepackage{hyperref}
\usepackage{enumerate}
\usepackage{longtable}

\numberwithin{equation}{section}

\newtheorem{proposition}{Proposition}[section]
\newtheorem{lemma}[proposition]{Lemma}
\newtheorem{corollary}[proposition]{Corollary}
\newtheorem{theorem}[proposition]{Theorem}

\newtheorem{ass}[proposition]{Assumptions}

\theoremstyle{definition}
\newtheorem{definition}[proposition]{Definition}
\newtheorem{remark}[proposition]{Remark}

\newcommand{\vertiiii}[1]{{\left\vert\kern-0.25ex\left\vert\kern-0.25ex\left\vert\kern-0.25ex\left\vert #1 \right\vert\kern-0.25ex\right\vert\kern-0.25ex\right\vert\kern-0.25ex\right\vert}}
\newcommand{\vertiii}[1]{{\left\vert\kern-0.25ex\left\vert\kern-0.25ex\left\vert #1 \right\vert\kern-0.25ex\right\vert\kern-0.25ex\right\vert}}
\newcommand{\norm}[1]{\|#1\|}

\newcommand{\mfu}{\mathfrak{u}}
\newcommand{\smfu}{\breve{\mathfrak{u}}}

\newcommand{\mrw}{\mathring{w}}
\newcommand{\mrH}{\mathring{H}}
\newcommand{\mrA}{\mathring{A}}

\newcommand{\ih}{\sqrt{\frac{3}{\Lambda}}}

\newcommand{\hmfu}{\hat{\mathfrak{u}}}
\newcommand{\emfu}{E^3\hat{\mathfrak{u}}}
\newcommand{\Rbb}{\mathbb{R}}
\newcommand{\Zbb}{\mathbb{Z}}
\newcommand{\Tbb}{\mathbb{T}}
\newcommand{\Pbb}{\mathbb{P}}

\newcommand{\nnb}{\nonumber}
\newcommand{\del}[1]{{\partial_{#1}}}

\newcommand{\AND}{{\quad\text{and}\quad}}

\newcommand{\Hs}{H^{s-1}}

\newcommand{\Hsss}{H^{s-2}}
\newcommand{\Li}{L^\infty}
\newcommand{\la}{\langle}
\newcommand{\ra}{\rangle}

\newcommand{\al}[2]{
\begin{align}\label{E:#1}
  #2
\end{align}
}
\newcommand{\ga}[2]{
	\begin{gather}\label{E:#1}
	#2
	\end{gather}
}
\newcommand{\ali}[1]{
\begin{align}
  #1
\end{align}
}
\newcommand{\gat}[1]{
	\begin{gather}
	#1
	\end{gather}
}
\newcommand{\als}[1]{
\begin{align*}
  #1
\end{align*}
}
\newcommand{\gas}[1]{
	\begin{gather*}
	#1
	\end{gather*}
}
\newcommand{\p}[1]{
\begin{pmatrix}
  #1
\end{pmatrix}
}

\allowdisplaybreaks

\DeclareMathOperator{\diag}{diag}

\begin{document}

\title{Newtonian limits of isolated cosmological systems on long time scales}

\address{School of Mathematical Sciences, 9 Rainforest Walk, Monash University, Clayton, VIC 3800, Australia}
\email{chao.liu.math@foxmail.com}
\email{todd.oliynyk@monash.edu}
\author{Chao Liu \and Todd A. Oliynyk}

\begin{abstract}
We establish the existence of $1$-parameter families of $\epsilon$-dependent solutions to the Einstein--Euler equations
with a positive cosmological constant $\Lambda >0$ and a linear equation of state $p=\epsilon^2 K \rho$, $0<K\leq 1/3$, for the parameter
values $0<\epsilon < \epsilon_0$. These solutions
exist globally to the future, converge as $\epsilon \searrow 0$ to solutions of the cosmological Poisson--Euler equations of Newtonian gravity,
and are inhomogeneous nonlinear perturbations of FLRW fluid solutions.
\end{abstract}

\maketitle

\section{Introduction} \label{S:INTRO}
Gravitating relativistic perfect fluids are governed by the Einstein--Euler equations. The dimensionless version of these equations
with a cosmological constant $\Lambda$ is given by
\begin{align}
    \tilde{G}^{\mu\nu}+\Lambda\tilde{g}^{\mu\nu}&=\tilde{T}^{\mu\nu}, \label{E:ORIGINALEESYSTEM.a}\\
    \tilde{\nabla}_\mu\tilde{T}^{\mu\nu}&=0, \label{E:ORIGINALEESYSTEM.b}
\end{align}
where $\tilde{G}^{\mu\nu}$ is the Einstein tensor of the metric
\begin{align*}
  \tilde{g}=\tilde{g}_{\mu\nu}d\bar{x}^\mu d\bar{x}^\nu,
\end{align*}
and
\begin{align*}
  \tilde{T}^{\mu\nu}=(\bar{\rho}+ \bar{p})\tilde{v}^\mu \tilde{v}^\nu +\bar{p}\tilde{g}^{\mu\nu}
\end{align*}
is the perfect fluid stress--energy tensor. Here, $\bar{\rho}$ and $\bar{p}$ denote
the fluid's proper energy density and pressure, respectively, while $\tilde{v}^\nu$ is the fluid four-velocity,
which we assume is normalized by
\begin{equation} \label{vnorm}
  \tilde{v}^\mu \tilde{v}_\mu=-1.
\end{equation}

In this article, we assume a positive cosmological constant $\Lambda > 0$ and
restrict our attention to barotropic fluids  with a linear equation of state of the form
\begin{align*}
  \bar{p}=\epsilon^2 K\bar{\rho}, \qquad 0 < K \leq \frac{1}{3}.
\end{align*}
The dimensionless parameter $\epsilon$ can be identified with the ratio
\begin{equation*}
\epsilon = \frac{v_T}{c},
\end{equation*}
where $c$ is the speed of light and  $v_T$ is a characteristic speed associated to the fluid.
Understanding the behavior of solutions to \eqref{E:ORIGINALEESYSTEM.a}--\eqref{E:ORIGINALEESYSTEM.b}
in the limit $\epsilon \searrow 0$ is known as the \textit{Newtonian limit}, which has been the subject of many investigations. Most work in
this subject has been carried out in the setting of isolated systems and has almost exclusively involved formal calculations, see
\cite{Blanchet:2014,Blanchet_et_al:2005,Chandrasekhar:1965,Dautcourt:1964,Ehlers:1986,Einstein_et_al:1938,FutamaseItoh:2007,Kunzle:1972,Kunzle:1976,KunzleDuval:1986} and
references therein, with a few exceptions \cite{oli1,Oliynyk:CMP_2009,ren} where rigorous
results were established. Due to questions surrounding the physical interpretation of large-scale cosmological simulations  using Newtonian gravity and
the role of Newtonian gravity in cosmological averaging, interest in the Newtonian limit and the related post-Newtonian expansions has shifted from
the isolated systems setting to the cosmological one. Here too, the majority of results are based on formal
calculations \cite{BuchertRasanen:2012,Clarkson_etal:2011,Ellis:2011,Green_Wald:2011,gre,Hwangetal:2008,HwangNoh:2013, HwangNoh:2014,
KopeikinPetrov:2013,KopeikinPetrov:2014,MatarreseTerranova:1996,Milillo_et_al:2015,NohHwang:2012,Rasanen:2010,Yamamoto_et_al:2016} with
the articles \cite{oli3,oli4,Oliynyk:PRD_2014,oli6} being the only exceptions where rigorous results have been obtained.

From a cosmological perspective, the most important family of solutions to the system \eqref{E:ORIGINALEESYSTEM.a}--\eqref{E:ORIGINALEESYSTEM.b}
are the Friedmann-Lema\^itre-Robertson-Walker (FLRW) solutions that represent a homogenous, fluid filled universe undergoing accelerated
expansion.  Letting $(\bar{x}^i)$, $i=1,2,3$, denote the standard periodic coordinates on the 3-torus\footnote{Here, $\mathbb{T}^n_\epsilon = [0,\epsilon]^n/\sim$ where
$\sim$ is equivalence relation that follows from the identification of the sides of the box $[0,\epsilon]^n$. When $\epsilon=1$, we will simply write
$\mathbb{T}^n$.} $\mathbb{T}^3_\epsilon$ and
$t=\bar{x}^0$ a time coordinate on the interval $(0,1]$, the
FLRW solutions on the manifold
\begin{equation*}
M_{\epsilon} = (0,1]\times \mathbb{T}^3_\epsilon
\end{equation*}
are defined by
\begin{align}
\tilde{h}(t) &= -\frac{3}{\Lambda t^2} dtdt + a(t)^2 \delta_{ij}d\bar{x}^i d\bar{x}^j, \label{FLRW.a}\\
\tilde{v}_{H}(t) &= -t\sqrt{\frac{\Lambda}{3}}\partial_{t}, \label{FLRW.b}\\
\rho_H(t) &= \frac{\rho_H(1)}{a(t)^{3(1+\epsilon^2 K)}} \label{FLRW.c},
\end{align}
where $\rho_H(1)$ is the initial density (freely specifiable) and $a(t)$ satisfies
\begin{equation}\label{E:TPTA}
-t a'(t) =a(t)\ih \sqrt{\frac{\Lambda}{3}+\frac{\rho_H(t)}{3}},\qquad a(1)=1.
\end{equation}

Throughout this article, we will refer to the global coordinates $(\bar{x}^\mu)$ on $M_\epsilon$ as \textit{relativistic coordinates}. In order to discuss the Newtonian
limit and the sense in which solutions converge as $\epsilon \searrow 0$,
we need to introduce the spatially rescaled coordinates
$(x^\mu)$ defined by
\begin{align}
  t=\bar{x}^0=x^0 \quad \text{and} \quad \bar{x}^i=\epsilon x^i, \qquad \epsilon > 0, \label{Ncoordinates}
\end{align}
which we refer to as \textit{Newtonian coordinates}.  We note that the Newtonian coordinates define a global coordinate system on the
$\epsilon$-independent manifold
\begin{equation*}
M := M_1 = (0,1]\times \mathbb{T}^3.
\end{equation*}

\begin{remark} \label{torient}
Due to our choice of time coordinate $t$ on $(0,1]$, the future lies in the direction of \textit{decreasing} $t$ and
timelike infinity is located at $t=0$.
\end{remark}

\begin{remark}
The nonstandard form of the FLRW solution and the  $\epsilon$-dependence in the manifold $M_\epsilon$ is a consequence of our starting point for the Newtonian limit, which
differs from the standard formulation in that the time interval has been compactified from $[0,\infty)$ to $(0,1]$ and the light cones of the metric \eqref{FLRW.a} do not flatten as $\epsilon \searrow 0$.
For comparison, we observe that
the standard formulation can be obtained by first switching to Newtonian coordinates, which removes the $\epsilon$-dependence from
the spacetime manifold, followed by the introduction of a new time coordinate
according to
\begin{align}\label{E:TIMETRANSFORMATION}
  t=e^{-\sqrt{\frac{\Lambda}{3}}\tau},
\end{align}
which undoes the compactification of the time interval. These new coordinates define a global coordinate system on
the $\epsilon$-independent manifold $[0,\infty)\times \Tbb^3$
on which the FLRW metric can be expressed as
\begin{equation*}
\hat{h}= - d\tau d\tau +  \epsilon^2 \hat{a}(\tau)\delta_{ij} dx^i dx^j
\end{equation*}
where $\hat{a}(\tau)=a(e^{-\sqrt{\frac{\Lambda}{3}}\tau})$.
%We note that standard presentation
%of the FRLW metric is obtained from this metric by setting $\epsilon=1$. Furthermore,
Dividing through by $\epsilon^2$ yields the metric
\begin{equation*}
\hat{h}_\epsilon = -\frac{1}{\epsilon^2} d\tau d\tau + \hat{a}(\tau)\delta_{ij} dx^i dx^j,
\end{equation*}
which is now in the standard form for taking the Newtonian limit. In particular,
we observe that the light cones of this metric flatten out as $\epsilon \searrow 0$.
\end{remark}

\begin{remark} Throughout this article, we take the homogeneous initial density $\rho_H(1)$ to be independent of
$\epsilon$. All of the results established in this article remain true if $\rho_H(1)$ is allowed to depend on $\epsilon$ in a $C^1$ manner,
that is the map $[0,\epsilon_0)\ni \epsilon \longmapsto \rho_H^\epsilon(1) \in \Rbb_{>0}$ is $C^1$ for
some $\epsilon_0 > 0$.
\end{remark}

\begin{remark} \label{FLRWlimrem}
As we show in \S\ref{FLRWanal},  FLRW solutions $\{a,\rho_H\}$ depend regularly on $\epsilon$ and have well-defined Newtonian limits.
Letting
\begin{equation} \label{arhoringdef}
\mathring{a} = \lim_{\epsilon\searrow 0} a \quad \text{and} \quad \mathring{\rho}_H = \lim_{\epsilon\searrow 0} \rho_H
\end{equation}
denote the Newtonian limit of $a$ and $\rho_H$, respectively, it then follows from \eqref{FLRW.c} and
\eqref{E:TPTA} that $\{\mathring{a},\mathring{\rho}_H\}$ satisfy
\begin{equation*}\label{rhoringdef}
\mathring{\rho}_H =  \frac{\mathring{\rho}_H(1)}{\mathring{a}(t)^{3}}
\end{equation*}
and
\begin{equation*} \label{aringdef}
-t \mathring{a}'(t) =\mathring{a}(t)\ih \sqrt{\frac{\Lambda}{3}+\frac{\mathring{\rho}_H(t)}{3}},\qquad \mathring{a}(1)=1,
\end{equation*}
which define the Newtonian limit of the FLRW equations.
\end{remark}

In the articles \cite{oli3,oli4}, the second author established the existence of $1$-parameter families of solutions\footnote{To convert the $1$-parameter solutions to the Einstein--Euler equations from \cite{oli3,oli4} to solutions
of \eqref{E:ORIGINALEESYSTEM.a}--\eqref{E:ORIGINALEESYSTEM.b}, the metric, four-velocity, time coordinate and spatial coordinates must each be rescaled by an appropriate powers of $\epsilon$, after which the rescaled time coordinate must be transformed
according to the formula  \eqref{E:TIMETRANSFORMATION}.}
$\{\tilde{g}^{\mu\nu}_\epsilon, \bar{\rho}_\epsilon,\tilde{v}^\mu_\epsilon\}$, $0<\epsilon < \epsilon_0$, to
\eqref{E:ORIGINALEESYSTEM.a}--\eqref{E:ORIGINALEESYSTEM.b}, which include the above family of FLRW solutions, on spacetime regions of the form
\begin{equation*}
(T_1,1] \times \Tbb^3_\epsilon \subset M_\epsilon,
\end{equation*}
for some $T_1 \in (0,1]$, that converge, in a suitable sense, as $\epsilon \searrow 0$ to solutions of the cosmological Poisson--Euler equations of Newtonian gravity. Although this result rigorously established the existence of a wide class of solutions to the Einstein--Euler equations that admit a (cosmological) Newtonian limit,  the local-in-time nature of the result left open the question of what happens on long time scales. In particular, the question of the existence of $1$-parameter families of solutions that converge globally to the future as $\epsilon \searrow 0$ was not addressed. In light of the importance of Newtonian gravity in cosmological simulations \cite{Crocce_et_al:2010,Evrard_et_al:2002,Springel:2005,Springel_et_al:2005}, this question needs to be resolved in order to understand on what time scales Newtonian cosmological simulations can be trusted to  approximate relativistic cosmologies. In this article, we resolve this question under a small initial data condition. Informally, we establish the existence of $1$-parameter families of $\epsilon$-dependent solutions
to \eqref{E:ORIGINALEESYSTEM.a}--\eqref{E:ORIGINALEESYSTEM.b} that:
\textit{(i)} are defined for $\epsilon \in (0,\epsilon_0)$ for some fixed constant $\epsilon_0>0$,
 \textit{(ii)} exist globally on $M_\epsilon$, % and are geodesically complete to the future,
  \textit{(iii)} converge, in a suitable sense, as $\epsilon \searrow 0$ to solutions of the cosmological Poisson--Euler equations of Newtonian gravity, and
  \textit{(iv)} are small, nonlinear perturbations of the FLRW fluid solutions \eqref{FLRW.a}--\eqref{E:TPTA}.
The precise statement of our results can be found in Theorem \ref{T:MAINTHEOREM}.

Before proceeding with the statement of Theorem \ref{T:MAINTHEOREM}, we first fix our notation and conventions, and define a number of new variables and equations.

\subsection{Notation\label{Notation}}

\subsubsection{Index of notation} An index containing frequently used definitions and nonstandard notation can be found in ``Appendix \ref{index}."

\subsubsection{Indices and coordinates\label{iandc}}  Unless stated otherwise, our indexing convention will be as follows: we use lower case Latin letters, e.g., $i, j,k$, for spatial indices that run from $1$ to $n$, and lower case Greek letters, e.g., $\alpha, \beta, \gamma$, for spacetime indices
that run from $0$ to $n$. When considering the Einstein--Euler equations, we will focus on the physical case where $n=3$, while all of the PDE results established in this article hold in arbitrary dimensions.

For scalar functions $f(t,\bar{x}^i)$ of the relativistic coordinates, we let
\begin{equation} \label{Neval}
\underline{f}(t,x^i) := f(t,\epsilon x^i)
\end{equation}
denote the representation of $f$ in Newtonian coordinates.

\subsubsection{Derivatives\label{Derivatives}}
Partial derivatives with respect to the  Newtonian coordinates $(x^\mu)=(t,x^i)$ and the relativistic coordinates $(\bar{x}^\mu)=(t,\bar{x}^i)$ will be denoted by $\partial_\mu = \partial/\partial x^\mu$ and
$\bar{\partial}_{\mu} = \partial/\partial \bar{x}^\mu$, respectively, and we use
$Du=(\partial_j u)$ and $\partial u = (\partial_\mu u)$ to denote the spatial and spacetime gradients, respectively, with respect to the Newtonian coordinates, and similarly $\bar{\partial} u = (\bar{\partial}_\mu u)$ to denote the spacetime gradient with
respect to the relativistic coordinates.
We also use Greek letters to denote multi-indices, e.g.,
$\alpha = (\alpha_1,\alpha_2,\ldots,\alpha_n)\in \mathbb{Z}_{\geq 0}^n$, and employ the standard notation $D^\alpha = \partial_{1}^{\alpha_1} \partial_{2}^{\alpha_2}\cdots
\partial_{n}^{\alpha_n}$ for spatial partial derivatives. It will be clear from context whether a Greek letter stands for a spacetime coordinate index or a multi-index.

Given a vector-valued map $f(u)$, where $u$ is a vector, we use $D f$ and $D_u f$ interchangeably to denote the derivative with respect to the vector $u$, and use the standard notation
\begin{equation*}
  D f(u)\cdot \delta u := \left.\frac{d}{dt}\right|_{t=0} f(u+t\delta u)
\end{equation*}
for the action of the linear operator $D f$ on the vector $\delta u$. For vector-valued maps $f(u,v)$ of two (or more)
variables, we use the notation $D_1 f$ and $D_u f$ interchangeably for the partial
derivative with respect to the first variable, i.e.,
\begin{equation*}
  D_u f(u,v)\cdot \delta u := \left.\frac{d}{dt}\right|_{t=0} f(u+t\delta u,v),
\end{equation*}
and a similar notation for the partial derivative with respect to the other variable.

%More generally, if $f(u)=\{a_{ij}(u)\}$ is a matrix and $w$, $u$ are vectors in the same dimension, then we denote
%$[D_u f(u)\cdot w]$ as a matrix in the same size as $f(u)$ with entries $D_u a_{ij}(u)\cdot w$, that is $[D_u f(u)\cdot w]=\{ w \cdot D_u a_{ij}(u)\}$.

\subsubsection{Function spaces\label{Functionspaces}}
Given a finite-dimensional vector space $V$, we let
$H^s(\mathbb{T}^n,V)$, $s\in \mathbb{Z}_{\geq 0}$,
denote the space of maps from $\mathbb{T}^n$ to $V$ with $s$ derivatives in $L^2(\Tbb^n)$. When the
vector space $V$ is clear from context, we write $H^s(\mathbb{T}^n)$ instead of $H^s(\mathbb{T},V)$.
Letting
\begin{equation*}
\langle{u,v\rangle} = \int_{\mathbb{T}^n} (u(x),v(x))\, d^n x,
\end{equation*}
where $(\cdot,\cdot)$
is a fixed inner product on $V$, denote the standard $L^2$ inner product, the $H^s$ norm is defined by
\begin{equation*}
\|u\|_{H^s}^2 = \sum_{0\leq |\alpha|\leq s} \langle D^\alpha u, D^\alpha u \rangle.
\end{equation*}
For any fixed basis $\{\mathbf{e}_I\}^N_{I=1}$ of $V$, we follow \cite{oli3} and define a subspace of $H^s(\mathbb{T}^n,V)$ by
\begin{equation*}
  \bar{H}^s(\mathbb{T}^n,V)=\biggl\{u(x)=\sum_{I=1}^N u^I(x)\mathbf{e}_I\in H^s(\mathbb{T}^n,V)\biggl|\langle 1,u^I\rangle=0 \text{ for } I=1, 2, \ldots, N\biggr\}.
\end{equation*}
Specializing to $n=3$, we define, for fixed $\epsilon_0 >0 $ and $r>0$, the spaces
\begin{equation*}
X^s_{\epsilon_0,r}(\mathbb{T}^3) = (-\epsilon_0,\epsilon_0)\times B_r(H^{s+1}(\mathbb{T}^3,\mathbb{S}_3)) \times H^s(\mathbb{T}^3,\mathbb{S}_3) \times B_r(\bar{H}^s(\mathbb{T}^3)) \times \bar{H}^s(\mathbb{T}^3,\mathbb{R}^3),
\end{equation*}
where $\mathbb{S}_N$ denotes the space of symmetric $N\times N$ matrices, and here and throughout this article,
we use, for any Banach space $Y$, $B_r(Y)=\{\, y\in Y \, | \, \|y\|_Y < r \,\}$ to denote the open ball of radius $r$.

To handle the smoothness of coefficients that appear in various equations, we introduce the spaces
\begin{equation*}
E^{p}((0,\epsilon_0)\times (T_1,T_2)\times U,V),\quad p \in \Zbb_{\geq 0},
\end{equation*}
which are defined to be the set of $V$-valued maps $f(\epsilon,t,\xi)$ that
are smooth on the open set $(0,\epsilon_0)\times (T_1,T_2)\times U$, where $U$ $\subset$ $\Tbb^n \times \Rbb^N$
is open, and for which there exist constants $C_{k,\ell}>0$, $(k,\ell)\in \{0,1,\ldots,p\}\times \Zbb_{\geq 0}$,
such that
\begin{equation*}
|\del{t}^k  D_\xi^\ell f(\epsilon,t,\xi)| \leq C_{k,\ell}, \quad \forall \,
(\epsilon,t,\xi) \in  (0,\epsilon_0)\times (T_1,T_2)\times U.
\end{equation*}
If $V=\Rbb$ or $V$ is clear from context, we will drop the $V$ and simply write $E^{p}((0,\epsilon_0)\times (T_1,T_2)\times U)$. Moreover,
we will use the notation $E^{p}((T_1,T_2)\times U,V)$ to denote the subspace of $\epsilon$-independent maps. Given $f\in E^{p}((0,\epsilon_0)\times (T_1,T_2)\times U,V)$, we note, by uniform continuity, that the limit
$f_0(t,\xi) := \lim_{\epsilon \searrow 0}f(\epsilon,t,\xi)$ exists and defines an element of $E^{p}((T_1,T_2)\times U,V)$.

\subsubsection{Constants\label{Constants}}
We employ that standard notation
\begin{equation*}
a \lesssim b
\end{equation*}
for inequalities of the form
\begin{equation*}
a \leq C b
\end{equation*}
in situations where the precise value or dependence on
other quantities of the constant $C$ is not required. On the other hand, when the dependence of the constant
on other inequalities needs to be specified, for example if the constant depends on the norms $\|u\|_{L^\infty}$ and $\|v\|_{L^\infty}$, we use the notation
\begin{equation*}
C = C(\|u\|_{L^\infty},\|v\|_{L^\infty}).
\end{equation*}
Constants of this type will always be nonnegative, non-decreasing, continuous functions of their arguments, and in general, $C$ will be used
to denote constants that may change from line to line. However, when we want to isolate
a particular constant for use later on, we will label the constant with a subscript, e.g., $C_1, C_2, C_3$, etc.

\subsubsection{Remainder terms\label{remainder}}
In order to simplify the handling of remainder terms whose exact form is not important, we will use, unless otherwise stated,
 upper case calligraphic letters, e.g.,
  $\mathcal{S}(\epsilon,t,x,\xi)$, $\mathcal{T}(\epsilon,t,x,\xi)$ and $\mathcal{U}(\epsilon,t,x,\xi)$, to denote vector-valued maps that are elements
of the space $E^0\bigl( (0,\epsilon_0)\times (0,2)\times \Tbb^n \times B_R\bigl(\mathbb{R}^N\bigr)\bigr)$, and upper case letters in typewriter font, e.g.,
$\texttt{S}(\epsilon,t,x,\xi)$, $\texttt{T}(\epsilon,t,x,\xi)$ and $ \texttt{U} (\epsilon,t,x,\xi)$, to denote vector-valued maps that are elements
of the space $E^1\bigl( (0,\epsilon_0)\times (0,2)\times \Tbb^n \times B_R\bigl(\mathbb{R}^N\bigr)\bigr)$.

\subsection{Conformal Einstein--Euler equations\label{conformalEinsteinEuler}}
The method we use to establish the existence of $\epsilon$-dependent families of solutions to the Einstein--Euler equations that
exist globally to the future is based on the one developed in \cite{oli5}.
Following \cite{oli5}, we introduce the conformal metric
\begin{equation}\label{E:CONFORMALTRANSF}
  \bar{g}^{\mu\nu}=e^{2\Psi}\tilde{g}^{\mu\nu}
\end{equation}
and the conformal four velocity
\begin{equation}\label{E:CONFORMALVELOCITY}
  \bar{v}^\mu=e^\Psi\tilde{v}^\mu.
\end{equation}
Under this change in variables, the Einstein equation \eqref{E:ORIGINALEESYSTEM.a} transforms as
\begin{equation}\label{E:CONFORMALEINSTEIN1}
  \bar{G}^{\mu\nu}=\bar{T}^{\mu\nu}:=e^{4\Psi}\tilde{T}^{\mu\nu}-e^{2\Psi}\Lambda\bar{g}^{\mu\nu}
  +2(\bar{\nabla}^\mu\bar{\nabla}^\nu\Psi-\bar{\nabla}^\mu\Psi\bar{\nabla}^\nu\Psi)
  -(2\bar{\Box}\Psi+|\bar{\nabla}\Psi|^2_{\bar{g}})
  \bar{g}^{\mu\nu},
\end{equation}
where $\bar{\Box}=\bar{\nabla}^\mu\bar{\nabla}_\mu$, $|\bar{\nabla}\Psi|^2_{\bar{g}}=\bar{g}^{\mu\nu}\bar{\nabla}_\mu\Psi\bar{\nabla}_\nu\Psi$, and here and in the following, unless otherwise specified, we raise and lower all coordinate tensor indices using the conformal metric.
Contracting the free indices of \eqref{E:CONFORMALEINSTEIN1} gives
\begin{align*}%\label{E:RT}
  \bar{R}=4\Lambda-\bar{T},
\end{align*}
where $\bar{T}=\bar{g}_{\mu\nu}\bar{T}^{\mu\nu}$, which we can use
to write \eqref{E:CONFORMALEINSTEIN1} as
\begin{align}
  -2\bar{R}^{\mu\nu}=-4\bar{\nabla}^\mu\bar{\nabla}^\nu\Psi+4\bar{\nabla}^\mu\Psi\bar{\nabla}^\nu\Psi
  -2\left[\bar{\Box}\Psi
  +2|\bar{\nabla}\Psi|^2+\left(\frac{1-\epsilon^2K}{2}\bar{\rho}+\Lambda\right)e^{2\Psi}\right]\bar{g}
  ^{\mu\nu}& \notag \\-2e^{2\Psi}(1+\epsilon^2K&)\bar{\rho} \bar{v}^\mu \bar{v}^\nu. \label{E:EXPANSIONOFEIN}
\end{align}
We will refer to these equations as the \textit{conformal Einstein equations}.

Letting $\tilde{\Gamma}^\gamma_{\mu\nu}$ and $\bar{\Gamma}^\gamma_{\mu\nu}$ denote the Christoffel symbols of the metrics $\tilde{g}_{\mu\nu}$
and $\bar{g}_{\mu\nu}$, respectively, the difference $\tilde{\Gamma}^\gamma_{\mu\nu}-
\bar{\Gamma}^\gamma_{\mu\nu}$ is readily calculated to be
\begin{equation*} \label{Gammadif}
\tilde{\Gamma}^\gamma_{\mu\nu}-\bar{\Gamma}^\gamma_{\mu\nu} =
\bar{g}^{\gamma \alpha}\bigl(\bar{g}_{\mu \alpha}\bar{\nabla}_\nu\Psi
+ \bar{g}_{\nu\alpha}\bar{\nabla}_\mu\Psi - \bar{g}_{\mu\nu} \bar{\nabla}_\alpha\Psi \bigr).
\end{equation*}
Using this, we can express the Euler equations \eqref{E:ORIGINALEESYSTEM.b} as
\begin{equation}\label{Confeul}
\bar{\nabla}_\mu \tilde{T}^{\mu \nu} = -6\tilde{T}^{\mu\nu}\nabla_\mu\Psi +\bar{g}_{\alpha\beta}\tilde{T}^{\alpha\beta}
\bar{g}^{\mu\nu}\bar{\nabla}_\mu\Psi,
\end{equation}
which we refer to as the \textit{conformal Euler equations}.

\begin{remark}
Due to our choice of time orientation, the conformal fluid four-velocity $\bar{v}^\mu$, which we assume is future oriented, satisfies
\begin{equation*} \label{future}
\bar{v}^0 < 0.
\end{equation*}
We also note that $\bar{v}^\mu$ is normalized by
\begin{equation} \label{normal}
\bar{v}^\mu\bar{v}_\mu = -1,
\end{equation}
which is a direct consequence of \eqref{vnorm}, \eqref{E:CONFORMALTRANSF} and \eqref{E:CONFORMALVELOCITY}.
\end{remark}

\subsection{Conformal factor\label{Conformalfactor}} Following  \cite{oli5}, we choose
\begin{align}\label{E:CONFORMALFACTOR}
  \Psi=-\ln{t}
\end{align}
for the conformal factor, and for later use, we introduce the background metric
\begin{align}\label{E:CONFORMALFLRW}
  \bar{h}=-\frac{3}{\Lambda}dtdt+E^2(t)\delta_{ij}d\bar{x}^id\bar{x}^j,
\end{align}
where
\begin{align}\label{E:DEFE}
  E(t)=a(t)t,
\end{align}
which is conformally related to the FLRW metric \eqref{FLRW.a}.
Using \eqref{E:TPTA},  we observe that $E(t)$ satisfies
\begin{align}\label{E:PTE}
  \partial_t E(t)=\frac{1}{t}E(t)\Omega(t),
\end{align}
where $\Omega(t)$ is defined by
\begin{align}\label{E:OMEGADEF}
  \Omega(t)=1-\ih\sqrt{\frac{\Lambda}{3}+\frac{\rho_H(t)}{3}}.
\end{align}
A short calculation then shows that the non-vanishing Christoffel symbols of the background metric \eqref{E:CONFORMALFLRW} are given by
\begin{align}\label{E:HOMCHRIS}
  \bar{\gamma}^0_{ij}=\frac{\Lambda}{3t}E^2\Omega\delta_{ij} \quad\text{and}\quad
  \bar{\gamma}^i_{j0}=\frac{1}{t} \Omega \delta^i_j,
\end{align}
from which we compute
\begin{align}\label{E:GAMMAFLRW}
  \bar{\gamma}^\sigma:=\bar{h}^{\mu\nu}\bar{\gamma}^\sigma_{\mu\nu}=\frac{\Lambda}{t}\Omega
  \delta^\sigma_0.
\end{align}

\subsection{Wave gauge\label{Wavegauge}}
For the hyperbolic reduction of the conformal Einstein equations, we use the \textit{wave gauge} from \cite{oli5}  that is defined by
\begin{align}\label{E:WAVEGAUGE}
   \bar{Z}^\mu = 0,
\end{align}
where
\begin{equation} \label{Zdef}
\bar{Z}^\mu = \bar{X}^\mu+\bar{Y}^\mu
\end{equation}
with
\begin{align}
  \bar{X}^\mu&:= \bar{\Gamma}^\mu
  - \bar{\gamma}^\mu =-\bar{\partial}_\nu\bar{g}^{\mu\nu}
  +\frac{1}{2}\bar{g}^{\mu\sigma}\bar{g}_{\alpha\beta}\bar{\partial}_\sigma\bar{g}^{\alpha\beta}
  -\frac{\Lambda}{t}\Omega\delta^\mu_0 \qquad  \bigl(\bar{\Gamma}^\mu=\bar{g}^{\sigma\nu}\bar{\Gamma}^\mu_{\sigma\nu}\bigr)
\label{E:XY}
\intertext{and}
  \bar{Y}^\mu&:=-2\bar{\nabla}^\mu\Psi+ \frac{2\Lambda}{3t} \delta^\mu_0 =-2(\bar{g}^{\mu\nu}-\bar{h}^{\mu\nu})\bar{\nabla}_\nu\Psi. \label{Ydef}
\end{align}

\subsection{Variables\label{vardefs}}
To obtain variables that are simultaneously suitable for establishing global existence and taking Newtonian limits, we switch to Newtonian coordinates $(x^\mu)=(t,x^i)$ and employ the following rescaled version
of the variables introduced in \cite{oli5}:
\begin{align}
  u^{0\mu}&=\frac{1}{\epsilon}\frac{\underline{\bar{g}^{0\mu}}-\bar{\eta}^{0\mu}}{2t}, \label{E:u.a} \\
  u^{0\mu}_0&=\frac{1}{\epsilon}\left(\underline{\bar{\partial}_0\bar{g}^{0\mu}}-\frac{3(\underline{\bar{g}^{0\mu}}
  -\bar{\eta}^{0\mu})}{2t}\right) \label{E:u.b},\\
  u^{0\mu}_i&=\frac{1}{\epsilon}\underline{\bar{\partial}_i\bar{g}^{0\mu}}, \label{E:u.c} \\
  u^{ij}(t,x) &=\frac{1}{\epsilon}\Bigl(\underline{\bar{\mathfrak{g}}^{ij}}-\bar{\eta}^{ij}), \label{E:u.d} \\
  u^{ij}_\mu &=\frac{1}{\epsilon}\underline{\bar{\partial}_{\mu}\bar{\mathfrak{g}}^{ij}}, \label{E:u.e} \\
  u&=\frac{1}{\epsilon}\underline{\bar{\mathfrak{q}}}, \label{E:u.f} \\
  u_\mu&=\frac{1}{\epsilon}\underline{\bar{\partial}_\mu\bar{\mathfrak{q}}}, \label{E:u.g} \\
 z_i&=\frac{1}{\epsilon}\underline{\bar{v}_i},\label{E:z.b}\\
 \zeta&=\frac{1}{1+\epsilon^2 K}\ln\bigl(t^{-3(1+\epsilon^2 K)}\underline{\bar{\rho}}\bigr), \label{E:ZETA}\\
 \intertext{and}
  \delta\zeta&=\zeta-\zeta_H, \label{E:DELZETA}
\end{align}
where
\begin{gather}
\bar{\mathfrak{g}}^{ij}=\alpha^{-1}\bar{g}^{ij},\qquad \alpha:=
(\det{\check{g}_{ij}})^{-\frac{1}{3}}=(\det{\bar{g}^{kl}})^{\frac{1}{3}}, \qquad \check{g}_{ij}=(\bar{g}^{ij})^{-1}, \label{E:GAMMA}\\
\bar{\mathfrak{q}}
=\bar{g}^{00}-\bar{\eta}^{00}-\frac{\Lambda}{3}\ln{\alpha}-\frac{2\Lambda}{3}\ln{E}\label{E:q}, \\
\bar{\eta}^{\mu\nu} = -\frac{\Lambda}{3}\delta^\mu_0\delta^\nu_0+\delta_{i}^\mu\delta^\nu_j\delta^{ij},  \label{etabardef}
\intertext{and}
  \zeta_H(t)=\frac{1}{1+\epsilon^2 K}\ln\bigl(t^{-3(1+\epsilon^2 K)}\rho_H(t)\bigr).
 \label{E:ZETAH1}
\end{gather}
As we show below in \S\ref{FLRWanal}, $\zeta_H$ is given by the explicit formula
\begin{equation} \label{E:ZETAH3}
  \zeta_H(t)=\zeta_H(1)-\frac{2}{1+\epsilon^2K}\ln{\left(\frac{C_0-t^{3(1+\epsilon^2K)}}{C_0-1}\right)},
\end{equation}
where the constants $C_0$ and $\zeta_H(1)$ are defined by
\begin{equation} \label{C0def}
 C_0=\frac{\sqrt{\Lambda+\rho_H(1)}+\sqrt{\Lambda}}{\sqrt{\Lambda+\rho_H(1)}-\sqrt{\Lambda}}>1
\end{equation}
and
\begin{equation}\label{zetaH1}
\zeta_H(1)=\frac{1}{1+\epsilon^2K}\ln{ \rho_H(1)},
\end{equation}
respectively. Letting
\begin{equation} \label{zetaHringdef}
\ring{\zeta}_H = \lim_{\epsilon\searrow 0} \zeta_H
\end{equation}
denote the Newtonian limit of $\zeta_H$, it is clear from the formula \eqref{E:ZETAH3} that
\begin{equation} \label{zetaHringform}
\ring{\zeta}_H (t) = \ln{ \rho_H(1)}- 2\ln{\left(\frac{C_0-t^{3}}{C_0-1}\right)}.
\end{equation}
For later use, we also define
\begin{align}
 z^i &= \frac{1}{\epsilon}\underline{\bar{v}^i}.\label{E:z.a}
\end{align}

\begin{remark}
It is important to emphasize that the above variables are defined on the $\epsilon$-independent manifold $M=(0,1]\times \mathbb{T}^3$. Effectively, we are treating components of
the geometric quantities with respect to the relativistic coordinates as scalars defined on $M_\epsilon$ and we
are pulling them back as scalars to $M$ by transforming to Newtonian coordinates. This process is necessary to obtain variables that have a well defined Newtonian limit.
We stress that for any fixed $\epsilon >0$, the gravitational and matter fields $\{\bar{g}^{\mu\nu},\bar{v}^\mu,\bar{\rho}\}$ on $M_\epsilon$ are completely determined by
the fields $\{u^{0\mu}, u^{ij},u,z_i,\zeta\}$ on $M$.
\end{remark}

\subsection{Conformal Poisson-Euler equations\label{CPEequations}}
The $\epsilon \searrow 0$ limit of the conformal Einstein--Euler equations on $M$
are the \textit{conformal cosmological Poisson--Euler equations}, which are defined by
\begin{align}
    \partial_t \mathring{\rho}+\sqrt{\frac{3}{\Lambda}}\partial_j\left(\mathring{\rho}\mathring{z}^j\right)
    &=\frac{3(1-\mathring{\Omega})}{t}\mathring{\rho}, \label{E:COSEULERPOISSONEQ.a}\\
    \sqrt{\frac{\Lambda}{3}}\mathring{\rho}\partial_t\mathring{z}^j+K
    \partial^j\mathring{\rho}+\mathring{\rho}\mathring{z}^i\partial_i\mathring{z}^j
    &=\sqrt{\frac{\Lambda}{3}}\frac{1}{t}\mathring{\rho}\mathring{z}^j-\frac{1}{2}
    \frac{3}{\Lambda}\mathring{\rho}\partial^j\mathring{\Phi}\quad\Bigl(\partial^j:= \frac{\delta^{ji}}{\mathring{E}^2} \partial_i\Bigr), \label{E:COSEULERPOISSONEQ.b}\\
    \Delta\mathring{\Phi}&=\frac{\Lambda}{3}\frac{\mathring{E}^2}{t^2} \Pi \mathring
    {\rho}  \qquad (\Delta:=\delta^{ij}\partial_i\partial_j), \label{E:COSEULERPOISSONEQ.c}
\end{align}
where $\Pi$ is the projection operator defined by
\begin{equation} \label{Pidef}
\Pi u = u - \langle 1, u \rangle,
\end{equation}
for $u\in L^2(\mathbb{T}^3)$,
\begin{equation}\label{Eringform}
\mathring{E}(t) =
  \left(\frac{C_0-t^{3}}{C_0-1}\right)^{\frac{2}{3}}
\end{equation}
and
\begin{equation} \label{Oringdef}
\ring{\Omega}(t) = \frac{2t^3}{t^3-C_0},
\end{equation}
with $C_0$ given by \eqref{C0def}.

It will be important for our analysis to introduce the modified density variable $\mathring{\zeta}$ defined by
\begin{equation*} \label{zetaringdef}
 \mathring{\zeta} = \ln(t^{-3}\mathring{\rho}),
\end{equation*}
which is the nonrelativistic version of the variable $\zeta$ introduced above, see \eqref{E:ZETA}. A short calculation then
shows that the conformal cosmological Poisson--Euler equations can be expressed in terms of this modified density as follows:
\begin{align}
    \partial_t \mathring{\zeta}+\sqrt{\frac{3}{\Lambda}}\bigl( \mathring{z}^j\partial_j \mathring{\zeta} + \partial_j\mathring{z}^j\bigr)
    &=-\frac{3\mathring{\Omega}}{t}, \label{CPeqn1}\\
    \sqrt{\frac{\Lambda}{3}}\partial_t\mathring{z}^j+ \mathring{z}^i\partial_i\mathring{z}^j+ K
    \partial^j\mathring{\zeta}
    &=\sqrt{\frac{\Lambda}{3}}\frac{1}{t}\mathring{z}^j-\frac{1}{2}
    \frac{3}{\Lambda}\partial^j\mathring{\Phi}, \label{CPeqn2}\\
    \Delta\mathring{\Phi}&=\frac{\Lambda}{3} t \mathring{E}^2\Pi e^{\mathring{\zeta}} .  \label{CPeqn3}
\end{align}

\subsection{Main Theorem}
We are in the position to state the main theorem of the article. The proof is given in \S \ref{S:MAINPROOF}.
%------------------------------
\begin{theorem}\label{T:MAINTHEOREM}
Suppose $s\in \mathbb{Z}_{\geq 3}$, $0<K\leq \frac{1}{3}$, $\Lambda >0$, $\rho_H(1)>0$, $r>0$ and the free initial data $\{\smfu^{ij},\smfu^{ij}_0,\breve{\rho}_0,\breve{\nu}^i\}$ is chosen so that $\smfu^{ij}\in B_r(H^{s+1}(\mathbb{T}^3,\mathbb{S}_3))$,
 $\smfu^{ij}_0\in H^{s}(\mathbb{T}^3,\mathbb{S}_3)$, $\breve{\rho}_0\in B_r(\bar{H}^s(\mathbb{T}^3))$,
$\breve{\nu}^i\in \bar{H}^s(\mathbb{T}^3,\mathbb{R}^3)$.
%-------------------------------------------------------------------------------------------------------------
Then for $r>0$ chosen small enough, there exists a constant $\epsilon_0>0$ and maps
$\breve{u}^{\mu\nu} \in C^\omega\bigl(X^s_{\epsilon_0,r}(\mathbb{T}^3),H^{s+1}(\mathbb{T}^3,\mathbb{S}_{4})\bigr)$,
$\breve{u} \in C^\omega\bigl(X^s_{\epsilon_0,r}(\mathbb{T}^3),H^{s+1}(\mathbb{T}^3)\bigr)$,
$\breve{u}^{\mu\nu}_0 \in C^\omega\bigl(X^s_{\epsilon_0,r}(\mathbb{T}^3),H^{s}(\mathbb{T}^3,\mathbb{S}_{4})\bigr)$, $\breve{u}_0 \in C^\omega\bigl(X^s_{\epsilon_0,r}(\mathbb{T}^3),H^{s}(\mathbb{T}^3)\bigr)$,
$\breve{z}=(\breve{z}_i) \in C^\omega\bigl(X^s_{\epsilon_0,r}(\mathbb{T}^3),H^{s}(\mathbb{T}^3,\mathbb{R}^3\bigr)\bigr)$, and
$\delta\breve{\zeta} \in C^\omega\bigl(X^s_{\epsilon_0,r}(\mathbb{T}^3),H^{s}(\mathbb{T}^3)\bigr)$,
such that\footnote{See Lemma \ref{L:INITIALTRANSFER} for details. }
\begin{align*}
     u^{\mu0}|_{t=1} &:= \breve{u}^{\mu0}(\epsilon, \smfu^{kl},\smfu^{kl}_0, \breve{\rho}_0, \breve{\nu}^k) = \epsilon \frac{\Lambda}{6}\Delta^{-1}\breve{\rho}_0 \delta^\mu_0+ \textrm{\em O}(\epsilon^2),\\
     u^{ij}|_{t=1} &:= \breve{u}^{ij}(\epsilon, \smfu^{kl},\smfu^{kl}_0, \breve{\rho}_0, \breve{\nu}^k) = \epsilon^2\left(\smfu^{ij}
     -\frac{1}{3}\smfu^{pq}\delta_{pq}\delta^{ij}\right) + \textrm{\em O}(\epsilon^3),\\
     u|_{t=1} &:= \breve{u}(\epsilon, \smfu^{kl},\smfu^{kl}_0, \breve{\rho}_0, \breve{\nu}^k) = \epsilon^2\frac{2\Lambda}{9}\smfu^{ij}\delta_{ij} + \textrm{\em O}(\epsilon^3),\\
     z_{i}|_{t=1} &:= \breve{z}_i(\epsilon, \smfu^{kl},\smfu^{kl}_0, \breve{\rho}_0, \breve{\nu}^k) = \frac{\breve{\nu}^j\delta_{ij}}{\rho_H(1)+\breve{\rho}_0} + \textrm{\em O}(\epsilon),\\
     \delta \zeta|_{t=1} &:= \delta\breve{\zeta}(\epsilon, \smfu^{kl},\smfu^{kl}_0, \breve{\rho}_0, \breve{\nu}^k) = \ln{\left(1+\frac{\breve{\rho}_0}{\rho_H(1)}\right)} + \textrm{\em O}(\epsilon^2), \\
u_0|_{t=1}&:=\breve{u}_0(\epsilon, \smfu^{ij},\smfu^{ij}_0, \breve{\rho}_0, \breve{\nu}^i) =   \textrm{\em O}(\epsilon)
\intertext{and}
u^{\mu\nu}_0|_{t=1}&:=\breve{u}^{\mu\nu}_0(\epsilon, \smfu^{kl},\smfu^{kl}_0, \breve{\rho}_0, \breve{\nu}^k) =   \textrm{\em O}(\epsilon)
\end{align*}
determine via the formulas \eqref{E:u.a}, \eqref{E:u.b}, \eqref{E:u.d}, \eqref{E:u.f}, \eqref{E:z.b}, \eqref{E:ZETA}, and
\eqref{E:DELZETA} a
solution of the gravitational and gauge constraint equations, see \eqref{E:CONSTRAINT}-\eqref{E:WAVECONSTRAINT}
and Remark \ref{wavegaugerem}.
Furthermore, there exists a $\sigma>0$, such that if
  \begin{align*}%\label{E:SMALLNESSOFDATA}
    \|\smfu^{ij}\|_{H^{s+1}}+\|\smfu^{ij}_0\|_{H^{s}}+\|\breve{\rho}_0\|_{H^s}+\|
  \breve{\nu}^i\|_{H^s} \leq \sigma,
  \end{align*}
  then there exist maps
  \begin{align*}
      u^{\mu\nu}_{\epsilon} & \in C^0((0,1], H^{s}(\mathbb{T}^3,\mathbb{S}_4))\cap C^1((0,1], H^{s-1}(\mathbb{T}^3,\mathbb{S}_4)),\\
u^{\mu\nu}_{\gamma,\epsilon} & \in C^0((0,1], H^{s}(\mathbb{T}^3,\mathbb{S}_4))\cap C^1((0,1], H^{s-1}(\mathbb{T}^3,\mathbb{S}_4)),\\
      u_\epsilon & \in C^0((0,1], H^{s}(\mathbb{T}^3))\cap C^1((0,1], H^{s-1}((\mathbb{T}^3)),\\
u_{\gamma,\epsilon} & \in C^0((0,1], H^{s}(\mathbb{T}^3))\cap C^1((0,1], H^{s-1}((\mathbb{T}^3)),\\
      \delta\zeta_\epsilon& \in C^0((0,1], H^{s}(\mathbb{T}^3))\cap C^1((0,1], H^{s-1}(\mathbb{T}^3)),\\
      z^\epsilon_i & \in C^0((0,1], H^{s}(\mathbb{T}^3),\mathbb{R}^3))\cap C^1((0,1], H^{s-1}(\mathbb{T}^3,\mathbb{R}^3)),
  \end{align*}
for $\epsilon \in (0,\epsilon_0)$, and
\begin{align*}
  \mathring{\Phi} & \in C^0((0,1], H^{s+2}(\mathbb{T}^3))\cap C^1((0,1], H^{s+1}(\mathbb{T}^3)),\\
  \delta \mathring{\zeta} & \in C^0((0,1], H^{s}(\mathbb{T}^3))\cap C^1((0,1], H^{s-1}(\mathbb{T}^3)),\\
  \mathring{z}_i & \in C^0((0,1], H^{s}(\mathbb{T}^3,\mathbb{R}^3))\cap C^1((0,1], H^{s-1}(\mathbb{T}^3,\mathbb{R}^3)),\\
\end{align*}
such that
\begin{enumerate}[(i)]
  \item \label{I} $\{u^{\mu\nu}_\epsilon(t, x), u_\epsilon(t, x), \delta\zeta_\epsilon(t, x), z_i^\epsilon(t, x)\}$ determines, via \eqref{E:CONFORMALTRANSF}, \eqref{E:CONFORMALVELOCITY}, \eqref{normal},  \eqref{E:u.a}, \eqref{E:u.d}, \eqref{E:u.f}, \eqref{E:z.b},
and \eqref{E:ZETA}-\eqref{etabardef}, a $1$-parameter family of solutions to the Einstein--Euler equations \eqref{E:ORIGINALEESYSTEM.a}-\eqref{E:ORIGINALEESYSTEM.b} in the wave gauge \eqref{E:WAVEGAUGE} on $M_\epsilon$,
  \item \label{II} $\{\mathring{\Phi}(t,x), \mathring{\zeta}(t,x):=\delta\mathring{\zeta}+\mathring{\zeta}_H, \mathring{z}^i(t,x):=
\mathring{E}(t)^{-2}\delta^{ij}\mathring{z}_j(t,x)\}$, with $\mathring{\zeta}_H$ and $\mathring{E}$ given
by \eqref{zetaHringform} and \eqref{Eringform}, respectively, solves the conformal cosmological Poisson-Euler equations \eqref{CPeqn1}-\eqref{CPeqn3} on $M$ with initial data $\mathring{\zeta}|_{t=1}= \ln(\rho_H(1)+ \breve{\rho}_0)$ and $\mathring{z}^i|_{t=1}=\breve{\nu}^i/(\rho_H(1)+\breve{\rho}_0)$,
  \item  \label{III} the uniform bounds
  \begin{gather*}
        \|\delta\mathring{\zeta}\|_{L^\infty((0,1],H^{s})} +  \|\mathring{\Phi}\|_{L^\infty((0,1],H^{s+2})} +  \|\mathring{z}_j\|_{L^\infty((0,1], H^{s})}+
      \|\delta\zeta_\epsilon\|_{L^\infty((0,1],H^{s})} +  \|\mathring{z}_j^\epsilon\|_{L^\infty((0,1], H^{s})} \lesssim 1
\intertext{and}
 \|u^{\mu\nu}_\epsilon\|_{L^\infty((0,1],H^{s})}+
\|u^{\mu\nu}_{\gamma,\epsilon}\|_{L^\infty((0,1],H^{s})} +  \|u_\epsilon\|_{L^\infty((0,1],H^{s})} +\|u_{\gamma,\epsilon}\|_{L^\infty((0,1],H^{s})} \lesssim 1,
\end{gather*}
hold for $\epsilon \in (0, \epsilon_0)$,
\item \label{IV} and the uniform error estimates
      \begin{gather*}
        \|\delta\zeta_\epsilon-\delta\mathring{\zeta}\|_{L^\infty((0,1],H^{s-1})} +
% \biggl\|\underline{v^0}-\sqrt{\frac{\Lambda}{3}}\biggr\|_{L^\infty((0,1],H^{s-1})}+
\| z_j^\epsilon-\mathring{z}_j\|_{L^\infty((0,1]\times H^{s-1})}\lesssim \epsilon, \\
        \|u^{\mu\nu}_{\epsilon,0}\|_{L^\infty((0,1],H^{s-1})}+
\|u^{\mu\nu}_{k,\epsilon}-\delta^\mu_0\delta^\nu_0\partial_k\mathring{\Phi}\|_{L^\infty((0,1],H^{s-1})}+
 \|u^{\mu\nu}_\epsilon\|_{L^\infty((0,1],H^{s-1})}\lesssim \epsilon
\intertext{and}
        \|u_{\gamma,\epsilon}\|_{L^\infty((0,1],H^{s-1})}+
\|u_\epsilon\|_{L^\infty((0,1],H^{s-1})}\lesssim \epsilon
      \end{gather*}
      hold for $\epsilon \in (0, \epsilon_0)$.
\end{enumerate}
\end{theorem}

\subsection{Future directions}
Although the $1$-parameter families of $\epsilon$-dependent solutions to the Einstein--Euler equations
from Theorem \ref{T:MAINTHEOREM} do provide a positive answer to the question of the existence of non-homogeneous relativistic
cosmological solutions that are globally approximated to the future by solutions of Newtonian gravity, it does not resolve the question for initial data that is relevant to our Universe. This is because these solutions have a characteristic size $\sim \epsilon$ and should be interpreted as cosmological versions of isolated systems \cite{gre,Oliynyk:PRD_2014,oli6}.  This defect was remedied on short time scales in \cite{oli6}. There the local-in-time existence of $1$-parameter families of $\epsilon$-dependent solutions to the Einstein--Euler equations that converge to solutions of the cosmological Poisson--Euler equations on \textit{cosmological spatial scales} was established.

In work that is currently in preparation \cite{liu1}, we combine the techniques developed in \cite{oli6} with a generalization of the ones developed in this article to extend the local-in-time existence
result from \cite{oli6} to a global-in-time result. This resolves the existence question of non-homogeneous relativistic cosmological solutions
that are globally approximated to the future on cosmological scales by solutions of Newtonian gravity, at least for initial data that is a small perturbation of FLRW initial data.
However, this is far from the end of the story because  there are relativistic effects that are important for precision cosmology that are not captured by the Newtonian solutions.
To understand these relativistic effects, higher-order post-Newtonian (PN) expansions are required starting with the 1/2-PN expansion, which is, by definition, the $\epsilon$ order correction
to the Newtonian gravity. In particular, it can be shown  \cite{OliRob} that the $1$-parameter families of solutions must admit a 1/2-PN expansion in order to view them on large scales as a linear perturbation of FLRW solutions. The importance of this result is that it shows it is possible to have rigorous solutions that fit within the standard cosmological paradigm of linear perturbations of FLRW metrics on large scales while, at the same time, are fully nonlinear on small scales of order $\epsilon$. Thus the natural next step is to extend the results of \cite{liu1} to include the existence of $1$-parameter families of $\epsilon$-dependent solutions to the Einstein--Euler equations
that admit 1/2-PN expansions globally to the future on cosmological scales. This is work that is currently in progress.

\subsection{Prior and related work} The future nonlinear stability of the FLRW fluid solutions for a linear equation of state $p=K\rho$
was first established  under the condition $0<K<1/3$ and the assumption of zero fluid vorticity
by Rodnianski and Speck \cite{RodnianskiSpeck:2013} using a generalization of a wave-based method developed by Ringstr\"{o}m in \cite{rin1}.
Subsequently, it has been shown \cite{Friedrich:2016,HadzicSpeck:2015,LubbeKroon:2013,spe}
that this future nonlinear stability result remains true for fluids with nonzero vorticity and
also for the equation of state parameter values $K=0$ and $K=1/3$, which correspond to dust and pure radiation, respectively.
It is worth noting that the stability results established in \cite{LubbeKroon:2013} and \cite{Friedrich:2016} for $K=1/3$ and
$K=0$, respectively,
rely on Friedrich's conformal method \cite{Friedrich:1986,Friedrich:1991},
which is completely different from the techniques used in \cite{HadzicSpeck:2015,RodnianskiSpeck:2013,spe} for the parameter values $0\leq K<1/3$.

In the Newtonian setting, the global existence to the future of solutions to the cosmological Poisson--Euler equations was established in
\cite{bra1} under a small initial data assumption and for a class of polytropic equations of state.
% that exclude the linear equations of state.

A new method was introduced in \cite{oli5} to prove the future nonlinear stability of the FLRW fluid solutions that was based on
a wave formulation of a conformal version of the Einstein--Euler equations. The global existence results in this article are established using
this approach. We also note that this method was recently used to establish the nonlinear stability of the FLRW fluid solutions that satisfy the generalized Chaplygin equation of state \cite{LeFlochWei:2015}.

\subsection{Overview}
In \S \ref{S:FORMULATIONOFEQ},  we employ the variables \eqref{E:u.a}--\eqref{E:DELZETA} and the wave gauge \eqref{E:WAVEGAUGE} to write the
conformal Einstein--Euler system, given by \eqref{E:EXPANSIONOFEIN} and \eqref{Confeul}, as a
non-local symmetric hyperbolic system, see \eqref{E:REALEQ}, that is jointly singular in $\epsilon$ and $t$.

In \S \ref{EEcont}, we state and prove a local-in-time existence and uniqueness result along with a continuation principle for solutions of the reduced
conformal Einstein--Euler
equations and discuss how solutions to the reduced conformal Einstein--Euler equations determine solutions to the singular system \eqref{E:REALEQ}. Similarly, in \S \ref{PEcont},
we state and prove a local-in-time existence and uniqueness result and continuation principle for solutions of the conformal cosmological Poisson--Euler equations \eqref{CPeqn1}--\eqref{CPeqn3}.

We establish in \S \ref{S:MODEL} uniform a priori estimates  for solutions to a class of symmetric hyperbolic equations that are
jointly singular in $\epsilon$ and $t$, and include both the formulation \eqref{E:REALEQ} of the conformal Einstein--Euler equations and the $\epsilon \searrow 0$
limit of these equations. We also establish \textit{error estimates}, that is, a priori estimates for the difference between solutions of
the singular hyperbolic equation and the corresponding $\epsilon \searrow 0$ limit equation.

In \S  \ref{S:INITIALIZATION}, we construct $\epsilon$-dependent $1$-parameter families of initial data for the reduced conformal Einstein--Euler equations that satisfy the constraint equations on
the initial hypersurface $t=1$ and limit as $\epsilon \searrow 0$ to solutions of the conformal cosmological Poisson--Euler equations.

Using the results from \S \ref{S:FORMULATIONOFEQ} to \S \ref{S:INITIALIZATION}, we complete the proof of Theorem \ref{T:MAINTHEOREM} in \S \ref{S:MAINPROOF}.

\section{A singular symmetric hyperbolic formulation of the conformal Einstein--Euler equations}\label{S:FORMULATIONOFEQ}
In this section, we employ the variables \eqref{E:u.a}--\eqref{E:DELZETA} and the wave gauge \eqref{E:WAVEGAUGE} to cast the
conformal Einstein--Euler system, given by \eqref{E:EXPANSIONOFEIN} and \eqref{Confeul}, into a form that is suitable for analyzing the limit $\epsilon \searrow 0$ globally to the future. More specifically, we show that the Einstein--Euler system can be
written as a symmetric hyperbolic system that is jointly singular in $\epsilon$ and $t$, and for which the singular terms have a specific structure. Crucially, the $\epsilon$-dependent singular terms are of a form that has been well studied beginning with the pioneering work of Browning, Klainerman, Kreiss and Majda \cite{bro,kla2,kla1,kre1,sch1,sch2}, while
the $t$-dependent singular terms are of
the type analyzed in \cite{oli5}.

\subsection{Analysis of the FLRW solutions\label{FLRWanal}}
As a first step in the derivation,  we find explicit formulas for the functions
 $\Omega(t)$, $\rho_H(t)$ and $E(t)$ that will be needed to show that the transformed conformal Einstein--Euler systems is of
the form analyzed in  \S \ref{S:MODEL}. We begin by differentiating \eqref{E:OMEGADEF} and observe, with the help of \eqref{FLRW.c}, \eqref{E:DEFE} and \eqref{E:PTE}, that it satisfies
the differential equation
\begin{align}\label{E:RICCATI1}
  -t\partial_t(1-\Omega)+\frac{3}{2}(1+\epsilon^2K)(1-\Omega)^2=\frac{3}{2}(1+\epsilon^2K).
\end{align}
Integrating gives
\begin{align}\label{E:OMEGAREP}
  \Omega(t)=\frac{2t^{3(1+\epsilon^2K)}}{t^{3(1+\epsilon^2K)}-C_0},
\end{align}
where $C_0$ is as defined above by \eqref{C0def}.
Then by \eqref{E:OMEGADEF}, we find that
\begin{align}\label{E:RHOHOM}
 \rho_H(t)=\frac{4C_0\Lambda t^{3(1+\epsilon^2K)}}{(C_0-t^{3(1+\epsilon^2K)})^2},
\end{align}
which, in turn, shows that $\zeta_H(t)$, as defined by \eqref{E:ZETAH1}, is given by the formula \eqref{E:ZETAH3}.

It is clear from the above formulas that $\Omega$, $\rho$ and $\zeta_H$, as functions of $(t,\epsilon)$, are in
$C^2([0,1]\times [0,\epsilon_0])\cap W^{3,\infty}([0,1]\times [-\epsilon_0,\epsilon_0])$ for any fixed $\epsilon_0 > 0$.
In particular, we can represent $t^{-1}\Omega$ and $\del{t}\Omega$ as
\begin{equation*}\label{E:OMEGATIMEIDENTITY}
  \frac{1}{t}\Omega= E^{-1}\partial_tE=t^{2+3\epsilon^2K}\texttt{Q}_1(t) \qquad \text{and} \qquad \partial_t\Omega=  t^{2+3\epsilon^2K}\texttt{Q}_2(t),
\end{equation*}
respectively, where we are employing the notation from \S \ref{remainder} for the remainder terms
$\texttt{Q}_1$ and $\texttt{Q}_2$. 

Using \eqref{E:OMEGAREP},  we can integrate \eqref{E:PTE} to obtain
\begin{align}\label{E:EREP}
  E(t)=  \exp{\left(\int_1^t\frac{2s^{2+3\epsilon^2K}}{s^{3(1+\epsilon^2K)}-C_0}ds\right)}=
  \left(\frac{C_0-t^{3(1+\epsilon^2K)}}{C_0-1}\right)^{\frac{2}{3(1+\epsilon^2K)}} \geq 1
\end{align}
for $t\in [0,1]$. From this formula, it is clear that $E\in C^2([0,1]\times  [-\epsilon_0,\epsilon_0])\cap W^{3,\infty}([0,1]\times  [-\epsilon_0,\epsilon_0])$, and that the Newtonian limit of $E$, denoted $\ring{E}$ and defined by
\begin{equation*} \label{Eringdef}
\ring{E}(t) = \lim_{\epsilon\searrow 0} E(t),
\end{equation*}
is given by the formula \eqref{Eringform}. Similarly, we denote the Newtonian limit of $\Omega$ by
\begin{equation*} \label{Oringlim}
\mathring{\Omega}(t) = \lim_{\epsilon\searrow 0} \Omega(t),
\end{equation*}
which we see from \eqref{E:OMEGAREP} is given by the formula \eqref{Oringdef}.

For later use, we observe that $E$, $\Omega$, $\rho_H$ and $\zeta_H$ satisfy
\begin{align}
  -E^{-1}\partial_t^2 E+\frac{1}{t}E^{-1}\partial_t E=\frac{1}{2\Lambda t^2} (1+3\epsilon^2K) \rho_H, \label{E:FRI1}\\
  E^{-1}\partial^2_t E+2E^{-2}(\partial_t E)^2-\frac{5}{t} E^{-1}\partial_t E=\frac{3}{2\Lambda t^2}(1-\epsilon^2K)\rho_H \label{E:FRI2}
\end{align}
and
\begin{align}\label{E:PTZETAH}
  \partial_t \zeta_H =-\frac{3}{t}\Omega=-3E^{-1}\partial_tE=-\bar{\gamma}^i_{i0}=-\bar{\gamma}^i_{0i}
  =t^{2+3\epsilon^2K}\texttt{Q}_3(t)
\end{align}   
as can be verified by a straightforward calculation using the formulas \eqref{E:ZETAH3} and \eqref{E:OMEGAREP}-\eqref{E:EREP}. By \eqref{zetaHringform} and \eqref{Oringdef}, it is easy to verify
\al{PTZETAH2}{
\del{t}\mathring{\zeta}_H=-\frac{3}{t}\mathring{\Omega}=\frac{6t^2}{C_0-t^3}.
}
We also record the following useful expansions of $t^{1+3\epsilon^2 K}$, $E(\epsilon, t)$ and $\Omega(\epsilon, t)$:  %which will be employed in the following context repeatedly. By Taylor's theorem, $t^{1+3\epsilon^2 K}$ can be expanded as
\al{CHI1}{
	t^{1+3\epsilon^2 K}=t+ \epsilon^2 \mathcal{X}(\epsilon,t) \quad \text{where} \quad  \mathcal{X}(\epsilon,t)=\frac{6K}{\epsilon^2}\int^\epsilon_0 \lambda t^{1+3\lambda^2 K}\ln t d\lambda
}
and
\al{EOEXP}{
	E(\epsilon, t)=\mathring{E}(t)+\epsilon\texttt{E}(\epsilon, t) \AND \Omega(\epsilon, t)=\mathring{\Omega}(t)+\epsilon \texttt{A}(\epsilon, t)
}
for $(\epsilon,t)\in (0,\epsilon_0)\times (0,1]$, where $\mathcal{X}$, $\texttt{E}$ and $\texttt{A}$ are all remainder terms
as defined in \S \ref{remainder}.

\subsection{The reduced conformal Einstein equations}
The next step in transforming the conformal Einstein--Euler system is
to replace the conformal Einstein equations  \eqref{E:EXPANSIONOFEIN} with the gauge
reduced version given by
\begin{align}
  -2\bar{R}^{\mu\nu}+&2\bar{\nabla}^{(\mu}\bar{Z}^{\nu)}+\bar{A}_\sigma^{\mu\nu}\bar{Z}^\sigma=-4\bar{\nabla}^\mu
  \bar{\nabla}^\nu\Psi+4\bar{\nabla}^\mu\Psi\bar{\nabla}^\nu\Psi  \notag \\
&-2\biggl[\bar{\Box}\Psi
  +2|\bar{\nabla}\Psi|^2
 +\biggl(\frac{1-\epsilon^2K}{2}\bar{\rho}+\Lambda\biggr)e^{2\Psi}\biggr]\bar{g}
  ^{\mu\nu}-2e^{2\Psi}(1+\epsilon^2K)\bar{\rho} \bar{v}^\mu \bar{v}^\nu,  \label{E:REDUCEDEINSTEIN}
\end{align}
where
\begin{align*}
  \bar{A}_\sigma^{\mu\nu}:=-\bar{X}^{(\mu}\delta^{\nu)}_\sigma+\bar{Y}^{(\mu}\delta^{\nu)}_\sigma.
\end{align*}
We will refer to these equations as the \textit{reduced conformal Einstein equations}.

\begin{proposition} \label{wgprop}
If the wave gauge \eqref{E:WAVEGAUGE} is satisfied, $\Psi$ is chosen as \eqref{E:CONFORMALFACTOR} and $\bar{\gamma}^\nu$ is given by \eqref{E:GAMMAFLRW}, then the following relations hold:
  \begin{align*}
  \bar{\nabla}^{(\mu}\bar{\gamma}^{\nu)}= \bar{g}^{0(\mu}\delta^{\nu)}_0\frac{\Lambda}{t}
  \left(\partial_t\Omega-\frac{1}{t}\Omega\right)-\frac{\Lambda}{2t}\Omega \partial_t \bar{g}^{\mu\nu} ,  \\
  \bar{\Box}\Psi=  \frac{1}{t^2}\bar{g}^{00}-\frac{1}{t} \bar{Y}^0+ \frac{1}{t} \bar{\gamma}^0,\qquad
  |\bar{\nabla}\Psi|^2 =  \frac{1}{t^2}\bar{g}^{00},\\
  \bar{Y}^\mu \bar{Y}^\nu= 4\bar{\nabla}^\mu \Psi \bar{\nabla}^\nu \Psi +\frac{8\Lambda}{3t^2} \delta^{(\mu}_0\bar{g}^{\nu)0} + \frac{4\Lambda^2}{9t^2} \delta^\mu_0 \delta^\nu_0 \\
\intertext{and}
  \bar{\nabla}^{(\mu}\bar{Y}^{\nu)}=-2\bar{\nabla}^\mu \bar{\nabla}^\nu \Psi -\frac{2\Lambda}{3t^2} \bar{g}^{0(\mu}\delta^{\nu)}_0 - \frac{\Lambda}{3t}\bar{\partial}_t \bar{g}^{\mu\nu}.
\end{align*}
\end{proposition}
\begin{proof}
The proof follows from the formulas \eqref{E:CONFORMALFACTOR}, \eqref{E:GAMMAFLRW} and \eqref{Zdef}-\eqref{Ydef} via
straightforward computation.
\end{proof}
\begin{remark} \label{wgrem}
For the purposes of proving a priori estimates, we can always assume that the wave gauge \eqref{E:WAVEGAUGE} holds
since this gauge condition is known to propagate for solutions of the reduced Einstein--Euler equations assuming that the gravitational constraint
equations and the gauge constraint $\bar{Z}^\mu=0$ are satisfied on the initial hypersurface. The implication for our strategy of
obtaining global solutions to the future by extending local-in-time solutions via a continuation principle through
the use of a priori estimates is that we can assume that the wave gauge $\bar{Z}^\mu=0$
is satisfied, which, in particular, means that we can freely use the relations\footnote{In fact, the only
relation from Proposition \ref{wgprop} that relies on the gauge condition $\bar{Z}^\mu=0$ being satisfied is
$\bar{\Box}\Psi=  \frac{1}{t^2}\bar{g}^{00}-\frac{1}{t} \bar{Y}^0+ \frac{1}{t} \bar{\gamma}^0$.}
from Proposition \ref{wgprop} in the following.
\end{remark}

A short computation using the relations from Proposition \ref{wgprop} then shows that the reduced conformal Einstein Equations \eqref{E:REDUCEDEINSTEIN} can be written as
\begin{align}
  -2\bar{R}^{\mu\nu}+2\bar{\nabla}^{(\mu}\bar{X}^{\nu)}-\bar{X}^\mu\bar{X}^\nu
  +\frac{2\Lambda}{t}\Omega\bar{g}^{\mu\nu}
  =\frac{2\Lambda}{3t}\partial_t \bar{g}^{\mu\nu}-\frac{4\Lambda}{3t^2}\left(\bar{g}^{00}+\frac{\Lambda}{3}\right)\delta^\mu_0\delta^\nu_0
   -\frac{4\Lambda}{3t^2}\bar{g}^{0k}\delta^{(\mu}_0\delta^{\nu)}_k \hspace{1.5cm}& \notag\\
   -\frac{2}{t^2}\bar{g}^{\mu\nu}\left(\bar{g}^{00}+\frac{\Lambda}{3}\right)
   -(1-\epsilon^2K)\frac{\bar{\rho}}{t^2}\bar{g}^{\mu\nu}-2(1+\epsilon^2K)\frac{\bar{\rho}}{t^2} \bar{v}^\mu \bar{v}^\nu.&  \label{RedEin}
\end{align}
Recalling the formula (e.g., see \cite{fri, rin3})
\begin{align*} %\label{E:REPOFR}
  \bar{R}^{\mu\nu}=\frac{1}{2}\bar{g}^{\lambda\sigma}\bar{\partial}_\lambda\bar{\partial}_\sigma \bar{g}^{\mu\nu}+\bar{\nabla}^{(\mu}\bar{\Gamma}^{\nu)}+\frac{1}{2}(Q^{\mu\nu}-\bar{X}^\mu \bar{X}^\nu),
\end{align*}
where
\begin{equation}\label{E:QMUNU}
\begin{aligned}
  Q^{\mu\nu}=\bar{g}^{\lambda\sigma}\bar{\partial}_\lambda(\bar{g}^{\alpha\mu}\bar{g}^{\rho\nu})
  \bar{\partial}_\sigma\bar{g}_{\alpha\rho}
  +2\bar{g}^{\alpha\mu}\bar{\Gamma}^\eta_{\lambda\alpha}\bar{g}_{\eta\delta}\bar{g}^
  {\lambda\gamma}\bar{g}^{\rho\nu}\bar{\Gamma}^\delta_{\rho
  \gamma}+4\bar{\Gamma}^\lambda_{\delta\eta}
  \bar{g}^{\delta\gamma}\bar{g}_{\lambda(\alpha}\bar{\Gamma}^\eta_{\rho)
  \gamma}\bar{g}^{\alpha\mu}\bar{g}^{\rho\nu}
  +(\bar{\Gamma}^\mu
  -\bar{\gamma}^\mu)(\bar{\Gamma}^\nu-\bar{\gamma}^\nu),
\end{aligned}
\end{equation}
we can express \eqref{RedEin} as
\begin{align}
  -\bar{g}^{\lambda\sigma}\bar{\partial}_\lambda\bar{\partial}_\sigma\bar{g}^{\mu\nu}
  -2\bar{\nabla}^{(\mu}\bar{\gamma}^{\nu)}-Q^{\mu\nu}+\frac{2\Lambda}{t^2}\Omega\bar{g}^{\mu\nu}=\frac{2\Lambda}{3t}\partial_t \bar{g}^{\mu\nu}-\frac{4\Lambda}{3t^2}\left(\bar{g}^{00}+\frac{\Lambda}{3}\right)\delta^\mu_0
  \delta^\nu_0 \hspace{3.2cm}  \notag \\
   -\frac{4\Lambda}{3t^2}\bar{g}^{0k}\delta^{(\mu}_0\delta^{\nu)}_k
   -\frac{2}{t^2}\bar{g}^{\mu\nu}\left(\bar{g}^{00}+\frac{\Lambda}{3}\right)
   -(1-\epsilon^2K)\frac{\bar{\rho}}{t^2}\bar{g}^{\mu\nu}-2(1+\epsilon^2K)\frac{\bar{\rho}}{t^2} \bar{v}^\mu \bar{v}^\nu.
 \label{E:REDUCEDCONFEINSTEINEQ1}
\end{align}

By construction, the quadruple $\{\Psi,\bar{h}_{\mu\nu},\rho_H,\bar{v}_H^\mu\}$, see \eqref{FLRW.b}, \eqref{FLRW.c}, \eqref{E:CONFORMALFACTOR} and \eqref{E:CONFORMALFLRW}, is the conformal representation of a FLRW solution, and as such, it satisfies the conformal Einstein equations \eqref{E:CONFORMALEINSTEIN1} under the replacement
$\{\bar{g}_{\mu\nu},\bar{\rho},\bar{v}\}$ $\mapsto$ $\{\bar{h}_{\mu\nu},\rho_H,e^{\Psi}\tilde{v}_H^\mu\}$. Since $\bar{X}^\mu$ and $\bar{Y}^\mu$
vanish when $\bar{g}_{\mu\nu} \longmapsto \bar{h}_{\mu\nu}$, it is clear that the conformal Einstein equations \eqref{E:CONFORMALEINSTEIN1}
and the reduced conformal Einstein equations \eqref{E:REDUCEDCONFEINSTEINEQ1} coincide under the
replacement $\{\bar{g}_{\mu\nu},\bar{\rho},\bar{v}\}$ $\mapsto$ $\{\bar{h}_{\mu\nu},\rho_H,e^{\Psi}\tilde{v}_H^\mu\}$, and thus it
follows that $\bar{h}_{\mu\nu}$ satisfies
\begin{equation}\label{E:FRIFORM1}
\begin{aligned}
  -\bar{h}^{00}\bar{\partial}^2_0\bar{h}^{\mu\nu}
  -2\bar{\nabla}_H^{(\mu}\bar{\gamma}^{\nu)}-Q_H^{\mu\nu}+\frac{2\Lambda}{t^2}\Omega\bar{h}^{\mu\nu}=\frac{2\Lambda}{3t}\partial_t \bar{h}^{\mu\nu}
   -(1-\epsilon^2K)\frac{\rho_H}{t^2}\bar{h}^{\mu\nu}-2(1+\epsilon^2K)\frac{\rho_H}{t^2} \frac{\Lambda}{3} \delta_0^\mu \delta_0^\nu,
\end{aligned}
\end{equation}
where $\bar{\nabla}_H$ is the Levi-Civita connection of $\bar{h}_{\mu\nu}$,
\begin{align*}
  \bar{\nabla}_H^{(\mu}\bar{\gamma}^{\nu)}=\bar{h}^{0(\mu}\delta^{\nu)}_0\frac{\Lambda}{t}
  \left(\partial_t\Omega-\frac{1}{t}\Omega\right)-\frac{2\Lambda}{t}\Omega \del{t}\bar{h}^{\mu\nu}
\end{align*}
and
\begin{equation*}\label{E:QMUNUH}
\begin{aligned}
  Q_H^{\mu\nu}=& \bar{h}^{\lambda\sigma}\bar{\partial}_\lambda(\bar{h}^{\alpha\mu}\bar{h}
  ^{\rho\nu})\bar{\partial}_\sigma\bar{h}_{\alpha\rho}
  +2\bar{h}^{\alpha\mu}\bar{\gamma}^\eta_{\lambda\alpha}\bar{h}_{\eta\delta}
  \bar{h}^{\lambda\gamma}\bar{h}^{\rho\nu}\bar{\gamma}^\delta_{\rho
  \gamma}+4\bar{\gamma}^\lambda_{\delta\eta}
  \bar{h}^{\delta\gamma}\bar{h}_{\lambda(\alpha}\bar{\gamma}^\eta_{\rho)
  \gamma}\bar{h}^{\alpha\mu}\bar{h}^{\rho\nu}.
\end{aligned}
\end{equation*}
Using the formulas \eqref{E:HOMCHRIS} for the Christoffel symbols of $\bar{h}_{\mu\nu}$, it is not difficult to verify via
a routine calculation that independent components of the equation \eqref{E:FRIFORM1} agree up to scaling by a constant
with the equations \eqref{E:FRI1}-\eqref{E:FRI2}.

Setting $\nu=0$ and subtracting \eqref{E:FRIFORM1} from \eqref{E:REDUCEDCONFEINSTEINEQ1}, we obtain the equation
\begin{align}
 & -\bar{g}^{\lambda\sigma}\bar{\partial}_\lambda\bar{\partial}_\sigma(\bar{g}^{\mu0}
  -\bar{h}^{\mu0})
  -2(\bar{\nabla}^{(\mu}\bar{\gamma}^{0)}-\bar{\nabla}_H^{(\mu}\bar{\gamma}^{0)})
  -(Q^{\mu0}-Q_H^{\mu0})
+\frac{2\Lambda}{t^2}\Omega(\bar{g}^{\mu0}-\bar{h}^{\mu0})   \notag \\
 & \hspace{2.5cm} =\frac{2\Lambda}{3t}\partial_t (\bar{g}^{\mu0}
  -\bar{h}^{\mu0})-\frac{4\Lambda}{3t^2}\left(\bar{g}^{00} +\frac{\Lambda}{3}\right)\delta^\mu_0
  \delta^0_0
   -\frac{4\Lambda}{3t^2}\bar{g}^{0k}\delta^{(\mu}_0\delta^{0)}_k
   -\frac{2}{t^2}\bar{g}^{\mu0}\left(\bar{g}^{00}+\frac{\Lambda}{3}\right) \notag \\
&  \hspace{1.0cm} -(1-\epsilon^2K)\frac{1}{t^2}\Bigl[(\bar{\rho}-\rho_H)\bar{g}^{\mu0}
+\rho_H(\bar{g}^{\mu0}
   -\bar{h}^{\mu0})\Bigr]
   -2(1+\epsilon^2K)\frac{1}{t^2} \left[(\bar{\rho}-\rho_H)\bar{v}^\mu \bar{v}^0+\rho_H\left(\bar{v}^\mu \bar{v}^0-\frac{\Lambda}{3}\delta^\mu_0\right)\right]
\label{E:REDUCEDCONFEINSTEINEQ2}
\end{align}
for the difference $\bar{g}^{\mu0}-\bar{h}^{\mu 0}$. This equation is close to the form that we are seeking. The final step needed to
complete the transformation is to introduce a non-local modification which effectively subtracts out the contribution due to the Newtonian
potential.

For the spatial components, a more complicated transformation is required to bring those equations into the desired form. The first step is
to contract the $\mu=i, \nu=j$  components of \eqref{E:REDUCEDCONFEINSTEINEQ1} with $\check{g}_{ij}$, where we recall
that $(\check{g}_{kl})=(\bar{g}^{kl})^{-1}$ . A straightforward calculation,
using the identity $\check{g}_{kl}\bar{\partial}_\mu \bar{g}^{kl}=-3\alpha \bar{\partial}_\mu \alpha^{-1}$ (recall
$\alpha = \det(\bar{g}^{kl})$) and \eqref{E:REDUCEDCONFEINSTEINEQ2}
with $\mu=0$, shows that $\mathfrak{\bar{q}}$, defined previously by \eqref{E:q},
satisfies the equation
\begin{align}
 &-\bar{g}^{\lambda\sigma}\bar{\partial}_\lambda\bar{\partial}_\sigma\bar{\mathfrak{q}}
    -2\Lambda\bar{g}^{00}\frac{1}{t}\left(\partial_t\Omega-\frac{1}{t}\Omega\right)
    -2\bar{g}^{\lambda 0}
    \bar{\Gamma}^0_{\lambda 0}\frac{\Lambda}{t}\Omega+\frac{2\Lambda^2}{9t}\Omega\bar{\Gamma}^k_{i0}\delta^i_k
    +\frac{2\Lambda^2}{9t}
    \Omega\bar{g}^{0(i}\bar{\Gamma}^{j)}_{00}\check{g}_{ij} \notag\\
& \hspace{1.5cm}-\frac{2\Lambda}{3}\bar{g}^{00}
    \left(E^{-1}\partial^2_t E-E^{-2} (\partial_t E)^2\right)-Q+\frac{2\Lambda}{t^2}\Omega \left( \bar{g}^{00}-\frac{\Lambda}{3} \right)=
\frac{2\Lambda}{3t}\bar{\partial}_0 \bar{\mathfrak{q}}
+ \frac{4\Lambda^2}{9t}E^{-1}\partial_t E  \notag \\
&\hspace{2.5cm} -\frac{2}{t^2} \left(\bar{g}^{00}+\frac{\Lambda}{3}\right)^2-(1-\epsilon^2 K) \frac{\bar{\rho}}{t^2}\left( \bar{g}^{00}-\frac{\Lambda}{3} \right) - 2(1+\epsilon^2 K) \frac{\bar{\rho}}{t^2} \left(\bar{v}^0\bar{v}^0-\frac{\Lambda}{9}\check{g}_{ij}\bar{v}^i\bar{v}^j \right), \label{E:QEQ}
\end{align}
where
\begin{align}
  Q=Q^{00}+\frac{\Lambda}{9}\bar{g}^{\lambda\sigma}\bar{\partial}_\lambda\check{g}_{ij} \bar{\partial}_\sigma\bar{g}^{ij}-\frac{\Lambda}{9}\check{g}_{ij}Q^{ij} \notag.
\end{align}
Under the replacement $\{\bar{g}_{\mu\nu},\bar{\rho},\bar{v}\}$ $\mapsto$ $\{\bar{h}_{\mu\nu},\rho_H,e^{\Psi}\tilde{v}_H^\mu\}$,
equation \eqref{E:QEQ} becomes
\begin{align}
    &-2\Lambda\bar{h}^{00}\frac{1}{t}\left(\partial_t\Omega-\frac{1}{t}\Omega\right)
    -2\bar{h}^{\lambda 0}
    \bar{\gamma}^0_{\lambda 0}\frac{\Lambda}{t}\Omega+\frac{2\Lambda^2}{9t}\Omega\bar{\gamma}^k_{i0}\delta^i_k
    +\frac{2\Lambda^2}{9t}
    \Omega\bar{h}^{0(i}\bar{\gamma}^{j)}_{00}\check{h}_{ij}-\frac{2\Lambda}{3}\bar{h}^{00}
    \left(E^{-1}\partial^2_t E-E^{-2} (\partial_t E)^2\right)\notag \\
& \hspace{2.0cm}-Q_H+\frac{2\Lambda}{t^2}\Omega \left( \bar{h}^{00}-\frac{\Lambda}{3} \right) = \frac{4\Lambda^2}{9t}E^{-1}\partial_t E -(1-\epsilon^2 K) \frac{\rho_H}{t^2}\left( \bar{h}^{00}-\frac{\Lambda}{3} \right) - 2(1+\epsilon^2 K) \frac{\rho_H}{t^2} \frac{\Lambda}{3}, \label{E:QEQH}
\end{align}
where
\begin{align}
  Q_H=Q_H^{00}+\frac{\Lambda}{9}\bar{h}^{\lambda\sigma}\bar{\partial}_\lambda\check{h}_{ij} \bar{\partial}_\sigma\bar{h}^{ij}-\frac{\Lambda}{9}\check{h}_{ij}Q_H^{ij} \qquad \text{and} \qquad \check{h}_{ij}:=(\bar{h}^{kl})^{-1}=E^2\delta_{ij}, \notag
\end{align}
which, for the reasons discussed above, is satisfied by the conformal FRLW solution  $\{\Psi,\bar{h}_{\mu\nu},\rho_H,\bar{v}_H^\mu\}$.
Taking the difference between \eqref{E:QEQ} and \eqref{E:QEQH} yields the following equation for $\mathfrak{\bar{q}}$:
\begin{align}
    &-\bar{g}^{\lambda\sigma}\bar{\partial}_\lambda\bar{\partial}_\sigma\bar{\mathfrak{q}}
    -2\Lambda(\bar{g}^{00}-\bar{h}^{00})\frac{1}{t}\left(\partial_t\Omega-\frac{1}{t}\Omega\right)
    -2(\bar{g}^{\lambda 0}
    \bar{\Gamma}^0_{\lambda 0}-\bar{h}^{\lambda 0}
    \bar{\gamma}^0_{\lambda 0})\frac{\Lambda}{t}\Omega+\frac{2\Lambda^2}{9t}\Omega(\bar{\Gamma}^k_{i0}-\bar{\gamma}^k_{i0})
    \delta^i_k \notag \\
& \hspace{0.5cm} +\frac{2\Lambda^2}{9t}
    \Omega(\bar{g}^{0(i}\bar{\Gamma}^{j)}_{00}\check{g}_{ij}-
    \bar{h}^{0(i}\bar{\gamma}^{j)}_{00}\check{h}_{ij})-\frac{2\Lambda}{3}(\bar{g}^{00}-\bar{h}^{00})
    \left(E^{-1}\partial^2_t E-E^{-2} (\partial_t E)^2\right) +\frac{2\Lambda}{t^2}\Omega \left( \bar{g}^{00}-\bar{h}^{00} \right) \notag \\
    & \hspace{1.0cm} -(Q-Q_H) = \frac{2\Lambda}{3t}\bar{\partial}_0 \bar{\mathfrak{q}} -\frac{2}{t^2} \left(\bar{g}^{00}+\frac{\Lambda}{3}\right)^2-(1-\epsilon^2 K) \frac{1}{t^2}(\bar{\rho}-\rho_H)\left( \bar{g}^{00}-\frac{\Lambda}{3} \right)  -(1-\epsilon^2 K) \frac{1}{t^2}\rho_H\left( \bar{g}^{00}+\frac{\Lambda}{3} \right)  \notag \\
& \hspace{4.0cm}- 2(1+\epsilon^2 K) \frac{1}{t^2}\left[(\bar{\rho}-\rho_H) \left(\bar{v}^0\bar{v}^0-\frac{\Lambda}{9}\check{g}_{ij}\bar{v}^i\bar{v}^j \right)+\rho_H \left(\bar{v}^0\bar{v}^0-\frac{\Lambda}{3}-\frac{\Lambda}{9}\check{g}_{ij}\bar{v}^i\bar{v}^j \right)\right]. \label{E:REDUCEDCONFEINSTEINEQ3}
\end{align}

Next, denote
\begin{align*}%\label{E:TRACEFREEOP}
  \mathcal{L}^{ij}_{kl}=\delta^i_k\delta^j_l-\frac{1}{3}\check{g}_{kl}\bar{g}^{ij},
\end{align*}
and apply $\frac{1}{\alpha}\mathcal{L}^{ij}_{lm}$ to \eqref{E:REDUCEDCONFEINSTEINEQ1} with $\mu=l$, $\nu=m$. A calculation using the identities
\begin{align*}
  \alpha^{-1}\mathcal{L}^{ij}_{lm}\bar{\partial}_\sigma\bar{g}^{lm}=\bar{\partial}_\sigma
  \bar{\mathfrak{g}}^{ij} \qquad \text{and} \qquad \mathcal{L}^{ij}_{lm}\bar{g}^{lm}=0,
\end{align*}
where we recall that $\mathfrak{\bar{g}}^{ij}$ is defined by \eqref{E:GAMMA},
then shows that $\mathfrak{\bar{g}}^{ij}$ satisfies the equation
\begin{align}
    -\bar{g}^{\lambda\sigma}\bar{\partial}_\lambda\bar{\partial}_\sigma(\bar{\mathfrak{g}}^{ij}
    -\delta^{ij})
    -\frac{2}{\alpha}\mathcal{L}^{ij}_{lm}\bar{\nabla}^{(l}\bar{\gamma}^{m)}-\tilde{Q}^{ij}
    =\frac{2\Lambda}{3t}\partial_t(\bar{\mathfrak{g}}^{ij}-\delta^{ij})-\frac{2(1+\epsilon^2 K)}{\alpha}\frac{\bar{\rho}}{t^2} \mathcal{L}^{ij}_{lm}\bar{v}^l\bar{v}^m, \label{E:FRAKGIJ}
\end{align}
where
\begin{align*}
  \tilde{Q}^{ij}=-\bar{g}^{\lambda\sigma}\bar{\partial}_\lambda\biggl(\frac{1}{\alpha}\mathcal{L}^{ij}
  _{lm}\biggr) \bar{\partial}_\sigma\bar{g}^{lm}+\frac{1}{\alpha}\mathcal{L}^{ij}_{lm}Q^{lm}.
\end{align*}
Making the replacement $\{\bar{g}_{\mu\nu},\bar{\rho},\bar{v}\}$ $\mapsto$ $\{\bar{h}_{\mu\nu},\rho_H,e^{\Psi}\tilde{v}_H^\mu\}$,
equation \eqref{E:FRAKGIJ} becomes
\begin{align}\label{E:FRAKGIJH}
  -\frac{2}{\alpha_H}\mathcal{L}^{ij}_{lm,H}\bar{\nabla}_H^{(l}\bar{\gamma}^{m)}-\tilde{Q}^{ij}_H=0,
\end{align}
where
\begin{align*}
  \tilde{Q}_H^{ij}=-\bar{h}^{\lambda\sigma}\bar{\partial}_\lambda\biggl(\frac{1}{\alpha_H}\mathcal{L}^{ij}
  _{lm,H}\biggr) \bar{\partial}_\sigma\bar{h}^{lm}+\frac{1}{\alpha_H}\mathcal{L}^{ij}_{lm,H}Q^{lm}_H,
  \qquad \alpha_{H} =(\det{\check{h}_{ij}})^{-\frac{1}{3}}=E^{-2}
\end{align*}
and
\begin{align*}
  \mathcal{L}^{ij}_{kl,H}=\delta^i_k\delta^j_l-\frac{1}{3}\delta_{kl}\delta^{ij}.
\end{align*}
Subtracting \eqref{E:FRAKGIJ} by \eqref{E:FRAKGIJH} gives
\begin{align}
    -\bar{g}^{\lambda\sigma}\bar{\partial}_\lambda\bar{\partial}_\sigma(\bar{\mathfrak{g}}^{ij}
    -\delta^{ij})
    -2\biggl(\frac{1}{\alpha}\mathcal{L}^{ij}_{lm}\bar{\nabla}^{(l}\bar{\gamma}^{m)}
    -\frac{1}{\alpha_H}&\mathcal{L}^{ij}_{lm,H}\bar{\nabla}_H^{(l}\bar{\gamma}^{m)}\biggr)
    -(\tilde{Q}^{ij}-\tilde{Q}^{ij}_H) \notag \\
    &=\frac{2\Lambda}{3t}\partial_t(\bar{\mathfrak{g}}^{ij}-\delta^{ij})-
\frac{2(1+\epsilon^2 K)}{\alpha}\frac{\bar{\rho}}{t^2} \mathcal{L}^{ij}_{lm}\bar{v}^l\bar{v}^m. \label{E:REDUCEDCONFEINSTEINEQ4}
\end{align}

\subsection{$\epsilon$-expansions and remainder terms\label{epexpansions}} The next step in the transformation of the reduced conformal Einstein--Euler equations requires us to understand the lowest-order $\epsilon$-expansion for a number of quantities.  We compute and collect together these
expansions in this section. Throughout this section, we work in Newtonian coordinates, and we frequently employ the notation
\eqref{Neval} for evaluation in Newtonian coordinates, and the notation from
\S\ref{remainder} for remainder terms.

First, we observe, using \eqref{E:u.a}, \eqref{E:u.f} and \eqref{E:q}, that we can write $\alpha$
as
\begin{align}\label{E:SMALLGAMMA}
\underline{\alpha} =E^{-2}\exp{\biggl(\epsilon \frac{3}{\Lambda}\left(2tu^{00}-u\right)\biggr)}.
%\quad\text{and}\quad
%\gamma^{-1}=E^2\exp{\left(\epsilon \frac{3}{\Lambda}\left(u-2tu^{00}\right)\right)}
\end{align}
Using this, we can write the conformal metric as
\begin{equation}\label{E:GIJ}
   \underline{ \bar{g}^{ij}}
    = E^{-2}\delta^{ij}+\epsilon \Theta^{ij},
\end{equation}
where
\begin{equation} \label{Thetaijdef}
  \Theta^{ij} = \Theta^{ij}(\epsilon,t,u,u^{\mu\nu}):= \frac{1}{\epsilon}\left(\underline{\alpha}
    -E^{-2}\right)(\delta^{ij}+\epsilon u^{ij})+E^{-2}u^{ij},
\end{equation}
and $\Theta^{ij}$ satisfies  $\Theta^{ij}(\epsilon,t,0,0)=0$  %$\Theta^{ij} \in E^{1,0}((0,\epsilon_0)\times(0,T_0]\times \Rbb \times \mathbb{S}_4)$ (by chain rule to differentiating $\Theta^{ij}$ with respect to $\epsilon$ with $\mathfrak{T}$, $\chi$, $E$ and $\Omega$ as the intermediate variables and applying \eqref{E:CHI1}-\eqref{E:ERINGORINGIN2}, one can verify this. And from now on, the subsequential remainder terms which are in $E^{1,0}$ are based on the same reason and for the simplicity of the notations, we drop the domain of space $E^{1,0}$ by denoting it simply as $E^{1,0}$, whose domain is clear for the target variables)
and the $E^1$-regularity properties of a remainder term, see \S\ref{remainder}.
By the definition of $u^{0\mu}$, see \eqref{E:u.a}, we have that
\begin{align}\label{E:G0MU}
 \underline{\bar{g}^{0\mu}}=\bar{\eta}^{0\mu}+2\epsilon t u^{0\mu}
\end{align}
and, see \eqref{E:u.b} and \eqref{E:u.c},  for the derivatives
\begin{align}\label{E:PIG}
 \underline{\bar{\partial}_0\bar{g}^{0\mu}}=\epsilon(u^{0\mu}_0+3u^{0\mu}),
%\quad 2t\partial_tu^{0\mu}=u^{0\mu}_0+u^{0\mu}
\quad \text{and}\quad \underline{\bar{\partial}_i \bar{g}^{0\mu}}=2t\partial_i u^{0\mu}=\epsilon u^{0\mu}_i.
\end{align}
Additionally, by \eqref{E:SMALLGAMMA}, we see, with the help of \eqref{E:PTE}, \eqref{E:u.a}--\eqref{E:u.c}
and \eqref{E:u.f}--\eqref{E:u.g},  that
\begin{align*}
  \partial_t \underline{\alpha}=-2\underline{\alpha}\frac{1}{t}\Omega+\epsilon\underline{\alpha}\frac{3}{\Lambda}
  (3u^{00}+u^{00}_0-u_0)\quad \text{and}\quad \partial_i\underline{\alpha}=\epsilon^2\frac{3}{\Lambda}
\underline{\alpha}( u^{00}_i-u_i).
\end{align*}
Then differentiating \eqref{E:GIJ}, we find, using the above expression and \eqref{E:u.d}--\eqref{E:u.e}, that
\begin{align}
  \underline{\bar{\partial}_\sigma\bar{g}^{ij}}
%&-\frac{2}{t}\Omega\delta^0_\sigma
  %(E^{-2}\delta^{ij}+\epsilon \Theta^{ij})+\epsilon \underline{\alpha} \left[u^{ij}_\sigma+ \frac
  %{3}{\Lambda}(3u^{00}\delta^0_\sigma+u^{00}_\sigma-u_\sigma)(\delta^{ij}+\epsilon u^{ij})\right]
  =\bar{\partial}_\sigma\bar{h}^{ij}-\epsilon\frac{2}{t}\delta^0_\sigma\Omega \Theta^{ij}+\epsilon \underline{\alpha} \left[u^{ij}_\sigma+ \frac
  {3}{\Lambda}(3u^{00}\delta^0_\sigma+u^{00}_\sigma-u_\sigma)(\delta^{ij}+\epsilon u^{ij})\right] .  \label{E:PGIJ}
\end{align}

Since   $\check{g}_{ij}$ is, by definition, the inverse of $\bar{g}^{ij}$,  it follows from \eqref{E:GIJ} and Lemma \ref{E:EXPANSIONOFINVERSE2}
that we can express $\check{g}_{ij}$ as
\begin{align}\label{E:CHECKGIJ}
  \underline{\check{g}_{ij}}
  =E^2\delta_{ij}+\epsilon \texttt{S}_{ij}(\epsilon,t,u,u^{\mu\nu}),
\end{align}
where $\texttt{S}_{ij}(\epsilon,t,0,0)=0$. From \eqref{E:GIJ}, \eqref{E:G0MU} and  Lemma \ref{E:EXPANSIONOFINVERSE2}, we then see that
\begin{align}\label{E:G_MUNU}
  \underline{\bar{g}_{\mu\nu}}=\bar{h}_{\mu\nu}+\epsilon \Xi_{\mu\nu}(\epsilon,t, u^{\sigma\gamma}, u),
\end{align}
where $\Xi_{\mu\nu}$ satisfies $\Xi_{\mu\nu}(\epsilon,t, 0, 0)=0$ and the $E^1$-regularity properties of a remainder term.
Due to the identity
\begin{align}\label{E:IDENTITYPG}
  \bar{\partial}_\lambda\bar{g}_{\mu\nu}=-\bar{g}_{\mu\sigma}\bar{\partial}_\lambda
\bar{g}^{\sigma\gamma}
\bar{g}_{\gamma\nu}
\end{align}
we can easily derive from \eqref{E:PGIJ} and \eqref{E:G_MUNU} that
\begin{align}\label{E:PGMUNU}
  \underline{\bar{\partial}_\sigma \bar{g}_{\mu\nu}}=\bar{\partial}_\sigma \bar{h}_{\mu\nu}+
\epsilon \texttt{S}_{\mu\nu\sigma}(\epsilon,t, u^{\alpha\beta}, u, u^{\alpha\beta}_\gamma, u_\gamma),
\end{align}
where $\texttt{S}_{\mu\nu\sigma}(\epsilon,t, 0, 0, 0, 0)=0$, 
which in turn, implies that 
\begin{align}\label{E:CHRISTOFFEL}
  \underline{\bar{\Gamma}^\sigma_{\mu\nu}}-\bar{\gamma}^\sigma_{\mu\nu}&=\epsilon
\texttt{I}^\sigma_{\mu\nu}(\epsilon,t, u^{\alpha\beta}, u, u^{\alpha\beta}_\gamma, u_\gamma),
\end{align}
where $\texttt{I}^\sigma_{\mu\nu}( \epsilon,t, 0,0,0,0)=0$.
Later, we will also need the explicit form of the next order term in the $\epsilon$-expansion for $\bar{\Gamma}^i_{00}$. To compute this, we first observe that the expansions 
\begin{align}
	\underline{\bar{\partial}_0\bar{g}_{k0}}&=\epsilon \frac{3}{\Lambda} \delta_{ki}E^2 [u^{0i}_0+(3+4\Omega)u^{0i}]+\epsilon^2 \texttt{S}_{k00}(\epsilon,t, u^{\alpha\beta}, u, u^{\alpha\beta}_\gamma, u_\gamma)\label{E:DGEP1}
\intertext{and}
	\underline{\bar{\partial}_k\bar{g}_{00}}&=-\epsilon \left(\frac{3}{\Lambda} \right)^2 u^{00}_k+\epsilon^2 \texttt{S}_{00k}(\epsilon,t, u^{\alpha\beta}, u, u^{\alpha\beta}_\gamma, u_\gamma),\label{E:DGEP2}
\end{align}
where $\texttt{S}_{k00}(\epsilon,t,0,0,0,0)=\texttt{S}_{00k}(\epsilon,t,0,0,0,0)=0$, 
follow from \eqref{E:GIJ}, \eqref{E:G0MU}, \eqref{E:PGIJ}, \eqref{E:IDENTITYPG} and a straightforward calculation.
Using \eqref{E:DGEP1} and \eqref{E:DGEP2},  it is then not difficult to verify that
\begin{align} \label{E:GAMMAI00}
	\underline{\bar{\Gamma}^i_{00}}=&\epsilon \frac{3}{\Lambda}[u^{0i}_0+(3+4\Omega)u^{0i}]+\epsilon\frac{1}{2}\left(\frac{3}{\Lambda}\right)^2E^{-2}\delta^{ik}u^{00}_k
+\epsilon^2\texttt{I}^i_{00}(\epsilon,t, u^{\alpha\beta}, u, u^{\alpha\beta}_\gamma, u_\gamma),  \\
	\underline{\bar{\Gamma}^i_{i0}}-\underline{\bar{\gamma}^i_{i0}}=%&-\frac{1}{2}(\bar{g}_{kj}-\bar{h}_{kj})\del{t}\bar{h}^{kj}-\frac{1}{2}\bar{h}_{kj}\del{t}(\bar{g}^{kj}-\bar{h}^{kj})+\epsilon^2\mathcal{\check{I}}^i_{i0}(\epsilon,t, u^{\alpha\beta}, u, u^{\alpha\beta}_\gamma, u_\gamma),\\
	& \epsilon \Xi_{kj}E^{-2}\frac{\Omega}{t}\delta^{kj}-\epsilon \frac{1}{2} E^2\delta_{kj}\bigl[-\frac{2}{t}\Omega \Theta^{ij}+E^{-2}\bigl(u^{ij}_0+\frac{3}{\Lambda}(3u^{00} + u^{00}_0-u_0 )\delta^{ij}\bigr)\bigr]+\epsilon^2\texttt{I}^i_{i0}(\epsilon,t, u^{\alpha\beta}, u, u^{\alpha\beta}_\gamma, u_\gamma), \label{E:GSUBSTG}
\end{align}
where $\texttt{I}^i_{00}(\epsilon,t,0,0,0,0)=0$ and $\texttt{I}^i_{i0}(\epsilon,t,0,0,0,0)=0$. 

Continuing on, we observe from \eqref{E:ZETA} that we can express the proper energy density
in terms of $\zeta$  by
\begin{align}\label{E:ZETA2}
  \rho:= \underline{\bar{\rho}}= t^{3(1+\epsilon^2 K)}e^{(1+\epsilon^2 K) \zeta},
\end{align}
and correspondingly,  by \eqref{E:ZETAH1},
\begin{align}\label{E:ZETAH2}
  \rho_H= t^{3(1+\epsilon^2 K)}e^{(1+\epsilon^2 K) \zeta_H }
\end{align}
for the FLRW proper energy density.
From \eqref{E:DELZETA}, \eqref{E:ZETA2} and \eqref{E:ZETAH2}, it is then clear that we can express the difference between
the proper energy densities $\rho$ and $\rho_H$ in terms of $\delta\zeta$ by
\begin{align}\label{E:DELRHO} \delta \rho:=\rho-\rho_H=t^{3(1+\epsilon^2K)}e^{(1+\epsilon^2K)\zeta_H}
  \Bigl(e^{(1+\epsilon^2K)\delta\zeta}-1\Bigr).
\end{align}

Due to the normalization $\bar{v}^\mu\bar{v}_\mu=-1$, only three components of $\bar{v}_\mu$ are independent. Solving
$\bar{v}^\mu\bar{v}_\mu=-1$
for $\bar{v}_0$ in terms of the components $\bar{v}_{i}$, we obtain
\begin{align*}
  \bar{v}_0=\frac{-\bar{g}^{0i}\bar{v}_i+\sqrt{(\bar{g}^{0i}\bar{v}_i)^2-\bar{g}^{00}
  (\bar{g}^{ij}\bar{v}_i\bar{v}_j+1)}}{\bar{g}^{00}},
\end{align*}
which, in turn, using definitions \eqref{E:u.a}, \eqref{E:u.d}, \eqref{E:u.f}, \eqref{E:z.b}, we can write as
\begin{equation} \label{E:V_0}
\underline{\bar{v}_0}=-\frac{1}{\sqrt{-\underline{\bar{g}^{00}}}}+ \epsilon^2 \texttt{V}_2(\epsilon,t, u, u^{\mu\nu}, z _j),
\end{equation}
where  $\texttt{V}_2(\epsilon,t, u, u^{\mu\nu},0)=0$.
From this and the definition
$\bar{v}^0 =  \bar{g}^{0\mu}\bar{v}_\mu$,  we get
\begin{equation} \label{E:V^0}
\underline{\bar{v}^0}=\sqrt{-\underline{\bar{g}^{00}}}+  \epsilon^2
\texttt{W}_2(\epsilon,t, u, u^{\mu\nu}, z _j),
\end{equation}
where  $\texttt{W}_2(\epsilon,t, u, u^{\mu\nu},0)=0$. 
We also observe that
\begin{align}\label{E:VELOCITY}
 \underline{\bar{v}^k}=\epsilon (2tu^{0k}\underline{\bar{v}_0}+
\underline{\bar{g}^{ik}} z _i)\quad \text{and} \quad z ^k=2tu^{0k}\underline{\bar{v}_0}+\underline{\bar{g}^{ik}} z _i
\end{align}
follow immediately from definitions \eqref{E:z.b} and \eqref{E:z.a}.
For later use,
note that $z^k$ can also be written in terms $(\underline{\bar{g}^{\mu\nu}}, z_j)$ by
\begin{equation}\label{E:Z_IANDZ^I}
   z ^i=\underline{\bar{g}^{ij} }z _j+\frac{\underline{\bar{g}^{i0}}}{\underline{\bar{g}^{00}}}
\left[-\underline{\bar{g}^{0j}} z _j
  +\frac{1}{\epsilon}\sqrt{-\underline{\bar{g}^{00}}}\sqrt{1-\frac{1}{\underline{\bar{g}^{00}}}\epsilon^2(\underline{\bar{g}^{0j}}
   z _j)^2+\epsilon^2\underline{\bar{g}^{jk}} z _j z _k}\right].
\end{equation}

Using the above expansions, we are able to expand  $Q^{\mu\nu}$, $Q$ and $\tilde{Q}^{ij}$.
\begin{proposition}\label{T:Q}
  $Q^{\mu\nu}$, $Q$ and $\tilde{Q}^{ij}$ admit the following expansions:
  \begin{align*}
    Q^{\mu\nu}-Q_H^{\mu\nu}=\epsilon\mathcal{Q}^{\mu\nu}, \quad  % \label{E:Q1.a}\\
    Q-Q_H=\epsilon\mathcal{Q},  % \label{E:Q1.b}
\quad \text{and} \quad
   \tilde{Q}^{ij}-\tilde{Q}_H^{ij}=\epsilon\mathcal{\tilde{Q}}{}^{ij}, %\label{E:Q1.c}
  \end{align*}
where
  \begin{align*}
    \mathcal{Q}^{\mu\nu}=E^{-2}\frac{\Omega}{t}\mathcal{R}^{\mu\nu \gamma}(t) u_\gamma^{00}+\frac{\Omega}{t}\mathcal{R}^{\mu\nu}(t,\mathbf{u})+\epsilon \mathcal{S}^{\mu\nu}(\epsilon,t,u^{\alpha\beta},u,u^{\alpha\beta}_\sigma, u_\sigma), \\ % \label{E:Q.a}\\
    \mathcal{Q}=E^{-2}\frac{\Omega}{t}\mathcal{R}^{\gamma}(t) u_\gamma^{00}+ \frac{\Omega}{t}\mathcal{R}(t, \mathbf{u})+\epsilon \mathcal{S}(\epsilon,t,u^{\alpha\beta},u,u^{\alpha\beta}_\sigma, u_\sigma) ,\\ %\label{E:Q.b} \\
    \mathcal{\tilde{Q}}^{ij}=E^{-2}\frac{\Omega}{t}\mathcal{\tilde{R}}^{ij \gamma}(t) u_{\gamma}^{00}+\frac{\Omega}{t}\mathcal{\tilde{R}}{}^{ij}(t,\mathbf{u})+\epsilon \mathcal{\tilde{S}}{}^{ij}(\epsilon,t,u^{\alpha\beta},u,u^{\alpha\beta}_\sigma, u_\sigma), 
%\label{E:Q.c}
  \end{align*}
with\footnote{Here, in line with our conventions, see  \S \ref{remainder}, the quantities written with calligraphic letters, e.g., $\mathcal{S}$ and $\mathcal{R}$,
denote remainder terms, and consequently also satisfy the properties of remainder terms.}    $\mathbf{u}=(u^{\alpha\beta},u,u^{0i}_\sigma,u^{ij}_\sigma, u_\sigma)$,
 $\{\mathcal{R}^{\mu\nu},\mathcal{R},\mathcal{\tilde{R}}{}^{ij}\}$  linear in $\mathbf{u}$,
$\{\mathcal{R}^{\mu\nu \gamma}, \mathcal{R}^{\gamma},\mathcal{\tilde{R}}{}^{ij \gamma}\}$ satisfy
  \begin{align*}  % \label{E:PTREST}
    |\partial_t \mathcal{R}^{\mu\nu k}(t)|+ |\partial_t \mathcal{R}^{k}(t)|+ |\partial_t \mathcal{\tilde{R}}{}^{ij k}(t)|\lesssim t^{2},
  \end{align*}
and the terms $\{\mathcal{S}^{\mu\nu},\mathcal{S},\mathcal{\tilde{S}}{}^{ij}\}$ vanish
for $(\epsilon,t,u^{\alpha\beta},u,u^{\alpha\beta}_\sigma, u_\sigma)$$=$$(\epsilon,t,0,0,0,0)$.
\end{proposition}
\begin{proof}
First, we observe that we can write  $Q^{\mu\nu}$ as  $Q^{\mu\nu}=Q^{\mu\nu}(\underline{\bar{g}},\underline{\bar{\partial}\bar{g}})$,
where $Q^{\mu\nu}(\underline{\bar{g}},\underline{\bar{\partial}\bar{g}})$ is quadratic in $\underline{\bar{\partial}\bar{g}}= (\underline{\bar{\partial}_\gamma\bar{g}^{\mu\nu}})$
and analytic in $\underline{\bar{g}}=(\underline{\bar{g}^{\mu\nu}})$ on the region $\det(\underline{\bar{g}}) < 0$.
Since $Q_H^{\mu\nu} = Q^{\mu\nu}(\bar{h},\bar{\partial}\bar{h})$, we can expand
$Q^{\mu\nu}(\bar{h}+\epsilon\mathcal{S},\bar{\partial}\bar{h}+\epsilon\mathcal{T})$ to get
\begin{equation} \label{Qexp}
Q^{\mu\nu}(\bar{h}+\epsilon\mathcal{S},\bar{\partial}\bar{h}+\epsilon\mathcal{T})
- Q^{\mu\nu}_H =\epsilon DQ_1^{\mu\nu}(\bar{h},\bar{\partial}\bar{h}) \cdot \mathcal{S}
+ \epsilon DQ_2 ^{\mu\nu}(\bar{h},\bar{\partial}\bar{h}) \cdot \mathcal{T}  +
\epsilon^2 \mathcal{G}^{\mu\nu}\bigl(\epsilon,\bar{h},\bar{\partial}\bar{h},\mathcal{S},\mathcal{T}\bigr)
\end{equation}
where  $\mathcal{G}^{\mu\nu}$ is analytic in all variables and vanishes for $(\mathcal{S},\mathcal{T})=(0,0)$,
and $DQ_1^{\mu\nu}$ and $DQ_2^{\mu\nu}$ are linear in their second variable. By \eqref{E:GIJ}, \eqref{E:G0MU}, \eqref{E:PIG}
and \eqref{E:PGIJ}, we can choose
\begin{equation*}
\mathcal{S}=(\mathcal{S}^{\mu\nu}(\epsilon,t,u,u^{\alpha\beta}))
\quad \text{and} \quad
\mathcal{T}=(\mathcal{T}^{\mu\nu}_\gamma(\epsilon,t,u,u^{\alpha\beta},u_\sigma,u^{\alpha\beta}_\sigma))
\end{equation*}
for appropriate
remainder terms $\mathcal{S}^{\mu\nu}$ and $\mathcal{T}^{\mu\nu}$, so that
\begin{equation*}
\underline{\bar{g}} = \bar{h}+\epsilon\mathcal{S} \quad \text{and} \quad
\underline{\bar{\partial}\bar{g}}= \bar{\partial}\bar{h}+\epsilon\mathcal{T}.
\end{equation*}
\begin{comment}
  From this and \eqref{Qexp}, we obtain the expansion
\begin{equation} \label{QexpA}
Q^{\mu\nu}
- Q^{\mu\nu}_H =\epsilon DQ_1^{\mu\nu}(\bar{h},\bar{\partial}\bar{h}) \cdot \mathcal{S}
+ \epsilon DQ_2 ^{\mu\nu}(\bar{h},\bar{\partial}\bar{h}) \cdot \mathcal{T}  +
\epsilon^2 \mathcal{G}^{\mu\nu}\bigl(\epsilon,\bar{h},\bar{\partial}\bar{h},\mathcal{S},\mathcal{T}\bigr).
\end{equation}
\end{comment}
Using the fact that
\begin{equation*} \label{hbarrep}
\bar{h} = \biggl(-\frac{\Lambda}{3}\delta^\mu_0\delta^\nu_0 + \frac{1}{E^2}\delta^\mu_i\delta^\nu_j \delta^{ij}\biggr)
\quad \text{and} \quad \bar{\partial}\bar{h} = E^{-2}\frac{\Omega}{t} \mathfrak{h},
\end{equation*}
where $\mathfrak{h}= \bigl(-2\delta_\gamma^0\delta^\mu_i\delta^\nu_j)$, we can, using the linearity on
the second variable of the derivatives $DQ^{\mu\nu}_\ell$, $\ell=1,2$, write \eqref{Qexp} as
 \begin{equation} \label{QexpB}
Q^{\mu\nu}
- Q^{\mu\nu}_H =  \epsilon E^{-2}\frac{\Omega}{t} DQ_1^{\mu\nu}(\bar{h},\mathfrak{h}) \cdot \mathcal{S}
+   \epsilon E^{-2}\frac{\Omega}{t} DQ_2 ^{\mu\nu}(\bar{h},\mathfrak{h}) \cdot \mathcal{T}  +
\epsilon^2 \mathcal{G}^{\mu\nu}\bigl(\epsilon,\bar{h},\bar{\partial}\bar{h},\mathcal{S},\mathcal{T}\bigr).
\end{equation}

Expanding $\Theta^{ij}$, recall \eqref{Thetaijdef}, as
\begin{equation*}
\Theta^{ij} = \frac{1}{E^2}\biggl(\frac{3}{\Lambda}(2tu^{00}-u)\delta^{ij} + u^{ij}\biggr) + \epsilon \mathcal{A}^{ij}(\epsilon,t,u,u^{\mu\nu}),
\end{equation*}
where $\mathcal{A}^{ij}(\epsilon,t,0,0)=0$, we see from \eqref{E:GIJ}, \eqref{E:G0MU}, \eqref{E:PIG} and \eqref{E:PGIJ} that
\begin{equation} \label{gexp2}
\mathcal{S}^{\mu\nu} = \frac{1}{\epsilon}\bigl(\underline{\bar{g}^{\mu\nu}} - \bar{h}^{\mu\nu}\bigr) =  2 t \delta^\mu_0\delta^\nu_0 u^{00}
+ 4 t\delta^{(\mu}_0 \delta^{\nu)}_j u^{0j} + \delta^{\mu}_i\delta^\nu_j
\frac{1}{E^2}\biggl(\frac{3}{\Lambda}(2tu^{00}-u)\delta^{ij} + u^{ij}\biggr) + \epsilon
\mathcal{B}^{\mu\nu}(\epsilon,t,u,u^{\alpha\beta}),
\end{equation}
where $\mathcal{B}^{\mu\nu}(\epsilon,t,0,0)=0$, and
\begin{align}
\mathcal{T}^{\mu\nu}_\gamma &= \frac{1}{\epsilon}\bigl(\underline{\bar{\partial}_\gamma\bar{g}^{\mu\nu}} -
\bar{\partial}_\gamma \bar{h}^{\mu\nu}\bigr)  =
\delta^\mu_0\delta^\nu_0 \delta^0_\gamma (u^{00}_0+3u^{00}) + \delta^\mu_0\delta^\nu_0 \delta^k_\gamma
u^{00}_k + 2\delta^{(\mu}_0\delta^{\nu)}_j \delta^0_\gamma (u^{0j}_0+3u^{0j})
+ 2\delta^{(\mu}_0\delta^{\nu)}_j \delta^k_\gamma u_k^{0j} \notag \\
& -\frac{2}{t}\Omega \delta^{\mu}_i\delta^\nu_j \delta^0_\gamma \frac{1}{E^2}\biggl[\frac{3}{\Lambda}(2tu^{00}-u)\delta^{ij} + u^{ij} \biggr] +
\frac{1}{E^2} \left[u^{ij}_\sigma+ \frac
  {3}{\Lambda}(3u^{00}\delta^0_\sigma+u^{00}_\sigma-u_\sigma)\delta^{ij}\right] +\epsilon
\mathcal{C}^{\mu\nu}_\gamma(\epsilon,t,u,u^{\alpha\beta},u_\sigma,u^{\alpha\beta}_\sigma),
\label{gexp3}
\end{align}
where $\mathcal{C}^{\mu\nu}_\gamma(\epsilon,t,0,0,0,0)=0$. The stated expansion for $Q^{\mu\mu}$ is then an immediate
consequence of \eqref{QexpB}, \eqref{gexp2}, \eqref{gexp3} and the boundedness and regularity properties of $E$ and $\Omega$, see \S\ref{FLRWanal} for details. The expansions for $Q$ and $\tilde{Q}^{ij}$ can be established in a similar fashion.
\end{proof}

Finally, we collect the last $\epsilon$-expansions that will be needed in the following proposition. The proof follows from the same
arguments used to prove Proposition \ref{T:Q} above.

\begin{proposition}\label{T:CORRECTIONTERM}
The following expansions hold:
  \begin{align*}
   &\hspace{3.0cm} 2(\underline{\bar{\nabla}^{(\mu}\bar{\gamma}^{0)}}-\bar{\nabla}_H^{(\mu}\bar{\gamma}^{0)})
  -\frac{2\Lambda}{t^2}\Omega(\underline{\bar{g}^{\mu0}}-\bar{h}^{\mu0})=  \epsilon  \mathcal{E}^{\mu0},   \\
      &2(\underline{\bar{g}^{00}}-\bar{h}^{00})\frac{\Lambda}{t}\left(\partial_t\Omega-\frac{1}{t}\Omega\right)
    +2(\underline{\bar{g}^{\lambda 0}\bar{\Gamma}^0_{\lambda 0}}-\bar{h}^{\lambda 0}\bar{\gamma}^0_{\lambda 0})\frac{\Lambda}{t}\Omega-\frac{2\Lambda^2}{9t}\Omega(\underline{\bar{\Gamma}^k_{i0}}-\bar{\gamma}^k_{i0})
    \delta^i_k\\
    &-\frac{2\Lambda^2}{9t}
    \Omega(\underline{\bar{g}^{0(i}\bar{\Gamma}^{j)}_{00}\check{g}_{ij}}-
    \bar{h}^{0(i}\bar{\gamma}^{j)}_{00}\check{h}_{ij})+\frac{2\Lambda}{3}(\underline{\bar{g}^{00}}-
\bar{h}^{00})
    \left(E^{-1}\partial^2_t E-E^{-2} (\partial_t E)^2\right)
     -\frac{2\Lambda}{t^2}\Omega \left( \underline{\bar{g}^{00}}-\bar{h}^{00} \right)
     =  \epsilon   \mathcal{E}
\intertext{and}
    &\hspace{3.0cm} 2\left(\underline{|g|^{\frac{1}{3}}\mathcal{L}^{ij}_{lm}\bar{\nabla}^{(l}\bar{\gamma}^{m)}}
    -|h|^{\frac{1}{3}}\mathcal{L}^{ij}_{lm,H}\bar{\nabla}_H^{(l}\bar{\gamma}^{m)}\right)=  \epsilon   \mathcal{\tilde{E}}^{ij},
  \end{align*}
 where
  \begin{align*}
    \mathcal{E}^{\mu0}&=
    E^{-2} \frac{\Omega}{t} \mathcal{F}^{\mu0 \gamma}(t)u^{00}_\gamma+\frac{\Omega}{t} \mathcal{F}^{\mu0}(t, \mathbf{u})+\epsilon \mathcal{S}^{\mu0}(\epsilon,t,u^{\alpha\beta},u,u^{\alpha\beta}_\sigma, u_\sigma)\\
    \mathcal{E}&= E^{-2} \frac{\Omega}{t} \mathcal{F}^{\gamma}(t)u^{00}_\gamma +\frac{\Omega}{t}\mathcal{F}(t,\mathbf{u})+\epsilon \mathcal{S}(\epsilon,t,u^{\alpha\beta},u,u^{\alpha\beta}_\sigma, u_\sigma),\\
    \mathcal{\tilde{E}}^{ij}&= E^{-2} \frac{\Omega}{t} \mathcal{\tilde{F}}{}^{ij \gamma}(t)u_\gamma+\frac{\Omega}{t}\mathcal{\tilde{F}}{}^{ij}(t, \mathbf{u})+\epsilon \mathcal{\tilde{S}}{}^{ij}(\epsilon,t,u^{\alpha\beta},u,u^{\alpha\beta}_\sigma, u_\sigma),
  \end{align*}
with $\mathbf{u}=(u^{\alpha\beta},u,u^{0i}_\sigma,u^{ij}_\sigma, u_\sigma)$, $\{\mathcal{F}^{\mu 0 \gamma}, \mathcal{F}^{\gamma}, \mathcal{\tilde{F}}{}^{ij \gamma}\}$  satisfy
  \begin{equation*}
    |\partial_t \mathcal{F}^{\mu\nu \gamma}(t)| + |\partial_t \mathcal{F}^{\gamma}(t)| +
 |\partial_t \mathcal{\tilde{F}}{}^{ij \gamma}(t)|\lesssim t^{2},
  \end{equation*}
$\{\mathcal{F}^{\mu0}, \mathcal{F},\mathcal{\tilde{F}}{}^{ij}\}$ are linear in
$\mathbf{u}$, and the remainder terms $\{\mathcal{S}^{\mu0}, \mathcal{S},\mathcal{\tilde{S}}{}^{ij}\}$
vanish for
$(\epsilon,t,u^{\alpha\beta},u,u^{\alpha\beta}_\sigma, u_\sigma)$$=$$(\epsilon,t,0,0,0,0)$.
\end{proposition}

\subsection{Newtonian potential subtraction\label{Npotsub}}
Switching to Newtonian coordinates, a straightforward calculation, with the help of Propositions \ref{T:Q} and \ref{T:CORRECTIONTERM}, shows that the reduced conformal Einstein equations given by  \eqref{E:REDUCEDCONFEINSTEINEQ2},  \eqref{E:REDUCEDCONFEINSTEINEQ3} and \eqref{E:REDUCEDCONFEINSTEINEQ4}
can be written in first-order form using the variables \eqref{E:u.a}--\eqref{E:DELZETA} and \eqref{E:z.a} as follows:
\begin{align}
\tilde{B}^0\partial_0\begin{pmatrix}
u^{0\mu}_0\\u^{0\mu}_k\\u^{0\mu}
\end{pmatrix}+\tilde{B}^k\partial_k\begin{pmatrix}
u^{0\mu}_0\\u^{0\mu}_l\\u^{0\mu}
\end{pmatrix}+\frac{1}{\epsilon}\tilde{C}^k\partial_k\begin{pmatrix}
u^{0\mu}_0\\u^{0\mu}_l\\u^{0\mu}
\end{pmatrix}&=\frac{1}{t}\mathfrak{\tilde{B}}\mathbb{P}_2\begin{pmatrix}
u^{0\mu}_0\\u^{0\mu}_l\\u^{0\mu}
\end{pmatrix}+\hat{S}_1, \label{E:EIN1}\\
\tilde{B}^0\partial_0\begin{pmatrix}
u^{ij}_0\\u^{ij}_k\\u^{ij}
\end{pmatrix}+\tilde{B}^k\partial_k\begin{pmatrix}
u^{ij}_0\\u^{ij}_l\\u^{ij}
\end{pmatrix}+\frac{1}{\epsilon}\tilde{C}^k\partial_k\begin{pmatrix}
u^{ij}_0\\u^{ij}_l\\u^{ij}
\end{pmatrix}&=-\frac{2 E^2\underline{\bar{g}^{00}}}{t}\breve{\mathbb{P}}_2 \begin{pmatrix}
u^{ij}_0\\u^{ij}_l\\u^{ij}
\end{pmatrix}+ \hat{S}_2, \label{E:EIN2}\\
\tilde{B}^0\partial_0\begin{pmatrix}
u_0\\u_k\\u
\end{pmatrix}+\tilde{B}^k\partial_k\begin{pmatrix}
u_0\\u_l\\u
\end{pmatrix}+\frac{1}{\epsilon}\tilde{C}^k\partial_k\begin{pmatrix}
u_0\\u_l\\u
\end{pmatrix}&=-\frac{2 E^2 \underline{\bar{g}^{00}}}{t}\breve{\mathbb{P}}_2 \begin{pmatrix}
u_0\\u_l\\u
\end{pmatrix}+\hat{S}_3, \label{E:EIN3}
\end{align}
where
\begin{align}\label{E:EINBk}
\tilde{B}^0=E^2\begin{pmatrix}
-\underline{\bar{g}^{00}}& 0 & 0\\
0 & \underline{\bar{g}^{kl}} & 0 \\
0 & 0 & -\underline{\bar{g}^{00}}
\end{pmatrix},
\qquad
\tilde{B}^k=E^2\begin{pmatrix}
-4tu^{0k} & -\Theta^{kl} & 0\\
-\Theta^{kl} & 0 & 0 \\
0 & 0 & 0
\end{pmatrix},
\end{align}
\begin{align} \label{E:EINCk}
\tilde{C}^k=\begin{pmatrix}
0 & - \delta^{kl} & 0\\
- \delta^{kl} & 0 & 0 \\
0 & 0 & 0
\end{pmatrix},
\qquad
\mathfrak{\tilde{B}}=E^2\begin{pmatrix}
-\underline{\bar{g}^{00}} & 0 & 0\\
0 & \frac{3}{2}\underline{\bar{g}^{ki}} & 0 \\
0 & 0 & -\underline{\bar{g}^{00}}
\end{pmatrix},
\end{align}
\begin{align} \label{E:EINP2}
\mathbb{P}_2=\begin{pmatrix}
\frac{1}{2} & 0 & \frac{1}{2}\\
0 & \delta^l_i & 0 \\
\frac{1}{2} & 0 & \frac{1}{2}
\end{pmatrix},
\qquad
\breve{\mathbb{P}}_2=\begin{pmatrix}
1 & 0 & 0\\
0 & 0 & 0 \\
0 & 0 & 0
\end{pmatrix},
\end{align}
\begin{align} \label{E:EINS1}
\hat{S}_1=E^2\begin{pmatrix}
\mathcal{Q}^{0\mu}+ \mathcal{E}^{\mu 0}+4\epsilon u^{00}u^{0\mu}_0-4\epsilon u^{0\mu}u^{00}+6\epsilon u^{0k}u^{0\mu}_k -2u^{0\mu}(1-\epsilon^2K) t^{2+3\epsilon^2K} e^{(1+\epsilon^2K)(\zeta_H+\delta\zeta)}+\hat{f}^{0\mu}\\
0 \\
0
\end{pmatrix},
\end{align}
\begin{align}
     \hat{f}^{0\mu}=&-2(1+\epsilon^2 K) t^{1+3\epsilon^2 K}e^{(1+\epsilon^2 K)(\zeta_H+\delta\zeta)} \left(\frac{1}{\epsilon}\left(\underline{\bar{v}^0}-\sqrt{\frac{\Lambda}{3}}\right)\left(\underline{\bar{v}^0}+\sqrt{\frac{\Lambda}{3}}\right)\delta^\mu_0 +
\underline{\bar{v}^0} z^k\delta^\mu_k\right) \notag \\
&\hspace{4.3cm}-\epsilon K \Lambda  t^{1+3\epsilon^2 K} e^{(1+\epsilon^2 K)\zeta_H} (e^{(1+\epsilon^2 K)\delta \zeta}-1)\delta^\mu_0-\frac{1}{\epsilon}\frac{\Lambda}{3}\frac{1}{t^2}\delta^\mu_0\delta \rho, \label{E:EINfhat}
 \end{align}

\begin{align} \label{E:EINS2}
\hat{S}_2=E^2\begin{pmatrix}
\mathcal{\tilde{Q}}^{ij}+ \mathcal{\tilde{E}}^{ij}+4\epsilon u^{00}u^{ij}_0+\hat{f}^{ij}\\
0 \\
-\underline{\bar{g}^{00}}u^{ij}_0
\end{pmatrix},
\end{align}
\begin{equation} \label{E:EINS3}
\hat{S}_3=E^2
\begin{pmatrix}
\mathcal{Q}+ \mathcal{E}+4\epsilon u^{00}u_0-8\epsilon(u^{00})^2 + \hat{f}\\
0 \\
-\underline{\bar{g}^{00}}u_0
\end{pmatrix},
\end{equation}
\begin{align}
\hat{f}^{ij} = -2\epsilon(1+\epsilon^2K)\underline{\alpha}^{-1}
\underline{\mathcal{L}^{ij}_{kl}} t^{1+3\epsilon^2K} e^{(1+\epsilon^2K)\zeta}
z ^k z ^l, \label{E:fij}
\end{align}
and
\begin{align}
\hat{f} = & -\epsilon K \frac{4\Lambda}{3} t^{1+3\epsilon^2K} e^{(1+\epsilon^2K) \zeta_H}\bigr(e^{(1+\epsilon^2 K)\delta\zeta}-1\bigr)+2\epsilon(1+\epsilon^2K)
\frac{\Lambda}{9}\underline{\check{g}_{ij}} t^{1+3\epsilon^2K} e^{(1+\epsilon^2K)\zeta} z ^i z ^j  \nnb \\
&  -2(1-\epsilon^2 K)u^{00} t^{2+3\epsilon^2K} e^{(1+\epsilon^2K)(\zeta_H+\delta\zeta)} -2 (1+\epsilon^2 K) t^{1+3\epsilon^2K} e^{(1+\epsilon^2K)\zeta}\left(\underline{v^0}+\sqrt{\frac{\Lambda}{3}}\right)
 \frac{1}{\epsilon}\left(\underline{\bar{v}^0}-\sqrt{\frac{\Lambda}{3}}\right). \label{E:f}
\end{align}

At this point, it is important to stress that equations \eqref{E:EIN1}-\eqref{E:EIN3} are completely equivalent
to the reduced conformal Einstein equations for $\epsilon > 0$. Moreover, these equations are almost of the form that we need in order
to apply the results of \S\ref{S:MODEL}. Since the term $\frac{1}{\epsilon}\left(\underline{\bar{v}^0}-\sqrt{\frac{\Lambda}{3}}\right)$ is easily seen to be regular in $\epsilon$ from the expansion \eqref{E:V^0} ,  the only $\epsilon$-singular term left
is $-\frac{1}{\epsilon}\frac{\Lambda}{3}\frac{1}{t^2}E^2\delta\rho\delta^\mu_0$, which can be found in the quantity $\hat{f}^{0\mu}$. Following the method introduced in \cite{oli1} and then adapted to the cosmological
setting in \cite{oli3}, we can remove the singular part of this term while preserving
the required structure via the introduction of the shifted variable
 \begin{align}\label{E:WPHI}
  w^{0\mu}_k=u^{0\mu}_k-\delta^0_0\delta^\mu_0\partial_k \Phi,
\end{align}
where $\Phi$ is the potential defined by solving the Poisson equation
\begin{align}\label{E:DEFOFPHI}
  \triangle\Phi:=\frac{\Lambda}{3}\frac{1}{t^2}E^2
  \Pi \rho^{\frac{1}{1+\epsilon^2K}} =\frac{\Lambda}{3} E^2
   t e^{ \zeta_H}
  \Pi e^{ \delta\zeta} \qquad (\Delta = \delta^{ij}\partial_i\partial_j),
\end{align}
which, as we shall show, reduces to the (cosmological) Newtonian gravitational field equations in the limit $\epsilon \searrow 0$.
In this sense, we can view \eqref{E:WPHI} as the subtraction of the gradient of the Newtonian potential from the
gravitational field component $u^{00}_k$.
%Similarly, we also define
%\begin{align}\label{E:DEFOFPHINEWTON}
 % \triangle\mathring{\Phi}=\delta^{ij}\partial_i\partial_j\mathring{\Phi}
 % =\frac{\Lambda}{3}\frac{1}{t^2}E^2
 % \Pi\delta\mathring{\rho}
%\end{align}

Using \eqref{E:DEFOFPHI}, we can decompose $-\frac{1}{\epsilon}\frac{\Lambda}{3}\frac{1}{t^2}E^2\delta\rho\delta^\mu_0$ as
\begin{align} \label{phidecomp}
  -\frac{1}{\epsilon}\frac{\Lambda}{3}\frac{1}{t^2}E^2\delta\rho\delta^\mu_0=-\frac{1}{\epsilon}\delta^\mu_0 \Delta\Phi-\frac{\Lambda}{3}  \delta^\mu_0 E^2  t^{1+3\epsilon^2 K}\phi+\epsilon E^2 \mathcal{S}^\mu(\epsilon, t, \delta\zeta),
\end{align}
where
\begin{equation}\label{E:DEFOFPHISMALL}
     \phi:=\frac{1}{\epsilon}\left\langle 1, \frac{1}{t^{3(1+\epsilon^2 K)}}\delta \rho \right\rangle= \frac{1}{\epsilon}  (\mathds{1}-\Pi) e^{(1+\epsilon^2K)\zeta_H}\left(e^{(1+\epsilon^2K)\delta\zeta}-1\right)
\end{equation}
and
\als%{RHOREM}
{
	\mathcal{S}^\mu(\epsilon, t, \delta\zeta)%=\frac{\Lambda}{3t^2} \frac{1}{\epsilon^2}\Pi (\delta\varrho-\delta\rho)\delta^\mu_0= \frac{\Lambda}{3t^2}E^2\frac{1}{\epsilon^2}\Pi\bigl[ (\rho^{\frac{1}{1+\epsilon^2K}}-\rho)-(\rho_H^{\frac{1}{1+\epsilon^2K}}-\rho_H)\bigr]\\
	=\frac{\Lambda}{3} \frac{1}{\epsilon^2}\Pi\bigl[te^{\zeta_H}(e^{\delta\zeta}-1)-t^{1+3\epsilon^2K}e^{(1+\epsilon^2K)\zeta_H}(e^{(1+\epsilon^2K)\delta\zeta}-1)\bigr]\delta^\mu_0,
	}
which obviously satisfies $\mathcal{S}^\mu(\epsilon, t, 0)=0$.
Although it is not obvious at the moment, $\phi$ is regular in $\epsilon$, and consequently, with this knowledge, it is clear from
\eqref{phidecomp} that $-\frac{1}{\epsilon}\delta^\mu_0 \Delta\Phi$ is the only $\epsilon$-singular term
in $ -\frac{1}{\epsilon}\frac{\Lambda}{3}\frac{1}{t^2}E^2\delta\rho$.
A straightforward computation using \eqref{E:DEFOFPHI} and \eqref{phidecomp} along with the expansions from
Propositions  \ref{T:Q} and \ref{T:CORRECTIONTERM} then shows that
replacing  $u^{0\mu}_k$ in \eqref{E:EIN1} with the shifted variable \eqref{E:WPHI} removes
the $\epsilon$-singular term  $-\frac{1}{\epsilon}\delta^\mu_0 \Delta\Phi$ and yields the equation
\begin{align}\label{E:FINALEINSTEINEQUATIONS1}
\tilde{B}^0\partial_0\begin{pmatrix}
u^{0\mu}_0\\w^{0\mu}_k\\u^{0\mu}
\end{pmatrix}+\tilde{B}^k\partial_k\begin{pmatrix}
u^{0\mu}_0\\w^{0\mu}_l\\u^{0\mu}
\end{pmatrix}+\frac{1}{\epsilon}\tilde{C}^k\partial_k\begin{pmatrix}
u^{0\mu}_0\\w^{0\mu}_l\\u^{0\mu}
\end{pmatrix}=\frac{1}{t}\mathfrak{\tilde{B}}\mathbb{P}_2\begin{pmatrix}
u^{0\mu}_0\\w^{0\mu}_l\\u^{0\mu}
\end{pmatrix}+\tilde{G}_1+\tilde{S}_1,
\end{align}
where
\gat{	\tilde{G}_1=E^2\begin{pmatrix}
		-E^{-2}\frac{\Omega}{t} \bigl[\mathcal{D}^{0\mu0}(t) u^{00}_0+ \mathcal{D}^{0\mu k}(t) w^{00}_k \bigr] -\frac {\Omega}{t}\mathcal{D}^{0\mu}(t, \mathbf{u})+4\epsilon u^{00}u^{0\mu}_0-4\epsilon u^{0\mu}u^{00}%+6\epsilon u^{0k}u^{0\mu}_k
		+f^{0\mu}\\
		0 \\
		0
	\end{pmatrix},  \label{E:G1}\\
	\tilde{S}_1 =\p{  -\frac{\Omega}{t} \mathcal{D}^{0\mu k}\del{k}\Phi+\Theta^{kl}\delta^\mu_0\partial_k\partial_l
		\Phi\\
		\frac{3}{2}\frac{1}{t}\delta^\mu_0\underline{\bar{g}^{kl}}
		\partial_l\Phi-\underline{\bar{g}^{kl}}\delta^\mu_0\partial_0\partial_l\Phi \\
		0
	}+\epsilon \p{ \mathcal{S}^{\mu}(\epsilon,t,u^{\alpha\beta}, u, u^{\alpha\beta}_\sigma, u_\sigma) \\0 \\ 0 }, \label{E:S1}  \\
	 \mathcal{D}^{0\mu \nu}(t)=  - \mathcal{R}^{0\mu \nu}(t) - \mathcal{F}^{0\mu \nu}(t), \quad \mathcal{D}^{0\mu}(t,\mathbf{u}) = - \mathcal{R}^{0\mu}(t,\mathbf{u}) - \mathcal{F}^{0\mu}(t,\mathbf{u}) \label{E:D1}
	 }
\begin{comment}
\begin{equation*}
\tilde{S}_1=E^2\begin{pmatrix}
\mathcal{Q}^{0\mu}+ \mathcal{E}^{\mu 0}+4\epsilon u^{00}u^{0\mu}_0-4\epsilon u^{0\mu}u^{00}+6\epsilon u^{0k}u^{0\mu}_k+\Theta^{kl}\delta^\mu_0\partial_k\partial_l
\Phi+f^{0\mu}\\
\frac{3}{2}\frac{1}{t}\delta^\mu_0\underline{\bar{g}^{kl}}
\partial_l\Phi-\underline{\bar{g}^{kl}}\delta^\mu_0\partial_0\partial_l\Phi  \\
0
\end{pmatrix}
\end{equation*}
\end{comment}
and
\begin{align}
    f^{0\mu}=&-2(1+\epsilon^2 K) t^{1+3\epsilon^2 K}e^{(1+\epsilon^2 K)(\zeta_H+\delta\zeta)}
\left(\frac{1}{\epsilon}\left(\underline{\bar{v}^0}-\sqrt{\frac{\Lambda}{3}}\right)\left(\underline{\bar{v}^0}+\sqrt{\frac{\Lambda}{3}}\right)\delta^\mu_0 +
\underline{\bar{v}^0} z^k\delta^\mu_k\right) \notag \\
& -\Lambda\epsilon K  t^{1+3\epsilon^2 K} e^{(1+\epsilon^2 K)\zeta_H} (e^{(1+\epsilon^2 K)\delta \zeta}-1)\delta^\mu_0  -2u^{0\mu}(1-\epsilon^2K) t^{2+3\epsilon^2K} e^{(1+\epsilon^2K)(\zeta_H+\delta\zeta)} \notag \\ &- \frac{\Lambda}{3}  t^{1+3\epsilon^2 K} \phi \delta^\mu_0  +\epsilon\mathcal{S}^\mu(\epsilon, t, \delta\zeta). \label{E:f2}
\end{align}
With the help of the expansions from Propositions  \ref{T:Q} and \ref{T:CORRECTIONTERM}, we further decompose $\hat{S}_2$ and $\hat{S}_3$ into a sum of local and non-local terms given by
\al{S2S3}{
	\hat{S}_2=\tilde{G}_2+\tilde{S}_2 \AND \hat{S}_3=\tilde{G}_3+\tilde{S}_3,
	}
where
\begin{gather}
%\tilde{S}_2=E^2\begin{pmatrix}
% \mathcal{\tilde{Q}}^{ij}+ \mathcal{\tilde{E}}^{ij}+4\epsilon u^{00}u^{ij}_0-2\epsilon(1+\epsilon^2K)\underline{\alpha}^{-1}
%\underline{\mathcal{L}^{ij}_{kl}} t^{1+3\epsilon^2K} e^{(1+\epsilon^2K)\zeta}
% z ^k z ^l\\
%0 \\
%-\underline{\bar{g}^{00}}u^{ij}_0
%\end{pmatrix},\\
\tilde{G}_2=E^2\p{-E^{-2}\frac{\Omega}{t} \bigl[\mathcal{\tilde{D}}^{ij0}(t) u^{00}_0+ \mathcal{\tilde{D}}^{ijk}(t) w^{00}_k \bigr] -\frac {\Omega}{t}\mathcal{\tilde{D}}^{ij}(t, \mathbf{u})+4\epsilon u^{00}u^{ij}_0+\hat{f}^{ij }\\
	0 \\
	-\underline{\bar{g}^{00}}u^{ij}_0
	},  \label{E:G2}   \\
	\tilde{S}_2 =\p{-\frac{\Omega}{t} \mathcal{\tilde{D}}^{ijk}\del{k}\Phi \\ 0 \\ 0}+\epsilon \p{ \mathcal{S}^{ij}(\epsilon,t,u^{\alpha\beta}, u, u^{\alpha\beta}_\sigma, u_\sigma) \\0 \\ 0 }, \label{E:S2}
	\end{gather}
	\begin{gather}
	%\tilde{S}_3=E^2
	%\begin{pmatrix}
	% \mathcal{Q}+ \mathcal{E}+4\epsilon u^{00}u_0-8\epsilon(u^{00})^2 + \hat{f}\\
	%0 \\
	%-\underline{\bar{g}^{00}}u_0
	%\end{pmatrix}, \\
	\tilde{G}_3=E^2
	\begin{pmatrix}
	-E^{-2}\frac{\Omega}{t} \bigl[\mathcal{D}^{0}(t) u^{00}_0+ \mathcal{D}^{k}(t) w^{00}_k \bigr] -\frac {\Omega}{t}\mathcal{D}(t, \mathbf{u})+4\epsilon u^{00}u_0-8\epsilon(u^{00})^2  + \hat{f}\\
	0 \\
	-\underline{\bar{g}^{00}}u_0
	\end{pmatrix}, \label{E:G3} \\
	\tilde{S}_3=\p{-\frac{\Omega}{t} \mathcal{D}^{k}\del{k}\Phi \\ 0 \\ 0}+\epsilon \p{ \mathcal{S} (\epsilon,t,u^{\alpha\beta}, u, u^{\alpha\beta}_\sigma, u_\sigma) \\0 \\ 0 },  \label{E:S3} \\
	\mathcal{\tilde{D}}^{ij \nu}(t)= -\mathcal{\tilde{R}}^{ij \nu}(t)- \mathcal{\tilde{F}}^{ij \nu}(t),    \quad  	\mathcal{\tilde{D}}^{ij}(t,\mathbf{u})= -\mathcal{\tilde{R}}^{ij}(t,\mathbf{u})- \mathcal{\tilde{F}}^{ij}(t,\mathbf{u}), \label{E:D4} \\
	\mathcal{D}^\mu(t)= - \mathcal{R}^{\mu}(t)-\mathcal{F}^{\mu}(t) \quad \text{and} \quad \mathcal{D} (t,\mathbf{u})= - \mathcal{R} (t,\mathbf{u})-\mathcal{F} (t,\mathbf{u}). \label{E:D5}
	\end{gather}
Not only is the system of equations \eqref{E:EIN2} \eqref{E:EIN3} and \eqref{E:FINALEINSTEINEQUATIONS1} completely equivalent to the reduced conformal Einstein equations for any $\epsilon > 0$, but  it is now of the form required to apply the results from \S\ref{S:MODEL}.
This completes our transformation of the reduced conformal Einstein equations.

\subsection{The conformal Euler equations\label{conformalEul}}
With the transformation of the reduced conformal Einstein equations complete, we now turn to the problem of transforming
the conformal Euler equations.
We begin observing that it follows from the computations from \cite[\S2.2]{oli6}  that conformal Euler equations \eqref{Confeul} can be written in
Newtonian coordinates as
\begin{align}\label{E:FINALEULEREQUATIONS1}
  \bar{B}^0\partial_0\begin{pmatrix}
    \zeta\\
     z ^i
  \end{pmatrix}+
  \bar{B}^k\partial_k\begin{pmatrix}
    \zeta\\
     z ^i
  \end{pmatrix}=
  \frac{1}{t}\mathcal{\bar{B}}\hat{\mathbb{P}}_2\begin{pmatrix}
    \zeta\\
     z ^i
  \end{pmatrix}+\bar{S},
\end{align}
where
  \begin{align*}
  \bar{B}^0 &=\begin{pmatrix}
    1 & \epsilon\frac{L^0_i}{\underline{\bar{v}^0}}\\
    \epsilon \frac{L^0_j}{\underline{\bar{v}^0}} & K^{-1}M_{ij}
  \end{pmatrix} %\label{barB0}
  ,\\
  \bar{B}^k &=\begin{pmatrix}
    \frac{1}{\epsilon}\frac{\underline{\bar{v}^k}}{\underline{\bar{v}^0}} & \frac{L^k_i}{\underline{\bar{v}^0}}\\
    \frac{L^k_j}{\underline{\bar{v}^0}} & K^{-1}M_{ij}\frac{1}{\epsilon}\frac{\underline{\bar{v}^k}}{\underline{\bar{v}^0}}
  \end{pmatrix}
  =\begin{pmatrix}
    \frac{1}{\underline{\bar{v}^0}} z ^k & \frac{1}{\underline{\bar{v}^0}}\delta^k_i\\
    \frac{1}{\underline{\bar{v}^0}}\delta^k_j & K^{-1}\frac{1}{\underline{\bar{v}^0}}M_{ij} z ^k
  \end{pmatrix} %\label{barBk}
  ,\\
  \mathfrak{\bar{B}} &=\begin{pmatrix}
    1 & 0\\
    0 & -K^{-1}(1-3\epsilon^2K)\frac{\underline{\bar{g}_{ik}}}{\underline{\bar{v}_0}\underline{\bar{v}^0}}
  \end{pmatrix}, \label{barBfr}  \\
  \hat{\mathbb{P}}_2&=\begin{pmatrix}
    0 & 0\\
    0 & \delta^k_{j}
  \end{pmatrix}, %\label{hatP2}
  \\
  \bar{S}&=\begin{pmatrix}
    -L^\mu_i\underline{\bar{\Gamma}^i_{\mu\nu}}\,\underline{\bar{v}^\nu}\frac{1}{\underline{\bar{v}^0}} \\
    -K^{-1}(1-3\epsilon^2K)\frac{1}{\underline{\bar{v}_0}}\underline{\bar{g}_{0j}}
    -K^{-1}M_{ij}\underline{\bar{v}^\mu}\frac{1}{\epsilon}\underline{\bar{\Gamma}^i_{\mu\nu}} \,
\underline{\bar{v}^\nu}\frac{1}{\underline{\bar{v}^0}}
  \end{pmatrix}, %\label{barS}
  \\
  L^\mu_i&=\delta^\mu_i-\frac{\underline{\bar{v}_i}}{\underline{\bar{v}_0}}\delta^\mu_0 %\label{Lmudef}
\intertext{and}
  M_{ij}&=\underline{\bar{g}_{ij}}-\frac{\underline{\bar{v}_i}}{\underline{\bar{v}_0}}
\underline{\bar{g}_{0j}}-\frac{\underline{\bar{v}_j}}{\underline{\bar{v}_0}}
  \underline{\bar{g}_{0i}}
  +\frac{\underline{\bar{g}_{00}}}{(\underline{\bar{v}_0})^2}\underline{\bar{v}_i}\,\underline{\bar{v}_j}.
%\label{Mijdef}
%=\bar{g}_{ij}+\epsilon^2\mathcal{S}().
\end{align*}
In order to bring \eqref{E:FINALEULEREQUATIONS1} into the required form, we perform a change in variables from $z^i$ to $z_j$,
which are related via  a map of the form
$z^i=z^i(z_j,\underline{\bar{g}^{\mu\nu}})$, see \eqref{E:Z_IANDZ^I}.
Denoting the Jacobian of the transformation by
\begin{equation*}
  J^{im}:=\frac{\partial z ^i}{\partial z _m},
\end{equation*}
we observe that
\begin{equation*}
  \partial_\sigma z ^i=J^{im}\partial_\sigma z _m+\delta_\sigma^0 \frac{\partial z ^i}{\partial\underline{\bar{g}^{\mu\nu}}}
\underline{\bar{\partial}_0\bar{g}^{\mu\nu}} +\epsilon \delta_\sigma^j \frac{\partial z ^i}{\partial\underline{\bar{g}^{\mu\nu}}}
\underline{\bar{\partial}_j\bar{g}^{\mu\nu}}.
\end{equation*}
Multiplying \eqref{E:FINALEULEREQUATIONS1} by the block matrix $\diag{(1, J^{jl})}$ and changing variables from $(\zeta,z^i)$ to
$(\delta \zeta,z_j)$, where we recall from \eqref{E:DELZETA} that $\delta \zeta = \zeta - \zeta_H$, we can write the conformal Euler equations
\eqref{E:FINALEULEREQUATIONS1} as
\begin{align}\label{E:FINALEULEREQUATIONS}
  B^0\partial_0\begin{pmatrix}
    \delta\zeta\\
     z _m
  \end{pmatrix}+
  B^k\partial_k\begin{pmatrix}
    \delta\zeta\\
     z _m
  \end{pmatrix}=
  \frac{1}{t}\mathfrak{B}\hat{\mathbb{P}}_2\begin{pmatrix}
    \delta\zeta\\
     z _m
  \end{pmatrix}+\hat{S},
\end{align}
where
  \begin{align}
  B^0&=\begin{pmatrix}
    1 & \epsilon\frac{L^0_i}{\underline{\bar{v}^0}}J^{im}\\
    \epsilon \frac{L^0_j}{\underline{\bar{v}^0}}J^{jl} & K^{-1}M_{ij}J^{jl}J^{im}
\end{pmatrix}, \label{E:EULERB0}\\
  B^k&=\begin{pmatrix}
    \frac{1}{\underline{\bar{v}^0}} z ^k & \frac{1}{\underline{\bar{v}^0}}J^{km}\\
    \frac{1}{\underline{\bar{v}^0}}J^{kl}& K^{-1}\frac{1}{\underline{\bar{v}^0}}M_{ij}J^{jl}J^{im} z ^k
  \end{pmatrix}, \label{E:EULERBk} \\
  \mathfrak{B}&=\begin{pmatrix}
    1 & 0\\
    0 & -K^{-1}(1-3\epsilon^2K)\frac{1}{\underline{\bar{v}_0}\underline{\bar{v}^0}}J^{ml}
  \end{pmatrix} \label{E:EULERBfr}
\end{align}
and
\begin{align}
  \hat{S}=&\begin{pmatrix}
    -L^0_i\underline{\bar{\Gamma}^i_{00}}-L^\mu_i\underline{\bar{\Gamma}^i_{\mu j}}\,\underline{\bar{v}^j} \frac{1}{\underline{\bar{v}^0}}+(\bar{\gamma}^i_{i0}-\underline{\bar{\Gamma}^i_{i0}}) \\
    -K^{-1}J^{jl}M_{ij}\underline{\bar{v}^\mu}\frac{1}{\epsilon}\underline{\bar{\Gamma}^i_{\mu\nu}}
\,\underline{\bar{v}^\nu}\frac{1}{\underline{\bar{v}^0}}+\epsilon\frac{L^0_j}{\underline{\bar{v}^0}}J^{jl}\bar{\gamma}^i_{i0}
  \end{pmatrix}-\begin{pmatrix}
    \epsilon\frac{L^0_i}{\underline{\bar{v}^0}}\frac{\partial z ^i}{\partial\underline{\bar{g}^{\mu\nu}}}
    \underline{\bar{\partial_0}\bar{g}^{\mu\nu}}+\epsilon\frac{\delta^k_i}{\underline{\bar{v}^0}}\frac{\partial z ^i}
    {\partial\underline{\bar{g}^{\mu\nu}}}
    \underline{\bar{\partial}_k\bar{g}^{\mu\nu}}\\
    K^{-1}M_{ij}J^{jl}\frac{\partial z ^i}{\partial\underline{\bar{g}^{\mu\nu}}}
\underline{\bar{\partial}_0\bar{g}^{\mu\nu}}
    +\epsilon K^{-1}\bar{M}_{ij}\frac{ z ^k}{\underline{\bar{v}^0}}J^{jl}\frac{\partial z ^i}{\partial\underline{\bar{g}^{\mu\nu}}}\underline{\bar{\partial}_k\bar{g}^{\mu\nu}}
  \end{pmatrix}. \nonumber
\end{align}

By direct calculation, we see from \eqref{E:Z_IANDZ^I} and the expansions \eqref{E:GIJ} and \eqref{E:G0MU}
that 
\begin{equation}\label{E:JACOBI}
  J^{ik}
=E^{-2}\delta^{ik}+\epsilon \Theta^{ik}+\epsilon^2 \texttt{S}^{ik}(\epsilon,t,u,u^{\mu\nu},z_j),
\end{equation}
where  $\texttt{S}^{ik}(\epsilon,t,0,0,0)=0$. 
Similarly, it is not difficult to see from \eqref{E:Z_IANDZ^I} and
the expansions  \eqref{E:GIJ} and \eqref{E:G0MU}--\eqref{E:PGIJ} that
\begin{gather}
\delta_\sigma^0 \frac{\partial z ^i}{\partial\underline{\bar{g}^{\mu\nu}}}
\underline{\bar{\partial}_0\bar{g}^{\mu\nu}} +\epsilon \delta_\sigma^j \frac{\partial z ^i}{\partial\underline{\bar{g}^{\mu\nu}}}\underline{\bar{\partial}_j\bar{g}^{\mu\nu}}
= -2\delta^0_\sigma\biggl(E^{-2}\frac{\Omega}{t}z_j\delta^{ij} +\sqrt{\frac{3}{\Lambda}}(u^{0i}_0 + 3u^{0i})\biggr)
+ \epsilon \mathcal{S}^i(\epsilon,t,u,u^{\alpha\beta},u_\gamma,u^{\alpha\beta}_\gamma,z_j)
\label{E:JACOBI2}%\\
\intertext{and}
\epsilon  \frac{\delta^k_i}{\underline{\bar{v}^0}}\frac{\partial z^i}{\partial \bar{g}^{\mu\nu}}\bar{\partial}_k\bar{g}^{\mu\nu}=- \epsilon \frac{6}{\Lambda} u^{0i}_k\delta^k_i+\epsilon^2 \mathcal{S} (\epsilon,t,u,u^{\alpha\beta},u_\gamma,u^{\alpha\beta}_\gamma,z_j),
\end{gather}
where  $\mathcal{S}^i(\epsilon,t,0,0,0,0,0)=0$ and $\mathcal{S}(\epsilon,t,0,0,0,0,0)=0$. %and $\mathcal{S}^i$, $\mathcal{S} \in E^{1,0}$.
We further note that  the term $-K^{-1}J^{jl}M_{ij}\underline{\bar{v}^\mu}\frac{1}{\epsilon}\underline{\bar{\Gamma}^j_{\mu\nu}}
\underline{\bar{v}^\nu}\frac{1}{\underline{\bar{v}^0}}$ found in $\hat{S}$
above is not singular in $\epsilon$. This can be seen from the
expansions \eqref{E:GIJ}, \eqref{E:G0MU}, \eqref{E:G_MUNU} and \eqref{E:GAMMAI00}, which can be used to calculate
\begin{align}
  \frac{1}{\epsilon}\underline{\bar{\Gamma}^j_{\mu\nu}} \underline{v^\mu}\underline{v^\nu}= & 2\underline{\bar{\Gamma}^j_{0i}} \underline{v^0} z^i+ \epsilon\underline{\bar{\Gamma}^j_{ik}} z^i z^k+\frac{1}{\epsilon}\underline{\bar{\Gamma}^j_{00}} \underline{v^0} \underline{v^0} \nonumber\\
  =& \sqrt{\frac{\Lambda}{3}}\frac{2\Omega}{t}E^{-2} z_j\delta^{ij}+ \bigl[u^{0i}_0+(3+4\Omega)u^{0i}\bigr]+\frac{1}{2}\frac{3}{\Lambda} E^{-2}\delta^{ik}u^{00}_k + \epsilon \mathcal{S}^j(\epsilon,t,u,u^{\alpha\beta},u_\gamma,u^{\alpha\beta}_\gamma,z_j), \label{Ssingterm}
\end{align}
where  $\mathcal{S}^j(\epsilon,t,0,0,0,0,0)=0$. % and $\mathcal{S}^j \in E^{1,0}$.

Using the expansions \eqref{E:JACOBI}, \eqref{E:JACOBI2} and \eqref{Ssingterm} in conjunction
with \eqref{E:GIJ}, \eqref{E:G0MU}, \eqref{E:G_MUNU}, \eqref{E:CHRISTOFFEL}, \eqref{E:V_0}, \eqref{E:V^0} and \eqref{E:VELOCITY},
we can expand the matrices $\{B^0, B^k,\mathfrak{B}\}$ and source term $S$ defined above as follows:
  \begin{align}
  B^0%=\p{
%  	1 & 0 \\
%  	0 & K^{-1}\bar{g}^{lm}
%  }+\epsilon^2 \mathcal{S}^0(\epsilon,t,u,u^{\alpha\beta},u_\gamma,u^{\alpha\beta}_\gamma,z_j)  \nnb \\&
=&\p{
  	1 & 0 \\
  	0 & K^{-1}E^{-2}\delta^{lm}
  }+\epsilon \p{0 & 0 \\ 0 & K^{-1} \Theta^{lm}}+\epsilon^2 \texttt{S}^0(\epsilon,t,u,u^{\alpha\beta},u_\gamma,u^{\alpha\beta}_\gamma,z_j), \label{E:B0REMAINDER}\\
  B^k%&= \p{ (-\bar{g}^{00})^{-\frac{1}{2}}z^k  &  (-\bar{g}^{00})^{-\frac{1}{2}} \bar{g}^{km}  \\  (-\bar{g}^{00})^{-\frac{1}{2}} \bar{g}^{kl}  &   K^{-1} (-\bar{g}^{00})^{-\frac{1}{2}}\bar{g}^{lm} z^k} +\epsilon^2 \mathcal{S}^k(\epsilon,t,u,u^{\alpha\beta},u_\gamma,u^{\alpha\beta}_\gamma,z_j)  \nnb \\
  =&\sqrt{\frac{3}{\Lambda}}\p{
  	z^k & E^{-2} \delta^{km}\\
  	E^{-2}\delta^{kl} & K^{-1} E^{-2}\delta^{lm}z^k
  	}+\epsilon \sqrt{\frac{3}{\Lambda}} \p{\frac{3}{\Lambda}t u^{00} z^k & \Theta^{km}+\frac{3}{\Lambda} t u^{00}E^{-2} \delta^{km} \\ \Theta^{kl}+\frac{3}{\Lambda} t u^{00}E^{-2} \delta^{kl} & K^{-1}\bigl(\Theta^{lm}+\frac{3}{\Lambda} t u^{00} E^{-2} \delta^{lm}\bigr) z^k } \nnb \\&+\epsilon^2 \texttt{S}^k(\epsilon,t,u,u^{\alpha\beta},u_\gamma,u^{\alpha\beta}_\gamma,z_j), \label{E:BkREMAINDER}  \\
  \mathfrak{B}%&=    \p{
  	%1 & 0 \\
  	%0 & K^{-1} (1-3\epsilon^2 K) \bar{g}^{lm}
  	%}  +\epsilon^2 \mathcal{\tilde{S}}(\epsilon,t,u,u^{\alpha\beta},u_\gamma,u^{\alpha\beta}_\gamma,z_j),   \\
    =& \p{
  	1 & 0 \\
  	0 & K^{-1} (1-3\epsilon^2 K) E^{-2}\delta^{lm}
  	}+\epsilon \p{0 & 0 \\ 0 &  K^{-1}\Theta^{lm}} +\epsilon^2 \texttt{S}(\epsilon,t,u,u^{\alpha\beta},u_\gamma,u^{\alpha\beta}_\gamma,z_j)  \label{E:BCALREMAINDER}
\intertext{and}
  \hat{S}%=&\p{0 \\
   	%-K^{-1}\left[\sqrt{\frac{3}{\Lambda}}\bigl(-u^{0l}_0+(-3+4\Omega)u^{0l}\bigr)+\frac{1}{2}\left(\frac{3}{\Lambda}\right)^{\frac{3}{2}}E^{-2}\delta^{lk}u^{00}_k\right]	
    % }+\epsilon \mathcal{\hat{S}}(\epsilon,t,u,u^{\alpha\beta},u_\gamma,u^{\alpha\beta}_\gamma,z_j), \label{E:SREMAINDER}\\
  =&\p{0 \\
  	-K^{-1}\left[\sqrt{\frac{3}{\Lambda}}\bigl(-u^{0l}_0+(-3+4\Omega)u^{0l}\bigr)+\frac{1}{2}\left(\frac{3}{\Lambda}\right)^{\frac{3}{2}}E^{-2}\delta^{lk}u^{00}_k\right]	
  }  \nnb  \\
  &  +\epsilon \p{-\Xi_{kj}E^{-2}\frac{\Omega}{t}\delta^{kj}+\frac{1}{2} E^2\delta_{kj}\bigl[-\frac{2}{t}\Omega \Theta^{kj}+E^{-2}\bigl(u^{kj}_0+\frac{3}{\Lambda}(3u^{00} + u^{00}_0-u_0 )\delta^{kj}\bigr)\bigr]+\frac{6}{\Lambda} u^{0i}_k\delta^k_i  \\  \mathcal{S}_1(\epsilon,t,u,u^{\alpha\beta},u_\gamma,u^{\alpha\beta}_\gamma,z_j)}  \nnb \\ &    +\epsilon^2 \mathcal{S}_2(\epsilon,t,u,u^{\alpha\beta},u_\gamma,u^{\alpha\beta}_\gamma,z_j), \label{E:SREMAINDER}
\end{align}
where the remainder terms $\texttt{S}^0$, $\texttt{S}^k$, $\texttt{S}$, $\mathcal{S}_1$ and  $\mathcal{S}_2$ all vanish
for $(u,u^{\alpha\beta},u_\gamma,u^{\alpha\beta}_\gamma,z_j)=(0,0,0,0,0)$. 
We further decompose $\hat{S}$ into a local and non-local part by writing
\al{S}{\hat{S}=G+S,}
where
\al{G}{G=\p{0 \\
		-K^{-1}\left[\sqrt{\frac{3}{\Lambda}}\bigl(-u^{0l}_0+(-3+4\Omega)u^{0l}\bigr) +\frac{1}{2}\left(\frac{3}{\Lambda}\right)^{\frac{3}{2}}E^{-2}\delta^{lk}w^{00}_k\right]
	},	}
and
\al{S2a}{
	S= \p{
		0\\ -K^{-1}\frac{1}{2}\left(\frac{3}{\Lambda}\right)^{\frac{3}{2}}E^{-2}\delta^{lk} \del{k}\Phi
	} + \epsilon \mathcal{S}(\epsilon,t,u,u^{\alpha\beta},u_\gamma,u^{\alpha\beta}_\gamma,z_j).
	}

\subsection{The complete evolution system\label{completeevolution}}
To complete the transformation of reduced conformal Einstein--Euler equations, we need to treat $\phi$, defined by \eqref{E:DEFOFPHISMALL},
as an independent variable and derive an evolution equation for it.  To do so, we see from \eqref{E:FINALEULEREQUATIONS} that
we can write the time derivative of $\delta\zeta$ as
\begin{align} \label{dtzeta}
\partial_t\delta\zeta = & e_0 (B^0)^{-1}\biggl[-B^k\partial_k\begin{pmatrix}\delta\zeta \\ z_m \end{pmatrix}
+ \frac{1}{t}\mathfrak{B}\mathbb{\hat{P}}_2\begin{pmatrix}\delta\zeta \\ z_m \end{pmatrix} + \hat{S}\biggr]  \nnb \\ =&-\sqrt{\frac{\Lambda}{3}}\bigl(z^k\del{k}\delta\zeta+E^{-2}\delta^{km}\del{k}z_m\bigr)+\epsilon \mathcal{S}(\epsilon,t,u,u^{\alpha\beta},u_\gamma,u^{\alpha\beta}_\gamma,z_j),
\end{align}
where $e_0=(1,0,0,0)$ and $\mathcal{S}$ vanishes
for $(u,u^{\alpha\beta},u_\gamma,u^{\alpha\beta}_\gamma,z_j)=(0,0,0,0,0)$.
Noting that \eqref{E:PTZETAH} is equivalent to
\begin{align*}
  \frac{1}{1+\epsilon^2 K}\partial_t \rho_H=\frac{3}{t}\rho_H-\frac{3}{t}\Omega\rho_H,
\end{align*}
it follows directly from the definition of $\delta\zeta$, $\rho$ and $\delta\rho$, see \eqref{E:DELZETA}, \eqref{E:ZETA2} and
\eqref{E:DELRHO}, that
\begin{align}\label{E:PTZETA}
  \partial_t(\delta\zeta)=\frac{1}{1+\epsilon^2 K}\frac{1}{\rho}\partial_t(\delta\rho)+\frac{3}{t}(\Omega-1)\frac{\delta\rho}{\rho}
\end{align}
and
\begin{equation} \label{dkzeta}
  \partial_k (\delta\zeta)=\frac{1}{1+\epsilon^2 K}\frac{1}{\rho}\partial_k \rho.
\end{equation}
Using \eqref{E:PTZETA}, \eqref{dkzeta} and \eqref{E:B0REMAINDER}-\eqref{E:SREMAINDER}, we can write \eqref{dtzeta} as
\begin{align}
  \frac{1}{1+\epsilon^2 K}\partial_t (\delta\rho)+\frac{3}{t}(\Omega-1)\delta\rho+ \sqrt{\frac{3}{\Lambda}}\partial_k(\rho z^k)=-\epsilon \left(\frac{3}{\Lambda}\right)^{\frac{3}{2}}t^{4+3\epsilon^2 K} \del{k}\bigl( u^{00} e^{(1+\epsilon^2 K)\zeta} z^k\bigr)  \nnb \\
  +\epsilon t^{3(1+\epsilon^2 K)} \check{S}+ \epsilon^2 t^{3(1+\epsilon^2 K)}  \mathcal{S}(\epsilon,t,u,u^{\alpha\beta},u_\gamma,u^{\alpha\beta}_\gamma,z_j,\partial_k z_j, \delta\zeta,\partial_k\delta \zeta) ,
\label{E:PTDRHO}
\end{align}
where
\begin{align*}
\check{S}
= &  %+\sqrt{\frac{3}{\Lambda}}z_l E^{-2} \biggl( u^{kl}_k+\frac{3}{\Lambda}(u^{00}_k-u_k)\delta^{kl} \biggr)
\frac{1}{2} E^2\delta_{kj}\biggl[-\frac{2}{t}\Omega \Theta^{kj}+E^{-2}\bigl(u^{kj}_0+\frac{3}{\Lambda}(3u^{00} + u^{00}_0-u_0 )\delta^{kj}\bigr)\biggr] e^{(1+\epsilon^2 K)(\zeta_H+\delta\zeta)} +\frac{6}{\Lambda} u^{0i}_k\delta^k_i e^{(1+\epsilon^2 K)(\zeta_H+\delta\zeta)}  \nnb \\
&-\Xi_{kj} E^{-2}\frac{\Omega}{t}\delta^{kj} e^{(1+\epsilon^2 K)(\zeta_H+\delta\zeta)} %\label{Ycaldef}
\end{align*}
and the remainder term $\mathcal{S}$ satisfies  $\mathcal{S}(\epsilon,t,0,0,0,0,0,0,\delta\zeta,0)=0$. %and $\mathcal{S} \in E^{1,0}$.
 Taking the $L^2$ inner product of  \eqref{E:PTDRHO} with $1$ and
then multiplying by $1/(\epsilon t^{3{(1+\epsilon^2K)}})$, we obtain the desired evolution equation for $\phi$ given by
\begin{equation} \label{E:EQUATIONSOFPHI}
\partial_t \phi = \acute{G}+\acute{S} ,
\end{equation}
where
\begin{gather}
	\acute{G}=(1+\epsilon^2 K) \bigl\langle 1, \check{S} \bigr\rangle
	- \frac{3(1+\epsilon^2K)\Omega}{t}\phi \\
	\intertext{and}
	\acute{S} =\epsilon (1+\epsilon^2 K) \bigl\langle 1,  \mathcal{S}(\epsilon,t,u,u^{\alpha\beta},u_\gamma,u^{\alpha\beta}_\gamma,z_j,\partial_k z_j, \delta\zeta,\partial_k\delta \zeta)  \bigr\rangle. \label{Ycaltildef}
\end{gather}
%and $\acute{S} \in E^{1,0}$.

Next, we incorporate the shifted variable \eqref{E:WPHI}
into our set of gravitational variables by defining the vector quantity
\begin{align}\label{E:U1}
 \mathbf{U}_1=(u^{0\mu}_0, w^{0\mu}_k, u^{0\mu}, u^{ij}_0, u^{ij}_k, u^{ij}, u_0, u_k, u)^T,
\end{align}
and then combine this with the fluid variables and $\phi$ by defining
\begin{align}\label{E:REALVAR}
 \mathbf{U}=(\mathbf{U}_1, \mathbf{U}_2, \phi)^T,
\end{align}
where
\begin{align}\label{E:REALVAR1}
  \mathbf{U}_2=(\delta\zeta, z _i)^T.
\end{align}
Gathering \eqref{E:EIN2}, \eqref{E:EIN3}, \eqref{E:FINALEINSTEINEQUATIONS1}, \eqref{E:FINALEULEREQUATIONS} and \eqref{E:EQUATIONSOFPHI} together, we arrive at the following complete evolution equation for $\mathbf{U}$:
\begin{equation}\label{E:REALEQ}
  \begin{aligned}
    \mathbf{B}^0\partial_t \mathbf{U}+\mathbf{B}^i\partial_i \mathbf{U}+\frac{1}{\epsilon}\mathbf{C}^i\partial_i\mathbf{U}=\frac{1}{t}\mathbf{B}\mathbf{P}
    \mathbf{U}+\mathbf{H}+\mathbf{F},
  \end{aligned}
\end{equation}
where
\begin{align}
  \mathbf{B}^0=\begin{pmatrix}
    \tilde{B}^0 & 0 & 0 & 0 & 0 \\
    0 & \tilde{B}^0 & 0 & 0 & 0 \\
    0 & 0 & \tilde{B}^0 & 0 & 0 \\
    0 & 0 & 0 & B^0 & 0 \\
    0 & 0 & 0 & 0 & 1
  \end{pmatrix},
  \quad
  \mathbf{B}^i=\begin{pmatrix}
    \tilde{B}^i & 0 & 0 & 0 & 0 \\
    0 & \tilde{B}^i & 0 & 0 & 0 \\
    0 & 0 & \tilde{B}^i & 0 & 0 \\
    0 & 0 & 0 & B^i & 0 \\
    0 & 0 & 0 & 0 & 0
  \end{pmatrix},
  \quad
  \mathbf{C}^i=\begin{pmatrix}
    \tilde{C}^i & 0 & 0 & 0 & 0 \\
    0 & \tilde{C}^i & 0 & 0 & 0 \\
    0 & 0 & \tilde{C}^i & 0 & 0 \\
    0 & 0 & 0 & 0 & 0 \\
    0 & 0 & 0 & 0 & 0
  \end{pmatrix}, \label{E:REALEQa}
\end{align}

\begin{align}
  \mathbf{B}=\begin{pmatrix}
    \mathfrak{\tilde{B}} & 0 & 0 & 0 & 0 \\
    0 & -2E^2\underline{\bar{g}^{00}}I & 0 & 0 & 0 \\
    0 & 0 & -2E^2\underline{\bar{g}^{00}}I & 0 & 0 \\
    0 & 0 & 0 & \mathfrak{B} & 0 \\
    0 & 0 & 0 & 0 & 1
  \end{pmatrix},
  \quad
  \mathbf{P}=\begin{pmatrix}
    \mathbb{P}_2 & 0 & 0 & 0 & 0 \\
    0 & \mathbb{\breve{P}}_2 & 0 & 0 & 0 \\
    0 & 0 & \mathbb{\breve{P}}_2 & 0 & 0 \\
    0 & 0 & 0 & \hat{\mathbb{P}}_2 & 0 \\
    0 & 0 & 0 & 0 & 0
  \end{pmatrix}, \label{E:REALEQb}
\end{align}
\begin{align}
  \mathbf{H}%=\mathbf{H}(\epsilon,t,\mathbf{U})
  =(\tilde{G}_1,
    \tilde{G}_2,
    \tilde{G}_3,
    G, \acute{G} ) ^T  \AND
\mathbf{F}%=\mathbf{F}(\epsilon,t,\mathbf{U}, \del{k}\Phi, \del{l}\del{k}\Phi, \del{t}\del{k}\Phi)
    =(\tilde{S}_1,
    \tilde{S}_2,
    \tilde{S}_3,
    S, \acute{S} ) ^T . \label{E:REALEQc}
\end{align}

The importance of equation \eqref{E:REALEQ} is threefold. First, solutions of the reduced conformal Einstein--Euler equations determine solutions of
\eqref{E:REALEQ} as we shall show in the following section. Second, equation \eqref{E:REALEQ} is of the required form so that the a priori estimates from \S\ref{S:MODEL} apply to its solutions. Finally, estimates for solutions of \eqref{E:REALEQ} that are determined from solutions of the reduced conformal Einstein--Euler equations imply estimates for solutions of the reduced conformal Einstein--Euler equations. In this way, we are able  to use the evolution equation \eqref{E:REALEQ} in conjunction with the a priori estimates from \S\ref{S:MODEL} to establish, for appropriate small data, the global existence of 1-parameter families of $\epsilon$-dependent solutions to the conformal Einstein--Euler equations that exist globally to the future and converge in the limit $\epsilon \searrow 0$
to solutions of the cosmological conformal Poisson-Euler equations of Newtonian gravity.

\section{Reduced conformal Einstein--Euler equations: local existence and continuation\label{EEcont}}

In this section, we consider the local-in-time existence and uniqueness of solutions to the reduced Einstein--Euler
equations and discuss how these solutions determine solutions of \eqref{E:REALEQ}. Furthermore, we
establish a continuation principle for the Einstein--Euler equations which is based on bounding the
$H^s$ norm of $\mathbf{U}$ for $s \in \Zbb_{\geq 3}$.

\begin{proposition} \label{rcEEexist}
Suppose $s\in \Zbb_{\geq 3}$, $\epsilon_0>0$, $\epsilon \in (0,\epsilon_0)$, $T_0 \in (0,1]$,  $(\bar{g}^{\mu\nu}_0)$ $\in$ $H^{s+1}(\mathbb{T}_\epsilon^3,\mathbb{S}_4)$,
and $(\bar{g}^{\mu\nu}_1)$ $\in$ $H^{s}(\mathbb{T}_\epsilon^3,\mathbb{S}_4)$,
 $(\bar{v}^\mu_0)$ $\in$  $H^{s}(\mathbb{T}_\epsilon^3,\mathbb{R}^4)$ and $\bar{\rho}_0$ $\in$ $H^{s}(\mathbb{T}_\epsilon^3)$,
where $\bar{v}^\mu_0$ is normalized by $\bar{g}_{0\mu\nu}\bar{v}^{\mu}_0\bar{v}^\nu_0=-1$, and $\det(\bar{g}^{\mu\nu}_0) < 0$
and $\bar{\rho}_0>0$ on
$\mathbb{T}^3_\epsilon$. Then there exists a $T_1 \in (0,T_0]$ and
a unique classical solution
\begin{equation*}
(\bar{g}^{\mu\nu},\bar{v}^{\mu},\bar{\rho})\in \bigcap_{\ell=0}^2 C^{\ell}((T_1,T_0],H^{s+1-\ell}(\mathbb{T}^3_\epsilon))
\times \bigcap_{\ell=0}^1  C^{\ell}((T_1,T_0],H^{s-\ell}(\mathbb{T}_\epsilon^3)) \times \bigcap_{\ell=0}^1  C^{\ell}((T_1,T_0],H^{s-\ell}(\mathbb{T}_\epsilon^3)),
\end{equation*}
of the reduced conformal Einstein--Euler equations, given by \eqref{Confeul} and \eqref{RedEin}, on the spacetime region
$(T_1,T_0]\times \mathbb{T}^3_\epsilon$ that satisfies
\begin{equation*}
(\bar{g}^{\mu\nu},\bar{\partial_0}\bar{g}^{\mu\nu},\bar{v}^\mu,\bar{\rho})|_{t=T_0} =( \bar{g}^{\mu\nu}_0,\bar{g}^{\mu\nu}_1,\bar{v}^\mu_0,\bar{\rho}_0).
\end{equation*}
 Morover,
\begin{itemize}
\item[(i)] there exists a unique $\Phi \in \bigcap_{\ell=0}^1  C^{\ell}((T_1,T_0],\bar{H}^{s+2-\ell}(\mathbb{T}^3))$ that solves
equation \eqref{E:DEFOFPHI},
\item[(ii)] the vector $\mathbf{U}$, see \eqref{E:REALVAR}, is well-defined, lies in the space
\begin{equation*}
\mathbf{U} \in  \bigcap_{\ell=0}^1  C^{\ell}((T_1,T_0],H^{s-\ell}(\mathbb{T}^3,\mathbb{V})),
\end{equation*}
where
\begin{equation*}
\mathbb{V} = \mathbb{R}^4\times\mathbb{R}^{12}\times \mathbb{R}^4\times\mathbb{S}_3\times  (\mathbb{S}_3)^3\times\mathbb{S}_3\times \mathbb{R}
\times\mathbb{R}^3\times\mathbb{R} \times \mathbb{R}\times\mathbb{R}^3\times\mathbb{R},
\end{equation*}
and solves \eqref{E:REALEQ} on the spacetime region $(T_1,T_0]\times \mathbb{T}^3$, and
\item[(iii)] there exists a constant $\sigma > 0$, independent of $\epsilon\in (0,\epsilon_0)$ and $T_1\in (0,T_0)$, such that if $\mathbf{U}$ satisfies
\begin{equation*}
\norm{\mathbf{U}}_{L^\infty((T_1,T_0],H^s(\mathbb{T}^3))} < \sigma,
\end{equation*}
then the solution $(\bar{g}^{\mu\nu},\bar{v}^{\mu},\bar{\rho})$ can be uniquely continued as a classical solution
with the same regularity
to the larger spacetime region $(T^*_1,T_0]\times \mathbb{T}^3_\epsilon$ for some $T^*_1 \in (0,T_1)$.
\end{itemize}
\end{proposition}

\begin{proof}
We begin by noting that the
reduced conformal Einstein--Euler equations are well defined as long as the conformal metric $\bar{g}^{\mu\nu}$ remains non-degenerate and the conformal fluid four-velocity remains future directed, that is,
\begin{equation} \label{welldef}
\det(\bar{g}^{\mu\nu})< 0 \AND \bar{v}^0 < 0.
\end{equation}
Since it is well known that the reduced Einstein--Euler equations can be written as a symmetric
hyperbolic system\footnote{This follows from writing the wave equation \eqref{E:REDUCEDCONFEINSTEINEQ1} in first order form and using one of the various methods for expressing  the relativistic conformal Euler equations as a symmetric hyperbolic system. One particular way of writing the conformal Euler equations in symmetric hyperbolic form is given in \S\ref{conformalEul} which is a variation of the method introduced by Rendall in \cite{Rendall:1992}. For other elegant approaches, see
\cite{braKarp,frauen,walton}.} provided that $\rho$ remains strictly positive, we obtain from standard local existence and continuation results for
symmetric hyperbolic systems, e.g. Theorems 2.1 and 2.2 of \cite{maj}, the  existence of a
unique local-in-time classical solution
\begin{equation}\label{locreg}
(\bar{g}^{\mu\nu},\bar{v}^{\mu},\bar{\rho})\in \bigcap_{\ell=0}^2 C^{\ell}((T_1,T_0],H^{s+1-\ell}(\mathbb{T}^3_\epsilon))
\times \bigcap_{\ell=0}^1  C^{\ell}((T_1,T_0],H^{s-\ell}(\mathbb{T}_\epsilon^3)) \times \bigcap_{\ell=0}^1  C^{\ell}((T_1,T_0],H^{s-\ell}(\mathbb{T}_\epsilon^3))
\end{equation}
of the reduced conformal Einstein--Euler equations, for some time $T_1\in (0,T_0)$, that
satisfies
\begin{equation*}
(\bar{g}^{\mu\nu},\bar{\partial_0}\bar{g}^{\mu\nu},\bar{v}^\mu,\bar{\rho})|_{t=T_0} =( \bar{g}^{\mu\nu}_0,\bar{g}^{\mu\nu}_1,\bar{v}^\mu_0,\bar{\rho}_0)
\end{equation*}
for given initial data
\begin{equation*}
(\bar{g}^{\mu\nu},\bar{\partial_0}\bar{g}^{\mu\nu},\bar{v}^\mu,\bar{\rho}_0)
\in H^{s+1}(\mathbb{T}_\epsilon^3,\mathbb{S}_4)\times H^{s}(\mathbb{T}_\epsilon^3,\mathbb{S}_4)
\times H^{s}(\mathbb{T}_\epsilon^3,\mathbb{R}^4)\times H^s(\mathbb{T}_\epsilon^3)
\end{equation*}
satisfying \eqref{welldef} and $\bar{\rho}_0>0$ on the initial hypersurface $t=T_0$. Moreover,
if the solution satisfies
\begin{gather}
\det(\bar{g}^{\mu\nu}(\bar{x}^\gamma))\leq c_1< 0, \quad \bar{v}^0(\bar{x}^\gamma) \leq c_2 < 0 \label{contA}
\intertext{and}
\bar{\rho}(\bar{x}^\gamma) \geq c_3 > 0 \notag
\end{gather}
for all $(\bar{x}^\gamma)$ $\in$ $(T_1,T_0]\times \mathbb{T}^3_\epsilon$,
for some constants $c_i$, $i=1,2,3$,
and
\begin{equation*}\label{contB}
\norm{\bar{g}^{\mu\nu}}_{L^\infty((T_1,T_0],W^{1,\infty}(\mathbb{T}^3_\epsilon))}+
\norm{\bar{\partial}\bar{g}^{\mu\nu}}_{L^\infty((T_1,T_0],W^{1,\infty}(\mathbb{T}^3_\epsilon))}
+ \norm{\bar{v}^{\mu}}_{L^\infty((T_1,T_0],W^{1,\infty}(\mathbb{T}^3_\epsilon))}
+ \norm{\bar{\rho} }_{L^\infty((T_1,T_0],W^{1,\infty}(\mathbb{T}^3_\epsilon))}< \infty,
\end{equation*}
then there exists a time $T^*_1 \in (0,T_1)$ such that the solution uniquely extends
to the spacetime region $(T^*_1,T_0]\times \mathbb{T}^3_\epsilon$ with the same regularity as
given by \eqref{locreg}.

Next, we set
\begin{equation*}
\mathbf{u} = (u^{\mu\nu},u^{\mu\nu}_\gamma,u,u_\gamma,\delta\zeta,z_i),
\end{equation*}
where $u^{\mu\nu}$, $u^{\mu\nu}_\gamma$, $u$, $u_\gamma$, $\delta\zeta$ and $z_i$ are
computed from the solution \eqref{locreg} via the definitions from \S\ref{vardefs}.
From the definitions \eqref{E:ZETA} and \eqref{E:DELZETA}, the formulas \eqref{E:ZETAH3}-\eqref{zetaH1},
the expansions \eqref{E:GIJ}-\eqref{E:G0MU}  and \eqref{E:V^0},  and Sobolev's inequality, see
Theorem \ref{sobolev}, that there exists
a constant $\sigma > 0$, independent of $T_1 \in (0,T_0)$ and $\epsilon \in (0,\epsilon_0)$, such that
 \begin{equation} \label{thetadef}
\norm{(u^{\mu\nu},u,\delta\zeta,z_i)}_{L^\infty((T_1,T_0],H^s(\mathbb{T}^3))}< \sigma
\end{equation}
implies that the inequalities \eqref{contA} and $t^{-3(1+\epsilon^2 K)}\bar{\rho}\geq c_3>0$ hold for some constants $c_i$, $i=1,2,3$.
Moreover, for $\sigma$ small enough, we see from the Moser inequality from
Lemma \ref{moserlemC} and the expansions \eqref{E:GIJ}-\eqref{E:PGIJ} and
\eqref{E:V^0}-\eqref{E:VELOCITY} that
\begin{align*}
&\norm{\bar{g}^{\mu\nu}}_{L^\infty((T_1,T_0],W^{1,\infty}(\mathbb{T}^3_\epsilon))}+
\norm{\bar{\partial}\bar{g}^{\mu\nu}}_{L^\infty((T_1,T_0],W^{1,\infty}(\mathbb{T}^3_\epsilon))} \notag \\
&\hspace{1.0cm}+ \norm{\bar{v}^{\mu}}_{L^\infty((T_1,T_0],W^{1,\infty}(\mathbb{T}^3_\epsilon))} + \norm{\bar{\rho} }_{L^\infty((T_1,T_0],W^{1,\infty}(\mathbb{T}^3_\epsilon))} \leq C(\sigma)\bigl(\norm{\mathbf{u}}_{L^\infty((T_1,T_0],H^s(\mathbb{T}^3))}+1\bigr). %\label{contC}
\end{align*}
Thus by the continuation principle,  there exists a $\sigma > 0$ such that if
\eqref{thetadef} holds and
\begin{equation} \label{contD}
\norm{\mathbf{u}}_{L^\infty((T_1,T_0],H^s(\mathbb{T}^3))} < \infty,
\end{equation}
then the solution \eqref{locreg} can be uniquely continued as a classical solution
with the same regularity
to the larger spacetime region $(T^*_1,T_0]\times \mathbb{T}^3_\epsilon$ for some $T^*_1 \in (0,T_1)$.

Since $\Delta : \bar{H}^{s+2}(\mathbb{T}^3) \longrightarrow \bar{H}^s(\mathbb{T}^3)$ is an
isomorphism,  we can solve \eqref{E:DEFOFPHI} to get
\begin{equation} \label{Phicont}
\frac{1}{t}\Phi = \frac{\Lambda}{3}E^2 e^{ \zeta_H}
\Delta^{-1}\Pi e^{ \delta\zeta} \in
\bigcap_{\ell=0}^1  C^{\ell}((T_1,T_0],\bar{H}^{s+2-\ell}(\mathbb{T}^3)).
\end{equation}
As $\zeta_H$ and  $E$ are uniformly bounded on $(0,1]$, see \eqref{E:ZETAH3} and \eqref{E:EREP}, it then follows via the
Moser inequality from
Lemma \ref{moserlemC} that the derivative $\del{k}\Phi$ satisfies the bound
\begin{equation*}
\norm{t^{-1}\partial_k \Phi(t)}_{H^{s+1}(\mathbb{T}^3))} \leq
C\bigl(\norm{\delta\zeta(t)}_{H^s(\mathbb{T}^3)}\bigr)\norm{\delta\zeta(t)}_{H^s(\mathbb{T}^3)}
\end{equation*}
uniformly for $(t,\epsilon)\in (T_1,T_0]\times (0,\epsilon_0)$, where $C$ is independent of initial data and the times $\{T_1,T_2\}$.
But, this implies via the definition of $\mathbf{U}$, see \eqref{E:REALVAR}, that
\begin{equation*}
\norm{\mathbf{u}}_{L^\infty((T_1,T_0],H^s(\mathbb{T}^3))} \leq C\bigl(\norm{\mathbf{U}}_{L^\infty((T_1,T_0],H^s(\mathbb{T}^3))}\bigr)
\norm{\mathbf{U}}_{L^\infty((T_1,T_0],H^s(\mathbb{T}^3))}.
\end{equation*}
Since
\begin{equation*}
\norm{(u^{\mu\nu},u,\delta\zeta,z_i)}_{L^\infty((T_1,T_0],H^s(\mathbb{T}^3))}\leq \norm{\mathbf{U}}_{L^\infty((T_1,T_0],H^s(\mathbb{T}^3))},
\end{equation*}
we find that
\begin{equation} \label{contF}
\norm{\mathbf{U}}_{L^\infty((T_1,T_0],H^s(\mathbb{T}^3))} < \sigma
\end{equation}
implies that the inequalities  \eqref{thetadef} and \eqref{contD} both hold.
In particular, this shows that  if \eqref{contF} holds for $\sigma>0$ small enough, then the solution \eqref{locreg} can be uniquely continued as a classical solution
with the same regularity
to the larger spacetime region $(T^*_1,T_0]\times \mathbb{T}^3_\epsilon$ for some $T^*_1 \in (0,T_1)$.
\end{proof}

\section{Conformal cosmological Poisson-Euler equations: local existence and continuation\label{PEcont}}

In this section, we consider the local-in-time existence and uniqueness of solutions to the conformal cosmological Poisson-Euler
equations, and we
establish a continuation principle that is based on bounding the
$H^s$ norm of $(\mathring{\zeta},\mathring{z}^j)$.

\begin{proposition} \label{PEexist}
Suppose $s\in \Zbb_{\geq 3}$, $\mathring{\zeta}_0 \in H^{s}(\mathbb{T}^3)$ and $ (\mathring{z}^i_0)\in  H^{s}(\mathbb{T}^3,\mathbb{R}^3)$.
Then there exists a  $T_1 \in (0,T_0]$ and a unique classical solution
\begin{equation*}
(\mathring{\zeta},\mathring{z}^i,\mathring{\Phi})\in \bigcap_{\ell=0}^1 C^{\ell}((T_1,T_0],H^{s-\ell}(\mathbb{T}^3))
\times \bigcap_{\ell=0}^1  C^{\ell}((T_1,T_0],H^{s-\ell}(\mathbb{T}^3,\mathbb{R}^3)) \times \bigcap_{\ell=0}^1  C^{\ell}((T_1,T_0],H^{s+2-\ell}(\mathbb{T}^3)),
\end{equation*}
of the conformal cosmological Poisson-Euler equations, given by \eqref{CPeqn1}-\eqref{CPeqn3}, on the spacetime region
$(T_1,T_0]\times \mathbb{T}^3$ that satisfies
\begin{equation*}
(\mathring{\zeta},\mathring{z}^i)|_{t=T_0} =( \mathring{\zeta}_0,\mathring{z}^i_0)
\end{equation*}
on the initial hypersurface $t=T_0$. Futhermore, if
\begin{equation*}
\norm{(\mathring{\zeta},\mathring{z}^i)}_{L^\infty((T_1,T_0],H^s)} < \infty,
\end{equation*}
then the solution $(\mathring{\zeta},\mathring{z}^i,\mathring{\Phi})$ can be uniquely continued as a classical solution with the same regularity to the larger spacetime region $(T_1^*,T_0]\times \mathbb{T}^3$ for some $T_1^*\in (0,T_1)$.
\end{proposition}
\begin{proof}
Using the fact that $\Delta \: :\: \bar{H}^{s+2} \longrightarrow \bar{H}^s$ is an isomorphism, we can solve the Poisson equation
\eqref{E:COSEULERPOISSONEQ.c} by setting
\begin{equation} \label{PEexist1}
\mathring{\Phi} = \frac{\Lambda}{3}t \mathring{E}^2 \Delta^{-1} \Pi e^{\mathring{\zeta}} .
\end{equation}
We  can use this to write \eqref{CPeqn1}-\eqref{CPeqn3} as
\begin{align}
    \partial_t \mathring{\zeta}+\sqrt{\frac{3}{\Lambda}}\bigl( \mathring{z}^j\partial_j \mathring{\zeta} + \partial_j\mathring{z}^j\bigr)
    &=-\frac{3\mathring{\Omega}}{t}, \label{PEexist2}\\
    \sqrt{\frac{\Lambda}{3}}\partial_t\mathring{z}^j+ \mathring{z}^i\partial_i\mathring{z}^j+ K
    \partial^j\mathring{\zeta}
    &=\sqrt{\frac{\Lambda}{3}}\frac{1}{t}\mathring{z}^j-\frac{1}{2}
    t \mathring{E}^2  \partial^j \Delta^{-1} \Pi e^{\mathring{\zeta}} . \label{PEexist3}
\end{align}
It is then easy to see that this system can be cast in symmetric hyperbolic form by multiplying \eqref{PEexist3} by
$\mathring{E}^2 K^{-1}\sqrt{\frac{3}{\Lambda}} $.  Even though the resulting system is non-local due to the
last term in \eqref{PEexist3}, all of the standard local existence and uniqueness results and continuation principles
that are valid for local symmetric hyperbolic systems, e.g. Theorems 2.1 and 2.2 of \cite{maj}, continue to apply.
Therefore it follows that there exists a
unique local-in-time classical solution
\begin{equation}\label{PElocreg}
(\mathring{\zeta},\mathring{z}^i)\in \bigcap_{\ell=0}^1 C^{\ell}((T_1,T_0],H^{s-\ell}(\mathbb{T}^3))
\times \bigcap_{\ell=0}^1  C^{\ell}((T_1,T_0],H^{s-\ell}(\mathbb{T}^3,\mathbb{R}^3))
\end{equation}
of \eqref{PEexist2}-\eqref{PEexist3} for some time $T_1\in (0,T_0)$ that
satisfies
\begin{equation*}
(\mathring{\zeta},\mathring{z}^i)_{t=T_0} =(\mathring{\zeta}_0,\mathring{z}^i_0)
\end{equation*}
for given initial data $(\mathring{\zeta}_0,\mathring{z}^i_0)
\in H^{s}(\mathbb{T}^3)\times H^{s}(\mathbb{T}^3,\mathbb{R}^3)$.
Moreover,
if the solution satisfies
\begin{equation*}
\norm{\mathring{\zeta}}_{L^\infty((T_1,T_0],W^{1,\infty})}+
\norm{\mathring{z}^i}_{L^\infty((T_1,T_0],W^{1,\infty})}
< \infty,
\end{equation*}
then there exists a time $T^*_1 \in (0,T_1)$ such that the solution \eqref{PElocreg} uniquely extends
to the spacetime region $(T^*_1,T_0]\times \mathbb{T}^3$ with the same regularity. By Sobolev's inequality, see Theorem \ref{sobolev},
this is clearly implied by the stronger condition
\begin{equation*}
\norm{(\mathring{\zeta},\mathring{z}^i)}_{L^\infty((T_1,T_0],H^s)} < \infty.
\end{equation*}
Finally from \eqref{PEexist1}, \eqref{PElocreg} and the Moser inequality from Lemma \ref{moserlemC}, it is clear that
\begin{equation*}
\mathring{\Phi} \in \bigcap_{\ell=0}^1  C^{\ell}((T_1,T_0],H^{s+2-\ell}(\mathbb{T}^3)).
\end{equation*}
\end{proof}

\begin{corollary} \label{PEcor}
If the initial modified density $\zeta_0 \in H^s(\Tbb^3)$ from Proposition \ref{PEexist} is chosen so that
\begin{equation*}
\mathring{\zeta}_0 = \ln\biggl(\frac{\mathring{\rho}_H(T_0) + \breve{\rho}_0}{T_0^3}\biggr),
\end{equation*}
where  $\mathring{\rho}_H  = \frac{4C_0 \Lambda t^3}{(C_0-t^3)^2}$, $\breve{\rho}_0 \in \bar{H}^s(\Tbb^3)$,
and $\mathring{\rho}_H(T_0) + \breve{\rho}_0 > 0$ in $\Tbb^3$,
then the solution $(\mathring{\zeta},\mathring{z}{}^i,\mathring{\Phi})$ to the conformal cosmological Poisson-Euler
equations from Proposition \ref{PEexist} satisfies
\begin{equation*}
\Pi \mathring{\rho} = \delta \mathring{\rho} := \mathring{\rho}-\mathring{\rho}_H \quad \text{in $(T_1,T_0]\times \Tbb^3$}.
\end{equation*}
\end{corollary}
\begin{proof}
Since $\mathring{\rho}=t^3 e^{\mathring{\zeta}}$ satisfies \eqref{E:COSEULERPOISSONEQ.a}, we see
after applying $\la 1, \cdot \ra$ to this equations that $\la 1, \mathring{\rho} \ra$ satisfies
\begin{equation*}
\frac{d\;}{dt}\la 1, \mathring{\rho}(t) \ra = \frac{3(1-\mathring{\Omega}(t))}{t} \la 1, \mathring{\rho}(t) \ra, \quad T_1 < t \leq T_0,
\end{equation*}
while from the choice of initial data, we have
\begin{equation*}
\la 1, \mathring{\rho}(T_0) \ra = \mathring{\rho}_H(T_0).
\end{equation*}
By a direct computation, we observe with the help of \eqref{Oringdef} that  $\mathring{\rho}_H  = \frac{4C_0 \Lambda t^3}{(C_0-t^3)^2}$ satisfies the differential equation
 \begin{equation}  \label{PEcor0}
\frac{d\;}{dt}\mathring{\rho}_H(t) = \frac{3(1-\mathring{\Omega}(t))}{t}\mathring{\rho}_H(t) \quad 0 < t \leq T_0,
\end{equation}
and hence, that
\begin{equation} \label{PEcor1}
\la 1, \mathring{\rho}(t) \ra = \mathring{\rho}_H(t), \quad T_1 < t \leq T_0,
\end{equation}
by the uniqueness of solutions to the initial value problem for ordinary differential equations. The proof now follows since
\begin{equation*}
\Pi \mathring{\rho} \overset{\eqref{Pidef}}{=} \mathring{\rho} - \la 1, \mathring{\rho}\ra
\overset{\eqref{PEcor1}}{=} \mathring{\rho} - \mathring{\rho}_H(t)
\quad \text{in $(T_1,T_0]\times \Tbb^3$}.
\end{equation*}
\end{proof}

\begin{remark} \label{PEcorrem}
Letting
\begin{equation} \label{deltazetaringdef}
\delta \mathring{\zeta} = \mathring{\zeta} - \mathring{\zeta}_H,
\end{equation}
where, see \eqref{arhoringdef}, \eqref{zetaHringform} and \eqref{E:RHOHOM},
\begin{equation} \label{deltazetaringH}
\mathring{\zeta}_H = \ln(t^{-3}\mathring{\rho}_H),
\end{equation}
it is clear that the initial condition
\begin{equation*}
\mathring{\zeta}|_{t=T_0} = \ln\biggl(\frac{\mathring{\rho}_H(T_0) + \breve{\rho}_0}{T_0^3}\biggr),
\end{equation*}
from Corollary \ref{PEcor} is equivalent to the initial condition
\begin{equation*}
\delta \mathring{\zeta}|_{t=T_0} = \ln\biggl(1+\frac{\breve{\rho}_0}{\mathring{\rho}_H(T_0)}\biggr)
\end{equation*}
for $\delta \mathring{\zeta}$.
\end{remark}

\section{Singular symmetric hyperbolic systems} \label{S:MODEL}

In this section, we establish uniform a priori estimates for solutions to a class of symmetric hyperbolic systems that are
jointly singular in $\epsilon$ and $t$, and include both the formulation of the reduced conformal Einstein--Euler equations given by \eqref{E:REALEQ} and the $\epsilon \searrow 0$
limit of these equations. We also establish \textit{error estimates}, that is, a priori estimates for the difference between solutions of the
$\epsilon$-dependent singular symmetric hyperbolic systems and their corresponding $\epsilon \searrow 0$ limit equations.

The $\epsilon$-dependent singular terms that appear in the symmetric hyperbolic systems we consider
are of a type that have been well studied, see \cite{bro,kla2,kla1,kre1,sch1,sch2}, while the
$t$-dependent singular terms are of the type analyzed in \cite{oli5}. The uniform a priori estimates established
here follow from combining the energy estimates from \cite{bro,kla2,kla1,kre1,sch1,sch2} with those
from \cite{oli5}.

\begin{remark}
In this section, we switch to the standard time orientation, where the future is located in the direction of increasing time, while keeping the singularity located at $t=0$. We
do this in order to make the derivation of the energy estimates in this section as similar as possible to those for non-singular symmetric hyperbolic systems, which we expect will
make it easier for readers familiar with such estimates to follow the arguments below. To get back to the time orientation used to formulate the conformal Einstein--Euler equations, see Remark \ref{torient}, we need only apply the trivial time transformation $t \mapsto -t$.
\end{remark}

\subsection{Uniform estimates\label{S:MODELuni}}
We will establish uniform a priori estimates for the following class of equations:
\begin{align}
  A^0  \partial_0 U+A^i  \partial_i U+\frac{1}{\epsilon}C^i\partial_i U=\frac{1}{t}\mathfrak{A} \mathbb{P}  U +H \quad &\mbox{in} \quad[T_0, T_1)\times\mathbb{T}^n, \label{E:MODELEQ2a}
%   U|_{t=T_0} =\p{\mrw^0(x)\\ \mru^0(x)}+\epsilon \p{s^0(\epsilon,x)\\ r^0(\epsilon,x)} \quad &\text{in} \quad \{T_0\}\times\mathbb{T}^n,
%\label{E:MODELEQ2b}
\end{align}
where
\begin{align*}
U &= (w, u)^T, \\
A^0&=\p{A^0_1(\epsilon,t,x,w) & 0\\
                        0 & A^0_2(\epsilon,t,x,w)}, \\
A^i&= \p{A^i_1 (\epsilon,t,x,w)& 0\\
                        0 & A^i_2(\epsilon,t,x,w)}, \\
C^i&=\p{C^i_1 & 0\\
0 & C^i_2 },  \quad  \Pbb=\p{\Pbb_1 & 0\\
0 & \Pbb_2 }, \\
\mathfrak{A}&= \p{\mathfrak{A}_1 (\epsilon,t,x,w)& 0\\
                        0 & \mathfrak{A}_2(\epsilon,t,x,w)},\\
H&= \p{H_1(\epsilon,t,x,w)\\ H_2(\epsilon,t,x,w,u)+R_2 }+\p{F_1(\epsilon,t,x)\\ F_2(\epsilon,t,x)},\\
R_2&=\frac{1}{t}
M_2(\epsilon, t,x,w,u) \Pbb_3 U,
\end{align*}
and the following assumptions hold for fixed constants $\epsilon_0,R >0$, $T_0 < T_1 < 0$ and $s\in \Zbb_{>n/2+1}$:

\begin{ass}\label{ASS1}$\;$

\begin{enumerate}
\item  \label{A:CONSC} The $C^i_a$, $i=1,\ldots,n$ and $a=1,2$, are constant, symmetric $N_a\times N_a$ matrices.
\item  The $\mathbb{P}_a$, $a=1,2$, are constant, symmetric $N_a\times N_a$ projection matrices, i.e. $\mathbb{P}_a^2= \mathbb{P}_a$. We use $\mathbb{P}_a^{\perp}=\mathds{1}-\mathbb{P}_a$  to denote the complementary projection matrix.
\item  \label{A:GH} The source terms $H_a(\epsilon,t,x,w)$, $a=1,2$, $F_a(\epsilon, t, x)$, $a=1,2$, and
$M_2(\epsilon,t,x,w,u)$ satisfy
$H_1 \in E^0\bigl((0,\epsilon_0)\times (2 T_0,0)\times \Tbb^n \times B_R(\Rbb^{N_1}),\Rbb^{N_1}\bigr)$,
$H_2 \in E^0\bigl((0,\epsilon_0)\times (2 T_0,0)\times \Tbb^n \times B_R(\Rbb^{N_1})
\times B_R(\Rbb^{N_2})\times B_R((\Rbb^{N_1})^n),\Rbb^{N_2}\bigr)$,
$F_a \in C^0\bigl((0,\epsilon_0)\times [T_0,T_1), H^s(\Tbb^n,\Rbb^{N_a})\bigr)$,
$M_2  \in E^0\bigl((0,\epsilon_0)\times (2 T_0,0)\times \Tbb^n \times B_R(\Rbb^{N_1})
\times B_R(\Rbb^{N_2}),\mathbb{M}_{N_2\times N_2}\bigr)$, and
\begin{equation*}
H_1(\epsilon,t,x,0) = 0, \quad H_2(\epsilon,t,x,0,0) = 0 \AND M_2(\epsilon,t,x,0,0) = 0
\end{equation*}
for all $(\epsilon,t,x)\in (0,\epsilon_0)\times (2 T_0,0)\times \Tbb^n $.

\item  \label{A:Bi} The matrix valued maps $A_a^i(\epsilon,t,x,w)$, $i=0,\ldots,n$ and $a=1,2$, satisfy $A^i_a \in E^0\bigl((0,\epsilon_0)\times (2 T_0,0)\times \Tbb^n \times B_R(\Rbb^{N_a}),\mathbb{S}_{N_a}\bigr)$.

\item \label{A:B0} The matrix valued maps
		$A_a ^0(\epsilon,t,x, w)$, $a=1,2$, and $\mathfrak{A}_a(\epsilon,t,x, w)$, $a=1,2$, can be decomposed as
		\begin{gather}
		A_a^0(\epsilon,t,x, w)=\mathring{A}_a^0(t)+\epsilon \tilde{A}_a^0(\epsilon,t,x, w),\label{E:DECOMPOSITIONOFA01}\\
			\mathfrak{A}_a(\epsilon,t,x, w)=\mathring{\mathfrak{A}}_a(t)+\epsilon \tilde{\mathfrak{A}}_a(\epsilon,t,x, w),\label{E:DECOMPOSITIONOFCALB}
		\end{gather}
		where $\mathring{A}_a^0 \in E^1\bigl((2 T_0,0),\mathbb{S}_{N_a}\bigr)$,
 $\mathring{\mathfrak{A}}_a \in E^1\bigl((2 T_0,0),\mathbb{M}_{N_a\times  N_a}\bigr)$,
$\tilde{A}_a^0 \in E^1\bigl((0,\epsilon_0)\times (2 T_0,0)\times \Tbb^n \times B_R(\Rbb^{N_1}),\mathbb{S}_{N_a}\bigr)$,
$\tilde{\mathfrak{A}}_a \in E^0\bigl((0,\epsilon_0)\times (2 T_0,0)\times \Tbb^n \times B_R(\Rbb^{N_1}),\mathbb{M}_{N_a\times N_a}\bigr)$, and\footnote{Or in other words,  the matrices $\tilde{\mathfrak{A}}_a|_{w=0}$ and
$\tilde{A}_a^0|_{w=0}$ depend only on $(\epsilon,t)$.}
\begin{equation} \label{DA0}
D_x\tilde{\mathfrak{A}}_a(\epsilon,t,x,0)=D_x\tilde{A}_a^0(\epsilon,t,x,0)=0
\end{equation}
for all $(\epsilon,t,x)\in (0,\epsilon_0)\times (2 T_0,0)\in \mathbb{T}^n$.
%\begin{equation*}
%\tilde{A}_a^0(\epsilon,t,x,0)=  \tilde{\mathfrak{A}}_a(\epsilon,t,x,0)=0
%\end{equation*}
%for all $(\epsilon,t,x)\in (0,\epsilon_0)\times (2 T_0,0)\times \Tbb^n $.

\item   \label{A:B}  For $a=1,2$, the matrix $\mathfrak{A}_a$ commutes with $\mathbb{P}_a$, i.e.
\begin{align}\label{E:COMMUTEPANFB}
[\mathbb{P}_a, \mathfrak{A}_a(\epsilon,t,x,w)]=0
\end{align}
for all $(\epsilon,t,x,w)\in(0, \epsilon_0)\times (2T_0, 0) \times\mathbb{T}^n \times B(\mathbb{R}^{N_1}) $.

\item  $\Pbb_3$ is a symmetric  $(N_1+N_2)\times (N_1+N_2)$ projection matrix that satisfies
\begin{gather}
\Pbb\Pbb_3 =\Pbb_3\Pbb=\Pbb_3,  \label{E:P32a} \\
 \Pbb_3 A^i(\epsilon,t,x,w) \Pbb_3^\perp =
\Pbb_3 C^i \Pbb_3^\perp= \Pbb_3 \mathfrak{A}(\epsilon,t,x,w) \Pbb_3^\perp = 0  \label{E:P32b}
\intertext{and}
[\Pbb_3,A^0(\epsilon,t,x,w)] = 0  \label{E:P32c}
\end{gather}
for all $(\epsilon,t,x,w)\in(0, \epsilon_0)\times (2T_0, 0) \times\mathbb{T}^n \times B_R(\mathbb{R}^{N_1})$,
where $\mathbb{P}_3^\perp = \mathds{1}-\mathbb{P}_3$ defines the complementary projection matrix.

\item   \label{E:B0ANDCALB}  There exists constants $\kappa$, $\gamma_1$, $\gamma_2>0$, such that
\begin{align}
\frac{1}{\gamma_1}\mathds{1}\leq A_a^0(\epsilon,t,x,w)\leq\frac{1}{\kappa}\mathfrak{A}_a(\epsilon,t,x,w)\leq \gamma_2\mathds{1} \label{E:KAPPAB0CALB}
\end{align}
for all $(\epsilon,t,x,w)\in(0, \epsilon_0)\times (2T_0, 0) \times\mathbb{T}^n \times B(\mathbb{R}^{N_1})$ and $a=1,2$.

\item  \label{A:PBP}  For $a=1,2$, the matrix $A_a^0$ satisfies
\begin{align}\label{E:PAP}
\mathbb{P}_a ^{\perp}A_a^0(\epsilon,t,x, \mathbb{P}_1 ^{\perp}w)\mathbb{P}_a =\mathbb{P}_a A_a^0(\epsilon,t,x, \mathbb{P}_1 ^{\perp}w)\mathbb{P}_a ^{\perp}=0
\end{align}
for all $(\epsilon,t,x,w)\in(0, \epsilon_0)\times (2T_0, 0)\times\mathbb{T}^n \times B(\mathbb{R}^{N_1}) $.

\item   \label{A:PDECOMPOSABLE} For $a=1,2$, the matrix $\Pbb_a^\perp[D_w A_a^0\cdot(A_1^0)^{-1}\mathfrak{A}_1\mathbb{P}_1w]
\Pbb_a^\perp$ can be decomposed as
\al{ADEC}{
\Pbb_a^\perp \bigl[D_w A_a^0(\epsilon,t,x,w)\cdot \bigl(A^0_1(\epsilon,t,x,w)\bigr)^{-1}\mathfrak{A}_1(\epsilon,t,x,w)\mathbb{P}_1w\bigr] \Pbb_a^\perp =
t\mathfrak{S}_a(\epsilon, t, x, w)+\mathfrak{T}_a(\epsilon, t, x, w, \Pbb_1 w)
			}
for some
$\mathfrak{S}_a \in E^0\bigl((0,\epsilon_0)\times (2T_0,0)\times \Tbb^n\times B_R(\Rbb^{N_1}),
\mathbb{M}_{N_a\times N_a}\bigr)$, $a=1,2$, and
$\mathfrak{T}_a \in E^0\bigl((0,\epsilon_0)\times (2T_0,0)\times \Tbb^n\times B_R(\Rbb^{N_1})\times \Rbb^{N_1},
\mathbb{M}_{N_a\times N_a}\bigr)$, $a=1,2$, where the  $\mathfrak{T}_a(\epsilon, t, x, w, \xi)$ are quadratic in $\xi$.
\end{enumerate}
\end{ass}

Before proceeding with the analysis, we take a moment to make a few
observations about the structure of the singular system \eqref{E:MODELEQ2a}. First, if $\mathfrak{A}=0$, then the singular term
$\frac{1}{t}\mathfrak{A}\mathbb{P}U$  disappears from \eqref{E:MODELEQ2a} and it becomes a regular symmetric hyperbolic system.
Uniform $\epsilon$-independent a priori estimates that are valid for $t\in [T_1,0)$ would then follow,  under a suitable small initial data assumption, as a direct consequence
of the energy estimates from \cite{bro,kla2,kla1,kre1,sch1,sch2}. When $\mathfrak{A}\neq 0$, the  positivity assumption
\eqref{E:KAPPAB0CALB} guarantees that the singular term $\frac{1}{t}\mathfrak{A}\mathbb{P}U$ acts like
a friction term. This allows us to generalize the energy estimates from \cite{bro,kla2,kla1,kre1,sch1,sch2} in such a way as to obtain, under a suitable small initial data assumption, uniform $\epsilon$-independent a priori estimates that are valid on the time interval $[T_1,0)$; see
\eqref{E:ENERGEST1}, \eqref{E:ENERGEST2} and \eqref{E:ENERGEST3}
 for the key differential inequalities used to derive these a priori estimates.

\begin{remark} \label{decouple}
The equation for $w$ decouples from the system
\eqref{E:MODELEQ2a} and is given by
	\begin{align}
	A_1^0  \partial_0 w+A_1^i  \partial_i w+\frac{1}{\epsilon}C_1^i\partial_i w =\frac{1}{t}\mathfrak{A}_1 \mathbb{P}_1  w +H_1+F_1 \quad &\text{in} \quad[T_0, T_1)\times\mathbb{T}^n.  \label{E:MODELEQ1a}
%	w|_{t=T_0} = \mrw^0(x) +\epsilon s^0(\epsilon,x)  \quad &\text{in} \quad \{T_0\}\times\mathbb{T}^n. \label{E:MODELEQ1b}
	\end{align}
\end{remark}

\begin{remark} $\;$

\begin{enumerate}
\item
By Taylor expanding $A^0_a(\epsilon,t,x,\Pbb^\perp_1 w+ \Pbb_1 w)$ in the variable $\Pbb_1 w$, it follows
from  \eqref{E:PAP} that there exist matrix valued maps $\hat{A}^0_a, \breve{A}^0_a \in E^1\bigl((0,\epsilon_0)\times (2 T_0,0)\times \Tbb^n \times B_R\bigl(\Rbb^{N_1}\bigr),\mathbb{M}_{N_a \times N_a}\bigr)$, $a=1,2$, such that
\begin{align}
			\mathbb{P} ^\perp_a  A^0_a(\epsilon,t,x, w) \mathbb{P}_a =\mathbb{P}_a ^\perp [\hat{A}^0_a(\epsilon,t,x, w)\cdot\mathbb{P}_1 w]\mathbb{P}_a \label{E:PPERPB0P}
\intertext{and}
			\mathbb{P}_a A^0_a(\epsilon,t,x, w) \mathbb{P}_a^\perp=\mathbb{P}_a  [\breve{A}^0_a(\epsilon,t,x, w)\cdot\mathbb{P}_1 w]\mathbb{P}_a ^\perp\label{E:PB0PPERP}
			\end{align}
for all $(\epsilon,t,x,w)\in(0, \epsilon_0)\times (2T_0, 0)\times\mathbb{T}^n \times B(\mathbb{R}^{N_1}) $.
\item
It is not difficult to see that the assumptions \eqref{E:KAPPAB0CALB} and  \eqref{E:PAP} imply that
\begin{align*}
\mathbb{P}_a ^{\perp}\bigl(A_a^0(\epsilon,t,x, \mathbb{P}_1 ^{\perp}w)\bigr)^{-1}\mathbb{P}_a =
\mathbb{P}_a \bigl(A_a^0(\epsilon,t,x, \mathbb{P}_1 ^{\perp}w)\bigr)^{-1}\mathbb{P}_a ^{\perp}=0
\end{align*}
for all $(\epsilon,t,x,w)\in(0, \epsilon_0)\times (2T_0, 0)\times\mathbb{T}^n \times B(\mathbb{R}^{N_1}) $.
By Taylor expanding $(A^0_a(\epsilon,t,x,\Pbb^\perp_1 w+ \Pbb_1 w))^{-1}$ in the variable $\Pbb_1 w$, it follows that
there exist matrix valued maps $\hat{B}^0_a, \breve{B}^0_a \in E^1\bigl((0,\epsilon_0)\times (2 T_0,0)\times \Tbb^n \times B_R\bigl(\Rbb^{N_1}\bigr),\mathbb{M}_{N_a \times N_a}\bigr)$, $a=1,2$, such that
\begin{align}
			\mathbb{P} ^\perp_a\bigl(A^0_a(\epsilon,t,x, w)\bigr)^{-1}\mathbb{P}_a =\mathbb{P}_a ^\perp [\hat{B}^0_a(\epsilon,t,x, w)\cdot\mathbb{P}_1 w]\mathbb{P}_a \label{E:PPERPB0Pa}
\intertext{and}
			\mathbb{P}_a\bigl(A^0_a(\epsilon,t,x, w)\bigr)^{-1}\mathbb{P}_a^\perp=\mathbb{P}_a  [\breve{B}^0_a(\epsilon,t,x, w)\cdot\mathbb{P}_1 w]\mathbb{P}_a ^\perp\label{E:PB0PPERPa}
			\end{align}
for all $(\epsilon,t,x,w)\in(0, \epsilon_0)\times (2T_0, 0)\times\mathbb{T}^n \times B(\mathbb{R}^{N_1}) $.
	\end{enumerate}
\end{remark}

To facilitate the statement and proof of our a priori estimates for solutions of the system \eqref{E:MODELEQ2a},
we introduce the following energy norms:
\begin{definition} \label{energynorms}
  Suppose  $w \in L^\infty([T_0,T_1)\times \mathbb{T}^n,\Rbb^{N_1})$, $k \in \mathbb{Z}_{\geq 0}$, and
$\{\mathbb{P}_a,A_a^0\}$, $a=1,2$, are as defined above. Then for maps $f_a$, $a=1,2$, and $U$ from the torus $\mathbb{T}^n$ into
$R^{N_a}$ and $R^{N_1}\times R^{N_2}$, respectively, the
\emph{energy norms},
denoted $\vertiii{f_a}_{a,H^s}$ and
$\vertiii{U}_{H^s}$, of $f_a$ and $U$ are defined by
  \begin{gather*}
    \vertiii{f_a}^2_{a,H^k}:=\sum_{0\leq |\alpha|\leq k}\langle D^\alpha f_a, A_a^0\bigl(\epsilon,t,\cdot,w(t,\cdot)\bigr) D^\alpha f_a\rangle\\
\intertext{and}
    \vertiii{U}^2_{H^k}:=\sum_{0\leq |\alpha|\leq k}\langle D^\alpha U, A^0\bigl(\epsilon,t,\cdot,w(t,\cdot)\bigr) D^\alpha U\rangle,
\end{gather*}
respectively. In addition to the energy norms, we also define, for $T_0 < T \leq T_1$, the spacetime norm
of maps $f_a$, $a=1,2$, from $[T_0,T)\times \mathbb{T}^n$ to $R^{N_a}$ by
\begin{equation*}
 \|f_a\|_{M_{\mathbb{P}_a ,k}^\infty([T_0, T)\times\mathbb{T}^n)}:=\|f_a\|_{L^{\infty}([T_0, T),H^k)}+ \left(-\int_{T_0}^T\frac{1}{t}\|\mathbb{P}_a f_a (t)\|^2_{H^k}dt\right)^\frac{1}{2}.
\end{equation*}
\end{definition}

\begin {remark}
For  $w \in L^\infty([T_0,T_1)\times \mathbb{T}^n,\Rbb^{N_1})$ satisfying $\norm{w}_{L^\infty([T_0,T_1)\times \mathbb{T}^n)} < R$,
we observe, by \eqref{E:KAPPAB0CALB}, that the standard Sobolev norm $\norm{\cdot}_{H^k}$ and
the energy norms  $\vertiii{\,\cdot\,}_{a,H^k}$, $a=1,2$, are equivalent since they satisfy
  \begin{align*}%\label{E:EQUIVNORM}
    \frac{1}{\sqrt{\gamma_1}}\|\cdot\|_{H^k}\leq \vertiii{\,\cdot\,}_{a,{H^k}} \leq \sqrt{\gamma_2}\|\cdot\|_{H^k}.
  \end{align*}
\end{remark}

With the preliminaries out of the way, we are now ready to state and prove a priori estimates for solutions of the system
 \eqref{E:MODELEQ2a} that are uniform in $\epsilon$.
\begin{theorem}\label{L:BASICMODEL}
  Suppose $R>0$, $s\in \mathbb{Z}_{\geq n/2+1}$, $T_0 < T_1 < 0$, $\epsilon_0 > 0$, $\epsilon\in(0, \epsilon_0)$, Assumptions \ref{ASS1} hold, the map
 \als%{MAXSOL}
 {
  	U=(w,u) \in \bigcap_{\ell=0}^1 C^\ell([T_0,T_1), H^{s-\ell}(\mathbb{T}^n,\Rbb^{N_1})) \times
\bigcap_{\ell=0}^1 C^\ell([T_0,T_1), H^{s-1-\ell}(\mathbb{T}^n,\Rbb^{N_2}))
  	}
defines a solution of the system \eqref{E:MODELEQ2a}, and for $t\in [T_0,T_1)$, the source terms $F_a$, $a=1,2$, satisfy
the estimates
\ga{F_I}{
			\|F_1(\epsilon,t)\|_{H^s}
			\leq C(\|w\|_{\Li([T_0,t),H^s)})\|w(t)\|_{H^s}
		}
and
\ga{F_I2}{
		\|F_2(\epsilon,t)\|_{\Hs}	\leq C\bigl(\|w\|_{\Li([T_0,t),H^s)},\|u\|_{\Li([T_0,t),\Hs)}\bigr)(\|w(t)\|_{H^s}+\|u(t)\|_{\Hs}),
		}
where the constants $C(\|w\|_{\Li([T_0,t),H^s)})$ and
$C\bigl(\|w\|_{\Li([T_0,t),H^s)},\|u\|_{\Li([T_0,t),\Hs)})$ are independent of $\epsilon \in (0,\epsilon_0)$
and $T_1 \in (T_0,0]$.
Then there exists a $\sigma>0$ independent of $\epsilon \in (0,\epsilon_0)$ and $T_1 \in (T_0,0)$, such that if initially
\begin{align*}%\label{E:INITIALDATAMODEL2}
  \|w(T_0)\|_{H^s}  \leq \sigma  \AND \|u(T_0)\|_{\Hs} \leq \sigma,
\end{align*}
then
\begin{equation*} %\label{wLinfty}
\norm{w}_{L^\infty([T_0,T_1)\times \Tbb^n)} \leq \frac{R}{2}
\end{equation*}
and there exists a constant $C>0$, independent of $\epsilon\in (0,\epsilon_0)$ and $T_1 \in (T_0,0)$, such that
\als%{IMPROVEDEST}
{
 \|w\|_{M^\infty_{\Pbb_1, s}([T_0,t)\times \Tbb^n)}+\|u\|_{M^\infty_{\Pbb_2, s-1}([T_0,t)\times \Tbb^n)} -\int_{T_0}^{t}
\frac{1}{\tau} \|\Pbb_3 U\|_{\Hs}\, d\tau  \leq C\sigma
}
for $T_0 \leq t < T_1$.
\end{theorem}
\begin{proof}
Letting $C_{\textrm{Sob}}$ denote the constant from the Sobolev inequality, we have that
\begin{equation*}
\norm{w(T_0)}_{L^\infty} \leq C_{\textrm{Sob}}\norm{w(T_0)}_{H^s} \leq C_{\textrm{Sob}}\sigma.
\end{equation*}
We then choose $\sigma$ to satisfy
\begin{equation} \label{sigmaC1}
\sigma \leq \min\biggl\{1, \frac{\hat{R}}{4}\biggr\},
\end{equation}
where $\hat{R} = \frac{R}{2 C_{\textrm{Sob}}}$,
so that
\begin{equation*} \label{winit1}
\norm{w(T_0)}_{L^\infty} \leq \frac{R}{8}.
\end{equation*}
Next, we define
\begin{align*}
K_1(t)=\|w\|_{L^\infty([T_0, t), H^s)} \AND  K_2(t)=\|u\|_{L^\infty([T_0, t), H^{s-1})},
\end{align*}
and observe that $K_1(T_0)+K_2(T_0) \leq \hat{R}/2$, and hence, by continuity,
either $K_1(t)+K_2(t) < \hat{R}$ for all $t\in [T_0,T_1)$, or
else there exists a first time $T_* \in (T_0,T_1)$ such that $K_1(T_*)+K_2(T_*) = \hat{R}$. Letting $T_* = T_1$ if the first case holds, we then have
that
\begin{equation} \label{K1ineq}
K_1(t)+K_2(t) < \hat{R}, \quad 0\leq t < T_*,
\end{equation}
where $T_* = T_1$ or else $T_*$ is the first time in $(T_0,T_1)$ for which $K_1(T_*) +K_2(T_*)= \hat{R}$.

Before proceeding the proof, we first establish a number of preliminary estimates, which we collect together in the following Lemma.
\begin{lemma}\label{L:PREEST}
There exists constants $C(K_1(t))$ and  $C(K_1(t),K_2(t))$, both independent of $\epsilon\in (0,\epsilon_0)$ and $T_*\in (T_0,T_1]$, such that the following estimates hold
for $T_0 \leq t < T_*<0$:
\gat{
    -\frac{2}{t}\sum_{|\alpha|\leq s }\langle D^\alpha w, A_1^0[(A_1^0)^{-1}\mathfrak{A}_1, D^\alpha]\mathbb{P}_1 w\rangle
   \leq    -\frac{1}{t}C(K_1 ) \|w\|_{H^s} \|\mathbb{P}_1 w\|^2_{H^s},  \label{E:INEQ1a}\\
    -\frac{2}{t}\sum_{|\alpha|\leq s-1 }\langle D^\alpha u, A_2^0[(A_2^0)^{-1}\mathfrak{A}_2, D^\alpha]\mathbb{P}_2 u\rangle
   \leq   -\frac{1}{t}C(K_1 ) (\|u\|_{\Hs}+\|w\|_{H^s})(\|\mathbb{P}_2u\|^2_{\Hs}+\|\mathbb{P}_2w\|^2_{H^s}),  \label{E:INEQ1b}
}
%\al{YPU}{-\frac{2}{t}\sum_{|\alpha|\leq s-1 }\la D^\alpha u, A^0_2 [D^\alpha, (A^0_2)^{-1}] \Pbb_2 Y_2 u] \ra\leq  -\frac{1}{t} C(K_1) (\|w\|_{H^s}+\|u\|_{\Hs}) (\|\Pbb_1 w\|_{H^s}^2+\|\Pbb_2 u\|_{\Hs}^2)}
\ali{
  -\sum_{ |\alpha|\leq s}\langle D^\alpha w, A_1^0[ D^\alpha,(A_1 ^0)^{-1}A_1^i]\partial_i w\rangle
  \leq & C(K_1 ) \|w\|_{H^s}^2,  \label{E:INEQ2a}\\
  -\sum_{ |\alpha|\leq s-1}\langle D^\alpha u, A_2 ^0[ D^\alpha,(A_2 ^0)^{-1}A_2^i]\partial_i u\rangle
  \leq  & C(K_1 ) \|u\|_{\Hs}^2,  \label{E:INEQ2b}
}
\gat{
    -\sum_{ |\alpha|\leq s}\langle D^\alpha w, [\tilde{A}_1^0, D^\alpha](A_1^0)^{-1}C_1^i\partial_i w\rangle \leq    C(K_1)\|w\|^2_{H^s},  \label{E:INEQ3a}\\
    -\sum_{ |\alpha|\leq s-1}\langle D^\alpha u, [\tilde{A}_2^0, D^\alpha](A_2^0)^{-1}C_2^i\partial_i u\rangle \leq  C(K_1 )\|u\|^2_{\Hs},  \label{E:INEQ3b}\\
    \sum_{|\alpha|\leq s}\langle D^\alpha w, (\partial_t A_1^0) D^\alpha w\rangle
    \leq    C(K_1) \|w\|^2_{H^s}-\frac{1}{t}C(K_1 )  \|w\|_{H^s} \|\mathbb{P}_1w\|^2_{H^s}, \label{E:INEQ5a}\\
    \sum_{|\alpha|\leq s-1}\langle D^\alpha u, (\partial_tA_2^0) D^\alpha u\rangle
    \leq C(K_1 )  \|u\|^2_{\Hs}-\frac{1}{t}C(K_1, K_2) (\|u\|_{\Hs}+\|w\|_{H^s})
    (\|\mathbb{P}_2u\|^2_{\Hs}+\|\mathbb{P}_1w\|^2_{H^s})\label{E:INEQ5b}
}
and
\al{INEQ6}{
	\sum_{|\alpha| \leq s-1} \la D^\alpha \Pbb_3 U, (\del{t} A^0) D^\alpha \Pbb_3 U \ra%= \la D^\alpha \Pbb_3 U, \Pbb (\del{t} A^0)\Pbb D^\alpha \Pbb_3 U \ra
	\leq -\frac{1}{t} C(K_1)\|\Pbb_1 w\|_{H^s}\|\Pbb_3 U\|^2_{\Hs}+C(K_1) \|\Pbb_3 U\|^2_{\Hs}.
}
\end{lemma}
\begin{proof}
Using the properties $\mathbb{P}_1^2=\mathbb{P}_1$, $\mathbb{P}_1 + \mathbb{P}_1^\perp = \mathds{1}$,
$\mathbb{P}_1^\textrm{T} = \mathbb{P}_1$, and $D\mathbb{P}_1 = 0$ of the projection matrix $\mathbb{P}_1$ repeatedly, we compute
\begin{align}
     & -\frac{2}{t}\sum_{|\alpha|\leq s}\langle D^\alpha w, A_1^0[(A_1 ^0)^{-1}\mathfrak{A}_1, D^\alpha]\mathbb{P}_1 w\rangle \notag \\
   &\hspace{0.5cm} = -\frac{2}{t}\sum_{|\alpha|\leq s}\langle D^\alpha \mathbb{P}_1 w, A_1^0[(A_1^0)^{-1}\mathfrak{A}_1, D^\alpha]\mathbb{P}_1 w\rangle
 -\frac{2}{t}\sum_{|\alpha|\leq s}\langle D^\alpha \mathbb{P}_1^\perp w, \mathbb{P}_1^\perp A_1^0 [(A_1^0)^{-1}\mathfrak{A}_1, D^\alpha]\mathbb{P}_1 w \rangle  \notag \\
 &\hspace{0.5cm}= -\frac{2}{t}\sum_{|\alpha|\leq s}\langle D^\alpha \mathbb{P}_1 w, A_1^0[(A_1^0)^{-1}\mathfrak{A}_1, D^\alpha]\mathbb{P}_1 w\rangle
 -\frac{2}{t}\sum_{|\alpha|\leq s}\langle D^\alpha \mathbb{P}_1^\perp w, \mathbb{P}_1^\perp A_1^0 [(A_1^0)^{-1}\mathbb{P}_1\mathfrak{A}_1, D^\alpha]\mathbb{P}_1 w \rangle  && \text{(by \eqref{E:COMMUTEPANFB})} \notag \\
&\hspace{0.5cm}= -\frac{2}{t}\sum_{|\alpha|\leq s}\langle D^\alpha \mathbb{P}_1 w, A_1^0[(A_1^0)^{-1}\mathfrak{A}_1, D^\alpha]\mathbb{P}_1 w\rangle
 -\frac{2}{t}\sum_{|\alpha|\leq s}\langle D^\alpha \mathbb{P}_1^\perp w, \mathbb{P}_1^\perp A_1^0\mathbb{P}_1^\perp [\mathbb{P}_1^\perp (A_1^0)^{-1}\mathbb{P}_1\mathfrak{A}_1, D^\alpha]\mathbb{P}_1 w \rangle   \nnb\\
&\hspace{7.5cm}-\frac{2}{t}\sum_{|\alpha|\leq s}\langle D^\alpha \mathbb{P}_1^\perp w, \mathbb{P}_1^\perp A_1^0\mathbb{P}_1 [\mathbb{P}_1(A_1^0)^{-1}\mathbb{P}_1 \mathfrak{A}_1, D^\alpha] \mathbb{P}_1 w\rangle. \notag
      \end{align}
From this expression,  we obtain, with the help the Cauchy-Schwarz inequality, the calculus inequalities from Appendix \ref{A:INEQUALITIES},
the expansions \eqref{E:DECOMPOSITIONOFA01}-\eqref{E:DECOMPOSITIONOFCALB},
the relations \eqref{DA0}, \eqref{E:PPERPB0P}, and \eqref{E:PPERPB0Pa}, and the inequality \eqref{K1ineq}, the estimate
\begin{align*}
&-\frac{1}{t}\sum_{|\alpha|\leq s}\langle D^\alpha w, A_1^0[(A_1 ^0)^{-1}\mathfrak{A}_1, D^\alpha]\mathbb{P}_1 w\rangle \notag
\\
&\hspace{0.5cm}  \lesssim \notag -\frac{1}{t}\biggl[\|A_1^0\|_{H^s} \|\mathbb{P}_1 w\|_{H^s}\| D\bigl((A_1^0)^{-1}\mathfrak{A}_1\bigr)]\|_{H^{s-1}}
+\|A_1^0\|_{H^s}
        \|\mathbb{P}_1^\perp w \|_{H^s} \|D\bigl(\mathbb{P}_1^\perp(A_1^0)^{-1}\mathbb{P}_1\mathfrak{A}_1\bigr)\|_{H^{s-1}}
         \notag \\
 & \hspace{0.5cm}  + \|\mathbb{P}_1^\perp A_1^0\mathbb{P}_1\|_{H^s}
        \|\mathbb{P}_1^\perp w\|_{H^s}  \|D\bigl(\mathbb{P}_1(A_1^0)^{-1}\mathbb{P}_1\mathfrak{A}_1\bigr)\|_{H^{s-1}}
\biggr] \|\mathbb{P}_1 w\|_{H^{s-1}}
        \leq  -C(K_1 )\frac{1}{t}\|w\|_{H^s}
        \|\mathbb{P}_1 w\|_{H^s}^2 %\label{E:1}
      \end{align*}
for $T_0\leq t < T_*$, where the constant $C(K_1)$ is independent of $\epsilon \in (0,\epsilon_0)$ and $T_* \in (T_0,T_1]$. This establishes
the estimate \eqref{E:INEQ1a}. By a similar calculation, we find that
      \begin{align*}
     & -\frac{2}{t}\sum_{|\alpha|\leq s-1}\langle D^\alpha u, A_2^0[(A_2 ^0)^{-1}\mathfrak{A}_2, D^\alpha]\mathbb{P}_2 u\rangle
=-\frac{2}{t}\sum_{|\alpha|\leq s-1}\langle D^\alpha \mathbb{P}_2u, A_2^0[(A_2^0)^{-1}\mathfrak{A}_2, D^\alpha]\mathbb{P}_2u\rangle
\notag \\
 &\hspace{0.5cm} -\frac{2}{t}\sum_{|\alpha|\leq s-1}\langle D^\alpha \mathbb{P}_2^\perp u, \mathbb{P}_2^\perp A_2^0\mathbb{P}_2^\perp [\mathbb{P}_2^\perp (A_2^0)^{-1}\mathbb{P}_2  \mathfrak{A}_2, D^\alpha]\mathbb{P}_2 u\rangle
-\frac{2}{t}\sum_{|\alpha|\leq s-1}\langle D^\alpha \mathbb{P}_2^\perp u, \mathbb{P}_2^\perp A_2^0\mathbb{P}_2 [\mathbb{P}_2(A_2^0)^{-1}\mathbb{P}_2\mathfrak{A}_2, D^\alpha]\mathbb{P}_2u\rangle   \nnb\\
        &\hspace{0.5cm}\leq  -\frac{1}{t}C(K_1) \|w\|_{H^s}\|\mathbb{P}_2u\|^2_{\Hs}
        -\frac{1}{t}C(K_1)\|u\|_{\Hs}\|\mathbb{P}_1 w\|_{H^s}\|\mathbb{P}_2u\|_{\Hs}
       -\frac{1}{t}C(K_1)\|u\|_{\Hs}\|\mathbb{P}_1 w\|_{H^s}\|\mathbb{P}_2u\|_{\Hs}   \nnb\\
        &\hspace{6.0cm}\leq  -\frac{1}{t}C(K_1 ) (\|u\|_{\Hs}+\|w\|_{H^s})(\|\mathbb{P}_2u\|^2_{\Hs}+\|\mathbb{P}_2w\|^2_{H^s}),  %\label{E:2}
      \end{align*}
which establishes the estimate \eqref{E:INEQ1b}.

Next, using the calculus inequalities from Appendix  \ref{A:INEQUALITIES}, we observe that
\begin{equation*}\label{E:COMMUTATOR2}
  \begin{aligned}
    \sum_{0\leq |\alpha|\leq s-1}\langle D^\alpha u, -A_2^0[ D^\alpha,(A_2^0)^{-1} A_2^i]\partial_i u\rangle
    \lesssim & \|A_2^0\|_{L^\infty}\|u\|_{\Hs}^2\| D((A_2^0)^{-1}A_2^i)\|_{\Hs}
    \leq    C(K_1 )\|u\|_{\Hs}^2,
  \end{aligned}
\end{equation*}
which establishes the estimate \eqref{E:INEQ2b}. Since the estimates \eqref{E:INEQ2a}, \eqref{E:INEQ3a} and \eqref{E:INEQ3b}
can be obtained in a similar fashion, we omit the details.

Finally, we consider the estimates \eqref{E:INEQ5a}-\eqref{E:INEQ5b}. We begin establishing these estimates by
writing \eqref{E:MODELEQ1a} as
\begin{align*}%\label{E:EPSILONPTU}
  \epsilon \partial_0 w=\epsilon\frac{1}{t}(A_1^0)^{-1} \mathfrak{A}_1\mathbb{P}_1 w -\epsilon(A_1^0)^{-1} A_1^i \partial_i w- (A_1^0)^{-1} C_1^i\partial_i w+\epsilon(A_1^0)^{-1} H_1 +\epsilon(A_1^0)^{-1} F_1.
\end{align*}
Using this  and the expansion \eqref{E:DECOMPOSITIONOFA01}, we can express the
time derivatives $\del{t}A^0_a$, $a=1,2$, as
\begin{align}
  \partial_t A_a^0=&  D_w A_a^0 \cdot \partial_t w +D_t A_a^0    \nnb\\
  =&-D_w A_a^0\cdot(A_1^0)^{-1} A_1^i \partial_i w -[ D_w \tilde{A}_a^0 \cdot (A_1^0)^{-1} C_1^i\partial_i w] \nnb\\
  &+[ D_w A_a^0 \cdot (A_1^0)^{-1} H_1]+D_t A_a^0   +[ D_w A_a^0\cdot(A_1^0)^{-1} F_1] + \frac{1}{t}[D_w A_a^0\cdot (A_1^0)^{-1} \mathfrak{A}_1\mathbb{P}_1 w]. \label{E:PTB0}
\end{align}
Using \eqref{E:PTB0} with  $a=2$, we see, with the help of the calculus inequalities from Appendix  \ref{A:INEQUALITIES},
the Cauchy-Schwarz inequality,
the estimate \eqref{E:F_I}, and  the expansion \eqref{E:ADEC} for $a=2$, that
  \begin{align*}
    \sum_{|\alpha|\leq s-1}\langle D^\alpha u, (\partial_tA_2^0) D^\alpha u\rangle
    \leq &\sum_{|\alpha|\leq s-1}\left[\langle D^\alpha u, \mathbb{P}_2^\perp(\partial_tA_2^0)\mathbb{P}_2^\perp D^\alpha u\rangle+\langle D^\alpha u, \mathbb{P}_2^\perp(\partial_tA_2^0)\mathbb{P}_2 D^\alpha u\rangle \right.\nnb\\
&\hspace{2.0cm} \left.+\langle D^\alpha u, \mathbb{P}_2(\partial_tA_2^0)\mathbb{P}_2^\perp D^\alpha u\rangle
    +\langle D^\alpha u, \mathbb{P}_2(\partial_tA_2^0)\mathbb{P}_2 D^\alpha u\rangle\right]\nnb\\
    \leq & C(K_1) \|u\|^2_{\Hs}    -\frac{2}{t}\|u\|_{\Hs}\|(A_1^0)^{-1}\mathfrak{A}_1\|_{L^\infty}\|D_wA_2^0\|_{L^\infty}
    \|\mathbb{P}_2u\|_{\Hs}\|\mathbb{P}_1w\|_{\Hs}\nnb\\
    &-\frac{1}{t}\|\mathbb{P}_1w\|_{H^s}\|(A^0)^{-1}\mathfrak{A}\|_{L^\infty}\|D_w A_2^0\|_{L^\infty}
    \|\mathbb{P}_2u
    \|_{\Hs}^2 - \frac{1}{t}\|u\|_{\Hs}^2C(K_1)\|\mathbb{P}_1w\|^2_{\Hs}\nnb\\
    \leq &  C(K_1 )  \|u\|^2_{\Hs}-\frac{1}{t}C(K_1, K_2) (\|u\|_{\Hs}+\|w\|_{H^s})
    (\|\mathbb{P}_2u\|^2_{\Hs}+\|\mathbb{P}_1w\|^2_{H^s}). \notag
  \end{align*}
This establishes the estimate \eqref{E:INEQ5b}. Since the estimate \eqref{E:INEQ5a} can be established using similar
arguments, we omit the details. The last estimate \eqref{E:INEQ6} can also be established using similar arguments with
the help of the identity $\Pbb_3\Pbb=\Pbb\Pbb_3=\Pbb_3$. We again omit the details.
\end{proof}

Applying $A^0 D^\alpha (A^0)^{-1}$ to both sides of \eqref{E:MODELEQ2a}, we find that
\begin{align}
A^0\partial_0 D^\alpha U+A^i\partial_i D^\alpha U+  \frac{1}{\epsilon}C^i\partial_i D^\alpha U
  =&-A ^0[ D^\alpha,(A ^0)^{-1}A^i]\partial_i U-[\tilde{A}^0, D^\alpha](A^0)^{-1}C^i\partial_i U \notag \\
  &+\frac{1}{t}\mathfrak{A} D^\alpha \mathbb{P} U +\frac{1}{t}A ^0[ D^\alpha, (A ^0)^{-1}\mathfrak{A}]\mathbb{P} U +A ^0 D^\alpha[(A ^0)^{-1} H],
\label{E:ABSOLUTEPTB0}
\end{align}
where in deriving this we have used
\begin{align*}
 \frac{1}{\epsilon}[A^0, D^\alpha](A ^0)^{-1}C^i\partial_i U &\overset{\eqref{E:DECOMPOSITIONOFA01}}{=}\frac{1}{\epsilon}[\mathring{A}^0+\epsilon\tilde{A}^0, D^\alpha](A ^0)^{-1}C^i\partial_i U=[\tilde{A}^0, D^\alpha](A ^0)^{-1}C^i\partial_i U
\intertext{and}
A^0[D^\alpha,(A^0)^{-1}]C^i\del{i}U &= A^0 D^\alpha \bigl( (A^0)^{-1}C^i\del{i}U\bigr)- D^\alpha \bigl(C^i\del{i}U\bigr) \\
&=  A^0 D^\alpha \bigl( (A^0)^{-1}C^i\del{i}U\bigr) - D^\alpha( A^{0} (A^0)^{-1} C^i \del{i} U \bigr) = [A^0, D^\alpha] (A^0)^{-1} C^i \del{i} U .
\end{align*}
Writing $A^0_a$, $a=1,2$, as $A ^0_a=(A ^0_a)^{\frac{1}{2}}(A ^0_a)^{\frac{1}{2}}$, which we can do since
 $A ^0_a$ is a real symmetric and positive-definite, we see from \eqref{E:KAPPAB0CALB} that
\begin{align} \label{Afrbound}
  (A ^0_a)^{-\frac{1}{2}}\mathfrak{A_a}(A_a^0)^{-\frac{1}{2}}\geq \kappa\mathds{1}.
\end{align}
Since, by \eqref{E:COMMUTEPANFB},
\begin{align*}
  \frac{2}{t} \langle D^\alpha f,\mathfrak{A}_a D^\alpha \mathbb{P}_a f\rangle=&\frac{2}{t} \langle D^\alpha \mathbb{P}_a f, (A^0)^{\frac{1}{2}}[(A^0_a)^{-\frac{1}{2}}\mathfrak{A}_a(A^0_a)^{-\frac{1}{2}}](A^0_a)^{\frac{1}{2}} D^\alpha \mathbb{P}_a f\rangle,
\quad a=1,2,
\end{align*}
it follows immediately from \eqref{Afrbound} that
\begin{align} \label{E:KAPPACONTR}
  \frac{2}{t}\sum_{0\leq |\alpha|\leq s-1}\langle D^\alpha u,\mathfrak{A}_2 D^\alpha \mathbb{P}_2  u\rangle \leq  \frac{2\kappa}{t}\vertiii{\mathbb{P}_2 u}^2_{2,H^{s-1}}\AND
  \frac{2}{t}\sum_{0\leq |\alpha|\leq s }\langle D^\alpha w,\mathfrak{A}_1 D^\alpha \mathbb{P}_1  w\rangle \leq \frac{2\kappa}{t}\vertiii{\mathbb{P}_1 w}^2_{1,H^s}.
\end{align}
Then, differentiating $\langle D^\alpha w, A^0_1 D^\alpha w\rangle$ with respect to $t$, we see, from
the identities  $\langle D^\alpha w, C^i_1\partial_i D^\alpha w\rangle = 0$
and $2 \langle D^\alpha w, A^i_1\partial_i D^\alpha w\rangle =  - \langle D^\alpha w, (\partial_i A^i_1) D^\alpha w\rangle$, the block decomposition of
\eqref{E:ABSOLUTEPTB0}, which
we can use to determine $ D^\alpha \del{t}w$, the estimates \eqref{E:F_I} and \eqref{E:KAPPACONTR} together with
those
from Lemma \ref{L:PREEST} and the calculus inequalities from Appendix  \ref{A:INEQUALITIES}, that
\begin{align}\label{E:DTW1}
  \partial_t\vertiii{w}^2_{1,H^s}=&\sum_{0\leq |\alpha|\leq s}\langle D^\alpha w, (\partial_tA_1^0) D^\alpha w\rangle+2\sum_{0\leq |\alpha|\leq s}\langle D^\alpha w, A_1^0 D^\alpha \partial_t w\rangle   \nnb \\
  \leq &  C(K_1 ) \|w\|^2_{H^s}-\frac{1}{t}C(K_1) \|w\|_{H^s}
     \|\mathbb{P}_1w\|^2_{H^s} +\sum_{0\leq |\alpha|\leq s}\langle D^\alpha w, (\del{i}A_1^i) D^\alpha w\rangle\nnb\\
     & -\frac{2}{\epsilon}\sum_{0\leq |\alpha|\leq s} \overset{\quad =0}{\overbrace{\langle D^\alpha w, C_1^i\partial_i D^\alpha w\rangle}}
-2\sum_{0\leq |\alpha|\leq s}\langle D^\alpha w, A_1^0[ D^\alpha,(A_1^0)^{-1}A_1^i]\partial_i w\rangle\nnb\\
     &-2\sum_{0\leq |\alpha|\leq s}\langle D^\alpha w, [\tilde{A}_1^0, D^\alpha](A_1^0)^{-1}C_1^i\partial_i w\rangle
  +\frac{2}{t}\sum_{0\leq |\alpha|\leq s}\langle D^\alpha w,\mathfrak{A}_1 D^\alpha \mathbb{P}_1 w\rangle  \nnb\\
  &+\frac{2}{t}\sum_{0\leq |\alpha|\leq s}\langle D^\alpha w, A_1^0[(A_1^0)^{-1}\mathfrak{A}_1, D^\alpha] \mathbb{P}_1 w\rangle+2\sum_{0\leq |\alpha|\leq s}\langle D^\alpha w, A_1^0 D^\alpha[(A_1^0)^{-1} (H_1+ F_1)]\rangle \nnb \\
  \leq & C(K_1)\vertiii{w}^2_{1,H^s}
  +\frac{1}{t}\bigl[2 \kappa -C_1(K_1) \|w\|_{H^s} \bigr] \vertiii{\mathbb{P}_1 w}^2_{1,H^s}
\end{align}
for $t\in [T_0, T_*)$.
By similar calculation, we obtain from differentiating $\langle D^\alpha u, A^0_2 D^\alpha u\rangle$ with respect to $t$ the estimate
\begin{align}
  \partial_t\vertiii{u}^2_{2,\Hs} =&\sum_{0\leq |\alpha|\leq s-1}\langle D^\alpha u, (\partial_t A_2^0) D^\alpha u \rangle+2\sum_{0\leq |\alpha|\leq s-1}\langle D^\alpha u, A_2^0 D^\alpha \partial_t u\rangle   \nnb  \\
  \leq &  C(K_1 ) \|u\|^2_{\Hs}-\frac{1}{t}C(K_1,K_2) (\|u\|_{\Hs}+\|w\|_{H^s})
  (\|\mathbb{P}_2u\|^2_{\Hs}+\|\mathbb{P}_1w\|^2_{H^s})   \nnb  \\
  & \sum_{0\leq |\alpha|\leq s-1}\langle D^\alpha u, (\del{i} A_2^i) D^\alpha u\rangle -\frac{2}{\epsilon}\sum_{0\leq |\alpha|\leq s-1}
\overset{\quad = 0}{ \overbrace{\langle D^\alpha u, C_2^i\partial_i D^\alpha u\rangle } }  \nnb  \\
  &-2\sum_{0\leq |\alpha|\leq s-1}\langle D^\alpha u, A_2^0[ D^\alpha,(A_2^0)^{-1}A_2^i]\partial_i u\rangle
   -2\sum_{0\leq |\alpha|\leq s-1}\langle D^\alpha u, [\tilde{A}_2^0, D^\alpha](A_2^0)^{-1}C_2^i\partial_i u\rangle
   \nnb   \\
  &+\frac{2}{t}\sum_{0\leq |\alpha|\leq s}\langle D^\alpha u,\mathfrak{A}_2 D^\alpha \mathbb{P}_2 u\rangle-\frac{2}{t}\sum_{0\leq |\alpha|\leq s-1}\langle D^\alpha u, A_2^0[(A_2^0)^{-1}\mathfrak{A}_2, D^\alpha]\mathbb{P}_2u\rangle \nnb  \\
  &+2\sum_{0\leq |\alpha|\leq s-1}\left\langle D^\alpha u, A_2^0 D^\alpha[(A_2^0)^{-1} \bigr(H_2+ \frac{1}{t} M_2 \Pbb_3 U +F_2  \bigr)]\right\rangle \nnb  \\
  \leq & C(K_1,K_2)(\vertiii{u}^2_{2,\Hs}+ \vertiii{w}_{1,H^s}^2)-\frac{1}{2t}C_2(K_1,K_2)(
    \|u\|_{\Hs}+\|w\|_{H^s})\vertiii{\mathbb{P}_1w}_{1,H^s}^2  \nnb  \\
  &+\frac{1}{t}\bigl[2\kappa-C_2(K_1,K_2)(\|u\|_{\Hs}+\|w\|_{H^s})\bigr] \vertiii{\mathbb{P}_2u}^2_{2,\Hs}
- C(K_1)\frac{1}{t} (\vertiii{u}^2_{2,\Hs}+\vertiii{w}^2_{1,H^s})\vertiii{\Pbb_3 U}_{\Hs} \label{E:DTU1}
\end{align}
for $t\in [T_0, T_*)$.

Applying the operator $A^0 D^\alpha \Pbb^3(A^0)^{-1}$ to \eqref{E:MODELEQ2a}, we see,
with the help of \eqref{E:P32a}-\eqref{E:P32c}, that
\al{PROJEQ}{
A^0\partial_0 D^\alpha \Pbb_3 U+\Pbb_3 A^i\Pbb_3 \partial_i D^\alpha\Pbb_3  U+  \frac{1}{\epsilon}\Pbb_3 C^i\Pbb_3 \partial_i D^\alpha \Pbb_3 U
  =&-A ^0[ D^\alpha,(A ^0)^{-1}\Pbb_3 A^i\Pbb_3 ]\partial_i \Pbb_3 U \notag \\
-[\tilde{A}^0, D^\alpha](A^0)^{-1}\Pbb_3 C^i\Pbb_3 \partial_i \Pbb_3 U
  +\frac{1}{t}\Pbb_3 \mathfrak{A} \Pbb_3 D^\alpha  \Pbb_3 U &+\frac{1}{t}A ^0[ D^\alpha, (A ^0)^{-1}
\Pbb_3 \mathfrak{A}\Pbb_3 ]\Pbb_3  U
  +A ^0 D^\alpha[(A ^0)^{-1}\Pbb_3  H].  }
Then, by similar arguments used to derive \eqref{E:DTW1} and \eqref{E:DTU1}, we obtain from \eqref{E:PROJEQ} the
estimate
\begin{align*}
  \partial_t\vertiii{\Pbb_3 U}^2_{\Hs}   =&\sum_{0\leq |\alpha|\leq s-1}\langle D^\alpha \Pbb_3 U, (\partial_t A^0) D^\alpha \Pbb_3 U \rangle+2\sum_{0\leq |\alpha|\leq s-1}\langle D^\alpha \Pbb_3 U, \Pbb_3 A^0 \Pbb_3 D^\alpha \partial_t \Pbb_3 U\rangle    \nnb  \\
  \leq &  -\frac{1}{t} C(K_1)\|\Pbb_1 w\|_{H^s}\|\Pbb_3 U\|^2_{\Hs}+C(K_1) \|\Pbb_3 U\|^2_{\Hs}  \nnb   \\
  & +\sum_{0\leq |\alpha|\leq s-1}\langle D^\alpha \Pbb_3 U, (\del{i}A^i) D^\alpha \Pbb_3 U\rangle
-\frac{2}{\epsilon}\sum_{0\leq |\alpha|\leq s-1} \overset{\quad =0}{\overbrace{\langle D^\alpha \Pbb_3 U, C ^i\partial_i D^\alpha \Pbb_3 U\rangle }}
 \nnb  \\
  &-2\sum_{0\leq |\alpha|\leq s-1}\langle D^\alpha \Pbb_3 U, A^0[ D^\alpha,(A ^0)^{-1}A ^i]\partial_i \Pbb_3 U+[\tilde{A}^0, D^\alpha](A^0)^{-1}C^i\partial_i \Pbb_3 U\rangle
   \nnb  \\
  &+\frac{2}{t}\sum_{0\leq |\alpha|\leq s}\langle D^\alpha \Pbb_3 U,\mathfrak{A} D^\alpha \Pbb_3 U\rangle+\frac{2}{t}\sum_{0\leq |\alpha|\leq s-1}\langle D^\alpha \Pbb_3 U, A^0[(A^0)^{-1}\mathfrak{A}, D^\alpha]\Pbb_3 U\rangle   \nnb  \\
  &\hspace{5.5cm} +2\sum_{0\leq |\alpha|\leq s-1}\left\langle D^\alpha \Pbb_3 U, A^0 D^\alpha[(A^0)^{-1} \Pbb_3 H]\right\rangle   \nnb  \\
  \leq &  C(K_1 ) \|\Pbb_3 U\|^2_{\Hs}  +C(K_1)\|\Pbb_3 U\|_{\Hs}\Bigl(\| H_1\|_{\Hs}+\|H_2\|_{\Hs}+\| F_1\|_{\Hs} \nnb \\
  &\hspace{1.5cm}+\|F_2\|_{\Hs}\Bigr)
 +\frac{1}{t}\Bigl(2 \kappa-C_2(K_1,K_2 ) \bigl(\|w\|_{H^s}+\|u\|_{\Hs}\bigr)\Bigr)\vertiii{\Pbb_3 U}^2_{\Hs} \nnb\\
&\leq  C(K_1 ) \vertiii{\Pbb_3 U}^2_{\Hs}  + C(K_1,K_2)\bigl(\vertiii{w}_{1,H^s}+\vertiii{u}_{2,\Hs})\bigr)\vertiii{\Pbb_3 U}_{\Hs} \nnb \\
&\hspace{4.1cm}+\frac{1}{t}\Bigl(2 \kappa-C_2(K_1,K_2 ) \bigl(\|w\|_{H^s}+\|u\|_{\Hs}\bigr)\Bigr)\vertiii{\Pbb_3 U}^2_{\Hs}.
\end{align*}
Dividing the above estimate by $ \vertiii{\Pbb_3 U}_{\Hs}$ gives
\begin{align}
 \partial_t\vertiii{\Pbb_3 U}_{\Hs}  &\leq  C(K_1 ) \vertiii{\Pbb_3 U}_{\Hs}  + C(K_1,K_2)\bigl(\vertiii{w}_{1,H^s}+
\vertiii{u}_{2,\Hs})\bigr) \notag \\
&\hspace{2.8cm}+\frac{1}{t}\biggl(\kappa-\frac{C_2(K_1,K_2 )}{2} \bigl(\|w\|_{H^s}+\|u\|_{\Hs}\bigr)\biggr)\vertiii{\Pbb_3 U}_{\Hs}.
\label{E:DTP3U}
\end{align}

Next, we choose $\sigma>0$ small enough so that
\begin{equation*}
\Bigl(C_1(\hat{R}) + 2C_2(\hat{R},\hat{R})\Bigr)
\sigma < \frac{\kappa}{2}
\end{equation*}
in addition to \eqref{sigmaC1}.
Then since
\begin{equation*}
2\kappa - \Bigl(C_1(K_1(T_0))\norm{w(T_0)}_{H^s} + C_2(K_1(T_0),K_2(T_0))\bigl(\norm{w(T_0)}_{H^s}
+ \norm{u(T_0)}_{H^{s-1}}\bigr)\Bigr)  > \kappa,
\end{equation*}
we see by continuity that either
\begin{equation*}
2\kappa - \Bigl(C_1(K_1(t))\norm{w(t)}_{H^s} + C_2(K_1(t),K_2(t))\bigl(\norm{w(t)}_{H^s}
+ \norm{u(t)}_{H^{s-1}}\bigr)\Bigr)  > \kappa, \quad 0\leq t < T_*,
\end{equation*}
or else there exists a first time $T^* \in (0,T_*)$ such that
\begin{equation*}
2\kappa - \Bigl(C_1(K_1(T^*))\norm{w(T^*)}_{H^s} + C_2(K_1(T^*),K_2(T^*))(\norm{w(T^*)}_{H^s}
+ \norm{u(T^*)}_{H^{s-1}}\Bigr)  = \kappa.
\end{equation*}
Thus if we let $T^*=T_*$ if the first case holds, then we have that
\begin{equation} \label{E:qq}
2\kappa - \Bigl(C_1(K_1(t))\norm{w(t)}_{H^s} + C_2(K_1(t),K_2(t))\bigl(\norm{w(t)}_{H^s}
+ \norm{u(t)}_{H^{s-1}}\bigr)\Bigr)  > \kappa, \quad 0\leq t < T^*\leq T_*.
\end{equation}
Taken together, the estimates \eqref{K1ineq}, \eqref{E:DTW1}, \eqref{E:DTU1}, \eqref{E:DTP3U} and \eqref{E:qq} imply that
\begin{align}
\partial_t\vertiii{w}^2_{1,H^s} \leq & C(\hat{R})\vertiii{w}^2_{1,H^s}
  +\frac{\kappa}{t} \vertiii{\mathbb{P}_1 w}^2_{1,H^s}, \label{E:ENERGEST1}\\
\partial_t\vertiii{u}^2_{2,\Hs} \leq & C(\hat{R} )\bigl(\vertiii{u}^2_{2,\Hs}+\vertiii{w}_{1,H^s}^2\bigr) - \frac{1}{t} C_3(\hat{R}) \bigl(\vertiii{u}^2_{2,\Hs}+\vertiii{w}^2_{1,H^s}\bigr) \vertiii{\Pbb_3 U}_{\Hs} \nnb \\
  & \hspace{6.8cm}
+\frac{\kappa}{2t}\vertiii{\mathbb{P}_1w}_{1,H^s}^2+\frac{\kappa}{t} \vertiii{\mathbb{P}_2u}^2_{2,\Hs}
 \label{E:ENERGEST2}
\intertext{and}
\partial_t\vertiii{\Pbb_3 U} _{\Hs} \leq &   C(\hat{R})\bigl(\vertiii{\Pbb_3 U}_{\Hs}  + \vertiii{w}_{1,H^s}+\vertiii{u}_{2,\Hs}\bigr)+
\frac{\kappa}{2t} \vertiii{\Pbb_3 U}_{\Hs}  \label{E:ENERGEST3}
\end{align}
for $0\leq t < T^* \leq T_*$.

Next, we set
\begin{align*}
X = \vertiii{w}^2_{1,H^s} + \vertiii{u}^2_{2,\Hs}, \quad Y = \vertiii{\Pbb_1 w}^2_{1,H^s} + \vertiii{\Pbb_2 u}^2_{2,\Hs},
\AND
Z = \vertiii{\Pbb_3U}_{\Hs}.
\end{align*}
Since $C_3(\hat{R})X(T_0)/\sigma \leq C(\hat{R})\sigma$, we can choose $\sigma$ small enough so that
$C_3(\hat{R})X(T_0)/\sigma < \kappa/4$. Then by continuity, either $ C_3(\hat{R})X(t)/\sigma \leq \kappa/4$ for $t\in [T_0,T^*)$,
or else there exists a first time $T\in (T_0,T^*)$ such that $C_3(\hat{R})X(T)/\sigma = \kappa/4$. Thus if we set
$T=T^*$ if the first case holds, then we have that
\begin{equation} \label{Tdef}
C_3(\hat{R})\frac{X(t)}{\sigma} < \kappa/4, \quad T_0 \leq t < T\leq T^* \leq T_*.
\end{equation}
Adding the inequalities \eqref{E:ENERGEST1} and \eqref{E:ENERGEST2} and dividing the results by $\sigma$, we
obtain, with the help of \eqref{Tdef}, the inequality
\begin{equation} \label{Xest1}
\del{t}\biggl(\frac{X}{\sigma}\biggr) \leq C(\hat{R})\frac{X}{\sigma} -\frac{\kappa}{4t}Z +\frac{\kappa}{2t}\frac{Y}{\sigma}, \quad T_0 \leq t < T\leq T^*\leq T_*,
\end{equation}
while the inequality
\begin{equation} \label{Zest1}
\del{t}Z \leq C(\hat{R})\biggl( Z + \sigma + \frac{X}{\sigma}\biggr) + \frac{\kappa}{2t}Z, \quad T_0 \leq t < T^*\leq T_*
\end{equation}
follows from \eqref{E:ENERGEST3} and Young's inequality. Adding \eqref{Xest1} and \eqref{Zest1}, we find that
\begin{equation} \label{XZest1}
\del{t}\biggl(\frac{X}{\sigma} +Z -\frac{\kappa}{4}\int_{T_0}^t \frac{1}{\tau}\biggl(\frac{Y}{\sigma}+Z\biggr)\, d\tau  + \sigma \biggr)
\leq C(\hat{R})  \biggl(\frac{X}{\sigma} +Z -\frac{\kappa}{4}\int_{T_0}^t \frac{1}{\tau}\biggl(\frac{Y}{\sigma}+Z\biggr)
\, d\tau  + \sigma \biggr)
\end{equation}
for $T_0 \leq t < T\leq T^*\leq T_*$. Since $X(T_0)\leq C(\hat{R})\sigma^2$ and $Z(T_0) \lesssim \sigma$, it follows
directly from  \eqref{XZest1} and Gr\"{o}nwall's inequality that
\begin{equation*}
\frac{X}{\sigma} +Z -\frac{\kappa}{4}\int_{T_0}^t \frac{1}{\tau}\biggl(\frac{Y}{\sigma}+Z\biggr)\, d\tau  + \sigma
\leq e^{C(\hat{R})(t-T_0)}C(\hat{R})\sigma,  \quad T_0 \leq t < T\leq T^*\leq T_*,
\end{equation*}
from which it follows that
\begin{equation} \label{XZest2}
 \|w\|_{M^\infty_{\Pbb_1, s}([T_0,t)\times \Tbb^n)}+\|u\|_{M^\infty_{\Pbb_2, s-1}([T_0,t)\times \Tbb^n)} -
\int_{T_0}^{t} \frac{1}{\tau} \|\Pbb_3 U\|_{\Hs}\, d\tau  \leq C(\hat{R})\sigma, \quad T_0 \leq t < T\leq T^*\leq T_*,
\end{equation}
where we stress that the constant $C(\hat{R})$ is independent of $\epsilon$ and the times $T$, $T^*$, $T_*$, and $T_1$. Choosing
$\sigma$ small enough, it is then clear from the estimate \eqref{XZest2} and the definition of the times $T$, $T^*$, and $T_1$
that $T=T^*=T_*=T_1$, which completes the proof.
\end{proof}

\subsection{Error estimates\label{S:MODELerr}}
In this section, we consider solutions of the singular initial value problem
  \begin{align}
  A_1^0 (\epsilon,t,x,w)\partial_0 w+A_1^i (\epsilon,t,x,w)\partial_i w+\frac{1}{\epsilon}C_1^i\partial_i w&=\frac{1}{t}\mathfrak{A}_1(\epsilon,t,x,w)\mathbb{P}_1  w+H_1+F_1  &&\mbox{in} \quad[T_0, T_1)\times\mathbb{T}^n, \label{E:MODELEQ3a}\\
  w(x) |_{t=T_0}&= \mrw^0(x) +\epsilon s^0(\epsilon,x) &&\text{in} \quad \{T_0\}\times\mathbb{T}^n, \label{E:MODELEQ3b}
  \end{align}
where the matrices $A_1^0$, $A_1^i$, $i=1,\ldots,n$, and $\mathfrak{A}_1$ and the source terms $H_1$ and $F_1$ satisfy the conditions from
Assumption \ref{ASS1}. Our aim is to use the uniform a priori estimates from Theorem \ref{L:BASICMODEL} to establish
uniform a priori estimates for solutions of \eqref{E:MODELEQ3a}-\eqref{E:MODELEQ3b} and to establish an error estimate
between solutions of  \eqref{E:MODELEQ3a}-\eqref{E:MODELEQ3b} and solutions of the \textit{limit equation}, which
is defined by
\begin{align}
  \mathring{A}_1^0\partial_0 \mathring{w}+\mathring{A}_1^i\partial_i \mathring{w}&=\frac{1}{t}\mathring{\mathfrak{A}}_1\mathbb{P}_1\mathring{w}-C_1^i\partial_i v+\mathring{H}_1+\mathring{F}_1
&& \mbox{in} \quad[T_0, T_1)\times\mathbb{T}^n,  \label{E:LIMITINGEQa}\\
  C_1^i\partial_i\mathring{w}&=0 && \mbox{in} \quad[T_0, T_1)\times\mathbb{T}^n, \label{E:LIMITINGEQb}\\
  \mathring{w}(x)|_{t=T_0}&=\mathring{w}^0(x) &&\text{in} \quad \{T_0\}\times\mathbb{T}^n. \label{E:LIMITINGEQc}
\end{align}
In this system, $\mathring{A}^0_1$ and $\mathring{\mathfrak{A}}_1$ are defined by \eqref{E:DECOMPOSITIONOFA01} and
\eqref{E:DECOMPOSITIONOFCALB} with $a=1$, respectively,
$\mathring{A}_1^i$ and $\mathring{H}_1$ are defined by the limits
\begin{equation} \label{E:HRIN}
			\mathring{A}_1^i(t,x, \mathring{w})=\lim_{\epsilon\searrow 0}A_1^i(\epsilon,t,x, \mathring{w}) \AND
\mathring{H}_1(t,x, \mathring{w})=\lim_{\epsilon\searrow 0} H_1(\epsilon,t,x,\mathring{w}),
\end{equation}
respectively, and the following assumptions hold for fixed constants
$R >0$, $T_0 < T_1 <0$ and $s\in \Zbb_{>n/2+1}$:

\begin{ass}\label{ASS3}$\;$

\begin{enumerate}
\item \label{A3a} The source terms\footnote{The source term $\mathring{F}_1$ should be thought of as the $\epsilon \searrow 0$ limit of
$F_1$. This is made precise by the hypothesis \eqref{HFLip} of Theorem  \ref{T:MAINMODELTHEOREM}.}
$\mathring{F}_1$ and $v$ satisfy  $\mathring{F}_1 \in C^0\bigl([T_0,T_1),H^s(\Tbb^n,\Rbb^{N_1})\bigr)$
and $v\in \bigcap_{\ell=0}^1 C^\ell \bigl([T_0,T_1),H^{s+1-\ell}(\Tbb^n,\Rbb^{N_1})\bigr)$.
\item \label{A3b} The matrices $\mathring{A}_1^i$, $i=1,\ldots, n$ and the source term $\mathring{H}_1$ satisfy\footnote{From
the assumptions, see Assumption \ref{ASS1}.\eqref{A:GH}-\eqref{A:Bi}, on $A_1^i$ and $H_1$, it follows directly from
the \eqref{E:HRIN} that
 $\mathring{A}_1^i \in  E^0\big((2T_0,0)\times \Tbb^n \times B_R\bigl(\Rbb^{N_1}\bigr), \mathbb{S}_{N_1}\bigr)$
and $\mathring{H}_1 \in  E^0\big((2T_0,0)\times \Tbb^n \times B_R\bigl(\Rbb^{N_1}\bigr), \mathbb{R}^{N_1}\bigr)$.
}
$t\mathring{A}_1^i \in E^1\big((2T_0,0)\times \Tbb^n \times B_R\bigl(\Rbb^{N_1}\bigr), \mathbb{S}_{N_1}\bigr)$,
$t\mathring{H}_1  \in E^1\big((2T_0,0)\times \Tbb^n \times B_R\bigl(\Rbb^{N_1}\bigr), \mathbb{R}^{N_1}\bigr)$,
and
\begin{equation*}
    D_{t}\bigl(t\mathring{H}_1(t,x,0\bigr)\bigr)=0.
\end{equation*}
%\item \label{A3c} The matrix  $\Pbb_1\mathring{\mathfrak{A}}_1\Pbb_1$ is symmetric, that is
%\begin{equation}
%\Pbb_1\mathring{\mathfrak{A}}_1(t)\Pbb_1 = \Pbb_1\mathring{\mathfrak{A}}_1(t)^{\mathrm{T}}\Pbb_1
%\end{equation}
%for all $t\in (2T_0,0)$.
\end{enumerate}
\end{ass}

We are now ready to state and establish uniform a priori estimates for solutions of the singular initial value problem \eqref{E:MODELEQ3a}-\eqref{E:MODELEQ3b} and the associated limit equation defined by
\eqref{E:LIMITINGEQa}-\eqref{E:LIMITINGEQc}.
\begin{theorem}\label{T:MAINMODELTHEOREM}
Suppose $R>0$, $s\in\mathbb{Z}_{>n/2+1}$,  $T_0< T_1 \leq 0$, $\epsilon_0>0$, $\mrw^0\in H^s(\Tbb^n, \Rbb^M)$,
$s^0\in L^\infty\bigl((0,\epsilon_0),H^s(\Tbb^n, \Rbb^{N_1})\bigr)$, Assumptions \ref{ASS1} and \ref{ASS3}
hold,  the maps
\als%{SPWMRW}
{
	(w, \mrw) \in \bigcap_{\ell=0}^1 C^\ell\bigl([T_0,T_1),H^{s-\ell}\bigl(\Tbb^n, \Rbb^{N_1}\bigr)\bigr) \times
\bigcap_{\ell=0}^1 C^\ell\bigl([T_0,T_1),H^{s-\ell}\bigl(\Tbb^n, \Rbb^{N_1}\bigr)\bigr)
	}
define a solution to the initial value problems
\eqref{E:MODELEQ3a}-\eqref{E:MODELEQ3b} and \eqref{E:LIMITINGEQa}-\eqref{E:LIMITINGEQc}, and
for $t\in [T_0,T_1)$, the following estimate holds:
\begin{align}
\|v(t)\|_{H^{s+1}}-\frac{1}{t}\|\Pbb_1 v(t)\|_{H^{s+1}}+\|\del{t} v(t)\|_{H^s} &\leq C\bigl(\|\mrw\|_{\Li([T_0,t),H^s)}\bigr)\|\mrw(t)\|_{H^s},
\label{E:VASS} \\
\|\mathring{F}_1(t)\|_{H^s}+\|t \partial_t \mathring{F}_1(t)\|_{\Hs}&\leq C\bigl(\|\mrw\|_{\Li([T_0,t),H^s)}\bigr)\|\mrw(t)\|_{H^s}, \label{E:F_IRIN}\\
\|F_1(\epsilon,t)\|_{H^s} &\leq C\bigl(\|w\|_{\Li([T_0,t),H^s)}\bigr)\|w(t)\|_{H^s}, \label{F1est} \\
\|A^i_1(\epsilon,t,\cdot,\mathring{w}(t))-\mathring{A}^i_1(t,\cdot,\mathring{w}(t))\|_{\Hs} &\leq \epsilon
C\bigl(\|\mrw(t)\|_{\Li([T_0,t),H^s)}\bigr) \label{AiLip}
\intertext{and}
\|H_1(\epsilon,t,\cdot,\mathring{w}(t))-\mathring{H}_1(t,\cdot,\mathring{w}(t))\|_{\Hs}+\|F_1(\epsilon,t)-\mathring{F}_1(t)&\|_{\Hs} \notag \\
	\leq  \epsilon
C\bigl(\|w\|_{\Li([T_0,t),H^s)},\|\mrw\|_{\Li([T_0,t),H^s)}\bigr)(&\|w(t)\|_{H^s}+\|z(t)\|_{\Hs}+\|\mrw(t)\|_{H^s}), \label{HFLip}
\end{align}
where
\begin{equation*}\label{E:Z}
    z=\frac{1}{\epsilon}\left(w-\mathring{w}-\epsilon v\right)
 \end{equation*}
and the constants $ C\bigl(\|w\|_{\Li([T_0,t),H^s)}\bigr)$, $C\bigl(\|\mrw\|_{\Li([T_0,t),H^s)}\bigr)$ and
$C\bigl(\|w\|_{\Li([T_0,T_t),H^s)},\|\mrw\|_{\Li([T_0,t),H^s)}\bigr)$
are independent of $\epsilon \in (0,\epsilon_0)$ and the time $T_1 \in (T_0,0)$.

Then there exists a small constant $\sigma>0$, independent of $\epsilon \in (0,\epsilon_0)$  and $T_1 \in (T_0,0)$, such that if
initially
\begin{align}\label{E:INITIALDATA3}
  \|\mathring{w}^0\|_{H^s}+\|s^0\|_{H^s}
  \leq \sigma \AND C_1^i\partial_i \mathring{w}^0=0,
\end{align}
then
\begin{equation} \label{wrmwsupest}
\max\{\norm{w}_{L^\infty([T_0,T_1)\times \Tbb^n)} ,\norm{w}_{L^\infty([T_0,T_1)\times \Tbb^n)} \} \leq \frac{R}{2}
\end{equation}
and there exists a constant $C>0$, independent of $\epsilon\in (0,\epsilon_0)$ and $T_1 \in (T_0,0)$, such that
\begin{align}
&\|w\|_{M^\infty_{\mathbb{P}_1,s}([T_0,T_1)\times\mathbb{T}^n)}  +  \|\mrw\|_{M^\infty_{\mathbb{P}_1,s}([T_0,T_1)\times\mathbb{T}^n)}
  +\|t\partial_t \mrw\|_{M^\infty_{\mathds{1},s-1}([T_0,T_1)\times\mathbb{T}^n)}  \notag \\
&\hspace{5.0cm} +\int_{T_0}^{t}\|\del{t}\mrw\|_{\Hs}d\tau-\int_{T_0}^{t}\frac{1}{\tau}\|\Pbb_1 \mrw\|_{\Hs} d\tau  \leq C \sigma, \label{E:FINALEST1a}
\end{align}
\begin{gather}
 \|w-\mathring{w}\|_{L^\infty([T_0, t),H^{s-1})}  \leq \epsilon C\sigma \label{E:FINALEST1b}
\intertext{and}
 -\int_{T_0}^{t}\frac{1}{\tau}\|\mathbb{P}_1(w-\mrw)\|^2_{H^{s-1}}d\tau  \leq \epsilon^2 C\sigma^2 \label{E:FINALEST1c}
\end{gather}
for $T_0 \leq t < T_1$.
\end{theorem}
\begin{proof}
First, we observe, by \eqref{E:DECOMPOSITIONOFA01} and \eqref{E:PAP}, that $A^0_1$ satisfies
\begin{equation} \label{Aringproj}
\Pbb_1^\perp \mathring{A}^0_1 \Pbb_1 = \Pbb_1 \mathring{A}^0_1 \Pbb_1^\perp.
\end{equation}
Using this, we find, after applying $\Pbb_1$ to the limit equation \eqref{E:LIMITINGEQa}, that
\al{d}{
b=\Pbb_1 \mrw
}
satisfies the equation
\al{PW2}{
\Pbb_1 \mrA^0_1\Pbb_1 \del{t} b +\Pbb_1 \mrA^i_1\Pbb_1 \del{i} b =\frac{1}{t}\Pbb_1 \mathfrak{\mrA}_1 \Pbb_1 b  +
\Pbb_1 \mrH_1+ \Pbb_1\bar{F}_2,
}
where
\als%{H2BAR}
{
\bar{F}_2 = -\Pbb_1 \mrA^i_1\Pbb_1^\perp \del{i}\mrw + \Pbb_1 \mathring{F}_1-\Pbb_1 C^i_1 \del{i} v.
}
Clearly, $\bar{F}_2$ satisfies
\begin{equation} \label{Fb2est}
\norm{\bar{F}_2(t)}_{H^{s-1}} \leq C\bigl(\norm{\mrw}_{L^\infty(T_0,t),H^s}\bigr)\norm{\mrw(t)}_{H^s}
\end{equation}
for $0\leq t < T_1$ by \eqref{E:VASS}, \eqref{E:F_IRIN} and the calculus inequalities from Appendix \ref{A:INEQUALITIES}, while
\begin{equation} \label{bidata}
\norm{b(T_0)}_{H^{s-1}} \leq \norm{ \mrw^0}_{H^s} \leq \sigma,
\end{equation}
by the assumption \eqref{E:INITIALDATA3} on the initial data, and  $\Pbb_1\mathring{H}_1(t,x,\mrw)$ satisfies
\begin{equation} \label{P1H1zerp}
\Pbb_1\mathring{H}_1(t,x,0) = 0
\end{equation}
by Assumption \ref{ASS1}.(3).

Next, we set
\als%{TPTWSHORT}
{y=t\del{t} \mrw.}
In order to derive an evolution equation for $y$, we apply $t\del{t}$ to \eqref{E:LIMITINGEQa} and use
the identity
\begin{equation*}
  t\partial_t f= t D_t f+ [D_{\mrw} f\cdot t\partial_t \mrw]
  =D_t (t f)- f+ [D_{\mrw} f\cdot t\partial_t \mrw], \quad f=f(t,x,\mrw(t,x)),
\end{equation*}
to obtain
\begin{equation}\label{E:MODELEQ4}
   \mrA_1^0\partial_t y+\mrA_1^i\partial_i y
   =\frac{1}{t}\bigr(\mathbb{P}_1\mathfrak{\mrA}_1\mathbb{P}_1 +\mrA_1^0\bigr ) y-\frac{1}{t}\mathfrak{\mrA}_1 b  + \tilde{R}_2  +\tilde{H}_2+\tilde{F}_2,
\end{equation}
where
\begin{align*}
  \tilde{H}_2&=   D_t (t \mathring{H}_1)-\mrH_1 + [D_{\mrw} \mrH_1 \cdot y]+(D_t \mathring{\mathfrak{A}}_1) b-(D_t \mrA^0_1)y
%\label{Ht2def}
\intertext{and}
  \tilde{F}_2 &=  -[D_{\mrw} \mrA_1^i \cdot y]\partial_i \mrw - D_t (t\mrA^i_1) \del{i} \mrw + \mrA^i_1 \del{i} \mrw + t \partial_t \mathring{F}_1 +tC^i_1\del{i}\del{t} v. %\label{Ft2def}
\end{align*}
Note that in deriving the above equation, we have used the identity
\begin{equation} \label{Afrcom}
\mathring{\mathfrak{A}}_1 \Pbb_1 =  \Pbb_1 \mathring{\mathfrak{A}}_1  =\Pbb_1 \mathring{\mathfrak{A}}_1\Pbb_1,
\end{equation}
which follows directly from \eqref{E:DECOMPOSITIONOFCALB} and \eqref{E:COMMUTEPANFB}.
We further note by \eqref{E:VASS}, \eqref{E:F_IRIN} and Assumption \ref{ASS1}.(4) and Assumption \ref{ASS3}.(2), it is clear that $\tilde{F}_2$
and $\tilde{H}_2 = \tilde{H}_2(t,x,\mrw,b,y)$ satisfy
\begin{equation} \label{Ft2est}
\norm{\tilde{F}_2(t)}_{\Hs} \leq C\bigl(\norm{\mrw}_{H^s}\bigr)( \norm{y}_{\Hs} + \norm{\mrw}_{H^s})
\end{equation}
for $T_0\leq t < T_1$ and
\begin{equation} \label{Ht2zero}
\tilde{H}_2(t,x,0,0,0)=0,
\end{equation}
respectively. Using \eqref{E:LIMITINGEQa} and \eqref{E:INITIALDATA3}, we see that
\begin{equation*}
   y|_{t=T_0}= \Bigl[(\mrA_1^0)^{-1} \mathfrak{\mrA}_1\mathbb{P}_1 \mrw- t(\mrA_1^0)^{-1} \mrA_1^i \partial_i \mrw- t(\mrA_1^0)^{-1} C_1^i\partial_i v + t(\mrA_1^0)^{-1} \mrH_1+ t(\mrA_1^0)^{-1} \mathring{F}_1\Bigr]\Bigl|_{t=T_0},
\end{equation*}
which in turn, implies, via \eqref{E:INITIALDATA3},  \eqref{E:VASS}-\eqref{E:F_IRIN}, and the calculus
inequalities from Appendix \ref{A:INEQUALITIES}, that
\begin{equation*}\label{E:P_UT0}
  \|y\|_{\Hs}(T_0) \leq C(\sigma) \sigma.
\end{equation*}

A short computation using \eqref{E:MODELEQ3a}, \eqref{E:LIMITINGEQa} and \eqref{E:LIMITINGEQb}  shows that
\begin{align}\label{E:DIFFEQ2}
    A_1^0\partial_ t z+A_1^i\partial_iz+\frac{1}{\epsilon}C_1^i\partial_iz
    =\frac{1}{t} \mathfrak{A}_1\mathbb{P}_1z +  \hat{R}_2 +\hat{F}_2,
 \end{align}
where
\begin{gather*}
  \hat{F}_2= \frac{1}{\epsilon} (H_1-\mathring{H}_1)+\frac{1}{\epsilon} (F_1-\mathring{F}_1)
  -\frac{1}{\epsilon}(A_1^i-\mathring{A}^i_1)\partial_i \mrw
  -  A_1^i\partial_i v -  A_1^0\partial_t v +\frac{1}{t} \Pbb_1  \mathfrak{A}_1 \mathbb{P}_1v %\label{E:FHAT2}
\intertext{and}
\hat{R}_2=- \frac{1}{t}\tilde{A}_1^0 y+ \frac{1}{t}\tilde{\mathfrak{A}}_1 b,  %\label{E:RHAT2}
\end{gather*}
and we recall that $\tilde{A}_1^0$ and $\tilde{\mathfrak{A}}_1$ are defined by the expansions
\eqref{E:DECOMPOSITIONOFA01}-\eqref{E:DECOMPOSITIONOFCALB}.
Next, we estimate
\al{FHATC}{
	\frac{1}{\epsilon} \|H_1(\epsilon,&t,\cdot,w(t))-\mathring{H}_1(t,\cdot, \mathring{w}(t))\|_{\Hs}   \notag \\
	&\leq  \frac{1}{\epsilon} \|H_1(\epsilon,t,\cdot,w(t))-H_1(\epsilon,t,\cdot, \mathring{w}(t))\|_{\Hs} +
\frac{1}{\epsilon} \|H_1(\epsilon,t,\cdot,\mathring{w}(t))-\mathring{H}_1(t,\cdot, \mathring{w}(t))\|_{\Hs}   \nnb  \\
	&\leq C\bigl(\|w\|_{\Li([T_0,t),H^s)},\|\mrw\|_{\Li([T_0,t),H^s)}\bigr)(\|w(t)\|_{H^s}+\|z(t)\|_{\Hs}+\|\mrw(t)\|_{H^s}),
	}
for $T_0\leq t < T_1$, where in deriving the second inequality, we used \eqref{HFLip}, Taylor's Theorem (in the last variable),
and the calculus inequalities.
By similar arguments and \eqref{AiLip}, we also
get that
\begin{align}
\frac{1}{\epsilon}\|(A_1^i(\epsilon,t,&\cdot,w(t))-\mathring{A}^i_1(t,\cdot,\mrw(t)) \|_{\Hs} \notag \\
&\leq  C\bigl(\|w\|_{\Li([T_0,t),H^s)},\|\mrw\|_{\Li([T_0,t),H^s)}\bigr)(\|w(t)\|_{H^s}+\|z(t)\|_{\Hs} +\|\mrw(t)\|_{H^s}),
\label{FHATD}
\end{align}
again for $T_0\leq t < T_1$. Taken together, the estimates \eqref{E:VASS}, \eqref{HFLip}, \eqref{E:FHATC} and \eqref{FHATD}
along with the calculus inequalities imply that
\begin{equation} \label{Fh2est}
	\|\hat{F}_2(\epsilon,t)\|_{\Hs}\leq C(\|w\|_{\Li([T_0,t),H^s)},\|\mrw\|_{\Li([T_0,t),H^s)})(\|w(t)\|_{H^s}+\|z(t)\|_{\Hs}+\|\mrw(t)\|_{H^s})
\end{equation}
for $T_0\leq t < T_1$. Furthermore, we see from \eqref{E:VASS} and \eqref{E:INITIALDATA3} that we can estimate $z$ at $t=T_0$ by
\al{INITIALZEST}{
\|z\|_{\Hs}(T_0) \leq C(\sigma)\sigma.
}

We can combine the two equations \eqref{E:MODELEQ3a} and \eqref{E:LIMITINGEQa} together into
the equation
\al{EQ1COR}{
&\p{A_1^0 & 0 \\ 0 & \mathring{A}_1^0}\partial_t \p{w \\ \mathring{w}}+\p{A_1^i & 0 \\ 0 & \mathring{A}_1^i} \partial_i \p{ w \\ \mathring{w}} +\frac{1}{\epsilon} \p{C_1^i & 0 \\ 0 & 0 }\del{i} \p{w \\ \mrw}  \nnb \\
&\hspace{3.0cm}=\frac{1}{t}\p{\mathfrak{A}_1 & 0\\0 & \mathring{\mathfrak{A}}_1} \p{\Pbb_1 & 0 \\ 0 & \mathbb{P}_1}\p{w \\ \mrw}+\p{H_1\\ \mathring{H}_1 }+\p{F_1\\ \mathring{F}_1- C_1^i\partial_i v},
}
and collect the three equations \eqref{E:PW2}, \eqref{E:MODELEQ4} and \eqref{E:DIFFEQ2}
together into the equation
\al{WHOLESYS1}{
A_2^0\del{t} \begin{pmatrix} b\\ y\\ z\end{pmatrix}  +A_2^i\del{i} \begin{pmatrix} b\\ y\\ z\end{pmatrix} +\frac{1}{\epsilon}C_2^i\del{i} \begin{pmatrix} b\\ y\\ z\end{pmatrix}=\frac{1}{t}\mathfrak{A}_2 \Pbb_2 \begin{pmatrix} b\\ y\\ z\end{pmatrix} +H_2+ R_2+F_2,
}
where
\al{WHOLESYS2}{
A_2^0:=\p{\Pbb_1 \mrA^0_1 \Pbb_1   & 0 & 0 \\   0  & \mrA^0_1 & 0\\  0 &   0 & A^0_1 }, \quad A_2^i:=\p{\Pbb_1 \mrA^i_1 \Pbb_1   & 0 & 0  \\  0   & \mrA^i_1 & 0\\ 0   & 0 & A^i_1 },
}
\al{WHOLESYS2.5}{
C_2^i:=\p{0 & 0 & 0 \\   0  & 0 & 0\\  0   & 0 & C^i_1 }, \quad  \Pbb_2:=\p{ \Pbb_1 & 0 & 0  \\ 0 &   \mathds{1} & 0\\ 0  & 0 & \Pbb_1 }, \quad \mathfrak{A}_2=\p{\Pbb_1\mathfrak{\mrA}_1\Pbb_1 & 0  & 0 \\  -\Pbb_1\mathring{\mathfrak{A}}_1\Pbb_1  & \Pbb_1 \mathring{\mathfrak{A}}_1 \Pbb_1+\mrA^0_1 & 0 \\  0 & 0 & \mathfrak{A}_1},
}
\al{WHOLESYS3}{
H_2:=\p{\Pbb_1 \mathring{H}_1 \\ \tilde{H}_2  \\ 0 },  \quad R_2:=\p{0 \\ 0  \\  \hat{R}_2} \AND F_2:=  \p{\Pbb_1\bar{F}_2 \\ \tilde{F}_2  \\ \hat{F}_2}.
}
We remark that due to the projection operator $\Pbb_1$ that appears in the definition \eqref{E:d} of $b$
and in the top row of \eqref{E:WHOLESYS2}, the vector $(b,y,z)^{\mathrm{T}}$ takes values
in the vector space $\Pbb_1\Rbb^{N_1} \times \Rbb^{N_1} \times \Rbb^{N_1}$ and  \eqref{E:WHOLESYS2}
defines a symmetric hyperbolic system, i.e. $A^0_2$ and $A_2^i$ define symmetric linear operators on  $\Pbb_1\Rbb^{N_1} \times \Rbb^{N_1} \times \Rbb^{N_1}$ and $A^0_2$ is non-degenerate.

Setting
\begin{equation*}
\Pbb_3:=\p{0 & 0 & 0 & 0 & 0\\ 0 & 0 & 0 & 0 & 0\\ 0 & 0 & \mathbb{P}_1 & 0 & 0 \\ 0 & 0 & 0 & \mathds{1} & 0\\ 0 & 0 & 0 & 0 & 0 },
\end{equation*}
it is then not difficult to verify from the estimates \eqref{E:VASS}, \eqref{F1est}, \eqref{Fb2est},  \eqref{Ft2est}
and \eqref{Fh2est}, the initial bounds \eqref{E:INITIALDATA3}, \eqref{bidata} and \eqref{E:INITIALZEST},
the relations \eqref{Aringproj}, \eqref{P1H1zerp}, \eqref{Afrcom} and \eqref{Ht2zero}, and the assumptions
on the coefficients $\{$$A^0_1$, $A^i_1$, $\mathring{A}^0_1$, $\mathring{A}^i_1$, $\mathfrak{A}_1$,
$\mathring{\mathfrak{A}}_1$, $H$, $F$$\}$, see Assumptions \ref{ASS1} and \ref{ASS3}, that the system
consisting of \eqref{E:EQ1COR} and \eqref{E:WHOLESYS1} and the solution $U=(w,\mrw,b,y,z)^{\mathrm{T}}$
satisfy the hypotheses of Theorem \ref{L:BASICMODEL}, and thus, for $\sigma >0$ chosen small enough,
there exists a constant $C>0$ independent of $\epsilon \in (0,\epsilon_0)$ and $T_1 \in (T_0,0)$ such that
\begin{equation} \label{wmrwLinfty}
\norm{(w,\mrw)}_{L^\infty([T_0,T_1)\times \Tbb^n)} \leq \frac{R}{2}
\end{equation}
and
\begin{equation} \label{Uproofest}
 \|(w,\mrw)\|_{M^\infty_{\Pbb_1, s}([T_0,t)\times \Tbb^n)}+\|(b,y,z)\|_{M^\infty_{\Pbb_2, s-1}([T_0,t)\times \Tbb^n)} -\int_{T_0}^{t}
\frac{1}{\tau} \|\Pbb_3 U\|_{\Hs}\, d\tau  \leq C\sigma
\end{equation}
 for $T_0 \leq t < T_1$.
This completes the proof since the estimates \eqref{wrmwsupest}-\eqref{E:FINALEST1c} follow immediately from \eqref{wmrwLinfty} and \eqref{Uproofest}.
\end{proof}

%---------------- begin section 6 -------------------------------------------------------------------------------
\section{Initial data}\label{S:INITIALIZATION}
As is well known, the initial data for the reduced conformal Einstein--Euler equations cannot be chosen freely on the initial hypersurface
\begin{equation*}
\Sigma_{T_0} = \{T_0\}\times \Tbb^3 \subset M=(0,T_0]\times \Tbb^3\qquad (T_0 > 0).
\end{equation*}
Indeed, a number of constraints, which we can separate into gravitational, gauge and velocity normalization, must be satisfied on
$\Sigma_{T_0}$. There are a number of distinct methods available to solve these constraint equations. Here, we will follow the method
used in  \cite{oli3,oli4}, which is an adaptation of the method introduced by Lottermoser in \cite{lot}.

The goal of this section is to construct $1$-parameter families of $\epsilon$-dependent solutions to the constraint equations that
behave appropriately in the limit $\epsilon\searrow 0$. In order to use the method from \cite{oli3,oli4} to solve the constraint equations, we need to introduce new gravitational variables $\hmfu^{\mu\nu}$
and $\hmfu^{\mu\nu}_{\sigma}$ defined via the formulas
\begin{align}\label{E:DEFHATG}
\hat{g}^{\mu\nu}:=\theta\underline{\bar{g}^{\mu\nu}}=
E^3 (\bar{h} ^{\mu\nu}+\epsilon^2 \hmfu^{\mu\nu})=\hat{h}^{\mu\nu}+\epsilon^2 E^3 \hmfu^{\mu\nu} \qquad \text{and} \qquad
\hmfu^{\mu\nu}_{\sigma}:=\bar{\partial}_\sigma \hmfu^{\mu\nu},
\end{align}
respectively, where
\begin{align}\label{E:SMALLTHETA}
\theta=\frac{\sqrt{|\underline{\bar{g}}|}}{\sqrt{|\bar{\eta}|}}=\sqrt{\frac{\Lambda}{3}|\underline{ \bar{g}}|}, \quad | \bar{g}|=-\det{ \bar{g}_{\mu\nu}},  \quad  \hat{h}^{\mu\nu}=E^3\bar{h}^{\mu\nu} \quad \text{and} \quad |\bar{\eta}|=-\det{\bar{\eta}_{\mu\nu}}=\frac{3}{\Lambda}.
\end{align}

\medskip

\noindent \textbf{Notation:}
In the following, we will use
upper case script letters, e.g. $\mathscr{Q}(\xi)$, $\mathscr{R}(\xi)$, $\mathscr{S}(\xi)$, to denote analytic maps of
the variable $\xi$ whose exact form is
not important. The domain of analyticity of these maps will be clear from context. Generally, we will use $\mathscr{S}$
to denote maps that may change line to line, while other letters will be used to denote maps that need to be distinguished for
later use. We also introduce the following derivative notation to facilitate the statements,
\als{
\hat{\partial}_\mu=\frac{1}{\epsilon} \delta^i_\mu \del{i}+\delta^0_\mu \del{0}.
}

\medskip

The total set of constraints that we need to solve on $\Sigma_{T_0}$ are:
\begin{align}
(\underline{\bar{G}^{0\mu}} -\underline{\bar{T}^{0\mu}})|_{t=T_0}&=0 \quad (\text{Gravitational Constraints}),\label{E:CONSTRAINT}\\
\left.\left(\hat{\partial}_{\nu}(E^3\hmfu^{\mu\nu})-\frac{2}{t}E^3\hmfu^{\mu0}- \frac{2\Lambda}{3t}
\frac{\theta-E^3}{\epsilon^2} \delta^\mu_0+\frac{\theta-E^3}{\epsilon^2}
\frac{\Lambda}{t}\Omega \delta^\mu_0 \right)\right|_{t=T_0}&=0 \quad (\text{Gauge constraint}) \label{E:WAVECONSTRAINT}
\intertext{and}
(\underline{\bar{v}^\mu\bar{v}_\mu}+1)|_{t=T_0}&=0 \quad (\text{Velocity Normalization}).\label{E:NORMALIZATION}
\end{align}

\begin{remark} \label{wavegaugerem}
	It is not difficult to verify that the constraint \eqref{E:WAVECONSTRAINT} is equivalent to the wave gauge condition
$\underline{\bar{Z}^\mu}=0$ on the initial
hypersurface $\Sigma_{T_0}$. Indeed, it is enough to notice that   $\bar{\partial}_\nu (\hat{h}^{\mu\nu})=- E^3 \frac{\Lambda}{t}\Omega \delta^\mu_0$ and
	\begin{align*}
	\underline{\bar{X}^\mu} = -\hat{\partial}_\nu \underline{\bar{g}^{\mu\nu}}-\underline{\bar{g}^{\mu\nu}}\frac{1}{\sqrt{|\underline{\bar{g}}|}}
	\hat{\partial}_\nu\sqrt{|\underline{\bar{g}}|}-\frac{\Lambda}{t}\Omega\delta^\mu_0
	=\frac{1}{\theta}(-\theta \hat{\partial}_\nu \underline{\bar{g}^{\mu\nu}}
	-\underline{\bar{g}^{\mu\nu}}
	\hat{\partial}_\nu\theta)-\frac{\Lambda}{t}\Omega\delta^\mu_0   = -\frac{1}{\theta}\hat{\partial}_\nu \hat{g}^{\mu\nu}-\frac{\Lambda}{t}\Omega\delta^\mu_0.
	\end{align*}
\end{remark}

\subsection{Reduced conformal Einstein-equations}

Before proceeding, we state in the following lemma a result that will be used repeatedly in this section.
The proof follows  from the definition of $\theta$, see \eqref{E:SMALLTHETA},  and  a direct calculation. We omit the details.
\begin{lemma}\label{L:IDENTITY}
	\begin{align}\label{E:THETAEXPANSION}
	\theta(\epsilon, \hmfu^{\mu\nu})=E^6\sqrt{-\frac{3}{\Lambda}\det{(\bar{h}^{\mu\nu}+\epsilon^2\hmfu^{\mu\nu}})}
	=E^3+\frac{1}{2}\epsilon^2 E^3 \left(-\frac{3}{\Lambda}\hmfu^{00}+ E^2 \hmfu^{ij}\delta_{ij}
	\right)+\epsilon^4 \mathscr{S}(\epsilon,t, E, \hmfu^{\mu\nu}),
	\end{align}
	where $\mathscr{S}(\epsilon,t, E, 0)=0$.
\end{lemma}
Using this lemma, we can express the gauge constraint \eqref{E:WAVECONSTRAINT} as follows:
\begin{equation}\label{E:PTU000J}\left\{
\begin{aligned}
\partial_t (E^3 \hmfu^{00}) %=&-\frac{1}{\epsilon}\partial_k(E^3\hmfu^{0k})+\frac{2}{t}E^3\hmfu^{00}+\frac {2\Lambda}{3t}
%\frac{\theta-E^3}{\epsilon^2} - \frac{\theta-E^3}{\epsilon^2}
%\frac{\Lambda}{t}\Omega   \\
=&-\frac{1}{\epsilon}\partial_k(E^3\hmfu^{0k})+\frac{2}{t}E^3\hmfu^{00}+\frac {\Lambda}{3t}
E^3\left(-\frac{3}{\Lambda}\hmfu^{00}+E^2\hmfu^{ij}\delta_{ij}\right) \\ &- \frac{1}{2}E^3
\frac{\Lambda}{t}\Omega \left(-\frac{3}{\Lambda}\hmfu^{00}+E^2\hmfu^{ij}\delta_{ij}\right) +\epsilon^2\mathscr{S}(\epsilon,t, E, \Omega/t , \hmfu^{\mu\nu})\\
\partial_t (E^3 \hmfu^{j0})=&-\frac{1}{\epsilon}\partial_k(E^3 \hmfu^{jk})+\frac{2}{t} \emfu^{j0}
\end{aligned}\right.
\end{equation}
where $\mathscr{S}(\epsilon,t, E, \Omega/t , 0)=0$. The importance of the relations \eqref{E:PTU000J} is that they allow us to determine the time derivatives $\partial_0 \hmfu^{\mu 0}$
on the initial hypersurface $\Sigma_{T_0}$  from the metric variables $\hmfu^{\mu\nu}$ and their spatial
derivatives on $\Sigma_{T_0}$.

\begin{lemma}\label{E:IDENTITYPTHETA}
	\begin{align}
	\partial_t(\theta-E^3)
	=&\frac{3}{2\Lambda}\epsilon E^3 \partial_k\hmfu^{0k} +\epsilon^2 \mathscr{A}(\epsilon,t, E, \Omega/t , \hmfu^{\mu\nu}, \partial_l\hmfu^{\mu \nu}, \hmfu^{ij}_0)\label{E:PTTHETA}
\intertext{and}
	\partial_i\theta
	=&-\frac{3}{2\Lambda}\epsilon^2 E^3 \partial_i\hmfu^{00}+\frac{1}{2}\epsilon^2 E^5 \delta_{kl}\partial_i\hmfu^{kl} +\epsilon^4
\mathscr{S}_i(\epsilon,t, E, \Omega/t ,\hmfu^{\mu\nu}, \partial_l\hmfu^{\mu \nu}, \hmfu^{ij}_0),\label{E:PITHETA}
	\end{align}
	where the $\mathscr{A}$ and $\mathscr{S}_i$ are linear in $(\partial_l\hmfu^{\mu \nu},\hmfu^{ij}_0)$ and  vanish for $(\epsilon,t, E, \Omega/t ,0,0,0)=0$.
\end{lemma}
\begin{proof}
	The proof of this Lemma follows from straightforward calculations; we only prove \eqref{E:PTTHETA}. Noticing
	\begin{equation}\label{E:ID5}
		\theta^{-1}\partial\theta=\frac{1}{2}\underline{\bar{g}^{\mu\nu}}\partial \underline{\bar{g}_{\mu\nu}}=\frac{1}{2}\hat{g}_{\mu\nu}\partial \hat{g}^{\mu\nu} \qquad (\hat{g}_{\mu\nu}=\theta^{-1}\underline{\bar{g}_{\mu\nu}}), 	
	\end{equation}
it is not difficult to verify that
	\begin{align*}
	\partial_t(\theta-E^3)= \frac{1}{2}\theta\hat{g}_{\mu\nu}\partial_0\hat{g}^{\mu\nu} - 3E^3\frac{\Omega}{t}
	= \frac{3}{2\Lambda}\epsilon E^3 \partial_k \hmfu^{0k}+3E^3\frac{\Omega}{t}-3E^3\frac{\Omega}{t}+\epsilon^2\mathscr{A}
	\end{align*}
follows from \eqref{E:PTU000J}.
\end{proof}
We proceed by differentiating \eqref{E:PTU000J} with respect to time $t$ to obtain, with the help of Lemma \ref{E:IDENTITYPTHETA}, the
following:
\begin{equation}\label{E:PPTU000J}\left\{
\begin{aligned}
\partial^2_t (\emfu^{00})=&\frac{1}{\epsilon^2}\partial_k\partial_i(\emfu^{ik})-\frac{1}{\epsilon}
\frac{4}{t}\partial_k (\emfu^{k0})+\frac{2}{t^2}\emfu^{00} +\left(\frac{2}{3}-\Omega\right)\frac{\Lambda}{t^2}
\frac{\theta-E^3}{\epsilon^2} \\
& - \frac{\theta-E^3}{\epsilon^2}\frac{\Lambda}{t}\partial_t\Omega
+\frac{\Lambda}{t}\left(\frac{2}{3}-\Omega\right) \partial_t\frac{\theta-E^3}{\epsilon^2}\\
%= & \frac{1}{\epsilon^2}\partial_k\partial_i(\emfu^{ik})-\frac{1}{\epsilon}
%\frac{3}{t} \left( 1 + \frac{1}{2} \Omega \right) E^3\partial_k \hmfu^{k0} +\frac{1}{2}  E^3 \frac{\Lambda}{t^2} \left(\frac{2}{3}-\Omega\right)
%\left(-\frac{3}{\Lambda}\hmfu^{00}+ E^2 \hmfu^{ij}\delta_{ij}
%\right)\\ &
%+\frac{2}{t^2}\emfu^{00} - \frac{1}{2}  E^3 \left(-\frac{3}{\Lambda}\hmfu^{00}+ E^2 \hmfu^{ij}\delta_{ij}
%\right)\frac{\Lambda}{t}\partial_t\Omega+\frac{\Lambda}{t}\left( \frac{2}{3}-\Omega\right)\mathscr{A}(\epsilon,t, E, \Omega/t ,\hmfu^{\mu\nu}, \partial_l\hmfu^{\mu k}, \hmfu^{ij}_0)\\ & +\epsilon^2\mathscr{S}(\epsilon,t, E, \Omega/t , \hmfu^{\mu\nu}, \partial_l\hmfu^{\mu k}, \hmfu^{ij}_0)\\
= & \frac{1}{\epsilon^2}\partial_k\partial_i(\emfu^{ik})-\frac{1}{\epsilon}
\frac{3}{t} \left( 1 + \frac{1}{2} \Omega \right) E^3\partial_k \hmfu^{k0}
+\frac{1}{t^2} \mathscr{S}(\epsilon,t, E, \Omega/t, \del{t}\Omega, \hmfu^{\mu\nu}, \partial_l\hmfu^{\mu k}, \hmfu^{ij}_0)\\ %& +\epsilon^2\mathscr{S}(\epsilon,t, E, \Omega/t , \hmfu^{\mu\nu}, \partial_l\hmfu^{\mu k}, \hmfu^{ij}_0)\\
\partial^2_t (\emfu^{j0})=&-\frac{1}{\epsilon}\partial_k\partial_0(\emfu^{jk})
+\frac{2}{t^2}\emfu^{j0}-\frac{2}{t}\frac
{1}{\epsilon}\partial_k(\emfu^{jk}) \\
= &-\frac{1}{\epsilon} E^3 \partial_k \hmfu^{jk}_0
+\frac{2}{t^2}\emfu^{j0}-\frac{2+3\Omega}{t}\frac
{1}{\epsilon}\partial_k(\emfu^{jk})
\end{aligned}\right.
\end{equation}
where $\mathscr{S}(\epsilon,t,E, \Omega/t ,\del{t}\Omega, 0,0,0)=0$. %and $\mathscr{S}(\epsilon,t,E, \Omega/t ,0,0,0)=0$.

Next, we consider the following reduced version of the conformal Einstein equations \eqref{E:CONFORMALEINSTEIN1}, which we write
using the metric variable $\hat{g}^{\mu\nu}$ defined by \eqref{E:DEFHATG}:
\begin{align}
\frac{1}{2\theta^2}\hat{g}^{\lambda\sigma}\hat{\partial}_\lambda \hat{\partial}_\sigma \hat{g}^{\mu\nu}+\underline{\bar{\nabla}^{(\mu}\bar{\Gamma}^{\nu)}}-\frac{1}{2\theta} \hat{g}^{\mu\nu}& \underline{\bar{\nabla}_\lambda \bar{\Gamma}^\lambda}+\frac{1}{\theta^2}\hat{Q}^{\mu\nu}(\hat{g}, \bar{\partial} \hat{g}, \theta)-\frac{1}{2}\underline{\bar{X}^\mu} \underline{\bar{X}^\nu}-\underline{\bar{\nabla}^{(\mu}\bar{Z}^{\nu)}}
+\frac{1}{2\theta}\hat{g}^{\mu\nu} \underline{\bar{\nabla}_\lambda \bar{Z}^\lambda}- \frac{1}{2}\underline{\bar{A}^{\mu\nu}_\lambda} \underline{\bar{Z}^\lambda}
\notag \\
= & e^{4\Phi}\underline{\tilde{T}^{\mu\nu}}-\frac{1}{\theta} e^{2\Phi}\Lambda\hat{g}^{\mu\nu} +2(\bar{\nabla}^\mu\bar{\nabla}^{\nu}\Psi- \bar{\nabla}^\mu \Psi \bar{\nabla}^\nu \Psi)-(2\bar{\Box} \Psi +|\bar{\nabla} \Psi|^2_{\bar{g}})\frac{1}{\theta}\hat{g}^{\mu\nu},
\label{idredEin1}
\end{align}
where
\begin{align*}
\hat{Q}^{\mu\nu}(\hat{g}, \hat{\partial} \hat{g}, \theta)= & \frac{1}{2}\theta^2 \underline{Q^{\mu\nu}} -\frac{1}{4}\hat{g}^{\mu\nu}\hat{g}_{\alpha\beta}(\theta^2 \underline{Q^{\alpha\beta}}-\theta^2 \underline{\bar{X}^\alpha} \underline{\bar{X}^\beta})-\frac{1}{2}\hat{g}^{\lambda\sigma} \hat{g}_{\alpha\beta} \hat{\partial}_\sigma\hat{g}^{\mu\nu} \hat{\partial}_\lambda \hat{g}^{\alpha\beta} \notag \\
& + \frac{1}{8} \hat{g}^{\lambda\sigma}\hat{g}^{\mu\nu}\hat{g}_{\gamma\rho}\hat{g}_{\alpha\beta} \hat{\partial}_\lambda \hat{g}^{\gamma\rho} \hat{\partial}_\sigma \hat{g}^{\alpha\beta}+\frac{1}{4} \hat{g}^{\lambda\sigma} \hat{g}^{\mu\nu} \hat{\partial}_\lambda \hat{g}_{\alpha\beta} \hat{\partial}_\sigma \hat{g}^{\alpha\beta} %\label{E:QHAT}
\end{align*}
with $Q^{\mu\nu}$ as defined previously by \eqref{E:QMUNU}. By \eqref{E:XY}, \eqref{E:QMUNU}, \eqref{E:DEFHATG}, \eqref{E:ID5}  and the identity
\begin{align}\label{E:GAMMAHATG}
\underline{\bar{\Gamma}^\lambda_{\mu\nu}}=-\hat{g}_{\sigma(\mu}\hat{\partial}_{\nu)}
\hat{g}^{\lambda\sigma}+\frac{1}{2}\hat{g}^{\lambda\sigma}\hat{g}_{\alpha\mu}
\hat{g}_{\beta\nu}\hat{\partial}_\sigma \hat{g}^{\alpha\beta}+\frac{1}{4}\left(2\hat{g}_{\alpha\beta}\delta^\lambda_{(\mu}\hat
{\partial}_{\nu)}\hat{g}^{\alpha\beta}-\hat{g}^{\lambda\sigma}\hat{g}_{\mu\nu}\hat{g}
_{\alpha\beta}\hat{\partial}_\sigma\hat{g}^{\alpha\beta}\right),
\end{align}
it is obvious that $\theta^2 \underline{Q^{\mu\nu}}$ is analytic in $\hat{g}^{\mu\nu}$, $\hat{\partial} \hat{g}^{\mu\nu}$ and $\theta$. From
this and the formula \eqref{E:GAMMAHATG}, it is clear that $\hat{Q}^{\mu\nu}$  is analytic in
 $\hat{g}^{\mu\nu}$, $\hat{\partial} \hat{g}^{\mu\nu}$ and $\theta$. Moreover, using
\eqref{E:PTU000J} and \eqref{E:GAMMAHATG}, it can be verified by a straightforward calculation that
$\hat{Q}^{\mu\nu}$ satisfies
	\begin{align*}
	\hat{Q}^{\mu\nu}(\hat{g}, \bar{\partial} \hat{g}, \theta)-\hat{Q}_H^{\mu\nu}(\hat{h}, \bar{\partial} \hat{h}, E^3)=\epsilon \mathcal{T}^{\mu\nu k}_{\alpha\beta}(t) \partial_k\hmfu^{\alpha\beta}+\epsilon^2\mathscr{\hat{Q}}{}^{\mu\nu}(\epsilon, t, E, \Omega/t ,x, \hmfu^{\alpha\beta}, \partial_k\hmfu^{\alpha\beta}, \hmfu^{ij}_0)
	\end{align*}
%	where $\hat{Q}_H^{\mu\nu}(\hat{h}, \bar{\partial} \hat{h}, E^3)$ is defined by taking $\hat{g}, \bar{\partial} \hat{g}, \theta$ as %$\hat{h}, \bar{\partial} \hat{h}, E^3$ in \eqref{E:QHAT} and
for coefficients $\mathcal{T}^{\mu\nu k}_{\alpha\beta}$ that depend only on $t$ and where
$\mathscr{\hat{Q}}{}^{\mu\nu}(\epsilon, t, E, \Omega/t ,x, \hmfu^{\alpha\beta},0,0)=0$.
%\footnote{One could calculate this $\mathcal{T}^{\mu\nu k}_{\alpha\beta}(t)$ explicitly, but it is not necessary to do that. Later, readers %will see this is enough for our proof. } $t$.

Using the easy to verify identities
\begin{gather*}
\underline{\bar{\Gamma}^\lambda_{\lambda 0}}= \frac{1}{2} \underline{\bar{g}^{\lambda\sigma}}\hat{\partial}_0\underline{\bar{g}_{\lambda\sigma}}= \frac{1}{2} \hat{g}_{\lambda\sigma}\hat{\partial}_0\hat{g}^{\lambda\sigma}=\frac{1}{\theta}
\hat{\partial}_0 \theta,  \\
\bar{\nabla}_\lambda \bar{\gamma}^\lambda=\frac{1}{t}\left(\partial_t\Omega -\frac{1}{t}\Omega +\bar{\Gamma}^\lambda_{\lambda0} \Omega \right) \quad \text{and} \qquad
\bar{\nabla}_\lambda\bar{Y}^\lambda =-2\bar{\Box}\Psi-\frac{2\Lambda}{3t^2}+\frac{2\Lambda}{3t}\bar{\Gamma}^\lambda_{\lambda 0},
\end{gather*}
we can write the reduced conformal Einstein equations \eqref{idredEin1}
as
\begin{align}
&\frac{1}{2\theta^2}\hat{g}^{\lambda\sigma}\hat{\partial}_\lambda\hat{\partial}_\sigma \hat{g}^{\mu\nu} + \bar{\nabla}^{(\mu}\bar{\gamma}^{\nu)}-\frac{1}{2\theta }\hat{g}^{\mu\nu}\frac{1}{t} \left(\partial_t\Omega-\frac{1}{t}\Omega+\underline{\bar{\Gamma}^\lambda_{\lambda 0}}\Omega \right) +\frac{1}{\theta^2}\hat{Q}^{\mu\nu} +\frac{1}{\theta}\hat{g}^{\mu\nu}\frac{1}{t}\bar{\gamma}^0 \notag \\
= & -\frac{\Lambda}{3t}\frac{1}{\theta}\hat{\partial}_0\hat{g}^{\mu\nu}+\frac{2\Lambda}{3t^2} \left[\left(\underline{\bar{g}^{00}}+\frac{\Lambda}{3}\right)\delta^\mu_0\delta^\nu_0+\underline{ \bar{g}^{0k}} \delta^{(\mu}_k
\delta^{\nu)}_0 \right]+\frac{1}{t^2}(1+\epsilon^2 K) \rho  \underline{\bar{v}^\mu} \underline{\bar{v}^\nu}+ \frac{1}{\theta}\frac{1}{t^2}\epsilon^2 K
 \rho  \hat{g}^{\mu\nu}. \label{E:INITIALEIN}
\end{align}
This equation is satisfied for the FLRW solutions \eqref{FLRW.a}-\eqref{E:TPTA}, i.e. we can substitute $\{\hat{g}^{\mu\nu},\bar{\rho},\bar{v}^\mu \}$ $\mapsto$ $\{\hat{h}^{\mu\nu},\rho_H,e^{\Psi}\tilde{v}^\mu_H\}$. Dividing the resulting FLRW
equation through by $\theta^2$, we get
\begin{align}
&\frac{1}{2 \theta^2}\hat{h}^{\lambda\sigma}\hat{\partial}_\lambda\hat{\partial}_\sigma \hat{h}^{\mu\nu} + \frac{E^6}{\theta^2}\bar{\nabla}_H^{(\mu}\bar{\gamma}^{\nu)}-\frac{E^3}{2 \theta^2}\hat{h}^{\mu\nu}\frac{1}{t} \left(\partial_t\Omega-\frac{1}{t}\Omega+\bar{\gamma}^\lambda_{\lambda 0}\Omega \right) +\frac{1}{\theta^2}\hat{Q}_H^{\mu\nu} +\frac{E^3}{\theta^2}\hat{h}^{\mu\nu}\frac{1}{t}\bar{\gamma}^0 \notag \\
= & -\frac{\Lambda}{3t}\frac{  E^3}{\theta^2}\hat{\partial}_0\hat{h}^{\mu\nu}+\frac{E^6}{\theta^2}\frac{2\Lambda}{3t^2} \left(\bar{h}^{00}+\frac{\Lambda}{3}\right)\delta^\mu_0\delta^\nu_0
+\frac{E^6}{\theta^2}\frac{\Lambda}{3}\frac{1}{t^2}(1+\epsilon^2 K)\rho_H \delta_0^\mu \delta_0^\nu+ \frac{E^3}{\theta^2}\frac{1}{t^2}\epsilon^2 K \rho_H \hat{h}^{\mu\nu}. \label{E:INITIALEINH}
\end{align}
Subtracting  \eqref{E:INITIALEINH} from \eqref{E:INITIALEIN} yields
\begin{align}
& \hat{g}^{\lambda\sigma}\hat{\partial}_\lambda\hat{\partial}_\sigma (\hat{g}^{\mu\nu}-\hat{h}^{\mu\nu}) +   (\hat{g}^{\lambda\sigma}-\hat{h}^{\lambda\sigma}) \hat{\partial}_\lambda\hat{\partial}_\sigma \hat{h}^{\mu\nu}+ 2\theta^2\left( \bar{\nabla}^{(\mu}\bar{\gamma}^{\nu)}
-\frac{E^6}{\theta^2}\bar{\nabla}_H^{(\mu}\bar{\gamma}^{\nu)}\right)  \nnb \\ &- \theta(\hat{g}^{\mu\nu}-\hat{h}^{\mu\nu})\frac{1}{t} \left(\partial_t\Omega-\frac{1}{t}\Omega\right)
+(E^3-\theta)\hat{h}^{\mu\nu}\frac{1}{t}\left(\del{t}\Omega-\frac{1}{t}\Omega \right)  +\theta(\hat{h}^{\mu\nu}-\hat{g}^{\mu\nu})\frac{1}{t}\bar{\gamma}^\lambda_{\lambda 0}\Omega \notag \\
& +(E^3-\theta)\hat{h}^{\mu\nu}\frac{1}{t}\bar{\gamma}^\lambda_{\lambda 0} \Omega +\theta\hat{g}^{\mu\nu} \frac{1}{t}(\underline{\bar{\gamma}^\lambda_{\lambda 0}}-\underline{\bar{\Gamma}^\lambda_{\lambda 0}})\Omega
%- \frac{1}{t}\theta\Omega\left(\hat{g}^{\mu\nu}\bar{\Gamma}^\lambda_{\lambda 0} - \hat{h}^{\mu\nu}\bar{\gamma}^\lambda_{\lambda 0} \right)
+ 2 (\hat{Q}^{\mu\nu} -\hat{Q}_H^{\mu\nu}) +2\theta \frac{1}{t}\bar{\gamma}^0\left(\hat{g}^{\mu\nu}-\frac{E^3} {\theta}\hat{h}^{\mu\nu}\right) \notag \\
= & -\frac{2\Lambda}{3t}\theta\hat{\partial}_0(\hat{g}^{\mu\nu} -\hat{h}^{\mu\nu})-\frac{2\Lambda}{3t}\theta\left(1-\frac{E^3}{\theta}\right) \hat{\partial}_0 \hat{h}^{\mu\nu} +\frac{4\Lambda}{3t^2} \theta^2 \left[\left(\left(\underline{\bar{g}^{00}}+\frac{\Lambda}{3}\right) -\frac{E^6}{\theta^2}  \left(\bar{h}^{00}+\frac{\Lambda}{3}\right)\right)\delta^\mu_0\delta^\nu_0+\underline{\bar{g}^{0k}} \delta^{(\mu}_k
\delta^{\nu)}_0 \right] \notag \\
& +2\theta^2\frac{1}{t^2}(1+\epsilon^2 K)\left( \rho  \underline{\bar{v}^\mu} \underline{\bar{v}^\nu} -\frac{E^6}{\theta^2}\frac{\Lambda}{3} \rho_H \delta_0^\mu \delta_0^\nu \right) + 2\theta\frac{1}{t^2}\epsilon^2 K \left( \rho  \hat{g}^{\mu\nu} - \frac{E^3}{\theta} \rho_H \hat{h}^{\mu\nu}\right). \label{E:INITIALEINDIFF}
\end{align}

\subsection{Transformation formulas}
Before proceeding, we collect in the following lemma a set of formulas that can be used to transform from the gravitational variables used in this
section to those introduced previously in \S \ref{vardefs} for the formulation of the evolution equations.

\begin{lemma}\label{L:RELATION1}
	The evolution variables $u^{0\mu}$, $u^{ij}$ and $u$ can be expressed in terms of the gravitational variables
$\hmfu^{\mu\nu}$ by the following expressions:
	\begin{align}
	u^{0\mu}=\frac{\epsilon}{2t}\left(\frac{1}{2}\hmfu^{00}\delta^\mu_0+ \hmfu^{0k}\delta^\mu_k+\frac{\Lambda}{6} E^2 \hmfu^{ij}\delta_{ij}\delta^\mu_0 \right)+\epsilon^3 \mathscr{S}^\mu(\epsilon,t,E, \Omega/t ,\hmfu^{\alpha\beta}),\label{E:U0MUANDINI}\\
	u = \epsilon \frac{2\Lambda}{9} E^2\hmfu^{ij}\delta_{ij} +\epsilon^3
\mathscr{S}(\epsilon,t,E, \Omega/t ,\hmfu^{\alpha\beta}),
\label{E:UANDINI}\\
	u^{ij}=  \epsilon E^2\left(\hmfu^{ij}-\frac{1}{3}\hmfu^{kl}\delta_{kl}\delta^{ij}\right)+
\epsilon^3\mathscr{S}^{ij}( \epsilon,t,E, \Omega/t ,\hmfu^{\alpha\beta}),\label{E:UIJANDINI}
	\end{align}
	where all of the remainder terms vanish for $(\epsilon,t,E, \Omega/t ,0)=0$. Moreover, the $0$-component of the conformal fluid four-velocity $\bar{v}^\mu$ can be written as
	\begin{align}\label{E:V0HAT}
	\underline{\bar{v}^0}=\sqrt{\frac{\Lambda}{3}}+\epsilon^2\mathscr{S}(\epsilon,t,E, \Omega/t, \hmfu^{\alpha\beta},z_j).
	\end{align}
\end{lemma}
\begin{proof}
First, we observe that  \eqref{E:U0MUANDINI} follows directly from \eqref{E:G0MU} and Lemma \ref{L:IDENTITY}.
Next, using \eqref{E:DEFHATG}, it is not hard to show that
\begin{align*} %\label{E:UDET}
	\det{(\underline{\bar{g}^{kl}})}=(\theta E^{-3})^{-3}(E^{-6}+\epsilon^2E^{-4}\hmfu^{ij}\delta_{ij})+\epsilon^4\mathscr{S}  = E^{-6}+\frac{1}{2}\epsilon^2 E^{-6} \left(\frac{9}{\Lambda}\hmfu^{00}-E^2\hmfu^{ij}\delta_{ij}\right) +\epsilon^4\mathscr{S}(\epsilon,t, E, \Omega/t, \hmfu^{\alpha\beta}),
\end{align*}
from which it follows that
	\begin{align} \label{idalpha}
	\underline{\alpha} E^2=1+\frac{1}{6}\epsilon^2 \left(\frac{9}{\Lambda}\hmfu^{00}-E^2\hmfu^{ij}\delta_{ij}\right)+\epsilon^4 \mathscr{S}(\epsilon,t,E, \Omega/t, \hmfu^{\alpha\beta})
	\end{align}
by \eqref{E:GAMMA}. Then by \eqref{E:u.f}, \eqref{E:q} and \eqref{idalpha}, we have
	\begin{align*}  %\label{E:U}
	u=2tu^{00}-\frac{1}{\epsilon}\frac{\Lambda}{3} \ln[1+(\underline{\alpha} E^2-1)]= \epsilon \frac{2\Lambda}{9} E^2\hmfu^{ij}\delta_{ij} +\epsilon^3\mathscr{S}(\epsilon,t,E, \Omega/t, \hmfu^{\alpha\beta}),
	\end{align*}
while
	\begin{align*}
	u^{ij}=\frac{1}{\epsilon}\left((\underline{\alpha}\theta)^{-1}\hat{g}^{ij}-E^{-1}\hat{h}^{ij}\right)=\epsilon E^2\left(\hmfu^{ij}-\frac{1}{3}\hmfu^{kl}\delta_{kl}\delta^{ij}\right)+\epsilon^3\mathscr{S}^{ij}( \epsilon,t,E, \Omega/t, \hmfu^{\alpha\beta})
	\end{align*}
follows from \eqref{E:GIJ}, \eqref{E:DEFHATG}, \eqref{idalpha} and
	\begin{align*}
	(\underline{\alpha} \theta)^{-1}=E^{-1}-\epsilon^2 \frac{1}{3}E\hmfu^{ij}\delta_{ij}+\epsilon^4\mathscr{S}(\epsilon,t,E, \Omega/t, \hmfu^{\mu\nu}).
	\end{align*}
Finally, \eqref{E:V0HAT} follows from \eqref{E:V^0}, \eqref{E:G0MU} and \eqref{E:U0MUANDINI}-\eqref{E:UIJANDINI}%\eqref{E:z.b}, \eqref{E:VLOW0_1}, and \eqref{E:DEFHATG}.
\end{proof}

\subsection{Solving the constraint equations}
We now need to write the constraint equations in a form that is suitable to used the methods
from  \cite{oli3,oli4}. We begin by defining the rescaled variables
\begin{align*} %\label{E:MFUIJ}
\hmfu^{ij}|_{t=T_0}= \epsilon\smfu^{ij}, \quad \hmfu^{ij}_0|_{t=T_0}=\smfu^{ij}_0, \quad \hmfu^{0\mu}|_{t=T_0}=\smfu^{0\mu}\quad \text{and} \quad \hmfu^{0\mu}_0|_{t=T_0}=\smfu^{0\mu}_0,
\end{align*}
and noting that
\begin{align}\label{E:PKMFU}
\partial_k\hmfu^{ij}|_{t=T_0}=\epsilon\partial_k\smfu^{ij}.
\end{align}
We  then observe that the following terms from \eqref{E:INITIALEINDIFF} can be represented as
\begin{align}
\epsilon^2 E^3 \hmfu^{\lambda\sigma} \hat{\partial}_\lambda \hat{\partial}_\sigma \hat{h}^{\mu 0} + 2(\theta^2 \bar{\nabla}^{(\mu}\bar{\gamma}^{0)}-E^6 \bar{\nabla}_H^{(\mu}\bar{\gamma}^{0)})-\epsilon^2 \theta E^3 \hmfu^{\mu 0} \frac{1}{t}\left(\partial_t\Omega -\frac{1}{t}\Omega \right) +\frac{2}{t}\bar{\gamma}^0(\theta\hat{g}^{0\mu}-E^3 \hat{h}^{\mu 0})\notag \\
+(E^3-\theta)\hat{h}^{\mu\nu}\frac{1}{t}\left(\del{t}\Omega-\frac{1}{t}\Omega \right)  +\theta(\hat{h}^{\mu\nu}-\hat{g}^{\mu\nu})\frac{1}{t}\bar{\gamma}^\lambda_{\lambda 0}\Omega  +(E^3-\theta)\hat{h}^{\mu\nu}\frac{1}{t}\bar{\gamma}^\lambda_{\lambda 0} \Omega+\frac{2\Lambda}{3t}(\theta-E^3)\partial_t\hat{h}^{0\mu}\notag \\
=\epsilon^2\mathscr{S}^\mu(\epsilon,t,E, \Omega/t, \del{t}\Omega, x, \hmfu^{\alpha\beta}, \partial_k\hmfu^{\alpha\beta}, \hmfu^{ij}_0). \label{idident1}
\end{align}
We also note that
\begin{align}
\underline{\bar{\Gamma}^\lambda_{\lambda 0}}-\bar{\gamma}^\lambda_{\lambda 0}=\frac{1}{\theta}\partial_t(\theta-E^3)+\frac{E^{-3}}{\theta E^{-3}}\partial_t E^3-\frac{3}{t}\Omega=\epsilon\frac{3}{2\Lambda} \partial_k\hmfu^{0k}+\epsilon^2\mathscr{S}(\epsilon,t, E, \Omega/t, \hmfu^{\mu\nu}, \partial_k \hmfu^{\mu\nu},  \hmfu^{ij}_0), \label{idident2}
\end{align}
Using  \eqref{E:PTU000J} and \eqref{E:PPTU000J} to replace the first and second time
derivatives of $\hmfu^{\mu 0}$  by spatial derivatives of $\hmfu^{\mu\nu}$
and the time derivatives $\hmfu^{ij}_0$ in \eqref{E:INITIALEINDIFF} with $\nu = 0$, we obtain, with the help of
\eqref{idident1}-\eqref{idident2}, the following elliptic equations
on $\Sigma_{T_0}$ for $\hmfu^{\mu 0}$:
\begin{align}
& \Delta\smfu^{00}-\frac{2\Lambda}{3T_0^2}E^2(T_0) \delta \rho  + \epsilon \left( \partial_k (\mathscr{T}^{00 k}_{\alpha\beta}\smfu^{\alpha\beta}) -   \frac{\Lambda}{3} E^2(T_0) \partial_k\partial_i \smfu^{ik}+ \bigl(\frac{\Lambda}{3t}+\frac{\Lambda+1}{2t}\Omega \bigr)E^2(T_0)\del{k}\mathfrak{\breve{u}}^{k0}\right)  \nnb \\&  \hspace{7cm} + \epsilon^2
\mathscr{R}^0(\epsilon,\smfu^{\mu\nu}, \partial_k \smfu^{\mu\nu}, \smfu^{ij}_0, \partial_i\partial_j \smfu^{0\mu}, \delta \rho, z^j)=0,
\label{E:FINALCONSTRAINT1}\\
&\Delta\smfu^{i0}+\epsilon\left(\frac{\Lambda}{3}E^2(T_0)\partial_k\smfu^{ik}_0 - \sqrt{\frac{\Lambda}{3}}\frac{2}{T_0^2}
E^2(T_0) \rho z ^i +\partial_k (\mathcal{T}^{0ik}_{\alpha\beta}\hmfu^{\alpha\beta})\right)+\epsilon^2\mathscr{R}^i(\epsilon,\smfu^{\mu\nu}, \partial_k \smfu^{\mu\nu}, \smfu^{ij}_0, \partial_i\partial_j \smfu^{0\mu}, \delta\rho, z^j) =0,
\label{E:FINALCONSTRAINT2}
\end{align}
where the coefficients $\mathcal{T}^{0\mu k}_{\alpha\beta}=\mathcal{T}^{0\mu k}_{\alpha\beta}(T_0)$ are constant on $\Sigma_{T_0}$ and the
remainder terms satisfy $\mathscr{R}^\mu(\epsilon,0,0,0,0,0,0)=0$.

\begin{remark}
From the above calculations, it is not difficult to see that
the elliptic equations \eqref{E:FINALCONSTRAINT1}-\eqref{E:FINALCONSTRAINT2} are equivalent to the gravitational constraint
equations \eqref{E:CONSTRAINT} provided that the gauge constraint \eqref{E:WAVECONSTRAINT} is also satisfied. Recalling that \eqref{E:PTU000J}
is equivalent to the gauge constraints, it is clear that we can solve the gauge constraints by using \eqref{E:PTU000J} to determine the time derivatives $\del{t}\smfu^{\mu 0}$ from the metric variables $\smfu^{\mu\nu}$ and their spatial derivatives.
\end{remark}

Decomposing $\delta \rho = \rho - \rho_H$ and $\rho z^i$ on $\Sigma_{T_0}$ as
\begin{align*}  %\label{E:REPRHO}
\delta \rho|_{t=T_0}=\breve{\rho}_0+\epsilon\breve{\phi} \quad \text{and}\quad (\rho z ^i)|_{t=T_0}
=\breve{\psi}^i+\breve{\nu}^i,
\end{align*}
where
\begin{align*}  %\label{E:DEFINIRHO}
\breve{\rho}_0:=\Pi \delta \rho|_{t=T_0}, \quad \breve{\phi}:=\frac{1}{\epsilon}\langle 1, \delta \rho\rangle|_{t=T_0}, \quad \breve{\psi}^i:=\langle 1, \rho z ^i\rangle|_{t=T_0} \quad \text{and} \quad \breve{\nu}^i=\Pi(\rho z ^i)|_{t=T_0},
\end{align*}
it is clear that $z^i|_{t=T_0}$ and $\delta\rho|_{t=T_0}$ depend analytically on $(\breve{\nu}^i, \breve{\psi}^i, \breve{\rho}_0, \breve{\phi})$,
and in particular,
\begin{align}\label{E:ZIANALYTICINITIALDATA}
z^i|_{t=T_0}=\frac{\breve{\nu}^i+\breve{\psi}^i}{\rho_H(T_0)+ \breve{\rho}_0+\epsilon \breve{\phi}}.
\end{align}
From this and the fact that the spatial derivatives $\partial_i: H^s(\mathbb{T}^n)\rightarrow H^{s-1}(\mathbb{T}^n)$ define bounded linear maps,
we can, by Lemmas \ref{L:MULTIPLICATION} and \ref{PSERIES} from Appendix \ref{A:ANALY}, view the remainder terms $\mathscr{R}^\mu$ from \eqref{E:FINALCONSTRAINT1}-\eqref{E:FINALCONSTRAINT2} as defining analytic maps
\begin{align}
 &(-\epsilon_0, \epsilon_0)\times B_r(H^{s+1}(\Tbb^3)) \times H^s(\Tbb^3)  \times B_r(H^s(\Tbb^3)) \times B_r(\Rbb) \times \Rbb^3  \times H^s(\Tbb^3) \ni (\epsilon ,\smfu^{\mu\nu}, \smfu^{ik}_0, \breve{\rho}_0, \breve{\phi}, \breve{\psi}^i, \breve{\nu}^i)
\notag \\
& \hspace{3.0cm} \longmapsto \mathscr{R}^\mu(\epsilon ,\smfu^{\mu\nu}, \smfu^{ik}_0, \breve{\rho}_0, \breve{\phi}, \breve{\psi}^i, \breve{\nu}^i) \in H^{s-1}(\Tbb^3) \label{Rscmap}
\end{align}
for $r>0$ chosen small enough. Using this observation, we can proceed with the existence proof for solutions
to the constraint equations.

\begin{theorem}\label{T:INITIALIZATION}
	Suppose $s\in \mathbb{Z}_{> n/2+1}$ and $r>0$, $\smfu^{ij}\in B_r(H^{s+1}(\mathbb{T}^3, \mathbb{S}_3))$, $\smfu^{ij}_0\in H^s(\mathbb{T}^3, \mathbb{S}_3)$, $\breve{\rho}_0\in B_r(\bar{H}^s(\mathbb{T}^3))$, $\breve{\nu}^i\in \bar{H}^s(\mathbb{T}^3, \mathbb{R}^3)$.
Then for $r>0$ chosen small enough so that the map \eqref{Rscmap} is well-defined and analytic, there exists an $\epsilon_0>0$,
	and analytic maps
	$\breve{\phi}\in C^\omega(X_{\epsilon_0,r}^s,\mathbb{R})$, $\breve{\psi}^l\in C^\omega(X_{\epsilon_0,r}^s,\mathbb{R}^3)$, $\smfu^{0\mu}\in C^\omega(X_{\epsilon_0,r}^s,\bar{H}^{s+1}(\mathbb{T}^3,\mathbb{R}^4))$ and $\smfu^{0\mu}_0 \in C^\omega(X_{\epsilon_0,r}^s,H^s(\mathbb{T}^3,\mathbb{R}^4))$ that satisfy
	\begin{equation*}\label{E:QUASILINEAR}
	\breve{\phi}( \epsilon,0,0,0,0)=0, \quad \breve{\psi}^l(\epsilon,0,0,0,0)=0, \quad \smfu^{0\mu}( \epsilon,0,0,0,0)=0
	\quad\text{and} \quad
	\smfu^{0\mu}_0( \epsilon,0,0,0,0)=0
	\end{equation*}
	such that
	\begin{align*}
	\rho|_{t=T_0}&=\rho_H(T_0)+\breve{\rho}_0+\epsilon \breve{\phi},%\label{E:INITIALRHODECOMP}
\\
	z^i|_{t=T_0}&= \frac{\breve{\nu}^i+\breve{\psi}^i}{\rho_H(T_0)+ \breve{\rho}_0+\epsilon \breve{\phi}},\\
	\mfu^{\mu\nu}|_{t=T_0}&=\begin{pmatrix}
	\smfu^{00} & \smfu^{0j}\\
	\smfu^{i0} & \epsilon\smfu^{ij}
	\end{pmatrix},\\
	\partial_0\mfu^{\mu\nu}|_{t=T_0}&=\begin{pmatrix}
	\smfu^{00}_0 & \smfu^{0j}_0\\
	\smfu^{i0}_0 & \smfu^{ij}_0
	\end{pmatrix},
	\end{align*}
	 where the $\smfu^{\mu0}_0$ are determined by \eqref{E:PTU000J},
	solve the constraints \eqref{E:CONSTRAINT}, \eqref{E:WAVECONSTRAINT} and \eqref{E:NORMALIZATION}. Moreover,
 the fields $\{ \breve{\phi}, \breve{\psi}^i, \smfu^{00},\smfu^{0i}\}$ satisfy the estimate
\begin{align*}%\label{E:INITIALDATAESTIMATE1}
	|\breve{\phi}|+|\breve{\psi}^i|+\|\smfu^{0\mu}\|_{H^{s+1}}
	+\|\smfu^{0\mu}_0\|_{H^s}\lesssim \|\smfu^{ik}\|_{H^{s+1}}+\|\smfu^{ik}_0\|_{H^s}+\|\breve{\rho}_0\|_{H^s}+\|\breve{\nu}^i\|_{H^s}
	\end{align*}
uniformly for $\epsilon \in (-\epsilon_0,\epsilon_0)$ and can be expanded as
	\begin{gather}
	\breve{\phi}=\epsilon \mathscr{S}(\epsilon,\smfu^{jk}, \smfu^{jk}_0, \breve{\rho}_0, \breve{\nu}^i), \quad
	\breve{\psi}^i= \epsilon \mathscr{S}^i(\epsilon,\smfu^{jk}, \smfu^{jk}_0, \breve{\rho}_0, \breve{\nu}^i),  \label{E:INITIALEXPANSIONa}\\
	\smfu^{00}=\frac{2\Lambda}{3 T_0^2} E^2(T_0)\Delta^{-1} \breve{\rho}_0+
\epsilon  \mathscr{S}(\epsilon,\smfu^{jk}, \smfu^{jk}_0, \breve{\rho}_0, \breve{\nu}^i)
\quad \text{and} \quad
	\smfu^{0i}= \epsilon \mathscr{S}^i(\epsilon,\smfu^{jk}, \smfu^{jk}_0, \breve{\rho}_0, \breve{\nu}^i),
\label{E:INITIALEXPANSIONb}
	\end{gather}
where the maps $\mathscr{S}$ and $\mathscr{S}^i$ that are analytic on $X_{\epsilon_0,r}^s$ and vanish for $(\epsilon,\smfu^{jk}, \smfu^{jk}_0, \breve{\rho}_0, \breve{\phi})=(\epsilon,0,0,0,0)$.
\end{theorem}
\begin{proof}
Acting on \eqref{E:FINALCONSTRAINT1} and \eqref{E:FINALCONSTRAINT2} with $\langle 1, \cdot\rangle$ and $\Pi$ , we obtain,
with the help of  \eqref{E:V0HAT}  and \eqref{Rscmap}, the equations
\begin{align}
	\breve{\phi} - \epsilon \left\langle 1,\mathscr{R}^0\left(\epsilon,\smfu^{\mu\nu}, \smfu^{jk}_0, \breve{\rho}_0, \breve{\phi}, \breve{\psi}^j, \breve{\nu}^j\right)\right\rangle=0, \label{consteqn1}\\
	\Delta \smfu^{00}-\frac{2\Lambda}{3 T_0^2} E^2(T_0)\breve{\rho}_0 +\epsilon  \Pi \mathscr{R}^0 \left(\epsilon,\smfu^{\mu\nu}, \smfu^{jk}_0,  \breve{\rho}_0, \breve{\phi}, \breve{\psi}^j, \breve{\nu}^j\right)=0, \label{consteqn2}\\
	\breve{\psi}^i+\epsilon \left\langle 1,\mathscr{R}^i \left(\epsilon,\smfu^{\mu\nu}, \smfu^{jk}_0,  \breve{\rho}_0, \breve{\phi}, \breve{\psi}^j, \breve{\nu}^j\right)\right\rangle=0 \label{consteqn3}
\intertext{and}
	\Delta\smfu^{0i}  +\epsilon
	\Pi
	\mathscr{R}^i \left(\epsilon,\smfu^{\mu\nu}, \smfu^{jk}_0,  \breve{\rho}_0, \breve{\phi}, \breve{\psi}^j, \breve{\nu}^j\right)=0,
\label{consteqn4}
	\end{align}
which are clearly equivalent to \eqref{E:FINALCONSTRAINT1}-\eqref{E:FINALCONSTRAINT2}.  Next, we let
	\begin{align*}
	\iota:=(\smfu^{ik},\smfu^{ik}_0, \breve{\rho}_0, \breve{\nu}^i) \quad \text{and} \quad \beta:=\left(\breve{\phi}, \breve{\psi}^i, \smfu^{0\mu}\right),
	\end{align*}
and write \eqref{consteqn1}-\eqref{consteqn2} more compactly as
	\begin{align}
	F(\epsilon, \iota,\beta) := L(\iota, \beta)+\epsilon M(\epsilon, \iota, \beta)=0, \label{E:L+M}
	\end{align}
	where
	\begin{align*}
	L(\iota, \beta)=\begin{pmatrix}
	\breve{\phi}\\
	\breve{\psi}^i\\
	\Delta\smfu^{0\mu}-\frac{2\Lambda}{3 T_0^2} E^2(T_0) \delta^\mu_0\breve{\rho}_0
	\end{pmatrix}.
	\end{align*}
Recalling that the Laplacian $\Delta$
defines an isomorphism from $\bar{H}^{s+1}(\mathbb{T}^3)$  to $\bar{H}^{s-1}(\mathbb{T}^3)$, we observe that
	\begin{align*}
	(0,\iota, \beta)=\left(0, \iota, \begin{pmatrix}
	0\\
	0\\
	\frac{2\Lambda}{3 T_0^2} E^2(T_0) \delta^\mu_0\Delta^{-1} \breve{\rho}_0
	\end{pmatrix}\right)
	\end{align*}
	solves \eqref{E:L+M}. Since $D_\beta F(0, \iota,\beta )\cdot \delta\beta = L(0,\delta\beta)$,  we can solve \eqref{E:L+M} via an analytic version of the Implicit Function
Theorem \cite[Theorem 15.3]{Deimling:1998}, at least for small  $\epsilon$, if we can show that
	\begin{align*}
	\tilde{L}(\delta\beta) = \begin{pmatrix}
	\delta \breve{\phi}\\
	\delta \breve{\psi}^i\\
	\Delta\delta\smfu^{0\mu}
	\end{pmatrix}
	\end{align*}
defines an isomorphism from  $\Rbb \times \Rbb^3 \times \bar{H}^{s+1}(\mathbb{T}^3, \mathbb{R}^4)$ to
$\Rbb \times \Rbb^3 \times \bar{H}^{s-1}(\mathbb{T}^3, \mathbb{R}^4)$.  But this is clear since $\Delta$ $\: :\:$ $\bar{H}^{s+1}(\mathbb{T}^3)$  $\mapsto$ $\bar{H}^{s-1}(\mathbb{T}^3)$ is an isomorphism. Thus, for $r>0$
chosen small enough and
any $R>0$, there exists an $\epsilon_0 >0$ and a unique analytic map
\begin{equation*}
P \: : \: X^s_{\epsilon_0,r} \mapsto \Rbb\times \Rbb^3 \times \bar{H}^{s+1}(\Tbb^3,\mathbb{R}^4)
\end{equation*}
that satisfies
\begin{equation*}
F(\iota,P(\epsilon,\iota),\epsilon) = 0
\end{equation*}
for all $(\epsilon, \iota)$ $\in$
$(-\epsilon_0,\epsilon_0) \times B_R\bigl( H^{s+1}(\Tbb^3,\mathbb{S}_3)\bigr)\times B_R\bigl( H^{s}(\Tbb^3,\mathbb{S}_3)\bigr)
\times B_r\bigl(\bar{H}^s(\Tbb^3)\bigr) \times  B_r\bigl(\bar{H}^s(\Tbb^3,\Rbb^3)\bigr)$
and
	\begin{align}
	P(\epsilon, \iota)=\begin{pmatrix}
	0\\
	0\\
	\frac{2\Lambda}{3 T_0^2}E^2(T_0) \delta^\mu_0\Delta^{-1} \breve{\rho}_0
	\end{pmatrix}+\text{O}(\epsilon). \label{Pidexp}
	\end{align}
Finally, the estimate
	\begin{align*}
	|\breve{\phi}| + |\breve{\psi}^i|+\|\smfu^{0\mu}\|_{H^{s+1}}
	+\|\mfu^{0\mu}_0\|_{H^s}\lesssim \|\smfu^{ik}\|_{H^{s+1}}+\|\smfu^{ik}_0\|_{H^s}+\|\breve{\rho}_0\|_{H^s}+\|\breve{\nu}^i\|_{H^s}
	\end{align*}
follows from analyticity of $P$,  $\eqref{Pidexp}$ and \eqref{E:PTU000J}.
\end{proof}

\subsection{Bounding $\mathbf{U}|_{t=T_0}$}
For the evolution problem, we need to bound $\mathbf{U}|_{t=T_0}$, see \eqref{E:REALVAR}, by the free initial data $\{\smfu^{ik}, \smfu^{ik}_0, \breve{\rho}_0, \breve{\nu}^i\}$ uniformly in $\epsilon$. The required bound is the content of
the following lemma.

\begin{lemma}\label{L:INITIALTRANSFER}
Suppose that the hypotheses of Theorem \ref{T:INITIALIZATION} hold, and that
$\breve{\phi}\in C^\omega(X_{\epsilon_0,r}^s,\mathbb{R})$, $\breve{\psi}^l\in C^\omega(X_{\epsilon_0,r}^s,\mathbb{R}^3)$, $\smfu^{0\mu}\in C^\omega(X_{\epsilon_0,r}^s,\bar{H}^{s+1}(\mathbb{T}^3,\mathbb{R}^4))$ and $\smfu^{0\mu}_0 \in C^\omega(X_{\epsilon_0,r}^s,H^s(\mathbb{T}^3,\mathbb{R}^4))$ are the analytic maps from that theorem.
Then on the initial hypersurface $\Sigma_{T_0}$, the gravitational and matter fields
\begin{equation*}
\{u^{\mu\nu},u^{ij}_\gamma,
w^{0\mu}_i,u^{0\mu}_0,u,u_\gamma,z_j,\delta \zeta\}
\end{equation*}
can be expanded as follows:
\begin{align*}
	u^{0\mu}|_{t=T_0}&= \epsilon \frac{\Lambda}{6T_0^3}E^2(T_0)\Delta^{-1}\breve{\rho}_0 \delta^\mu_0+\epsilon^2
\mathscr{S}^\mu(\epsilon,\smfu^{kl},\smfu^{kl}_0, \breve{\rho}_0, \breve{\nu}^l), %\label{E:U0MUANDINI2}
\\
	u|_{t=T_0}&=\epsilon^2\frac{2\Lambda}{9}E^2(T_0)\smfu^{ij}\delta_{ij}+\epsilon^3 \mathscr{S}(\epsilon,\smfu^{kl},\smfu^{kl}_0, \breve{\rho}_0, \breve{\nu}^l),  %\label{E:UANDINI2}
\\
	u^{ij}|_{t=T_0}&=\epsilon^2E^2(T_0)\left(\smfu^{ij}
	-\frac{1}{3}\smfu^{kl}\delta_{kl}\delta^{ij}\right)+\epsilon^3\mathscr{S}^{ij}(\epsilon,\smfu^{kl},\smfu^{kl}_0, \breve{\rho}_0, \breve{\nu}^l), %\label{E:UIJANDINI2}
\\
	z_j|_{t=T_0}&= E^2(T_0)\frac{\breve{\nu}^i\delta_{ij}}{\rho_H(T_0)+\breve{\rho}_0}+\epsilon\mathscr{S}_j(\epsilon,\smfu^{kl},\smfu^{kl}_0, \breve{\rho}_0, \breve{\nu}^l),%\label{E:ZJINI}
\\
	\delta\zeta|_{t=T_0}&=\frac{1}{1+\epsilon^2 K}\ln{\left(1+\frac{\breve{\rho}_0+\epsilon\breve{\phi}}{\rho_H(T_0)} \right)}=\ln{\left(1+\frac{\breve{\rho}_0}{\rho_H(T_0)}\right)}+\epsilon^2\mathscr{S}( \epsilon,\smfu^{kl},\smfu^{kl}_0, \breve{\rho}_0, \breve{\nu}^l), %\label{E:DELTAZETA},
\\
 w^{0\mu}_i|_{t=T_0} &= \epsilon \mathscr{S}^\mu_i(\epsilon,\smfu^{kl},\smfu^{kl}_0, \breve{\rho}_0, \breve{\nu}^l), %\label{wexp}
 \\
u^{0\mu}_0|_{t=T_0} &= \epsilon \mathscr{S}^\mu(\epsilon,\smfu^{kl},\smfu^{kl}_0, \breve{\rho}_0, \breve{\nu}^l), %\label{u0muexp}
\\
u_\gamma|_{t=T_0} &= \epsilon \mathscr{S}_\gamma(\epsilon,\smfu^{kl},\smfu^{kl}_0, \breve{\rho}_0, \breve{\nu}^l)%\label{ugammaexp}
\intertext{and}
u^{ij}_\gamma|_{t=T_0} &= \epsilon \mathscr{S}^{ij}_\gamma(\epsilon,\smfu^{kl},\smfu^{kl}_0, \breve{\rho}_0, \breve{\nu}^l),%\label{uijgammaexp}
\end{align*}
for maps $\mathscr{S}$ that are analytic on $X^s_{\epsilon_0,r}$. Moreover, the estimates
	\begin{equation*}\label{E:INITIALDATAESTIMATE2}
	\begin{aligned}
	 \|u^{\mu\nu}|_{t=T_0}\|_{H^{s+1}}+\|u|_{t=T_0}\|_{H^{s+1}}+\|w^{0\mu}_i|_{t=T_0}
	\|_{H^s}+\|u^{0\mu}_0|_{t=T_0}\|_{H^s}+
    \|u_\mu|_{t=T_0}\|_{H^s}& \\
+\|u^{ij}_\mu|_{t=T_0}\|_{H^s}+|\phi(T_0)|  \lesssim  \epsilon(\|\smfu^{ij}\|_{H^{s+1}}+\|\smfu^{ij}_0\|_{H^{s}}+\|\breve{\rho}_0&\|_{H^s}+\|
	\breve{\nu}^i\|_{H^s})
	\end{aligned}
	\end{equation*}
	and
	\begin{align*}%\label{E:INITIALDATAESTIMATE3}
	\|z_j|_{t=T_0}\|_{H^s}+\|\delta\zeta|_{t=T_0}\|_{H^s}\lesssim \|\smfu^{ij}\|_{H^{s+1}}+\|\smfu^{ij}_0\|_{H^{s}}+\|\breve{\rho}_0\|_{H^s}+\|
	\breve{\nu}^i\|_{H^s}
	\end{align*}
hold uniformly for $\epsilon \in (-\epsilon_0,\epsilon_0)$.
\end{lemma}
\begin{proof}

First, we observe by \eqref{E:PIG},  \eqref{E:WPHI}, \eqref{E:DEFOFPHI}, \eqref{phidecomp}, \eqref{E:DEFHATG}, \eqref{E:PKMFU} and Lemma \ref{L:IDENTITY} that
	\begin{align}\label{E:REPW0MUI1}
	w^{0\mu}_i|_{t=T_0}= \frac{1}{2}\partial_i\smfu^{00}\delta^\mu_0-\delta^\mu_0
	\frac{\Lambda}{3T_0^2}E^2(T_0)\partial_i \Delta^{-1}\breve{\rho}_0+\partial_i \smfu^{0k}\delta^\mu_k+\epsilon^2
\mathscr{S}^\mu_i(\epsilon,\smfu^{kl},\smfu^{kl}_0, \breve{\rho}_0, \breve{\nu}^l),
	\end{align}
	where $\mathscr{S}_i^\mu(\epsilon,0,0,0,0)=0$, which in turn, implies by \eqref{E:INITIALEXPANSIONb} that
	\begin{align}\label{E:REPW0MUI}
	w^{0\mu}_i|_{t=T_0}=\epsilon \mathscr{S}^\mu_i(\epsilon,\smfu^{kl},\smfu^{kl}_0, \breve{\rho}_0, \breve{\nu}^l),
	\end{align}
where again $\mathscr{S}_i^\mu(\epsilon,0,0,0,0)=0$.
Furthermore, by \eqref{E:PIG}, \eqref{E:DEFHATG}, \eqref{E:U0MUANDINI}, Lemma \ref{L:IDENTITY} and Theorem \ref{T:INITIALIZATION},
we see that
	\begin{equation}\label{E:U0MU0ANDINI}
	u^{0\mu}_0|_{t=T_0}= \frac{1}{\epsilon} \frac{1}{\theta}\partial _0\hat{g}^{0\mu}-\frac{1}{\epsilon} \hat{g}^{0\mu}\frac{1}{\theta^2}\partial _0 \theta-3u^{0\mu} = \epsilon\mathscr{S}^\mu(\epsilon,\smfu^{ik},\smfu^{ik}_0, \breve{\rho}_0, \breve{\nu}^i),
	\end{equation}
	where $\mathscr{S}^\mu(\epsilon,0,0,0,0)=0$.

Next, we see from \eqref{E:PTE}, \eqref{E:GAMMA}, \eqref{E:DEFHATG}, \eqref{E:PTTHETA}, \eqref{E:PITHETA}  and Theorem \ref{T:INITIALIZATION},
that we can express $\partial_\mu \underline{\alpha}$ as
	\begin{align}
	\underline{\alpha}^{-3}\partial_t\underline{\alpha}^3=3\underline{\alpha}^{-1}\partial_t \underline{\alpha}=\underline{\check{g}_{kl}}\partial_t\underline{\bar{g}^{kl}}= -\frac{6\Omega(T_0)}{T_0}+\epsilon^2\mathscr{S}(\epsilon,\smfu^{ik},\smfu^{ik}_0, \breve{\rho}_0, \breve{\nu}^i)\label{E:PTGAMMA}
	\end{align}
	and
	\begin{align}
	\underline{\alpha}^{-3}\partial_j \underline{\alpha}^3=&3\underline{\alpha}^{-1}\partial_j \underline{\alpha}=\epsilon^2\frac{9}{2\Lambda}\partial_j \smfu^{00}+\epsilon^3\mathscr{S} (\epsilon,\smfu^{ik},\smfu^{ik}_0, \breve{\rho}_0, \breve{\nu}^i) \notag \\ =& \epsilon^2 \frac{3}{T_0^2}E^2(T_0)\Delta^{-1}\partial_j\breve{\rho}_0 +\epsilon^3\mathscr{S} (\epsilon,\smfu^{ik},\smfu^{ik}_0, \breve{\rho}_0, \breve{\nu}^i),
	\label{E:PJGAMMA}
	\end{align}
where the error terms $\mathscr{S}$ vanish for $(\epsilon,\smfu^{ik},\smfu^{ik}_0, \breve{\rho}_0, \breve{\nu}^i)=(\epsilon,0,0,0,0)$.
Using  \eqref{E:U0MUANDINI} and \eqref{E:U0MU0ANDINI}, we then find with the help of \eqref{E:u.g}  and \eqref{E:PTGAMMA}  that
	\begin{align}
	u_0|_{t=T_0}= 3u^{00}+u^{00}_0-\frac{1}{\epsilon}
	\frac{\Lambda}{9} \underline{\alpha}^{-3} \partial_t \underline{\alpha}^3-\frac{1}{\epsilon}\frac{2\Lambda}{3}\frac{\Omega(T_0)}{T_0}=\epsilon \mathscr{S}_0(\epsilon,\smfu^{ik},\smfu^{ik}_0, \breve{\rho}_0, \breve{\nu}^i)\label{E:U0},
	\end{align}
while we note that
	\begin{equation}
	\begin{aligned}
	u_k|_{t=T_0}=&w^{00}_k+\frac{\Lambda}{3T^2_0} E^2 (T_0) \partial_k\Delta^{-1}\breve{\rho}_0-\frac{1}{\epsilon^2}
	\frac{\Lambda}{9}\underline{\alpha}^{-3} \partial_k \underline{\alpha}^3
	= \epsilon \mathscr{S}_k(\epsilon,\smfu^{ij},\smfu^{ij}_0, \breve{\rho}_0, \breve{\nu}^i)\label{E:UK}
	\end{aligned}
	\end{equation}
follows from \eqref{E:DEFOFPHI}, \eqref{E:PJGAMMA} and \eqref{E:REPW0MUI1}. Again the error terms  $\mathscr{S}_\mu$ vanish for $(\epsilon,\smfu^{ik},\smfu^{ik}_0, \breve{\rho}_0, \breve{\nu}^i)=(\epsilon,0,0,0,0)$.
Starting from \eqref{E:u.e} and \eqref{E:GAMMA},  we see, with the help of \eqref{E:U0},  Lemmas \ref{L:IDENTITY} and \ref{E:IDENTITYPTHETA} along with Theorem \ref{T:INITIALIZATION}, that
	\begin{equation}\label{E:UIJ0ANDINI}
	u^{ij}_0|_{t=T_0}=  \frac{1}{\epsilon}\partial_0(\underline{\alpha}^{-1}\theta^{-1}\hat{g}^{ij})
	=  \epsilon\mathscr{S}^{ij}(\epsilon,\smfu^{kl},\smfu^{kl}_0, \breve{\rho}_0, \breve{\nu}^l),
	\end{equation}
where $\mathscr{S}^{ij}(\epsilon,0,0,0,0)=0$. By a similar calculation, we find with the help of \eqref{E:REPW0MUI} and\eqref{E:UK} that
	\begin{equation}\label{E:UIJKANDINI}
	u^{ij}_k|_{t=T_0}=  \frac{1}{\epsilon}\partial_k(\underline{\alpha}^{-1}\theta^{-1}\hat{g}^{ij}) =   \epsilon\mathscr{S}^{ij}(\epsilon,\smfu^{kl},\smfu^{kl}_0, \breve{\rho}_0, \breve{\nu}^l),
	\end{equation}
where $\mathscr{S}^{ij}(\epsilon,0,0,0,0)=0$. Noting that
\begin{equation*}
\phi(T_0)=\frac{1}{T_0^{3(1+\epsilon^2 K)}}\breve{\phi}=\epsilon \mathscr{S}(\epsilon,\smfu^{ik},\smfu^{ik}_0, \breve{\rho}_0, \breve{\nu}^i)
\end{equation*}
by Theorem \ref{T:INITIALIZATION}, the estimate
\begin{align*}
	 \|u^{\mu\nu}|_{t=T_0}\|_{H^{s+1}}+\|u|_{t=T_0}\|_{H^{s+1}}+\|w^{0\mu}_i|_{t=T_0}
	\|_{H^s}+\|u^{0\mu}_0|_{t=T_0}\|_{H^s}+
    \|u_\mu|_{t=T_0}\|_{H^s}& \\
+\|u^{ij}_\mu|_{t=T_0}\|_{H^s}+|\phi(T_0)|  \lesssim  \epsilon(\|\smfu^{ij}\|_{H^{s+1}}+\|\smfu^{ij}_0\|_{H^{s}}+\|\breve{\rho}_0&\|_{H^s}+\|
	\breve{\nu}^i\|_{H^s}),
	\end{align*}
which holds uniformly for $\epsilon \in (-\epsilon_0,\epsilon_0)$, follows directly from
\eqref{E:REPW0MUI}, \eqref{E:U0MU0ANDINI}, \eqref{E:U0}, \eqref{E:UK}, \eqref{E:UIJ0ANDINI}, \eqref{E:UIJKANDINI}, Lemma \ref{L:RELATION1}
and Theorem \ref{T:INITIALIZATION}.

Next, we observe from $z_j=\frac{1}{\epsilon}\underline{\bar{g}_{j0}}\underline{\bar{v}^0}+\underline{\bar{g}_{ij}}z^i$, \eqref{E:V^0},  \eqref{E:DEFHATG},
\eqref{E:U0MUANDINI}-\eqref{E:UIJANDINI},  \eqref{E:ZIANALYTICINITIALDATA} and Theorem \ref{T:INITIALIZATION} that we can write $z_j|_{t=T_0}$ as
	\begin{equation} \label{zjidata}
	z_j|_{t=T_0}= E^2(T_0)\frac{\breve{\nu}^i\delta_{ij}}{\rho_H(T_0)+\breve{\rho}_0}+\epsilon\mathscr{S} (\epsilon,\smfu^{ik},\smfu^{ik}_0, \breve{\rho}_0, \breve{\nu}^i),
	\end{equation}
where $\mathscr{S} (\epsilon,0,0,0,0)=0$. In addition, we note that
	\begin{equation} \label{delzetaidata}
	\delta\zeta|_{t=T_0}=\frac{1}{1+\epsilon^2 K}\ln{\left(1+\frac{\breve{\rho}_0+\epsilon\breve{\phi}}{\rho_H(T_0)} \right)}=\ln{\left(1+\frac{\breve{\rho}_0}{\rho_H(T_0)}\right)}+\epsilon^2\mathscr{S}( \epsilon,\smfu^{ik},\smfu^{ik}_0, \breve{\rho}_0, \breve{\nu}^i)
	\end{equation}
follows  from \eqref{E:ZETAH2}, \eqref{E:DELRHO} and Theorem \ref{T:INITIALIZATION}, where
$\mathscr{S} (\epsilon,0,0,0,0)=0$. Together, \eqref{zjidata} and \eqref{delzetaidata} imply that the estimate
	\begin{align*}
	\|z_j|_{t=T_0}\|_{H^s}+\|\delta\zeta|_{t=T_0}\|_{H^s}\lesssim \|\smfu^{ij}\|_{H^{s+1}}+\|\smfu^{ij}_0\|_{H^{s}}+\|\breve{\rho}_0\|_{H^s}+\|
	\breve{\nu}^i\|_{H^s}
	\end{align*}
holds uniformly for $\epsilon \in (-\epsilon_0,\epsilon_0)$.
\end{proof}
%---------------- end section 6 --------------------------------------------------------------------------------
\section{Proof of Theorem \ref{T:MAINTHEOREM}}\label{S:MAINPROOF}

\subsection{Transforming the conformal Einstein--Euler equations\label{proof:EEeqn}}
The first step of the proof is to observe that the non-local formulation of the conformal Einstein--Euler equations
given by \eqref{E:REALEQ} can be transformed into the form \eqref{E:MODELEQ3a} analyzed in \S \ref{S:MODEL}
by making the simple change of time coordinate
\al{TIMECHANGE}{
	t\mapsto \hat{t}:=-t
	}
and the substitutions
\begin{gather}
w(\hat{t},x)=\mathbf{U}(-\hat{t},x), \quad A_1^0(\epsilon,-\hat{t},w)=\mathbf{B}^0(\epsilon,-\hat{t},\mathbf{U}),  \quad A_1^i(\epsilon,\hat{t},w) =-\mathbf{B}^i(\epsilon,-\hat{t},\mathbf{U}),  \quad \mathfrak{A}_1(\epsilon,\hat{t},w) =\mathbf{B}(\epsilon,-\hat{t},\mathbf{U}), \label{singdefA} \\  C_1^i=-\mathbf{C}^i, \quad
\mathbb{P}_1=\mathbf{P}, \quad     H_1(\epsilon,\hat{t},w)=-\mathbf{H}(\epsilon,-\hat{t},\mathbf{U}) \AND F_1(\epsilon,\hat{t},x)=-\mathbf{F}(\epsilon,-\hat{t},x,\mathbf{U},\del{k}\Phi,\del{t}\del{k}\Phi,\del{k}\del{l}\Phi). \label{singdefB}
\end{gather}
With these choices, it is clear that the evolution equations \eqref{E:REALEQ} on the spacetime region
$t\in (T_1,1]$, $0< T_1 < 1$, are equivalent to
 \begin{equation*} \label{proofeq1}
  A_1^0\partial_{\hat{t}} w+A_1^i\partial_i w+\frac{1}{\epsilon}C_1^i\partial_i w =\frac{1}{\hat{t}}\mathfrak{A}_1\mathbb{P}_1  w+H_1+F_1
\qquad \text{for $(\hat{t},x)\in[-1,-T_1)\times \Tbb^3$},
  \end{equation*}
which is of the form studied in \S \ref{S:MODELerr}, see \eqref{E:MODELEQ3a}.
Furthermore,  it is not difficult to verify (see \cite[\S 3]{oli5} for details)  that matrices
$\{A^\mu_1, C^i_1, \mathfrak{A}_1,\mathbb{P}_1\}$ and
the source term $H_1$ satisfy the Assumptions \ref{ASS1}.\eqref{A:CONSC}-\eqref{A:PBP} from
\S \ref{S:MODELuni}
for some positive constants $\kappa, \gamma_1, \gamma_2 > 0$.

To see that Assumption \ref{ASS1}.\eqref{A:PDECOMPOSABLE} is also satisfied is more involved. First, we
note that this assumption is equivalent to verifying
$\mathbf{P}^\perp[D_\mathbf{U}\mathbf{B}^0\cdot (\mathbf{B}^0)^{-1}\mathbf{B}\mathbf{P}\mathbf{U}]\mathbf{P}^\perp$
admits an expansion of the type \eqref{E:ADEC}. To see why this is the case, we recall that $\mathbf{B}^0$ and $\mathbf{P}$ are block matrices, see \eqref{E:REALEQa}-\eqref{E:REALEQb}, from which it is clear using
\eqref{E:EINBk}-\eqref{E:EINP2}  that  we can expand $\mathbf{P}^\perp[D_\mathbf{U}\mathbf{B}^0\cdot (\mathbf{B}^0)^{-1}\mathbf{B}\mathbf{P}\mathbf{U} ]\mathbf{P}^\perp$ as
\al{PBPVER}{
	\p{\Pbb_2^\perp[D_\mathbf{U}\tilde{B}^0\cdot \mathbf{W}]\Pbb_2^\perp & 0 & 0 & 0 & 0\\
		0 & \breve{\Pbb}_2^\perp[D_\mathbf{U}\tilde{B}^0\cdot \mathbf{W}]\breve{\Pbb}_2^\perp & 0 & 0 & 0 \\
		0 & 0 & \breve{\Pbb}_2^\perp[D_\mathbf{U}\tilde{B}^0\cdot \mathbf{W}]\breve{\Pbb}_2^\perp & 0 & 0 \\
		0 & 0 & 0 & \hat{\Pbb}_2^\perp[D_\mathbf{U}B^0\cdot \mathbf{W}]\hat{\Pbb}_2^\perp & 0 \\
		0 & 0 & 0 & 0 & 0
		},
	}
where
\als%{WWW}
{	\mathbf{W}:=(\mathbf{B}^0)^{-1}\mathbf{B}\mathbf{P}\mathbf{U}&=\p{(\tilde{B}^0)^{-1} \tilde{\mathfrak{B}}\Pbb_2 & 0 & 0 & 0 & 0\\
	 0 & -2E^2\underline{\bar{g}^{00}}(\tilde{B}^0)^{-1} \breve{\Pbb}_2 & 0 & 0 & 0	\\
	0 & 0 & -2E^2\underline{\bar{g}^{00}}(\tilde{B}^0)^{-1} \breve{\Pbb}_2 &  0 & 0 \\
	0 & 0 & 0 & (B^0)^{-1}\mathfrak{B}\hat{\Pbb}_2 & 0 \\
	0 & 0 & 0 & 0 & 0
    }\mathbf{U}  \nnb \\
    &=
%\p{ \Pbb_2 & 0 & 0 & 0 & 0\\
  %  	0 &   \breve{\Pbb}_2 & 0 & 0 & 0	\\
   % 	0 & 0 &   \breve{\Pbb}_2 &  0 & 0 \\
    %	0 & 0 & 0 & \mathds{1} & 0 \\
    %	0 & 0 & 0 & 0 & 0}
\mathbf{P}
\p{\mathbf{Y}  & 0 & 0 & 0 & 0\\
    	0 & -2\mathds{1}   & 0 & 0 & 0	\\
    	0 & 0 & -2\mathds{1}  &  0 & 0 \\
    	0 & 0 & 0 & (B^0)^{-1}\mathfrak{B}\hat{\Pbb}_2 & 0 \\
    	0 & 0 & 0 & 0 & 0
    }\mathbf{U}
}
with
\begin{equation*}
\mathbf{Y}=\p{1 & 0 & 0 \\
    	0 & \frac{3}{2} \delta^i_j & 0 \\
    	0 & 0 & 1}.
\end{equation*}
Next, by \eqref{E:u.d}, \eqref{E:GAMMA}, \eqref{E:SMALLGAMMA}, \eqref{E:G0MU}, and \eqref{E:EINBk},  we observe
that $\tilde{B}^0$ can be expressed as
\begin{align*}
\tilde{B}^0%(\epsilon,t,\mathbf{U}_1)
=E^2\begin{pmatrix}
\frac{\Lambda}{3}-2\epsilon t u^{00} & 0 & 0\\
0 & (\delta^{ij}+\epsilon u^{ij})E^{-2}\exp{\left(\epsilon\frac{3}{\Lambda}(2tu^{00}-u)\right)} & 0\\
0 & 0 & \frac{\Lambda}{3}-2\epsilon t u^{00}
\end{pmatrix}.
\end{align*}
%Let $\mathbf{B}_1^0=\diag{(\tilde{B}^0,\tilde{B}^0,\tilde{B}^0)}$, $\mathbf{B}_1=\diag{(\mathfrak{\tilde{B}},-2E^2\underline{\bar{g}^{00}}I, -2E^2\underline{\bar{g}^{00}}I)}$ and $\mathbf{P}_1=\diag{(\mathbb{P}_2,\breve{\mathbb{P}}_2,\breve{\mathbb{P}}_2)}$.
Noting from definition \eqref{E:U1} of $\mathbf{U}_1$  that $u^{ij}$ and $u$ are components of  the vector $\mathbf{P}_1^\perp \mathbf{U}_1$, where
\begin{equation*}
\mathbf{P}_1=\diag{(\mathbb{P}_2,\breve{\mathbb{P}}_2,\breve{\mathbb{P}}_2)},
\end{equation*}
it is clear that $\tilde{B}^0$, as
a map, depends only on the the variables $(\epsilon,t\mathbf{U}_1, \mathbf{P}_1^\perp\mathbf{U}_1)$. To make this
explicit, we define the map   $\hat{B}^0 (\epsilon,t\mathbf{U}_1, \mathbf{P}_1^\perp\mathbf{U}_1)
:=\tilde{B}^0(\epsilon,t,\mathbf{U})$. Letting $\mathscr{P}$ denote linear maps that projects out the components $\mathbf{U}_1$ from
$\mathbf{U}$, i.e.
\als%{U1b}
{
	\mathbf{U}_1=\mathscr{P}\mathbf{U},
	}
we can then differentiate $\tilde{B}^0$ with respect to $\mathbf{U}$ in the direction $\mathbf{W}$ to get
\begin{align}\label{E:A10b}
  & D_{\mathbf{U}} \tilde{B}^0\cdot \mathbf{W}%=D_{(\mathbf{U}_1,\mathbf{U}_2, \phi)^T} \tilde{B}^0\cdot \mathbf{W}
  =D_{\mathbf{U} } \hat{B}^0(\epsilon,t\mathbf{U}_1, \mathbf{P}_1^\perp\mathbf{U}_1)\cdot \mathbf{W}
  =\bigl(D_2\hat{B}^0 D_\mathbf{U}(t\mathbf{U}_1)+D_3\hat{B}^0 D_\mathbf{U}( \mathbf{P}_1^\perp\mathbf{U}_1)\bigr)\cdot \mathbf{W} \nnb \\
  & \hspace{3cm}=\bigl(t D_2\hat{B}^0 D_\mathbf{U}(\mathscr{P}\mathbf{U})+D_3\hat{B}^0 D_\mathbf{U}( \mathbf{P}_1^\perp\mathscr{P}\mathbf{U})\bigr)\cdot\mathbf{W}
  \nnb \\
  & \hspace{3cm} = t D_2\hat{B}^0 \mathscr{P} \mathbf{W} +D_3\hat{B}^0 ( \mathbf{P}_1^\perp\mathscr{P}) \mathbf{W} =
  t D_2\hat{B}^0 \mathscr{P}\mathbf{W},
%  =\bigl(D_2\hat{B}^0\cdot t\mathds{1}+D_3\hat{B}^0\cdot \mathbf{P}_1^\perp\bigr)\cdot \mathbf{W}
\end{align}
where in the above calculations, we employed the identities
\gas{
	\mathscr{P}\mathbf{W}=\p{ \Pbb_2 & 0 & 0 & 0 & 0\\
		0 &   \breve{\Pbb}_2 & 0 & 0 & 0	\\
		0 & 0 &   \breve{\Pbb}_2 &  0 & 0  } \p{\mathbf{Y}  & 0 & 0 & 0 & 0\\
		0 & -2\mathds{1}   & 0 & 0 & 0	\\
		0 & 0 & -2\mathds{1}  &  0 & 0 \\
		0 & 0 & 0 & (B^0)^{-1}\mathfrak{B}\hat{\Pbb}_2 & 0 \\
		0 & 0 & 0 & 0 & 0
	}\mathbf{U}=\p{\mathbf{Y} & 0 & 0 \\ 0 & -2\mathds{1} & 0 \\0 & 0 & -2\mathds{1}} \mathbf{P}_1\mathbf{U}_1\\
	\intertext{and}
	\mathbf{P}_1^\perp \mathscr{P} \mathbf{W}=0.
	}
By \eqref{E:EULERB0}, it is not difficult to see that
\begin{align}\label{E:A10c}
  (\hat{\mathbb{P}}_2)^\perp[D_{\mathbf{U}}B^0
\mathbf{W}](\hat{\mathbb{P}}_2)^\perp=\begin{pmatrix}
  D_{\mathbf{U}}1 \cdot
\mathbf{W} & 0\\
0 & 0
\end{pmatrix} = 0,
\end{align}
which in turn, implies via \eqref{E:PBPVER}, \eqref{E:A10b} and \eqref{E:A10c} that
\als{
	&\mathbf{P}^\perp[D_\mathbf{U}\mathbf{B}^0\cdot (\mathbf{B}^0)^{-1}\mathbf{B}\mathbf{P}\mathbf{U} ]\mathbf{P}^\perp \nnb \\
	& \hspace{2cm} = t \diag{\bigl(  \Pbb_2^\perp D_2\hat{B}^0 \mathscr{P}\mathbf{W}\Pbb_2^\perp,  \breve{\Pbb}_2^\perp D_2\hat{B}^0 \mathscr{P} \mathbf{W}\breve{\Pbb}_2^\perp, \breve{\Pbb}_2^\perp D_2\hat{B}^0 \mathscr{P} \mathbf{W}\breve{\Pbb}_2^\perp,0,0\bigr)}.
	}
From this it is then clear that $\mathbf{P}^\perp[D_\mathbf{U}\mathbf{B}^0\cdot (\mathbf{B}^0)^{-1}\mathbf{B}\mathbf{P}\mathbf{U} ]\mathbf{P}^\perp $ satisfies Assumption  \ref{ASS1}.\eqref{A:PDECOMPOSABLE}.

\subsection{Limit equations\label{proof:limiteq}}
Setting
\begin{equation} \label{Uringdef}
\mathring{\mathbf{U}}=(\mathring{u}^{0\mu}_0, \mathring{w}^{0\mu}_k, \mathring{u}^{0\mu}, \mathring{u}^{ij}_0, \mathring{u}^{ij}_k, \mathring{u}^{ij}, \mathring{u}_0, \mathring{u}_k, \mathring{u}, \delta\mathring{\zeta}, \mathring{z}_i, \mathring{\phi})^\textrm{T},
\end{equation}
the limit equation, see \S \ref{S:MODELerr}, associated to \eqref{E:REALEQ}  on the spacetime region $(T_2,1]\times \Tbb^3$,
$0<T_2 < 1$, is given by
 \begin{align}
 \mathring{\mathbf{B}}^0\partial_t \mathring{\mathbf{U}}+\mathring{\mathbf{B}}^i\partial_i \mathring{\mathbf{U}} +\mathbf{C}^i\partial_i
 \mathbf{V}= & \frac{1}{t}
 \mathring{\mathbf{B}} \mathbf{P}
 \mathring{\mathbf{U}} + \mathring{\mathbf{H}} + \mathring{\mathbf{F}} \label{E:REALLIMITINGEQUATIONa} &&
\text{in $(T_2,1]\times \Tbb^3$},\\
 \mathbf{C}^i\partial_i\mathring{\mathbf{U}} = & 0   \label{E:REALLIMITINGEQUATIONb} &&
\text{in $(T_2,1]\times \Tbb^3$},
\end{align}
where
 \al{BLIM}{
 	\mathbf{\mathring{B}}{}^\mu(t,\mathbf{\mathring{U}}):=\lim_{\epsilon\searrow 0}\mathbf{B}^{\mu}(\epsilon,t,\mathbf{\mathring{U}}), \quad \mathbf{\mathring{B}}(t,\mathbf{\mathring{U}}):=\lim_{\epsilon\searrow 0}\mathbf{B}(\epsilon,t,\mathbf{\mathring{U}}), \quad  \mathbf{\mathring{H}}(t,\mathbf{\mathring{U}}):=\lim_{\epsilon\searrow 0}\mathbf{H}(\epsilon,t,\mathbf{\mathring{U}}),
 	}
 and
 \gat{
	%\mathring{\mathbf{B}}^i=\diag{\left(0,0,0, \sqrt{\frac{3}{\Lambda}}\p{\mathring{z}^k & \mathring{E}^{-2} \delta^{km} \\ \mathring{E}^{-2} \delta^{kl}  & K^{-1}\mathring{E}^{-2}\delta^{lm}\mathring{z}^k},0\right)}\label{E:BH1}\\
	\mathbf{\mathring{F}} := \biggl(-\frac{\mathring{\Omega}}{t}\mathcal{D}^{0\mu j}\del{j}\mathring{\Phi} , \frac{3}{2t}\delta^\mu_0\mathring{E}^{-2}\delta^{kl}\del{l}\mathring{\Phi}-\mathring{E}^{-2}\delta^{kl}\delta^\mu_0\del{0}\del{l}\mathring{\Phi}, 0,  -\frac{\mathring{\Omega}}{t}\mathcal{\tilde{D}}^{ijr}\del{r}\mathring{\Phi},  \nnb \\
	 0, 0,   -\frac{\mathring{\Omega}}{t}\mathcal{D}^{j}\del{j}\mathring{\Phi},0,0,0,-K^{-1}\frac{1}{2}\left(\frac{3}{\Lambda}\right)^{\frac{3}{2}}\mathring{E}^{-2}\delta^{lk}\del{k}\mathring{\Phi}, 0 \biggr)^\textrm{T}.\label{E:BH2}
}
In $\mathbf{\mathring{F}}$,  the coefficients $\mathcal{D}^{0\mu j}$ and $\tilde{\mathcal{D}}{}^{ijr}$ are as defined by \eqref{E:D1} and \eqref{E:D4},
$\mathring{\Phi}$ is the Newtonian potential, see \eqref{CPeqn3}, and $\mathring{E}$ and $\mathring{\Omega}$ are defined by
\eqref{Eringform} and \eqref{Oringdef}, respectively.

We then observe that under the change of time coordinate \eqref{E:TIMECHANGE} and the substitutions
\begin{gather}
\mrw(\hat{t},x)=\mathbf{\mathring{U}}(-\hat{t},x), \quad \mathring{A}_1^0( \hat{t},w)=\mathbf{\mathring{B}}{}^0(-\hat{t},\mathbf{\mathring{U}}), \quad \mathring{A}_1^i( \hat{t},w)=-\mathbf{\mathring{B}}{}^i(-\hat{t},\mathbf{\mathring{U}}),
\quad  \mathfrak{\mathring{A}}_1( \hat{t},w) =\mathbf{\mathring{B}}(-\hat{t},\mathbf{\mathring{U}}), \quad  C_1^i=-\mathbf{C}^i,
\label{singdefc} \\
 v(\hat{t},x)=\mathbf{V}(-\hat{t},x), \quad \mathbb{P}_1=\mathbf{P}, \quad \mathring{H}_1( \hat{t},w) =-\mathring{\mathbf{H}}(-\hat{t},\mathbf{\mathring{U}}) \quad \text{and} \quad \mathring{F}_1(\hat{t},x) =-\mathbf{\mathring{F}}(-\hat{t},x), \label{singdefd}
\end{gather}
the limit equation \eqref{E:REALLIMITINGEQUATIONa}-\eqref{E:REALLIMITINGEQUATIONb} transforms into
\begin{align*}
  \mathring{A}_1^0\partial_{\hat{t}} \mathring{w}+\mathring{A}_1^i\partial_i \mathring{w}&=\frac{1}{\hat{t}}\mathring{\mathfrak{A}}_1\mathbb{P}_1\mathring{w}-C_1^i\partial_i v+\mathring{H}_1+\mathring{F}_1
&& \mbox{in} \quad[-1, -T_2)\times\mathbb{T}^3, \\
  C_1^i\partial_i\mathring{w}&=0 && \mbox{in} \quad[-1, -T_2)\times\mathbb{T}^3,
\end{align*}
which is of the form analyzed in \S \ref{S:MODELerr}, see \eqref{E:LIMITINGEQa}-\eqref{E:LIMITINGEQb} and \eqref{E:HRIN}.
It is also not difficult to verify the matrices $\mathring{A}_1^i$ and the source term $\mathring{H}_1$
satisfy the Assumptions \ref{ASS3}.(2) from \S \ref{S:MODELerr}.

\subsection{Local existence and continuation\label{proof:loccont}}
For fixed $\epsilon \in (0,\epsilon_0)$,  we know from Proposition \ref{rcEEexist} that for $T_1 \in (0,1)$ chosen close enough to $1$ there exists a unique solution
\begin{equation*}
\mathbf{U} \in \bigcap_{\ell=0}^1 C^\ell\bigl( (T_1,1],H^{s-\ell}(\Tbb^3,\mathbb{V})\bigr)
\end{equation*}
to \eqref{E:REALEQ} satisfying the initial condition
\begin{equation*}
 \mathbf{U}|_{t=1}=
\bigl(u^{0\mu}_0|_{t=1}, w^{0\mu}_k|_{t=1}, u^{0\mu}|_{t=1}, u^{ij}_0|_{t=1}, u^{ij}_k|_{t=1}, u^{ij}|_{t=1}, u_0|_{t=1}, u_k|_{t=1}, u|_{t=1},\delta\zeta|_{t=1}, z _i|_{t=1},\phi|_{t=1}\bigr)^{\textrm{T}},
\end{equation*}
where the initial data, $u^{0\mu}_0|_{t=1}$, $w^{0\mu}_k|_{t=1}$, $\ldots$, is determined from Lemma \ref{L:INITIALTRANSFER}.
Moreover, this solution can be continued beyond $T_1$ provided that
\begin{equation*}
\sup_{ t\in (T_1,1]} \norm{\mathbf{U}(t)}_{H^s} < \infty.
\end{equation*}

Next, by  Proposition \ref{PEexist},  there exists, for some $T_2 \in (0,1]$, a unique solution
\begin{equation} \label{PEproofsol}
(\mathring{\zeta},\mathring{z}^i,\mathring{\Phi})\in \bigcap_{\ell=0}^1 C^{\ell}((T_1,T_0],H^{s-\ell}(\mathbb{T}^3))
\times \bigcap_{\ell=0}^1  C^{\ell}((T_1,T_0],H^{s-\ell}(\mathbb{T}^3,\mathbb{R}^3)) \times \bigcap_{\ell=0}^1  C^{\ell}((T_1,T_0],H^{s+2-\ell}(\mathbb{T}^3)),
\end{equation}
to the conformal cosmological Poisson-Euler equations, given by \eqref{CPeqn1}-\eqref{CPeqn3}, satisfying
the initial condition
\begin{equation} \label{PEproofsolid}
(\mathring{\zeta},\mathring{z}_i)|_{t=1} =  \biggl(\ln\bigl(\rho_H(1)+\breve{\rho}_0\bigr),\frac{
\breve{\nu}^i\delta_{ij}}{\rho_H(1)+\breve{\rho}_0}\biggr).
\end{equation}
Setting
\begin{align}\label{E:V}
 \mathbf{V}=\left(V^{0 \mu}_0, V^{0\mu}_k, V^{0 \mu}, 0, V^{ij}_k, 0, 0, V_k, 0,0,0,0\right),
 \end{align}
 where
 \begin{gather}
 V^{0 \mu}_0=- \mathring{E}^2\frac{3}{2t}\delta^\mu_0 \mathring{\Phi}
 +\mathring{E}^2\delta^\mu_0 \partial_t\mathring{\Phi}
 =- \frac{1}{2t} \mathring{E}^2\delta^\mu_0\mathring{\Phi}+\delta^\mu_0t \mathring{E}^2\partial_t\biggl(
 \frac{\mathring{\Phi}}{t}\biggr),\label{E:V1}\\
V^{0\mu}_k=\frac{\mathring{\Omega}}{t} \mathcal{D}^{0\mu j} \Delta^{-1} \partial_k \partial_j\mathring{\Phi}+2 \mathring{E}^2\sqrt{\frac{\Lambda}{3}}\frac{1}{t^2}\Delta^{-1}\partial_k (\mathring{\rho}\mathring{z}^j)\delta_j^\mu,  \label{E:V2}\\
 V^{0 \mu}
 =\frac{1}{2}\delta^\mu_0 \mathring{E}^2 \frac{\mathring{\Phi}}{t}+ \delta^\mu_0 \frac{\Lambda}{3 t^3} \mathring{E}^4\mathring{\Omega} \Delta^{-1}\delta \mathring{\rho},   \label{E:V3}\\
 V^{ij}_k=\frac{\mathring{\Omega}}{t}   \mathcal{\tilde{D}}^{ij r}    \Delta^{-1} \partial_k \partial_r\mathring{\Phi}, \label{E:V4}
 \intertext{and}
 V_k=\frac{\mathring{\Omega}}{t}\mathcal{D}^j\Delta^{-1} \partial_k \partial_j\mathring{\Phi}, \label{E:V5}
 \end{gather}
it follows from Corollary \ref{PEcor} and \eqref{PEproofsol} that $\mathbf{V}$ is well-defined and lies in the
space
\begin{equation*}
\mathbf{V} \in \bigcap_{\ell=0}^1 C^\ell\bigl( (T_2,1],H^{s-\ell}(\Tbb^3,\mathbb{V})\bigr).
\end{equation*}
Defining
 \begin{align}\label{E:URINGVALUE}
 \mathring{\mathbf{U}}=(0,0,0,0,0,0,0,0,0,\delta\mathring{\zeta},\mathring{z}_i,0),
 \end{align}
where we recall, see \eqref{deltazetaringdef}, \eqref{deltazetaringH} and Theorem \ref{T:MAINTHEOREM}.(ii), that
 \begin{align}\label{LIMITRHOZ}
\delta\mathring{\zeta} = \mathring{\zeta}-\mathring{\zeta}_H \AND  \mathring{z}^i=\mathring{E}^{-2}\delta^{ij} \mathring{z}_j,
 \end{align}
we see from Remark \ref{PEcorrem},  \eqref{Uringdef} and \eqref{PEproofsol}-\eqref{PEproofsolid} that  $\mathring{\mathbf{U}}$ lies in the space
\begin{equation*}
  \mathring{\mathbf{U}} \in \bigcap_{\ell=0}^1 C^\ell\bigl( (T_2,1],H^{s-\ell}(\Tbb^3,\mathbb{V})\bigr)
\end{equation*}
and satisfies
\begin{align*}%\label{E:URINGa}
 \mathring{\mathbf{U}}|_{t=1} = \biggl(0,0,0,0,0,0,0,0,0,\ln\biggl(1+\frac{\breve{\rho}_0}{\rho_H(1)}\biggr),\frac{
\breve{\nu}^i\delta_{ij}}{\rho_H(1)+\breve{\rho}_0},0\biggr)^{\textrm{T}}.
\end{align*}
It can be verified by a direct calculation that the pair $(\mathbf{V}, \mathring{\mathbf{U}})$ determines
a solution of the limit equation  \eqref{E:REALLIMITINGEQUATIONa}-\eqref{E:REALLIMITINGEQUATIONb}. Moreover,
by Proposition  \ref{PEexist}, it is clear that this solution can be continued past $T_2$ provided that
\begin{equation*}
\sup_{ t\in (T_2,1]} \norm{\mathring{\textbf{U}}(t)}_{H^s} < \infty.
\end{equation*}

\subsection{Global existence and error estimates}
For the last step of the proof, we will use the a priori estimates from Theorem \ref{T:MAINMODELTHEOREM}
to show that the solutions $\mathbf{U}$ and $(\mathbf{V},\mathring{\mathbf{U}})$ to the reduced conformal Einstein--Euler equations and the corresponding limit equation, respectively, can be continued all the way to $t=0$, i.e. $T_1=T_2=0$, with uniform bounds and error estimates. In order to apply Theorem \ref{T:MAINMODELTHEOREM}, we need to verify that the estimates
\eqref{E:VASS}-\eqref{HFLip} hold for the solutions  $\mathbf{U}$ and $(\mathbf{V},\mathring{\mathbf{U}})$. We begin by observing
the equation
\begin{equation} \label{PEdeltarpho}
\del{t}\delta\mathring{\rho} =- \sqrt{\frac{3}{\Lambda}}\del{j} \bigl(\mathring{\rho}\mathring{z}{}^j\bigr)
+ \frac{3(1-\mathring{\Omega})}{t}\delta \mathring{\rho}
\end{equation}
holds in $(T_2,1]\times \Tbb^3$
by \eqref{E:COSEULERPOISSONEQ.a}, \eqref{PEcor0} and the equivalence of the two formulations
 \eqref{E:COSEULERPOISSONEQ.a}- \eqref{E:COSEULERPOISSONEQ.c} and \eqref{CPeqn1}-\eqref{CPeqn3} of
the conformal Poisson-Euler equations. From this equation, \eqref{Oringdef} and the calculus inequalities from Appendix  \ref{A:INEQUALITIES},
we obtain the estimate
\begin{equation}\label{E:PTRHOTEST}
  \left\|\partial_t\left(\frac{\delta\mathring{\rho}}{t^3}\right)\right\|_{H^{s-1}} \leq
C\bigl(\|\delta\mathring{\zeta}\|_{\Li((t,1],H^s)}, \|\mathring{z}_i\|_{\Li((t,1],H^s)}\bigr) (\|\delta\mathring{\zeta}(t)\|_{H^s}+\|\mathring{z}_i(t)\|_{H^s}),  \quad T_2 < t \leq 1
\end{equation}
Recalling that we can write the Newtonian potential as
\begin{align}\label{E:PHIRING}
  \mathring{\Phi}=\frac{\Lambda}{3}\frac{1}{t^2}\mathring{E}^2\Delta^{-1}\delta\mathring{\rho}= \frac{\Lambda}{3} t \mathring{E}^2 e^{\mathring{\zeta}_H} \Delta^{-1} (e^{\delta\mathring{\zeta}}-1) \quad \text{in $(T_2,1]\times \Tbb^3$}
\end{align}
by \eqref{zetaHringform}, \eqref{PEexist1} and Corollary \ref{PEcor} we see, using the calculus inequalities from Appendix \ref{A:INEQUALITIES}
and invertibility of the Laplacian $\Delta$ $:$ $\bar{H}^{k+1}(\Tbb^3)$ $\rightarrow$ $\bar{H}^{k-1}(\Tbb^3)$, $k\in \mathbb{Z}_{\geq 1}$,  that we can estimate $\frac{1}{t}\mathring{\Phi}$ by
\al{PHIRINEST}{
	\left\|\frac{1}{t}\mathring{\Phi}(t)\right\|_{H^{s+1}} \leq
C\bigl(\|\delta\mathring{\zeta}\|_{\Li((t,1],H^{s-1})}\bigr)
\|\delta\mathring{\zeta}(t)\|_{\Hs}, \quad T_2 < t \leq 1.
	}
Dividing \eqref{E:PHIRING} by $t$ and then differentiating with respect to $t$, we find
using \eqref{PEdeltarpho} that
\begin{equation}\label{E:PTPHI}
  \partial_t \left( \frac{\mathring{\Phi}}{t}\right)+\frac{1}{3}\frac{\Lambda}{t^4}\mathring{E}^2\mathring{\Omega} \Delta^{-1}\delta\mathring{\rho} +\sqrt{\frac{\Lambda}{3}} \frac{1}{t^3}\mathring{E}^2  \partial_k\Delta^{-1}
  \left( \mathring{\rho} \mathring{z}^k\right)=0,
\end{equation}
which, by \eqref{E:PHIRING}, is also equivalent to
\begin{equation}  \label{E:PTPHIRING}
  \partial_t \mathring{\Phi}=\frac{\Lambda}{3}\frac{1}{t^3}\mathring{E}^2(1-\mathring{\Omega})\Delta^{-1} \delta\mathring{\rho}-\sqrt{\frac{\Lambda}{3}}\frac{1}{t^2} \mathring{E}^2 \partial_k \Delta^{-1} \left(\mathring{\rho}\mathring{z}^k\right).
\end{equation}
From \eqref{E:PTPHI} and  \eqref{E:PTPHIRING}, we then obtain, with the help
of the calculus inequalities and the invertibility of the Laplacian, the estimate
\al{DTPHITEST}{
	\|\del{t}\mathring{\Phi}\|_{H^{s+1}}+\left\|t\del{t}\left(\frac{\mathring{\Phi}}{t}\right)\right\|_{H^{s+1}} \leq
	C(\|\delta\mathring{\zeta}\|_{\Li((t,1],H^s)}, \|\mathring{z}_i\|_{\Li((t,1]),H^s)}) (\|\delta\mathring{\zeta}(t)\|_{H^s}+\|\mathring{z}_i(t)\|_{H^s}), \quad T_2 < t \leq 1.
	}
Continuing on, we differentiate \eqref{E:PTPHIRING} with respect to $t$ to get
\begin{equation} \label{E:PPTPHIRING}
  \partial^2_t\mathring{\Phi}=  \frac{\Lambda}{3}\frac{1}{t^4}\mathring{\Omega}\left(\frac{5}{2} \mathring{\Omega}-4\right) \mathring{E}^2 \Delta^{-1}\delta\mathring{\rho} - \sqrt{\frac{\Lambda}{3}}\frac{1}{t^2} \mathring{E}^2 \Delta^{-1}\partial_k \partial_t \left(\mathring{\rho}\mathring{z}^k \right)+\sqrt{\frac{\Lambda}{3}}\frac{1}{t^3} \mathring{E}^2 (1-\mathring{\Omega})\Delta^{-1}\partial_k \left(\mathring{\rho}\mathring{z}^k\right),
\end{equation}
where in deriving this we have used the fact that $\mathring{\Omega}$ satisfies  \eqref{E:RICCATI1} with $\epsilon=0$ and
that $ \Delta^{-1}\delta\mathring{\rho} $ is well defined by Corollary \ref{PEcor}.  Adding the conformal cosmological Poisson-Euler equations \eqref{E:COSEULERPOISSONEQ.a}-\eqref{E:COSEULERPOISSONEQ.b} together, we obtain the following
equation for $\mathring{\rho}\mathring{z}^j$:
    \begin{align*}%\label{E:PTRHOZ}
 	  \partial_t\left(\mathring{\rho}\mathring{z}^j\right)+\sqrt{\frac{3}{\Lambda}}K \partial^j \mathring{\rho}+\sqrt{\frac{3}{\Lambda}}\partial_i\left( \mathring{\rho}\mathring{z}^i\mathring{z}^j\right)=\frac{4-3\mathring{\Omega}}{t}\mathring{\rho} \mathring{z}^j -\frac{1}{2}\left(\frac{3}{\Lambda}\right)^{\frac{3}{2}}\mathring{\rho} \partial^j \mathring{\Phi}.
 	\end{align*}
Substituting this into \eqref{E:PPTPHIRING} yields the estimate
\al{DT2PHIEST}{
	\|\partial^2_t\mathring{\Phi}\|_{H^s} \leq 	C(\|\delta\mathring{\zeta}\|_{\Li((t,1],H^s)}, \|\mathring{z}_i\|_{\Li((t,1]),H^s)}) (\|\delta\mathring{\zeta}(t)\|_{H^s}+\|\mathring{z}_i(t)\|_{H^s}), \quad T_2 < t \leq 1,
	}
by \eqref{E:PHIRINEST},  the invertibility of the Laplacian $\Delta$ $:$ $\bar{H}^{k+1}(\Tbb^3)$ $\rightarrow$ $\bar{H}^{k-1}(\Tbb^3)$, $k\in \mathbb{Z}_{\geq 1}$,  and the calculus inequalities from Appendix \ref{A:INEQUALITIES}.
Next, from the definition of $\mathbf{P}$, see \eqref{E:REALEQb},  and \eqref{E:PTPHI}, we compute
\begin{equation} \label{PVt}
\frac{1}{t}\mathbf{P}\mathbf{V}=\left(\frac{1}{2t}(V^{0 \mu}_0+V^{0 \mu}), \frac{1}{t}V^{0\mu}_i, \frac{1}{2t}(V^{0 \mu}_0+V^{0 \mu}), 0, 0, 0, 0, 0, 0,0,0,0\right)^\mathrm{T},
\end{equation}
where the components are given by
\begin{align}
  \frac{1}{2t} (V^{0 \mu}_0+V^{0 \mu})=&\frac{1}{2}\delta^\mu_0 \mathring{E}^2 \partial_t\left(\frac{\mathring{\Phi}}{t}\right)
  +\frac{1}{2t}\delta^\mu_0 \frac{1}{3}\frac{\Lambda}{t^3}\mathring{E}^4\mathring{\Omega} \Delta^{-1}\delta \mathring{\rho} = -\delta^\mu_0 \frac{1}{2} \sqrt{\frac{\Lambda}{3}} \frac{1}{t^3}\mathring{E}^4  \partial_k\Delta^{-1}
  \left( \mathring{\rho} \mathring{z}^k\right), \label{E:TIV1}\\
\frac{1}{t}V_k^{0 \mu}=& \frac{\mathring{\Omega}}{t^2} \mathcal{D}^{0\mu j} \Delta^{-1} \partial_k \partial_j\mathring{\Phi}+2E^2\sqrt{\frac{\Lambda}{3}}\frac{1}{t^3}\Delta^{-1}\partial_k (\mathring{\rho}\mathring{z}^j)\delta_j^\mu \label{E:TIV2}.
\end{align}
Routine calculations also
show that the components of $\del{t} \mathbf{V}$ are given by
\begin{align}
  \partial_t V^{0\mu}_0= &  \mathring{E}^2\delta^\mu_0\left(2\mathring{\Omega}-\frac{3}{2}\right)\partial_t \left(\frac{\mathring{\Phi}}{t}\right)+ \mathring{E}^2 \delta^\mu_0\partial^2_t \mathring{\Phi} - \frac{\mathring{\Omega}}{t^2}  \mathring{E}^2 \delta^\mu_0\mathring{\Phi}, \label{E:DTV1}  \\
  \partial_t V^{0\mu}= & \delta^\mu_0 \mathring{E}^2\frac{\mathring{\Omega}}{t}\frac{\mathring{\Phi}}{t}+\frac{1}{2} \delta^\mu_0 \mathring{E}^2 \partial_t\left(\frac{\mathring{\Phi}}{t}\right)+\delta^\mu_0\frac{\Lambda}{3}\mathring{E}^4 \left(4\frac{\mathring{\Omega}^2}{t}+ \partial_t\mathring{\Omega}\right) \Delta^{-1}\frac{\delta\mathring{\rho}}{t^3}+ \delta^\mu_0 \frac{\Lambda}{3} \mathring{E}^4 \mathring{\Omega}\Delta^{-1}\partial_t\left(\frac{\delta\mathring{\rho}}{t^3}\right), \label{E:DTV2} \\
  \partial_t V^{ij}_k=& \partial_t\left(\frac{\mathring{\Omega}}{t}\right)\tilde{\mathcal{D}}^{ijr} \Delta^{-1}\partial_k\partial_r \mathring{\Phi}+ \frac{\mathring{\Omega}}{t} (\partial_t \tilde{\mathcal{D}}^{ijr}) \Delta^{-1}\partial_k\partial_r \mathring{\Phi}+ \frac{\mathring{\Omega}}{t}\tilde{\mathcal{D}}^{ijr} \Delta^{-1}\partial_k\partial_r \left(\partial_t \mathring{\Phi}\right), \label{E:DTV3}\\
  \partial_t V_k=& \partial_t\left(\frac{\mathring{\Omega}}{t}\right) \mathcal{D}^{j} \Delta^{-1}\partial_k\partial_j \mathring{\Phi}+ \frac{\mathring{\Omega}}{t}(\partial_t \mathcal{D}^{j}) \Delta^{-1}\partial_k\partial_j \mathring{\Phi}+ \frac{\mathring{\Omega}}{t} \mathcal{D}^{j} \Delta^{-1}\partial_k\partial_j \left(\partial_t \mathring{\Phi}\right)\label{E:DTV4}\\
\intertext{and}
  \partial_t V^{0\mu}_k=   & \partial_t\left(\frac{\mathring{\Omega}}{t}\right) \mathcal{D}^{0\mu j} \Delta^{-1}\partial_k\partial_j \mathring{\Phi}+ \frac{\mathring{\Omega}}{t}(\partial_t \mathcal{D}^{0\mu j}) \Delta^{-1}\partial_k\partial_j \mathring{\Phi}+ \frac{\mathring{\Omega}}{t} \mathcal{D}^{0\mu j} \Delta^{-1}\partial_k\partial_j \left(\partial_t \mathring{\Phi}\right) \nnb\\
      &+2 E^2\sqrt{\frac{\Lambda}{3}}\frac{1}{t^2}\delta^\mu_j \left[\frac{2\Omega}{t} \Delta^{-1} \partial_k(\mathring{\rho} \mathring{z}^j)-\frac{2}{t} \Delta^{-1}\partial_k (\mathring{\rho} \mathring{z}^j)+\Delta^{-1}\partial_k \partial_t (\mathring{\rho} \mathring{z}^j)\right] \label{E:DTV5}.
  \end{align}
Recalling the the coefficients $\mathcal{D}^{0\mu\nu}$, $\tilde{\mathcal{D}}{}^{ijk}$ and $\mathcal{D}^j$ are remainder terms as defined in \S  \ref{remainder}, it is then clear that the estimate
\al{NORMVEST}{
	\|\mathbf{V}(t)\|_{H^{s+1}}+\|t^{-1}\mathbf{P}\mathbf{V}(t)\|_{H^{s+1}}+\|\del{t}\mathbf{V}(t)\|_{H^s}\leq 	C(\|\delta\mathring{\zeta}\|_{\Li((t,1],H^s)}, \|\mathring{z}_i\|_{\Li((t,1),H^s)}) (\|\delta\mathring{\zeta}(t)\|_{H^s}+\|\mathring{z}_i(t)\|_{H^s}),
	}
which holds for $ T_2 < t \leq 1$,
follows from the formulas \eqref{Eringform}, \eqref{Oringdef}, \eqref{E:V}-\eqref{E:V5}  and \eqref{PVt}-\eqref{E:DTV5}, the estimates \eqref{E:PTRHOTEST}, \eqref{E:PHIRINEST}, \eqref{E:DTPHITEST}, \eqref{E:DT2PHIEST}, the calculus inequalities and the invertibility of the Laplacian. By similar reasoning, it is also not difficult to
verify that $\mathring{\mathbf{F}}$, defined by \eqref{E:BH2}, satisfies the estimate
\begin{equation}\label{E:TDTFCHE}
	 \|\mathring{\mathbf{F}}(t)\|_{H^s}+\|t\del{t}\mathring{\mathbf{F}}(t)\|_{\Hs}
%\nnb \\ \leq &C(\|\delta\mathring{\zeta}\|_{\Li([T_0,T_1),H^s)}, \|\mathring{z}_i\|_{\Li([T_0,T_1),H^s)})\bigl( \|t^{-1}\mathring{\Phi}\|_{H^{s+1}}+\|
%	\delta \mathring{\zeta}\|_{H^s}+\|\mathring{z}_j\|_{H^s}+\|\del{t}\del{l} \mathring{\Phi}\|_{H^s}+\|\partial_t^2\mathring{\Phi}\|_{H^s} \bigr) \nnb \\
	\leq  C(\|\delta\mathring{\zeta}\|_{\Li((t,1],H^s)}, \|\mathring{z}_i\|_{\Li((t,1],H^s)}) (\|\delta\mathring{\zeta}(t)\|_{H^s}+\|\mathring{z}_i(t)\|_{H^s}), \quad T_2 < t \leq 1.
\end{equation}

From the definition of $\mathbf{F}$, see \eqref{E:REALEQc}, along with \eqref{E:S1}, \eqref{E:S2}, \eqref{E:S3}, \eqref{E:S2a}
and \eqref{Ycaltildef},  the definitions \eqref{E:WPHI} and \eqref{E:U1}-\eqref{E:REALVAR1}, and the calculus inequalities, we see
that $\mathbf{F}$ can be estimated as
\begin{equation} \label{bfFest1}
\|\mathbf{F}(t)\|_{H^s}\leq  C\bigl(\|\mathbf{U}\|_{\Li((t,1],H^s)}, \|\del{k}\Phi\|_{\Li((t,1],H^s)}\bigr)(\|\mathbf{U}(t)\|_{H^s}+\|\del{l}\del{k}\Phi(t)\|_{H^s}+\|\del{t}\del{k}\Phi(t)\|_{H^s}), \quad T_1 < t < 1.\\
%\leq & C(\|\mathbf{U}\|_{\Li([-1,-T_1),H^s)}, \|\delta\zeta\|_{\Li([-1,-T_1),\Hs)})(\|\mathbf{U}\|_{H^s}+\|\delta\zeta\|_{\Hs})\\
%\leq & C(\|\mathbf{U}\|_{\Li([-1,t),H^s)})\|\mathbf{U}\|_{H^s}.
\end{equation}
Appealing again to the invertibility of the map $\Delta$ $:$ $\bar{H}^{k+1}(\Tbb^3)$ $\rightarrow$ $\bar{H}^{k-1}(\Tbb^3)$, $k\in \mathbb{Z}_{\geq 1}$, it follows from \eqref{E:DEFOFPHI} and the calculus inequalities that we can estimate the spatial derivatives  of $\Phi$ as follows:
\begin{equation} \label{Phiestimate}
 \|\del{k}\Phi(t)\|_{H^s} + \|\del{l}\del{k}\Phi(t)\|_{H^s} \leq  C\bigl(\|\mathbf{U}\|_{\Li((t,1],H^s)}\bigr)\|\mathbf{U}(t)\|_{H^s}
\end{equation}
for $T_1 < t < 1$. Using \eqref{E:PTZETAH} and \eqref{Phicont}, we see that
$\del{t}\del{k}\Phi$ satisfies
\begin{equation*}
\del{t}\del{k}\Phi=\frac{\Lambda}{3} E^2 e^{\zeta_H}(1-\Omega) \del{k} \Delta^{-1}\Pi e^{\delta\zeta}+ \frac{\Lambda}{3}E^2t e^{\zeta_H}\del{k}\Delta^{-1}\biggl(e^{\delta\zeta}\del{t}\delta \zeta\biggr).
\end{equation*}
Replacing $\del{t}\delta\zeta$ in the above equation with the right hand side of \eqref{dtzeta}, we see, with the
help of the calculus properties and the invertibility of the Laplacian that
 $\del{t}\del{k}\Phi$ can be estimated by
\al{DTDSPHI}{
\|\del{t}\del{k}\Phi\|_{H^s}\leq C(\|\delta\zeta\|_{\Li((t,1],H^s)})(\|\del{t}\delta\zeta\|_{\Hs}+\|\delta\zeta\|_{\Hs})%\leq C(\|\delta\zeta\|_{\Li([-1,-T_1),H^s)})(\|\del{t}\delta\zeta\|_{\Hs}+\|\delta\zeta\|_{\Hs})
\leq C(\|\mathbf{U}\|_{\Li((t,1],H^s)})\|\mathbf{U}\|_{H^s}
}
for $T_1 < t < 1$. Combining the estimates
\eqref{bfFest1}-\eqref{E:DTDSPHI} gives
\begin{equation} \label{bfFest2}
\|\mathbf{F}(t)\|_{H^s}\leq  C\bigl(\|\mathbf{U}\|_{\Li((t,1],H^s)}\bigr)\|\mathbf{U}(t)\|_{H^s}, \quad T_1 < t < 1.\\
\end{equation}
Together, \eqref{E:NORMVEST}, \eqref{E:TDTFCHE} and \eqref{bfFest2} show that source terms
$\{F_1,\mathring{F}_1,v\}$, as defined by \eqref{singdefB} and \eqref{singdefd},  satisfy the
estimates \eqref{E:VASS}-\eqref{F1est}
from Theorem \ref{T:MAINMODELTHEOREM} for times $-1 \leq  \hat{t} < -T_3$, where
\begin{equation*}
T_3 = \max\{T_1,T_2\}.
\end{equation*}

This leaves us to verify the Lipchitz estimates \eqref{AiLip}-\eqref{HFLip}. We begin by noticing,
with the help of \eqref{E:EINBk}, \eqref{E:BkREMAINDER} and \eqref{E:URINGVALUE}, that
\begin{align*}
\tilde{B}{}^i (\epsilon,t,\mathring{\mathbf{U}})&= 0,\\
B^i(\epsilon,t,\mathring{\mathbf{U}}) &=\sqrt{\frac{3}{\Lambda}}\p{\mathring{z}^i & E^{-2} \delta^{im} \\ E^{-2}\delta^{il} & K^{-1} E^{-2} \delta^{lm}\mathring{z}^i
}+\epsilon^2\mathcal{S}^i(\epsilon,t,\mathring{\mathbf{U}})
\intertext{and}
B^i(0,t,\mathring{\mathbf{U}})&=\sqrt{\frac{3}{\Lambda}}\p{\mathring{z}^i & \mathring{E}^{-2} \delta^{im} \\ \mathring{E}^{-2}\delta^{il} & K^{-1} \mathring{E}^{-2} \delta^{lm}\mathring{z}^i }.
\end{align*}
%where $\mathcal{S}^i(\epsilon,t,\mathring{\mathbf{U}})$ is the remainder term from \eqref{E:BkREMAINDER}.
%\co{This $\mathcal{S}^i(\epsilon,t,\mathring{\mathbf{U}})$ may be $0$ (Need more detailed structure in Section 2), but it is not influence our proof here, so I leave it there. }
From the above expressions, \eqref{E:EOEXP} and the calculus inequalities, we then obtain the estimate
\als%{BLIP}
{
\|\mathbf{B}^i(\epsilon,t,\mathring{\mathbf{U}}) -\mathring{\mathbf{B}}^i(t,\mathring{\mathbf{U}})\|_{\Hs}
%\lesssim & \|E^{-2}-\mathring{E}^{-2}\|_{\Hs}+\|(E^{-2}-\mathring{E}^{-2}) \mathring{z}^i\|_{\Hs}+\epsilon^2\|\mathring{\mathbf{U}}\|_{\Hs} \nnb \\ &\hspace{6.5cm}
\leq \epsilon C(\|\mathring{{\mathbf{U}}}\|_{\Li((t,1],H^s)}), \quad T_3 < t \leq 1.
}
Next, using \eqref{E:fij}, \eqref{E:f}, \eqref{E:G1}, \eqref{E:f2}, \eqref{E:G2}, \eqref{E:G3} and   \eqref{E:URINGVALUE}, we compute the components of $\mathbf{H}(\epsilon,t,\mathring{\mathbf{U}})$, see \eqref{E:REALEQc},  as follows:
\begin{gather*}
\tilde{G}_1(\epsilon,t,\mathring{\mathbf{U}})=\biggl(
-2(1+\epsilon^2 K)E^2 t^{1+3\epsilon^2 K}e^{(1+\epsilon^2 K)(\zeta_H+\delta\mathring{\zeta})}  \sqrt{\frac{\Lambda}{3}}   \mathring{z}{}^k\delta^\mu_k + \epsilon\mathcal{S}^\mu(\epsilon, t, \mathring{\mathbf{U}})  ,0,0\biggr)^\mathrm{T}, \\
\tilde{G}_2(\epsilon,t,\mathring{\mathbf{U}})= \bigl( \epsilon\mathcal{S}^{ij}(\epsilon, t, \mathring{\mathbf{U}})
,0,0\bigr)^\mathrm{T},
\quad
\tilde{G}_3(\epsilon,t,\mathring{\mathbf{U}})= \bigl(\epsilon \mathcal{S} (\epsilon, t, \mathring{\mathbf{U}})
 ,0,0\bigr)^\mathrm{T}, \\
G(\epsilon,t,\mathring{\mathbf{U}})=(0,0)^\mathrm{T} \quad \text{and} \quad \acute{G}(\epsilon,t,\mathring{\mathbf{U}})=0,
\end{gather*}
where $\mathcal{S}^\mu$, $\mathcal{S}^{ij}$ and $\mathcal{S}$ all vanish for $\mathring{\mathbf{U}}=0$.
It follows immediately from these expressions and the definitions \eqref{E:REALEQc} and \eqref{E:BLIM} that
\als%{HRING2}
{
	\mathring{\mathbf{H}}(t,\mathring{\mathbf{U}})=\biggl(-2\mathring{E}^2 t e^{  \mathring{\zeta}_H+\delta\mathring{\zeta}}  \sqrt{\frac{\Lambda}{3}}   \mathring{z}{}^k\delta^\mu_k,0,0,0,0,0,0,0,0,0,0,0\biggr)^T,
	}
and,  with the help of the calculus inequalities and \eqref{E:CHI1}-\eqref{E:EOEXP}, that
\al{HSUBH}{
	\|\mathbf{H}(\epsilon,t,\mathring{\mathbf{U}})-\mathring{\mathbf{H}}(t,\mathring{\mathbf{U}})\|_{\Hs} \leq \epsilon C(\|\mathring{\mathbf{U}}\|_{\Li((t,1],H^s)}) \|\mathring{\mathbf{U}}\|_{\Hs}, \quad T_3 < t \leq 1.
}

To proceed, we define
\al{ZREAL}{
	\mathbf{Z}=\frac{1}{\epsilon}(\mathbf{U}-\mathring{\mathbf{U}}-\epsilon \mathbf{V}),
	}
and set
\begin{equation} \label{E:zreal}
z(\hat{t},x) = \mathbf{Z}(-\hat{t},x).
\end{equation}
In view of the definitions \eqref{E:REALEQc} and \eqref{E:BH2}, we see that the estimate
\al{FSUBF}{
	\|\mathbf{F}(\epsilon,t,\cdot)-\mathring{\mathbf{F}}(t,\cdot)\|_{\Hs}
	\leq & C\bigl(\|\mathbf{U}\|_{\Li((t,1],H^s)}\bigr)\Bigl(\epsilon\|\mathbf{U}\|_{\Hs}+\|\del{k}\Phi\|_{\Hs}+\|\del{k}\del{l}\Phi\|_{\Hs}+\|\del{0}\del{l}\Phi\|_{\Hs}  \nnb \\
 &\hspace{3.0cm}+\|\del{k}(E^{-2}\Phi-\mathring{E}^{-2}\mathring{\Phi})\|_{\Hs}+\|\del{0}\del{l}(E^{-2}\Phi-\mathring{E}^{-2}\mathring{\Phi})\|_{\Hs}\Bigr)  \nnb \\
	\leq & C\bigl(\|\mathbf{U}\|_{\Li((t,1],H^s)}\bigr)\Bigl( \epsilon\|\mathbf{U}\|_{H^s}    +\|\del{k}(E^{-2}\Phi-\mathring{E}^{-2}\mathring{\Phi})\|_{\Hs}+\|\del{0}\del{l}(E^{-2}\Phi-\mathring{E}^{-2}\mathring{\Phi})\|_{\Hs}
\Bigr),
	}
which holds for $T_3 < t \leq 1$, follows from  \eqref{E:CHI1}-\eqref{E:EOEXP}, the estimates \eqref{E:NORMVEST},  \eqref{Phiestimate} and \eqref{E:DTDSPHI}, the calculus inequalities, and the estimate
\als{
	\|\acute{S}\|_{\Hs}\lesssim \epsilon \la 1,\mathcal{S} \ra \lesssim \epsilon \|\mathcal{S}\|_{L^2}\leq \epsilon C(\|\mathbf{U}\|_{\Li((t,1],H^s)}) \|\mathbf{U}\|_{H^1} \leq \epsilon C(\|\mathbf{U}\|_{\Li((t,1],H^s)}) \|\mathbf{U}\|_{\Hs}.
	}
By \eqref{E:ZETAH1}, \eqref{zetaHringform}, \eqref{E:CHI1}, \eqref{E:DEFOFPHI}, \eqref{E:PHIRING}, \eqref{E:NORMVEST},
\eqref{E:ZREAL},
the invertibility of the Laplacian and the calculus inequalities,  we see also that
\al{DSSUBPHI}{
	\|\del{k}(E^{-2}\Phi-\mathring{E}^{-2}\mathring{\Phi})\|_{\Hs}&=\|E^{-2}\Phi-\mathring{E}^{-2}\mathring{\Phi} \|_{H^s} \lesssim \|e^{\zeta_H}\Pi e^{\delta\zeta}-e^{\mathring{\zeta}_H}\Pi e^{\delta\mathring{\zeta}}\|_{\Hsss} \nnb \\
	&\lesssim | e^{\zeta_H}-e^{\mathring{\zeta}_H} |\|\Pi (e^{\delta\zeta}-1)\|_{\Hsss}+\|e^{\mathring{\zeta}_H}\Pi( e^{\delta\zeta}-e^{\delta\mathring{\zeta}})\|_{\Hsss} \nnb \\
	&\leq  C\bigl(\|\delta\zeta\|_{\Li((t,1],H^s)},\|\delta\mathring{\zeta}\|_{\Li((t,1],H^s)}\bigr)\bigl(| \zeta_H- \mathring{\zeta}_H| \| e^{\delta\zeta}-1 \|_{\Hsss}
	+ \|\delta\zeta-\delta\mathring{\zeta}\|_{\Hs} \bigr) \nnb  \\
	& \leq  \epsilon  C\bigl(\|\delta\zeta\|_{\Li((t,1],H^s)},\|\delta\mathring{\zeta}\|_{\Li((t,1],H^s)}\bigr)\bigl(\epsilon  \|  \delta\zeta  \|_{\Hs}
	+   \|\mathbf{Z}\|_{\Hs}+ \|\mathbf{V}\|_{\Hs} \bigr)  \nnb \\
& \leq \epsilon  C\bigl(\|\mathbf{U}\|_{\Li((t,1],H^s)},\| \mathring{\mathbf{U}}\|_{\Li((t,1],H^s)}\bigr)\bigl(\| \mathbf{U}\|_{H^s}
+   \|\mathbf{Z}\|_{\Hs}+ \|\mathring{\mathbf{U}}\|_{H^s} \bigr)
	}
for $T_3 < t \leq 1$, while similar calculations using \eqref{E:PTZETAH}, \eqref{E:PTZETAH2} and \eqref{E:EOEXP} show that
\al{DTDSPHIEST}{
	&\|\del{0}\del{l}(E^{-2}\Phi-\mathring{E}^{-2}\mathring{\Phi})\|_{\Hs}=\|\del{0} (E^{-2}\Phi-\mathring{E}^{-2}\mathring{\Phi})\|_{H^s}\lesssim \|\del{t}(e^{\zeta_H}\Pi e^{\delta\zeta}-e^{\mathring{\zeta}_H}\Pi e^{\delta\mathring{\zeta}})\|_{\Hsss} \nnb \\
	& \hspace{1cm} \leq C\bigl(\|\delta\zeta\|_{\Li((t,1],H^s)},\|\delta\mathring{\zeta}\|_{\Li((t,1],H^s)}\bigr)\Bigl(\epsilon  \|  \delta\zeta \|_{\Hs}  +\epsilon^2\|\del{t}\delta\zeta\|_{\Hs} +\|\delta\mathring{\zeta}-\delta\zeta\|_{\Hs} \nnb \\
	& \hspace{1.5cm} +\|e^{\delta\zeta}\del{t}(\delta\zeta-\delta\mathring{\zeta})\|_{\Hsss}+\|\delta\zeta-\delta\mathring{\zeta}\|_{\Hs}\|\del{t}\delta\mathring{\zeta}\|_{\Hs}
	 \Bigr)  \nnb \\
	 & \hspace{1cm} \leq C\bigl(\|\delta\zeta\|_{\Li((t,1],H^s)},\|\delta\mathring{\zeta}\|_{\Li((t,1],H^s)}\bigr)\Bigl(\epsilon  \|  \delta\zeta \|_{\Hs}  +\epsilon^2\|\del{t}\delta\zeta\|_{\Hs} +\epsilon\|\mathbf{Z}\|_{\Hs}+\epsilon\|\mathbf{V}\|_{\Hs} \nnb \\
	 & \hspace{1.5cm} +\|e^{\delta\zeta}\del{t}(\delta\zeta-\delta\mathring{\zeta})\|_{\Hsss}+\epsilon(\|\mathbf{Z}\|_{\Hs}+ \|\mathbf{V}\|_{\Hs})\|\del{t}\delta\mathring{\zeta}\|_{\Hs}
	 \Bigr)
	}
for $T_3 < t \leq 1$.

Next, by \eqref{E:PTZETAH2}, it is easy to see that \eqref{CPeqn1} is equivalent to
\begin{equation*}
	\partial_t \delta \mathring{\zeta}+\sqrt{\frac{3}{\Lambda}}\bigl( \mathring{z}{}^j\partial_j \delta\mathring{\zeta} + \partial_j\mathring{z}{}^j\bigr)
	=0.
\end{equation*}
Using this, we derive the estimate
\begin{align}
	\|\del{t} \delta\mathring{\zeta}\|_{\Hs} \leq&
 C(\|\mathring{z}_j\|_{\Li((t,1],H^s)})\bigl(\|\delta\mathring{\zeta}\|_{H^s}+\|\mathring{z}_j\|_{H^s}\bigr),
\quad T_3 < t \leq 1,
\end{align}
while we see from \eqref{dtzeta} and \eqref{Phiestimate} that
\al{DTZSUBZ2}{
	\|\del{t} \delta \zeta\|_{\Hs} \leq & C(\|\mathbf{U}\|_{\Li((t,1],H^s)}, \|\del{k}\Phi\|_{\Li((t,1],H^s)})\bigl(\|\delta\zeta\|_{H^s}+\|z_j\|_{H^s}+\epsilon(\|\mathbf{U}\|_{\Hs}+\|\del{k}\Phi\|_{\Hs})\bigr) \nnb  \\
	\leq & C(\|\mathbf{U}\|_{\Li((t,1],H^s)}) \|\mathbf{U}\|_{H^s}
}
for $T_3 < t \leq 1$.
We also observe that
\al{111}{
	& \|e^{\delta\zeta} \mathring{z}^k \del{k}(\delta\zeta-\delta\mathring{\zeta})\|_{\Hsss}\leq \|\del{k}\bigl[e^{\delta\zeta}\mathring{z}^k (\delta\zeta-\delta\mathring{\zeta})\bigr]\|_{\Hsss}+ \|\del{k}\bigl[e^{\delta\zeta}\mathring{z}^k \bigr](\delta\zeta-\delta\mathring{\zeta})\|_{\Hsss} \nnb\\
	& \hspace{6cm} \leq C(\|\delta\zeta\|_{\Li((t,1],H^s)},\|\mathring{z}^k\|_{\Li((t,1],H^s)})\| \delta\zeta-\delta\mathring{\zeta} \|_{\Hs}
	}
and
\al{222}{
	& \|e^{\delta\zeta} \del{k}(z_m-\mathring{z}_m)\|_{\Hsss} %\leq \|\del{k}\bigl[e^{\delta\zeta} (z_m-\mathring{z}_m)\bigr]\|_{\Hsss} + \|\bigl(\del{k}e^{\delta\zeta}\bigr) (z_m-\mathring{z}_m)\|_{\Hsss}\nnb \\ & \hspace{7.5cm}
\leq C(\|\delta\zeta\|_{\Li((t,1],H^s)})\| z_m-\mathring{z}_m \|_{\Hs}
	}
hold for $T_3 < t \leq 1$. Furthermore,
by \eqref{E:GIJ}, \eqref{E:G0MU}, \eqref{E:V_0}, \eqref{E:VELOCITY}, \eqref{LIMITRHOZ} and \eqref{E:ZREAL}, we see
that
	\al{333}{
		\|z^k-\mathring{z}^k\|_{\Hs}%=\|2tu^{0k}\underline{\bar{v}^0}+\underline{\bar{g}^{kj}}z_j -\mathring{E}^{-2}\delta^{kj}\mathring{z}_j\|_{\Hs} %\nnb \\
		%\leq & C(\|\mathbf{U}\|_{\Li((t,1],H^s)} )\bigl(\|u^{0k}\|_{\Hs}+\|(\underline{\bar{g}^{kj}}-E^{-2}\delta^{kj})z_j\|_{\Hs}+\|(E^{-2}-\mathring{E}^{-2})z_j\|_{\Hs}+\|z_j-\mathring{z}_j\|_{\Hs}\bigr) \nnb \\
	%	\leq & C(\|\mathbf{U}\|_{\Li((t,1],H^s)} )\bigl(\|u^{0k}\|_{\Hs}+\epsilon \| z_j\|_{\Hs} +\|z_j-\mathring{z}_j\|_{\Hs}\bigr) \nnb \\
		\leq  \epsilon C(\|\mathbf{U}\|_{\Li((t,1],H^s)} )\bigl(\|\mathbf{Z}\|_{\Hs}+\|\mathbf{V}\|_{\Hs}+ \| z_j\|_{\Hs} \bigr)
		}
and, with the help of \eqref{E:EOEXP}, \eqref{E:NORMVEST}, \eqref{Phiestimate} and \eqref{E:111}-\eqref{E:333},
that
\al{DTZSUBZ3}{
	&\|e^{\delta\zeta}\partial_t (\delta \zeta- \delta \mathring{\zeta})\|_{\Hsss} \lesssim  \|e^{\delta\zeta}(z^k\del{k}\delta\zeta-\mathring{z}^k\del{k}\delta\mathring{\zeta})\|_{\Hsss}+\|e^{\delta\zeta}(E^{-2} \del{k}z_m-\mathring{E}^{-2} \del{k}\mathring{z}_m)\|_{\Hsss}+\epsilon\|e^{\delta\zeta}\mathcal{S}\|_{\Hsss} \nnb \\
	&\hspace{1.5cm} \lesssim \|e^{\delta\zeta}(z^k-\mathring{z}^k)\del{k}\delta\zeta\|_{\Hsss}+\|e^{\delta\zeta}\mathring{z}^k\del{k}(\delta\zeta-\delta\mathring{\zeta})\|_{\Hsss}+\|e^{\delta\zeta}(E^{-2}-\mathring{E}^{-2})\del{k}z_m\|_{\Hsss}\nnb \\
	&\hspace{2cm}  +\|e^{\delta\zeta} \del{k}(z_m-\mathring{z}_m)\|_{\Hsss} + \epsilon \|e^{\delta\zeta}\mathcal{S}\|_{\Hs} \nnb \\
%	&\leq  C(\|\delta\zeta\|_{\Li((t,1],H^s)})\|z^k-\mathring{z}^k\|_{\Hs}+C(\|\delta\zeta\|_{\Li((t,1],H^s)},\|\mathring{z}^k\|_{\Li((t,1],H^s)})\| \delta\zeta-\delta\mathring{\zeta} \|_{\Hs}  \nnb  \\
%	&\quad+\epsilon C(\|\delta\zeta\|_{\Li((t,1],H^s)}) \|\del{k}z_m\|_{\Hs} +C(\|\delta\zeta\|_{\Li((t,1],H^s)})\| z_m-\mathring{z}_m \|_{\Hs}\nnb  \\
%	&\quad+\epsilon C(\|\mathbf{U}\|_{\Li((t,1],H^s)}, \|\del{k}\Phi\|_{\Li((t,1],H^s)})(\|\mathbf{U}\|_{\Hs}+\|\del{k}\Phi\|_{\Hs}) \nnb \\
%	&\leq  \epsilon C(\|\mathring{z}^k\|_{\Li((t,1],H^s)}, \|\mathbf{U}\|_{\Li((t,1],H^s)},\|\del{k}\Phi\|_{\Li((t,1],H^s)})(\|z\|_{\Hs}+\|\mathbf{V}\|_{H^s}+\|\mathbf{U}\|_{H^s}+\|\del{k}\Phi\|_{\Hs}) \nnb\\
	&\hspace{2cm} \leq  \epsilon C(  \|\mathbf{U}\|_{\Li((t,1],H^s)},\|\mathring{\mathbf{U}}\|_{\Li((t,1],H^s)} )(\|\mathbf{Z}\|_{\Hs}+\|\mathring{\mathbf{U}}\|_{H^s}+\|\mathbf{U}\|_{H^s}),
}	
where both estimates hold for $T_3 < t \leq 1$. %and note %we have used the following estimates \eqref{E:111}-\eqref{E:333}:
We also observe that \eqref{E:NORMVEST} and \eqref{E:DTDSPHIEST}-\eqref{E:DTZSUBZ3} imply
\al{DT1}{
	\|\del{0}\del{l}(E^{-2}\Phi-\mathring{E}^{-2}\mathring{\Phi})\|_{\Hs}\leq \epsilon C(\|\mathbf{U}\|_{\Li((t,1],H^s)},\|\mathring{\mathbf{U}}\|_{\Li((t,1],H^s)}) \bigl(\|\mathbf{U}\|_{H^s}+\|\mathbf{Z}\|_{\Hs}+\|\mathring{\mathbf{U}}\|_{H^s}\bigr)
	}	
for $T_3 < t \leq 1$.
Gathering \eqref{E:FSUBF}, \eqref{E:DSSUBPHI} and \eqref{E:DT1} together, we obtain the estimate
\al{FSUBF2}{
	\|\mathbf{F}(\epsilon,t,\cdot)-\mathring{\mathbf{F}}(t,\cdot)\|_{\Hs}
	\leq   \epsilon C(\|\mathbf{U}\|_{\Li((t,1],H^s)},\|\mathring{\mathbf{U}}\|_{\Li((t,1],H^s)}) (\|\mathbf{U}\|_{H^s}+\|\mathbf{Z}\|_{\Hs}+\|\mathring{\mathbf{U}}\|_{H^s}),  \quad T_3 < t \leq 1.
}
The estimates \eqref{E:HSUBH} and \eqref{E:FSUBF2} show that source terms $\{H_1,\mathring{H}_1, F_1,\mathring{F}_1\}$,
as defined by \eqref{singdefB} and \eqref{singdefd}, and $z$, defined by \eqref{E:zreal}, verify the Lipschitz estimate \eqref{HFLip} from Theorem \ref{T:MAINMODELTHEOREM} for times $-1 \leq \hat{t} < -T_3$.

Having verified that all of the hypotheses of Theorem \ref{T:MAINMODELTHEOREM} are satisfied, we conclude, with
the help of Lemma \ref{L:INITIALTRANSFER}, that
 there exists a  constant $\sigma>0$, independent of $\epsilon \in (0,\epsilon_0)$, such that if the free initial data is
chosen so that
\begin{align*}  %\label{E:INITIAL}
\|\smfu^{ij}\|_{H^{s+1}}+\|\smfu^{ij}_0\|_{H^{s}}+\|\breve{\rho}_0\|_{H^s}+\|
\breve{\nu}^i\|_{H^s}  \leq  \sigma,
\end{align*}
 then the estimates
\begin{equation}
  \|\mathbf{U}\|_{L^\infty((T_3,1], H^s)}\leq C\sigma, \quad
  \|\mathring{\mathbf{U}}\|_{L^\infty((T_3,1], H^s)}\leq C\sigma  \AND \|\mathbf{U}-\mathring{\mathbf{U}}\|_{L^\infty((T_3,1], H^{s-1})}\leq \epsilon C \sigma \label{UUringbounds}
\end{equation}
hold
for some constant $C>0$, independent of $T_3\in (0,1)$ and  $\epsilon \in (0,\epsilon_0)$.
Furthermore, from the continuation criterion discussed in \S \ref{proof:loccont}, it is clear that
the bounds \eqref{UUringbounds} imply that the solutions $\mathbf{U}$ and $\mathring{\mathbf{U}}$ exist globally on $M=(0,1]\times \mathbb{T}^3$ and satisfy the estimates
\eqref{UUringbounds}  with $T_3=0$ and uniformly for $\epsilon\in (0,\epsilon_0)$.
In particular, this implies via the definitions \eqref{E:REALVAR} and
\eqref{Uringdef} of $\mathbf{U}$ and $\mathring{\mathbf{U}}$ that
\begin{gather*}
        \|\delta\zeta(t)-\delta\mathring{\zeta}(t)\|_{H^{s-1}}\leq \epsilon C \sigma, \quad \| z_j(t)-\mathring{z}_j(t)\|_{H^{s-1}}\leq \epsilon C \sigma, \\
        \|u^{\mu\nu}_0(t)\|_{H^{s-1}}\leq \epsilon C \sigma, \quad \|u^{\mu\nu}_k(t)-
\delta^\mu_0\delta^\nu_0\partial_k\mathring{\Phi}(t)\|_{H^{s-1}}\leq \epsilon C\sigma, \quad \|u^{\mu\nu}(t)\|_{H^{s-1}}\leq \epsilon C \sigma, \\
        \|u_0(t)\|_{H^{s-1}}\leq \epsilon C \sigma, \quad \|u_k(t)\|_{H^{s-1}}\leq \epsilon C \sigma \quad \text{and} \quad \|u(t)\|_{H^{s-1}}\leq \epsilon C \sigma
\end{gather*}
for $0 < t \leq 1$, while, from \eqref{E:V^0}, we see that
\begin{equation*}
  \left\|\underline{\bar{v}^0}(t)-\sqrt{\frac{\Lambda}{3}}\right\|_{H^{s-1}}\leq C\epsilon\sigma
\end{equation*}
holds for $0 < t \leq 1$. This concludes the proof of Theorem \ref{T:MAINTHEOREM}.

\section*{Acknowledgement}
\noindent This work was partially supported by the ARC grant FT120100045. Part of this work was completed during
a visit by the authors to the IHP as part of the Mathematical Relativity Program in 2015. We are grateful to the Institute
for its support and hospitality during our stay. We also thank the referee for their comments and
criticisms, which have served to improve the content and exposition of this article.

\appendix

\section{Calculus Inequalities} \label{A:INEQUALITIES}

%\textbf{Chao: fix the statement of the lemmas below so that there are correct and the notation is consistent, and add appropriate references to the literature}

We use the following Sobolev-Moser inequalities throughout this article. The proofs can be found in \cite{tay}.

\begin{theorem}{\emph{[Sobolev's inequality]}} \label{sobolev}
If $s \in \mathbb{Z}_{>n/2}$, then
\begin{equation*}
\norm{f}_{L^\infty} \lesssim \norm{f}_{H^s}
\end{equation*}
for all $f\in H^s(\mathbb{T}^n)$.
\end{theorem}

%\begin{lemma} \label{moserlemA}
%Suppose $s\in \Zbb_{\geq 0}$ and $f_i \in H^s(\mathbb{T}^n)\cap L^{\infty}(\mathbb{T}^n)$ for $1\leq i \leq l$. Then there is a constant $C$ depending %on $s$ and $l$ such that
%\begin{align*}
%\|D^{\beta_1}f_1\ldots D^{\beta_l}f_l\|_{L^2}\leq C\sum^l_{i=1}\biggl(\prod_{j\neq i}\|f_j\|_{L^{\infty}}\sum_{|\beta|=s}\| D^{\beta}f_i\|_{L^2}\biggr)
%\label{E:MOSERESTIMATE1}
%\end{align*}
%for all   multi-indices $\beta_1, \ldots, \beta_l$ satisfying $|\beta_1|+\ldots+|\beta_l|=s$.
%\end{lemma}

\begin{lemma}  \label{moserlemB}
Suppose $s\in \mathbb{Z}_{\geq 1}$, $l \in \mathbb{Z}_{\geq 2}$, $f_i \in L^{\infty}(\mathbb{T}^n)$ for $1\leq i \leq l$, $f_l \in H^s(\mathbb{T}^n)$, and $ D f_i \in H^{s-1}(\mathbb{T}^n)$ for $1 \leq i \leq l-1$. Then there exists a constant $C>0$, depending
on $s$ and $l$, such that
\begin{align*}
\|f_1\ldots f_l\|_{H^s}\leq C\biggl({\|f_l\|_{H^s}\prod_{i=1}^{l-1}\|f_i\|_{L^{\infty}}+\sum_{i=1}^{l-1}\| D f_i\|_{H^{s-1}}\prod_{i\neq j}\|f_j\|_{L^{\infty}}}\biggr).
%\label{E:MOSERESTIMATE2}
\end{align*}
\end{lemma}

\begin{lemma} \label{moserlemC}
Suppose $s\in \mathbb{Z}_{\geq 1}$, $f \in L^\infty(\Tbb^n,V)\cap H^s(\Tbb^n,V)\cap C^0(\Tbb^n,V)$, $W,U\subset V$ are open with $U$ bounded and $\overline{U}\subset W$,
$f(x)\in U$ for all  $x\in \Tbb^n$   and $F\in C^s(W)$.  Then there exists a constant $C>0$, depending on $s$, such that
\begin{equation*}
\|D^{\alpha}(F\circ f)\|_{L^2}\leq C \norm{DF}_{W^{s-1,\infty}(U)} \norm{f}_{L^\infty}^{s-1} \left(\sum_{|\beta|=s}\|D^{\beta}f\|_{L^2}\right)^{\frac{1}{2}}
\end{equation*}
for any multi-index $\alpha$ satisfying $|\alpha|=s$.

%Suppose $s\in \mathbb{Z}_{\geq 1}$,  $F \in C^{\infty}(I)$ for some open interval $I\subset \mathbb{R}$, $f \in H^s(\mathbb{T}^n)\bigcap L^{\infty}(\mathbb{T}^n)$, and $J\subset I$, where $J=[a,b]$ with $a$ and $b$ the essential infimum and supremum of $f$,
%respectively. Then there exists a constant $C>0$, depending on $s$, the supremum of $F$ and its derivatives up to order $s$ on $J$ and $\|f\|_{L^{\infty}}$, such that
%\begin{align*}
%\|D^{\alpha}(F\circ f)\|_{L^2}\leq C  \sum_{|\beta|=s}\|D^{\beta}f\|_{L^2}
%\label{E:COMPOSITIONESTIMATE}
%\end{align*}
%for any multi-index $\alpha$ satisfying $|\alpha|=s$.
\end{lemma}

\begin{lemma} \label{moserlemD}
If $s \in \mathbb{Z}_{\geq 1}$ and $|\alpha|\leq s$, then
\begin{align*}
\|D^{\alpha}(f g)-fD^{\alpha}g\|_{L^2}\lesssim \| Df\|_{H^{s-1}}\|g\|_{L^{\infty}}+\| D f\|_{L^{\infty}}\|g\|_{H^{s-1}} %\label{E:COMMUTATORESTIMATE}
\end{align*}
for all $f, g$ satisfying $Df,g \in L^\infty(\Tbb^n) \cap \Hs(\Tbb^n) $.
\end{lemma}

\section{Matrix relations}

\begin{lemma}\label{A:INVERSEOFA}
  Suppose
  \begin{align*}
    A=\begin{pmatrix}
      a & b\\
      b^\mathrm{T} &c
    \end{pmatrix}
  \end{align*}
  is an $(n+1)\times (n+1)$ symmetric matrix, where $a$ is an $1\times 1$ matrix, $b$ is an $1\times n$ matrix and c is an $n\times n$ symmetric matrix. Then
  \begin{align*}
    A^{-1}=\begin{pmatrix}
      a & b\\
      b^\mathrm{T} &c
    \end{pmatrix}^{-1}=\begin{pmatrix}
      \frac{1}{a}\Bigr[1+b\Bigl(c-\frac{1}{a}b^\mathrm{T}b\Bigr)^{-1}b^\mathrm{T}\Bigr] & -\frac{1}{a}b\Bigl(c-\frac{1}{a}b^\mathrm{T}b\Bigr)^{-1}\\
      -\Bigl(c-\frac{1}{a}b^\mathrm{T}b\Bigr)^{-1}\frac{1}{a}b^\mathrm{T} & \Bigl(c-\frac{1}{a}b^\mathrm{T}b\Bigr)^{-1}
    \end{pmatrix}
  \end{align*}
\end{lemma}
\begin{proof}
Follows from direct computation.
\end{proof}

We also recall the well-known Neumann series expansion.
\begin{lemma} \label{E:EXPANSIONOFINVERSE2}
  If $A$ and $B$ are $n\times n$ matrices with $A$ invertible,  then there exists an $\epsilon_0>0$ such that the map
\begin{equation*}
(-\epsilon_0,\epsilon_0) \ni \epsilon \longmapsto (A+\epsilon B)^{-1} \in \mathbb{M}_{n\times n}
\end{equation*}
is analytic and can be expanded as
  \begin{align*}
    (A+\epsilon B)^{-1}=A^{-1}+\sum_{n=1}^\infty (-1)^n\epsilon^n (A^{-1}B)^nA^{-1}, \quad  |\epsilon|< \epsilon_0.
  \end{align*}
\end{lemma}

\section{Analyticity}\label{A:ANALY}
We list some well-known properties of analytic maps that will be used throughout this article. We refer interested readers to \cite{hei} or \cite{oli3} for the proofs.
\begin{lemma}\label{L:MULTIPLICATION}
  Let $X$, $Y$ and $Z$ be Banach spaces with $U \subset X$ and $V \subset Y$ open.
  \begin{enumerate}
    \item If $L: X\rightarrow Y$ is a continuous linear map, then $L \in C^\omega(X,Y)$;
    \item If $B: X\times Y\rightarrow Z$ is a continuous bilinear map, then $B\in C^\omega(X\times Y, Z)$;
    \item If $f \in C^\omega(U, Y)$, $g\in C^\omega(V,Z)$ and $\text{ran}(f)\subset V$, then $g\circ f\in C^\omega(U, Z)$.
  \end{enumerate}
\end{lemma}

\begin{lemma} \label{PSERIES}
  Suppose $s \in \mathbb{Z}_{>n/2}$, $F\in C^\omega\bigl(B_R(\mathbb{R}^N), \mathbb{R}\bigr)$,  and that
  \begin{align*}
    F(y_1, \cdots, y_N)=F_0+\sum_{|\alpha|\geq 1}c_\alpha y_1^{\alpha_1}\cdots y_N^{\alpha_N}
  \end{align*}
  is the powerseries expansion for $F(y)$ about $0$. Then there exists a constant $C_s$ such that the map
  \begin{align*}
    \bigl(B_{R/C_s}\bigl(H^s(\mathbb{T}^n\bigr)\bigr)^N\ni(\psi_1, \psi_2, \cdots, \psi_N) \mapsto F(\psi_1, \psi_2, \cdots, \psi_N)\in
H^s(\mathbb{T}^n)
  \end{align*}
  is in $C^\omega\bigl(\bigl(B_{R/C_s}\bigl(H^s(\mathbb{T}^n\bigr)\bigr)^N, H^s(\mathbb{T}^n)\bigr)$, and
  \begin{align*}
    F(\psi_1, \cdots, \psi_N)=F_0+\sum_{|\alpha|\geq 1}c_\alpha\psi_1^{\alpha_1}\psi_2^{\alpha_2}\cdots \psi_N^{\alpha_N}
  \end{align*}
  for all $(\psi_1, \cdots, \psi_N)\in \bigl(B_{R/C_s}\bigl(H^s(\mathbb{T}^n\bigr)\bigr)^N$.
\end{lemma}

\section{Index of notation} \label{index}

\bigskip

\begin{longtable}{ll}
$\tilde{g}_{\mu\nu}$ & Physical spacetime metric; \S \ref{S:INTRO} \\
$\tilde{v}^\mu$ & Physical fluid four-velocity; \S \ref{S:INTRO} \\
$\bar{\rho}$ & Fluid proper energy density; \S \ref{S:INTRO} \\
$\bar{p} = \epsilon^2 K\bar{\rho}$ & Fluid pressure; \S \ref{S:INTRO} \\
$\epsilon = \frac{v_T}{c}$ & Newtonian limit parameter; \S \ref{S:INTRO} \\
$M_\epsilon=(0,1]\times \mathbb{T}^3_\epsilon$ & Relativistic spacetime manifold; \S \ref{S:INTRO}\\
$M = M_1$ & Newtonian spacetime manifold; \S \ref{S:INTRO}\\
$a(t)$ & FLRW scale factor; \S \ref{S:INTRO}, eqns. \eqref{FLRW.a} and \eqref{E:TPTA}\\
$\tilde{v}_H(t)$ & FLRW fluid four-velocity; \S \ref{S:INTRO}, eqn. \eqref{FLRW.b} \\
$\rho_H(t)$ & FLRW proper energy density; \S \ref{S:INTRO}, eqn. \eqref{FLRW.c} (see also \eqref{E:RHOHOM} and \eqref{E:ZETAH2})  \\
$(\bar{x}^\mu) = (t,\bar{x}^i)$ & Relativistic coordinates; \S \ref{S:INTRO}\\
$(x^\mu)=(t,x^i)$ & Newtonian coordinates; \S \ref{S:INTRO}, eqn. \eqref{Ncoordinates}\\
$\mathring{a}(t)$ & Newtonian limit of $a(t)$; \S \ref{S:INTRO}, eqn. \eqref{arhoringdef}\\
$\mathring{\rho}_H(t)$ & Newtonian limit of $\rho_H(t)$; \S \ref{S:INTRO}, eqn. \eqref{arhoringdef}  \\
$\underline{f}(t,x^i)$ & Evaluation in Newtonian coordinates; \S \ref{iandc}, eqn. \eqref{Neval}\\
$X^s_{\epsilon_0,r}(\mathbb{T}^3)$  & Free initial data function space; \S \ref{Functionspaces}\\
%$E^p( (0,\epsilon_0)\times (T_1,T_2)\times U, V)$ & remainder coefficient space; \S \ref{Functionspaces}\\
$\mathcal{S}(\epsilon,t,\xi)$, $\mathcal{T}(\epsilon,t,\xi)$, $\ldots$ & Remainder terms that are elements of
$E^0$, \S \ref{remainder}   \\
$\texttt{S}(\epsilon,t,\xi)$, $\texttt{T}(\epsilon,t,\xi)$, $\ldots$ & Remainder terms that are elements of
$E^1$, \S \ref{remainder}   \\
$\bar{g}^{\mu\nu}$ & Conformal metric; \S \ref{conformalEinsteinEuler}, eqn.
\eqref{E:CONFORMALTRANSF} \\
$\bar{v}^\mu$ & Conformal four-velocity; \S \ref{conformalEinsteinEuler}, eqn.
\eqref{E:CONFORMALVELOCITY}\\
$\Psi$ & Conformal factor; \S \ref{Conformalfactor}, eqn. \eqref{E:CONFORMALFACTOR}\\
$\bar{h}$ & Conformal FLRW metric; \S \ref{Conformalfactor}, eqn. \eqref{E:CONFORMALFLRW}\\
$E(t)$ & Modified scale factor; \S \ref{Conformalfactor}, eqn. \eqref{E:DEFE} (see also \eqref{E:EREP})\\
$\Omega(t)$ & Modified density;  \S \ref{Conformalfactor}, eqn. \eqref{E:OMEGADEF} (see also \eqref{E:OMEGAREP})\\
$\bar{\gamma}^0_{ij}$, $\bar{\gamma}^i_{j0}$ & Non-vanishing Christoffel symbols of $\bar{h}$;
 \S \ref{Conformalfactor}, eqn. \eqref{E:HOMCHRIS}\\
$\bar{\gamma}^\sigma$ & Contracted Christoffel symbols of $\bar{h}$;
 \S \ref{Conformalfactor}, eqn. \eqref{E:GAMMAFLRW}\\
$\bar{Z}^\mu$ & Wave gauge vector field;  \S \ref{Wavegauge}, eqn. \eqref{Zdef}\\
$\bar{X}^\mu$ & Contracted Christoffel symbols;  \S \ref{Wavegauge}, eqn. \eqref{E:XY}\\
$\bar{Y}^\mu$ & Gauge source vector field;  \S \ref{Wavegauge}, eqn. \eqref{Ydef}\\
$u^{\mu\nu}$, $u$ & Modified conformal metric variables;  \S \ref{vardefs}, eqns. \eqref{E:u.a}, \eqref{E:u.d} and \eqref{E:u.f}\\
$u^{\mu\nu}_\gamma$ & First order metric field variables;  \S \ref{vardefs}, eqns. \eqref{E:u.b}, \eqref{E:u.c}, \eqref{E:u.e}
and \eqref{E:u.g}\\
$z_i$ &  Modified lower conformal fluid 3-velocity;  \S \ref{vardefs}, eqn. \eqref{E:z.b}\\
$\zeta$ & Modified density;  \S \ref{vardefs}, eqn. \eqref{E:ZETA}\\
$\delta \zeta$ & Difference between $\zeta$ and $\zeta_H$; \S \ref{vardefs}, eqn. \eqref{E:DELZETA}\\
$\bar{\mathfrak{g}}^{ij}$ & Densitized conformal 3-metric; \S \ref{vardefs}, eqn. \eqref{E:GAMMA} \\
$\alpha$ & Cube root of conformal 3-metric determinant;  \S \ref{vardefs}, eqn. \eqref{E:GAMMA} \\
$\check{g}_{ij}$ & Inverse of the conformal 3-metric $\bar{g}^{ij}$;  \S \ref{vardefs}, eqn. \eqref{E:GAMMA} \\
$\bar{\mathfrak{q}}$ & Modified conformal 3-metric determinant; \S \ref{vardefs}, eqn. \eqref{E:q} \\
$\bar{\eta}$ & Background Minkowski metric; \S \ref{vardefs}, eqn. \eqref{etabardef}\\
$\zeta_H(t)$ & FLRW modified density; \S \ref{vardefs}, eqns. \eqref{E:ZETAH1} and \eqref{E:ZETAH3}\\
$C_0$ & FLRW constant; \S \ref{vardefs}, eqn. \eqref{C0def}\\
$\mathring{\zeta}_{H}(t)$ & Newtonian limit of $\zeta_H(t)$, \S \ref{vardefs}, eqns. \eqref{zetaHringdef}
and \eqref{zetaHringform}  (see also \eqref{deltazetaringH})\\
$z^i$ &  Modified upper conformal fluid 3-velocity;  \S \ref{vardefs}, eqn. \eqref{E:z.a}\\
$\mathring{\rho}$ & Newtonian fluid density; \S \ref{CPEequations}\\
$\mathring{z}^j$ & Newtonian fluid 3-velocity; \S \ref{CPEequations}\\
$\mathring{\Phi}$ & Newtonian potential; \S \ref{CPEequations}\\
$\Pi$ & Projection operator; \S \ref{CPEequations}, eqn. \eqref{Pidef}\\
$\mathring{E}(t)$ & Newtonian limit of $E(t)$; \S \ref{CPEequations}, eqn. \eqref{Eringform}\\
$\mathring{\Omega}(t)$ & Newtonian limit of $\Omega(t)$; \S \ref{CPEequations}, eqn. \eqref{Oringdef}\\
$\mathring{\zeta}$ & Modified Newtonian fluid density; \S \ref{CPEequations}\\
$\rho$ & Fluid proper energy density in Newtonan coordinates; \S \ref{epexpansions}, eqn. \eqref{E:ZETA2}\\
$\delta\rho$ & Difference between $\rho$ and $\rho_H$; \S \ref{epexpansions}, eqn. \eqref{E:DELRHO}\\
$w^{0\mu}_k$ & Shifted first order gravitational variable; \S \ref{Npotsub}, eqn. \eqref{E:WPHI}\\
$\Phi$ & Gravitational potential; \S \ref{Npotsub}, eqn. \ref{E:DEFOFPHI}\\
$\phi$ & Renormalized spatially average density; \S \ref{Npotsub}, eqn. \ref{E:DEFOFPHISMALL}\\
$\mathbf{U}_1$ & Gravitational field vector; \S \ref{completeevolution}, eqn. \eqref{E:U1}\\
$\mathbf{U}$ & Combined gravitational and matter field vector; \S \ref{completeevolution}, eqn. \eqref{E:REALVAR}\\
$\mathbf{U}_2$ & Matter field vector; \S \ref{completeevolution}, eqn. \eqref{E:REALVAR1}\\
$\delta \mathring{\zeta}$ & Difference between $\mathring{\zeta}$ and $\mathring{\zeta}_H$;
\S \ref{PEcont},  eqn. \eqref{deltazetaringdef}\\
$\vertiii{\cdot}_{a,H^k}$, $\vertiii{\cdot}_{H^k}$,
$\|\cdot\|_{M^\infty_{\mathbb{P}_a,k}([T_0,T)\times \mathbb{T}^n)}$ & Energy norms; \S \ref{S:MODELuni}, Definition \ref{energynorms}\\
$\mathscr{Q}(\xi)$, $\mathscr{R}(\xi)$, $\mathscr{S}(\xi)$, $\ldots$ & Analytic remainder terms; \S \ref{S:INITIALIZATION} \\
\end{longtable}

\end{document}